\documentclass[preprint]{elsarticle}
\usepackage{graphicx} % Required for inserting images
\usepackage[margin=1in]{geometry}
\usepackage{amsthm,amsmath,amssymb}
\usepackage{subcaption}
\usepackage{enumerate}
\usepackage{comment}
\usepackage{tikz}
\usepackage{titlesec}

\titleformat{\section}
{\normalfont\fontfamily{phv}\fontsize{12}{17}\bfseries\itshape}{\thesection}{1em}{}

\titleformat{\subsection}
{\normalfont\fontfamily{phv}\fontsize{10}{17}\bfseries\itshape}{\thesubsection}{1em}{}

\titleformat{\subsubsection}
{\normalfont\fontfamily{phv}\fontsize{10}{17}\bfseries\itshape}{\thesubsubsection}{1em}{}

\usetikzlibrary{positioning,shapes.geometric,arrows.meta,calc}

\newcommand{\notimplies}{%
  \mathrel{{\ooalign{\hidewidth$\not\phantom{=}$\hidewidth\cr$\implies$}}}}

\newdefinition{note}{Note}
\newdefinition{remark}{Remark}
\newdefinition{definition}{Definition}
\newdefinition{example}{Example}
\newtheorem{theorem}{Theorem}
\newtheorem{observation}{Observation}
\newtheorem{lemma}{Lemma}
\newtheorem{proposition}{Proposition}
\newtheorem{corollary}{Corollary}
\newtheorem{conjecture}{Conjecture}
\newtheorem{op}{Open Problem}

\title{Graphs with single interval Cayley configuration spaces in 3-dimensions\tnoteref{t1}}

\author[1]{William Sims}
\ead{w.sims@ufl.edu}

\author[1,2]{Meera Sitharam\corref{cor1}}

\cortext[cor1]{Corresponding author}
\affiliation[1]{
organization={CISE Department, University of Florida},
country={USA}
}
\affiliation[2]{
organization={Department of Mathematics (Affiliate), University of Florida},
country={USA}
}

\date{}

\begin{document}

\begin{abstract}
 We prove a conjectured graph theoretic characterization of a geometric property of 3 dimensional linkages  posed  15 years ago by Sitharam and Gao, motivated by their equivalent characterization for $d\le 2$ that does not generalize to $d\ge 3$.   A \emph{linkage} $(G,\ell)$ contains a finite simple undirected graph $G$ and a map $\ell$ that assigns squared Euclidean lengths to the edges of $G$.  
A \emph{$d$-realization} of $(G,\ell)$ is an assignment of points in $\mathbb{R}^d$ to the vertices of $G$ for which pairwise squared distances between points agree with $\ell$.  
 For any positive integer $d \leq 3$, we characterize pairs $(G,f)$, where $f$ is a nonedge of $G$, such that, for any  linkage $(G,\ell)$,  the lengths attained by $f$ form a single interval  - over the (typically a disconnected set of) $d$-realizations of $(G,\ell)$.   
 
Although related to the minor closed class of $d$-flattenable graphs, the class of pairs $(G,f)$ with the above property is not closed under edge deletions, has no obvious well quasi-ordering, and there are infinitely many minimal graph-nonedge pairs - with respect to edge contractions - in the complement class.  
Our characterization overcomes these obstacles, is based on the forbidden minors for $d$-flattenability for $d \leq 3$, and contributes to the theory of Cayley configurations with many applications.  Helper results and corollaries provide new tools for reasoning about configuration spaces and completions of partial 3-tree linkages,  (non)convexity of Euclidean measurement sets in $3$-dimensions, their projections, fibers and sections.  
Generalizations to higher dimensions and efficient algorithmic characterizations are conjectured. 
\end{abstract}


\maketitle

\section{Introduction}
\label{sec:intro}

A \emph{linkage} $(G,\ell)$ is a pair containing a (finite simple undirected) graph $G$ and a \emph{squared edge-length map} $\ell$ that sends each edge in the edge set $E(G)$ to a nonnegative real number.  
For any dimension $d \geq 0$, a \emph{(Euclidean) $d$-realization} $p$ of $(G,\ell)$ is a map that sends each vertex in the vertex set $V(G)$ to a point in $\mathbb{R}^d$ such that $\ell(uv) = \|p(u) - p(v)\|^2$ for each edge $uv \in E(G)$.  
We omit ``$d$'' from this terminology when it is clear from context.  
The \emph{configuration space (CS)} $\mathcal{C}^d(G,\ell)$ is the set of all $d$-realizations of $(G,\ell)$.  
 Characterizing and sampling $\mathcal{C}^d(G,\ell)$ are difficult and well-studied problems \cite{fudos1997graph,gao2009characterizing,joan2003transforming,sitharam2006solution,sitharam2018handbook,van2005constructive,zhang2006well} with applications in many domains \cite{Prabhu2020JCIM,sacks2010configuration,ying1995robot}.  

\begin{figure}[htb]
   \centering
   \begin{subfigure}[t]{0.28\linewidth}
       \centering
       \includegraphics[width=0.9\linewidth]{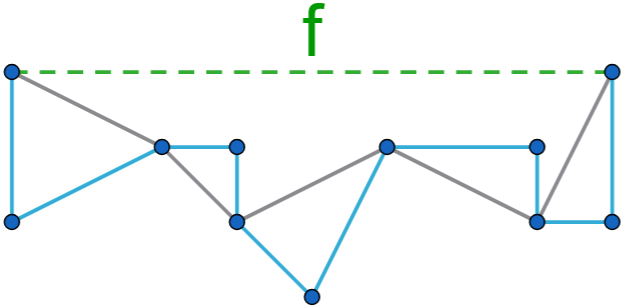}
       \caption{}
       \label{fig:k3-e}
   \end{subfigure}
   \begin{subfigure}[t]{0.28\linewidth}
       \centering
       \includegraphics[width=0.9\linewidth]{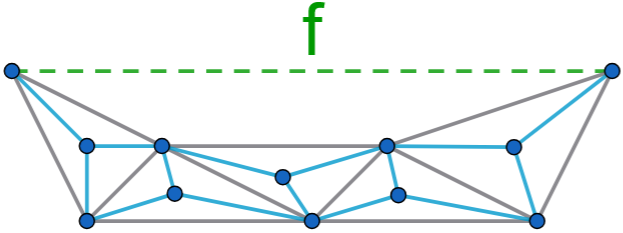}
       \caption{}
       \label{fig:k3-e_realization}
   \end{subfigure}
   \begin{subfigure}[t]{0.28\linewidth}
       \centering
       \includegraphics[width=0.9\linewidth]{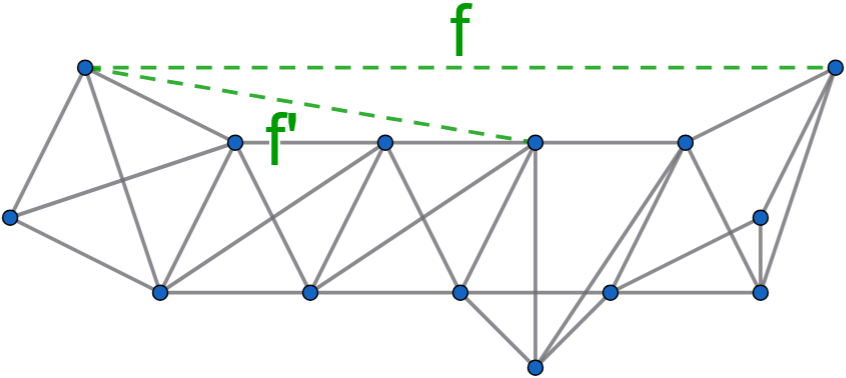}
       \caption{}
       \label{fig:k4-e}
   \end{subfigure}
   % \begin{subfigure}[t]{0.32\linewidth}
   %     \centering
   %     \includegraphics[width=0.8\linewidth]{linkage.png}
   %     \caption{}
   %     \label{fig:k3-e}
   % \end{subfigure}
   % \begin{subfigure}[t]{0.32\linewidth}
   %     \centering
   %     \includegraphics[width=\linewidth]{2-realization.png}
   %     \caption{}
   %     \label{fig:k3-e_realization}
   % \end{subfigure}
   % \begin{subfigure}[t]{0.24\linewidth}
   %     \centering
   %     \includegraphics[width=0.5\linewidth]{k4-e.png}
   %     \caption{}
   %     \label{fig:k4-e}
   % \end{subfigure}
   \begin{subfigure}[t]{0.14\linewidth}
       \centering
       \includegraphics[width=0.5\linewidth]{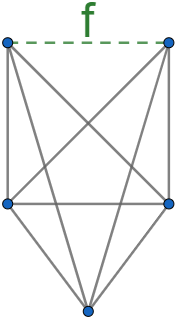}
       \caption{}
       \label{fig:k5-e}
   \end{subfigure}
   \caption{Graph-nonedge pairs.  
   See Example \ref{ex:1} and the discussion in this section, Example \ref{ex:3} and the discussion in Section \ref{sec:results}, and the discussions in Sections \ref{sec:proof_obstacles} and \ref{sec:forward}.  }
   \label{fig:linkage_obstacles}
\end{figure}

\begin{figure}[htb]
    \centering
    \includegraphics[width=0.4\linewidth]{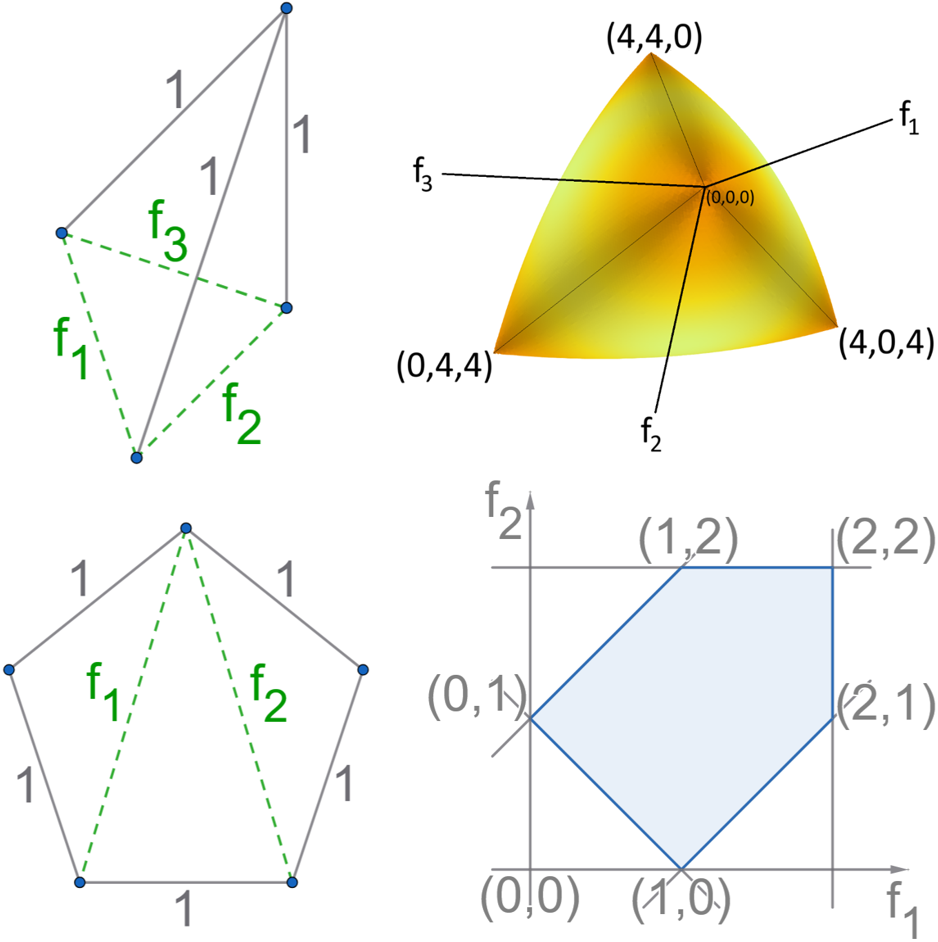}
    \caption{Left: linkages $(G,\ell)$ and sets $F$ of nonedges (dashed line-segments).  
    Top-right: the CCS $\Omega^3_F(G,\ell)$ for the top-left linkage.  
    Top-right recreated from \cite{vinzant2014spectrahedron}.  
    Bottom-right: the CCS $\Omega^2_F(G,\sqrt\ell)$ for the bottom-left linkage.  
    Bottom recreated from \cite{sitharam2010characterizing}.  See Example \ref{ex:1} and the discussion in this section, as well as Example \ref{ex:3} in Section \ref{sec:results}.  }
    \label{fig:2_and_3-convexity_ex}
\end{figure}

To tackle these problems, the paper \cite{sitharam2010characterizing}  began characterizing graphs $G$ for which, for any $\ell$,  the CS $\mathcal{C}^d(G,\ell)$, while typically disconnected,  has a (branched) covering map to a base space that is  convex. 
The notable feature is a graph theoretic characterization of a geometric property of linkages.
For such graphs, the difficulty of sampling the CS in applications reduces to a much easier problem of sampling a convex base ``configuration" space and inverting the covering map.  
More precisely, let $F = \{u_1v_1,\dots,u_mv_m\}$ be an indexed set of \emph{nonedges} of $G$, i.e., pairs of distinct vertices not in $E(G)$.  
The \emph{Cayley map} $\phi_F$ sends a realization $p$ in $\mathcal{C}^d(G,\ell)$ to the vector $(\|p(u_1)-p(v_1)\|^2,\dots,\|p(u_m)-p(v_m)\|^2)$ of squared lengths attained by $F$ under $p$.  
The \emph{Cayley configuration space (CCS)} $\Omega^d_F(G,\ell)$ is the image of $\mathcal{C}^d(G,\ell)$ under $\phi_F$.  
Note that the preimage of any point in $\Omega^d_F(G,\ell)$ under $\phi_F$ is a non-empty subset of $\mathcal{C}^d(G,\ell)$.  
See Figure \ref{fig:2_and_3-convexity_ex} for examples.  
The pair $(G,F)$  is \emph{$d$-convex (resp. $d$-Cayley-connected)} if $\Omega^d_F(G,\ell)$ is convex (resp. connected) for any squared edge-length map $\ell$.  

\begin{note}
    In this paper, all sets are closed and compact, all maps are closed, and ``connected'' means  ``path-connected'' unless otherwise specified.  
    Also, $(G,F)$ is trivially $d$-convex and $d$-Cayley-connected if $d = 0$,  hence we assume that $d$ is an integer greater than $0$, unless stated otherwise.  
\end{note}

\iffalse
If a graph-nonedge-set pair $(G,F)$ is $d$-convex, then the CCS $\Omega^d_F(G,\ell)$ can be efficiently traversed and sampled using well-known methods.  
Additionally, if $G$ belongs to certain large and natural classes of graphs, then there are efficiently computable maps from $\Omega^d_F(G,\ell)$ to $\mathcal{C}^d(G,\ell)$ \cite{SitharamWang2014Beast}.  
\fi
Characterizations of $d$-convexity for special classes of graphs have led to improvements in molecular and particle assembly modeling and kinematic mechanism analysis and design, respectively, via the opensource software EASAL \cite{Ozkan2018ACMTOMS,Prabhu2020JCIM} and CayMos \cite{SitharamWang2014Beast,Wang2014cayley}, and to efficient algorithms for distance constraint graph realization \cite{Baker2015}.  

We study the following property as a first step towards completely characterizing $d$-convex pairs $(G,F)$.  
If $F$ contains exactly one nonedge $f$, then we call $(G,f)$ a \emph{graph-nonedge pair}.  
\begin{definition}[$d$-SIP]
    \label{def:d-sip}
    For any dimension $d \geq 1$, a graph-nonedge pair has the \emph{$d$-single interval property (SIP)} if it is $d$-convex (equivalently, $d$-Cayley-connected).
\end{definition}

% In the above definitions, Sitharam and Gao \cite{sitharam2010characterizing} defined $F$ to be a set of nonedges.
% We allow $F$ to also contain edges so that $d$-convexity and $d$-Cayley-connectedness are closed under edge contractions (see Section \ref{sec:results}).  
% However, note that that this generalization is purely cosmetic: letting $F'$ be the subset of all nonedges in $F$, $(G,F')$ is $d$-convex (resp. $d$-Cayley-connected) if and only if $(G,F)$ is $d$-convex (resp. $d$-Cayley-connected).  
% Hence, we assume that $F$ is a set of nonedges unless stated otherwise, and we refer to $(G,F)$ as a \emph{graph-nonedge-set pair}, or as a \emph{graph-nonedge pair} if $F$ contains exactly one nonedge.  

The main result of this paper is a characterization for any $d \leq 3$ of  graph-nonedge pairs   $(G,f)$  with $d$-SIP,  proving a conjecture posed in \cite{sitharam2010characterizing}, which  proved an equivalent characterization for $d \leq 2$ that does not generalize to $d \geq 3$.
\begin{example}
\label{ex:1}
Refer to the graph-nonedge pairs $(G,f)$ in Figure \ref{fig:linkage_obstacles}.  
If $G$ is the graph in Figure \ref{fig:k3-e} without the blue edges, then there exists an $\ell$ for which both the CS $\mathcal{C}^1(G,\ell)$ (modulo translations) and CCS $\Omega^1_f(G,\ell)$ are finite non-singleton sets and hence $(G,f)$ does not have 1-SIP. 
On the other hand, the CS $\mathcal{C}^2(G,\ell)$ and CCS $\Omega^2_f(G,\ell)$ are clearly connected and hence $(G,f)$  has 2-SIP.
Adding in the blue edges with almost  any lengths  causes $\mathcal{C}^2(G,\ell)$ to become disconnected while, for any lengths, all the connected components of $\mathcal{C}^2(G,\ell)$  induce the same  interval   $\Omega^2_f(G,\ell)$,  and hence $(G,f)$ has 2-SIP. 

If $G$ is the graph in Figure \ref{fig:k3-e_realization} without the blue edges, then there exists an $\ell$ for which both the CS $\mathcal{C}^2(G,\ell)$ (modulo translations and rotations) and CCS $\Omega^2_f(G,\ell)$ are finite non-singleton sets and hence $(G,f)$ does not have 2-SIP. 
On the other hand, the CS $\mathcal{C}^3(G,\ell)$ and CCS $\Omega^3_f(G,\ell)$ are clearly connected and hence $(G,f)$  has 3-SIP.
Adding in the blue edges with almost  any lengths  causes $\mathcal{C}^3(G,\ell)$ to become disconnected (modulo   translations and rotations) while, for any lengths, all the connected components of $\mathcal{C}^3(G,\ell)$  induce the same  interval   $\Omega^3_f(G,\ell)$,  and hence $(G,f)$ has 3-SIP.  

For similar reasons, if $G$ is the graph in Figure \ref{fig:k5-e}, $(G,f)$ has the $d$-SIP for any $d \geq 4$ but not for any $d < 4$. 

The  linkage in Figure \ref{fig:k4-e} also has the property that $\mathcal{C}^3(G,\ell)$ is disconnected.  
  If we fix one of the two reflections of each of the 3 tetrahedra,  $3$-SIP ($3$-convexity) of $(G,f')$ is intuitively clear when restricted to this choice of reflections because the configuration space is connected. Although the  different reflections could result in disconnected components of the CS $\mathcal{C}^3(G,\ell)$, with some effort it can be shown that  all these components have the same attainable lengths for the nonedge $f'$,  which is sufficient  (although not necessary) to prove $\Omega^3_{f'}(G,\ell)$ is connected  and hence $(G,f')$  has 3-SIP. Showing  a weaker version of this property  in general is  sufficient to show $3$-SIP for $(G,f)$ e.g. in Figure \ref{fig:k4-e} (see also Figure \ref{fig:converse_4}). This is  a key  property established in   the converse direction of   Theorem \ref{thm:3-sip_characterization} in Section \ref{sec:results} - our main result.

\iffalse
For the linkage $(G,\ell)$ in Figure \ref{fig:k3-e}, the CS $\mathcal{C}^1(G,\ell)$ contains exactly two realizations up to translation and overall reflection, and the CCS $\Omega^1_f(G,\ell)$ is the set $\{1,9\}$.  
Hence, $(G,f)$ does not have the $1$-SIP.  
On the other hand, for any $\ell$ and any dimension $d \geq 2$, the triangle-inequality applied to $(G,\ell)$ ensures that the CCS $\Omega^d_f(G,\ell)$ is a single interval, and hence $(G,f)$ has the $d$-SIP. 
\fi
% , illustrated for $d=2$ in Figure \ref{fig:k3-e_realization}.  

Next, refer to the  pairs $(G,F)$ in Figure \ref{fig:2_and_3-convexity_ex}.  
For the linkage $(G,\ell)$ in the bottom-left image, the  CCS $\Omega^2_F(G,\sqrt\ell)$ shown in the bottom-right image is the intersection of half-spaces determined by triangle-inequalities on $\sqrt\ell$, and  it is not hard to see that $\Omega^2_F(G,\ell)$ is also convex.  
Since this is true for any $\ell$, $(G,F)$ is $d$-convex  for any $d\ge 2$.  
On the other hand, by choosing $\ell$ so that one edge has length $0$ and the rest have length $1$, we get that $\Omega^1_f(G,\ell)$ has cardinality two, and so $(G,F)$ is not $1$-convex.  
In the top-left image, a similar argument using tetrahedral-inequalities shows that $(G,F)$ is $d$-convex for any $d \geq 3$ but not for any $d < 3$.  
The top-right image shows the convex CCS $\Omega^3_F(G,\ell)$ for the linkage $(G,\ell)$ in the top-left image.  
The CS $\mathcal{C}^3(G,\ell)$ is clearly connected for any $\ell$, and so $\Omega^3_F(G,\ell)$ is connected for any $\ell$ since $\phi_F$ is continuous.  
Therefore, $(G,F)$ is $2$-Cayley-connected.  
However, an argument similar to those above shows that $(G,F)$ is not $1$-Cayley-connected.  

\begin{figure}[htb]
    \centering
    \includegraphics[width=0.4\textwidth]{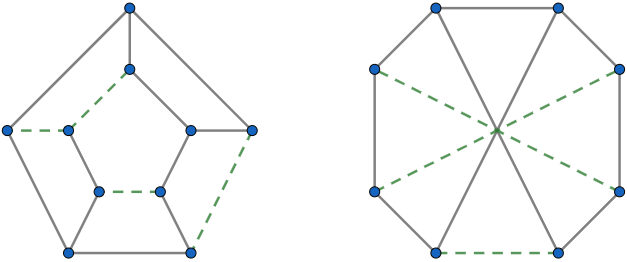}
    \caption{The graphs $C_5 \times C_2$ (left) and $V_8$ (right) along with nonedge sets (shown as dashed line-segments) that together are $d$-convex pairs for any dimension $d \geq 3$ but not for any $d < 3$.  
    See Example \ref{ex:1} in this section, Example \ref{ex:3} and the discussion in Section \ref{sec:results}, and Example \ref{ex:minimal} in Section \ref{sec:forward}.  }
    \label{fig:d-convexity_ex}
\end{figure}

Finally, the  pairs $(G,F)$ in Figure \ref{fig:d-convexity_ex} are $d$-convex for any dimension $d \geq 3$ but not for any $d < 3$.  
The case $d \geq 3$ can be verified for the first two pairs using Theorem \ref{thm:sitharam_willoughby} in Section \ref{sec:related}. 
The case $d<3$  can be verified intuitively, and using Theorems \ref{thm:2-convexity} and \ref{thm:2-SIP} in Section \ref{sec:related}.  
\qed
\end{example}

\subsection{Organization}
Previous results on $d$-conxevity, $d$-Cayley-connectedness, and the $d$-SIP are collected in Section \ref{sec:related}.  
We state our contributions in Section \ref{sec:results}, including Theorem \ref{thm:3-sip_characterization}, our main result, which gives a characterization of graph-nonedge pairs that have the $3$-SIP; and several frequently used and independently interesting tool lemmas, byproducts, and corollaries, which are proved in Sections \ref{sec:tool_lemmas} and \ref{sec:tool-thms}.  
The majority of this paper is spent proving Theorem \ref{thm:3-sip_characterization}.  
The primary proof obstacles are detailed in Section \ref{sec:proof_obstacles}.  
Since our proof techniques involve very technical work on graph minors,   notation and terminology to facilitate this is introduced in Section \ref{sec:minor-notation}.  
The proof is given in Sections \ref{sec:reverse} (converse direction) and \ref{sec:forward} (forward direction).  
Finally, we state several open problems in Section \ref{sec:conclusions}, including a conjectured characterization of graph-nonedge pairs that have the $d$-SIP for any $d \geq 1$.
The appendix collects proofs of independently interesting results about clique separators, partial 3-trees, and related covering maps, and long and tedious but not particularly illuminating proofs for helper lemmas.

\subsection{Related Works}
\label{sec:related}

Theorems \ref{thm:2-convexity} and \ref{thm:2-SIP}, below, are complete characterizations of $d$-convexity and the $d$-SIP for any $d \leq 2$, given in \cite{sitharam2010characterizing} and restated using our terminology.  
The proof of the former strongly relies on the latter.  
See Example \ref{ex:1}, in Section \ref{sec:intro}, and Example \ref{ex:2}, below, for applications of these theorems.  
We require the following concepts to state these theorems.  
Given a graph $G$ and a set of its edges and nonedges $F$, $G \cup F$ is the graph obtained from $G$ by adding $F$ to $E(G)$.  
A subgraph $H$ of $G$ is \emph{induced} if, for any $u,v \in V(H)$, $uv$ is an edge of $H$ if and only if it is an edge of $G$.  
Furthermore, $H$ is \emph{induced by} a set $U \subseteq V(G)$ if $H$ is induced and $V(H) = U$.  
The subgraph induced by $V(G) \setminus U$ is written as $G \setminus U$.  
If either $F$ or $U$ contains exactly one element $f$ or $u$, respectively, then we replace $F$ with $f$ and $U$ with $u$ in this notation.  
$U$ is a \emph{separator} of $G$ if $G$ is connected and $G \setminus U$ is disconnected, a \emph{clique} of $G$ if its induces a clique subgraph, and a \emph{clique separator} of $G$ if it is both a clique and a separator.  
An \emph{atom} is a graph with no clique separator.  
An \emph{atom of $G$} is a vertex-maximal induced subgraph with no clique separator.  
See Figure \ref{fig:clique_sep_atoms} for examples.  

\begin{figure}[htb]
   \centering
   \begin{subfigure}[t]{0.32\linewidth}
       \centering
       \includegraphics[width=0.75\linewidth]{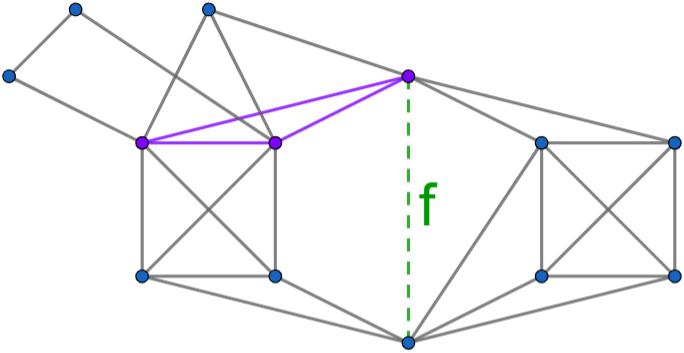}
       \caption{}
       \label{fig:clique_sep}
   \end{subfigure}
   \begin{subfigure}[t]{0.33\linewidth}
       \centering
       \includegraphics[width=0.9\linewidth]{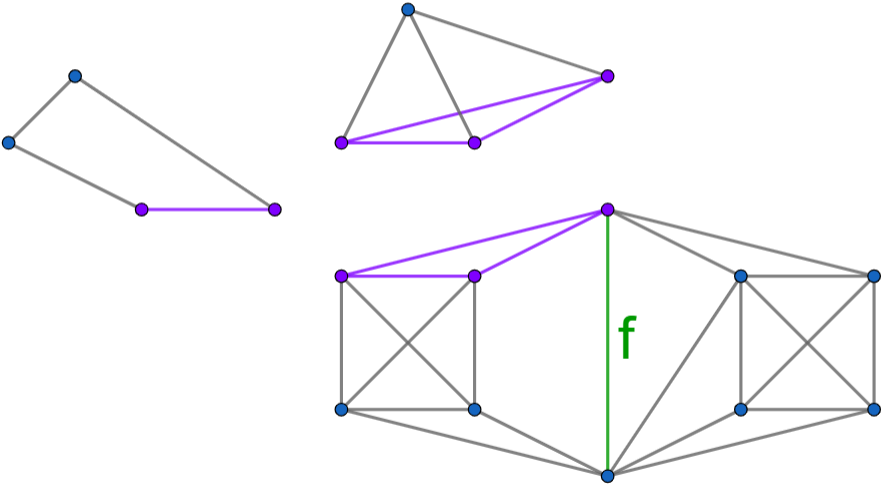}
       \caption{}
       \label{fig:clique_sep_components}
   \end{subfigure}
   \begin{subfigure}[t]{0.33\linewidth}
       \centering
       \includegraphics[width=\linewidth]{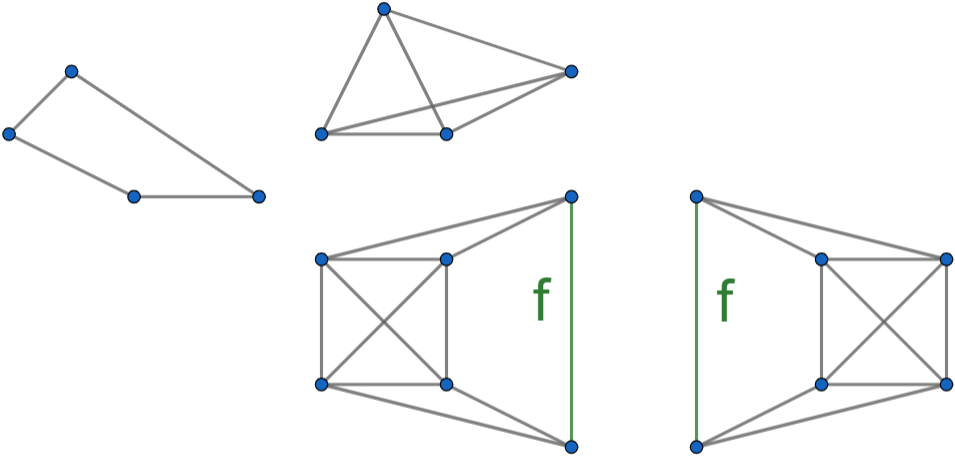}
       \caption{}
       \label{fig:atoms}
   \end{subfigure}
   \caption{(a) a graph-nonedge pair $(G,f)$ with $G \cup f$ having a clique separator $C$, shown in purple, (b) the $C$-components of $G \cup f$, and (c) the atoms of $G \cup f$.  
   See the discussion in this section, Example \ref{ex:3} in Section \ref{sec:results}, and Example \ref{ex:minimal} in Section \ref{sec:forward}.  }
   \label{fig:clique_sep_atoms}
\end{figure}

A graph $M$ is a \emph{minor} of $G$ if it is isomorphic to some graph obtained from $G$ via some sequence of edge deletions and (edge) contractions - i.e.,  replacing two adjacent vertices by a single vertex whose neighborhood is the union of their neighborhoods.  

\begin{theorem}[\cite{sitharam2010characterizing} Theorem 5.11]
    \label{thm:2-convexity}
    For any dimension $d \leq 2$,  and graph $G$  with nonedge set $F$, the pair $(G,F)$ is $d$-convex if and only if each atom of $G \cup F$ that contains a non-empty subset of $F$ has no $K_{d+2}$ minor.  
\end{theorem}

\begin{theorem}[\cite{sitharam2010characterizing} Theorem 5.1]
    \label{thm:2-SIP}
    For any dimension $d \leq 2$, a graph-nonedge pair $(G,f)$ has the $d$-SIP if and only if each atom of $G \cup f$ that contains $f$ has no $K_{d+2}$ minor.  
\end{theorem}

The above theorem does not directly generalize to $d\ge 3$ as discussed in more detail in Example \ref{ex:3} in Section \ref{sec:results}   (see Figure \ref{fig:k5_k222_e_contracted} for a counterexample). Hence the following modification was conjectured in \cite{sitharam2010characterizing} for $d = 3$, 15 years ago, and is the main result in this paper.  

\begin{conjecture}[\cite{sitharam2010characterizing}, Conjecture 6.4]
    \label{conj:sitharam-gao}
    A graph-nonedge pair $(G,f)$ has the $3$-SIP if and only if $f$ must be contracted in any $K_5$ or $K_{2,2,2}$ minor of any atom of $G \cup f$ that contains $f$.  
\end{conjecture}

The above minors are central to another property of graphs called \emph{$d$-flattenability}, introduced in \cite{belk2007realizability1,belk2007realizability2}.  
A graph $G$ is \emph{$d$-flattenable} if, for any squared edge-length map $\ell$ and any dimension $n \geq d$, the CS $\mathcal{C}^d(G,\ell)$ is non-empty whenever the CS $\mathcal{C}^n(G,\ell)$ is non-empty.  
Intuitively, this means that any high-dimensional realization of any linkage $(G,\ell)$ can be flattened to live in $d$-dimensions while preserving squared edge-lengths.  
For example, $K_3$ is clearly $2$-flattenable but not $1$-flattenable, since choosing a squared edge-length map $\ell$ that assigns $1$ to each edge of $K_3$ ensures the CS $\mathcal{C}^1(K_3,\ell)$ is empty even though the CS $\mathcal{C}^2(K_3,\ell)$ is not.  

\begin{figure}[htb]
    \centering
    \includegraphics[width=0.2\textwidth]{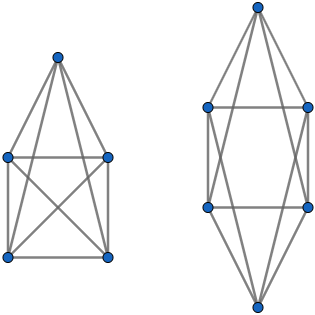}
    \caption{The $3$-flattenability forbidden minors: $K_5$ (left) and $K_{2,2,2}$ (right).  
    See the discussions in this section and Section \ref{sec:results}.  }
    \label{fig:fbms}
\end{figure}

A property of graphs is \emph{minor closed} if it is preserved under edge deletions and contractions.  
More precisely: given a graph $G$ with the property, any graph obtained by deleting or contracting edges of $G$ also has the property.  
In other words, the set of graphs that have the property is \emph{closed} under edge deletions and contractions.  
It is easy to see that $d$-flattenability is minor closed.  
% Alternatively, if a graph is $d$-flattenable, then so are all of its minors.
Hence, Robertson and Seymour's famous Graph Minor Theorem \cite{robertson2004graph} guarantees the existence of a finite set of \emph{forbidden} minors such that a graph is $d$-flattenable if and only if it has no minor in this set.  
It is folklore that $K_{d+2}$ is the only graph in this set for $d \in \{1,2\}$, and proofs can be found in \cite{belk2007realizability1}, where they proved that $K_5$ and $K_{2,2,2}$ are the only graphs in this set for $d = 3$ (see Figure \ref{fig:fbms}).  
Finite forbidden minor characterizations for $d$-flattenability have also been studied for normed spaces other than Euclidean \cite{sitharam2005combinatorial,fiorini2017excluded,dewar2024edge}. 
With these characterizations, the right-hand side of Theorem \ref{thm:2-SIP} is equivalent to ``each atom of $G \cup f$ that contains $f$ has no $d$-flattenability forbidden minor''.  
 As mentioned earlier, we will see in Example \ref{ex:3}, in Section \ref{sec:results}, that this restatement does not extend to $d = 3$ (see Figure \ref{fig:k5_k222_e_contracted} for a counterexample).  

A relationship was established in \cite{sitharam2010characterizing} between $d$-convexity and $d$-flattenability for $d = 3$, and later extended in \cite{sitharam2014flattenability} to any $d \geq 0$ and to other norms.  
This relationship makes use of previous work on understanding the Euclidean and non-Euclidean distance cones \cite{ball1990isometric,schoenberg1935remarks}.  
Given a graph $G$ and a set $F$ of its edges and nonedges, $G \setminus F$ is the graph obtained from $G$ by deleting $F$ from $E(G)$.  
If $F$ contains exactly one element $f$, then we replace $F$ with $f$ in this notation.  

\begin{theorem}[\cite{sitharam2010characterizing} Theorem 5.15, \cite{sitharam2014flattenability} Theorem 1]
\label{thm:sitharam_willoughby}
     For any dimension $d \geq 0$ and any graph $G$ the following are equivalent:
     \begin{enumerate}
         \item  $G$ is $d$-flattenable.
         \item for any subset $F \subseteq E(G)$, the pair $(G \setminus F,F)$ is $d$-convex.
     \end{enumerate}      
\end{theorem}

\begin{figure}[htb]
    \centering
    \includegraphics[width=0.2\textwidth]{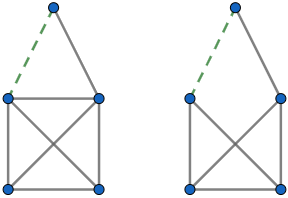}
    \caption{The graph-nonedge pair on the left has the $2$-SIP, while the one on the right does not (nonedges are shown as dashed-line segments).  
    This illustrates that $d$-convexity is not closed under edge deletions.  
    See Example \ref{ex:2} in this section, Examples \ref{ex:3} and \ref{ex:5} in Section \ref{sec:results}, and the discussion in Section \ref{sec:proof_obstacles}.  }
    \label{fig:not_minor closed}
\end{figure}

\begin{example}
    \label{ex:2}
    We examine applications of Theorems \ref{thm:2-convexity} - \ref{thm:sitharam_willoughby} to the graph-nonedge pairs $(G,f)$ (left) and $(G',f)$ (right) in Figure \ref{fig:not_minor closed}.  
    The only atom of $G \cup f$ that contains $f$ is $K_3$, which does not have a $K_4$ minor.  
    On the other hand, $G' \cup f$ is an atom with a $K_4$ minor.
    Hence, Theorem \ref{thm:2-SIP} shows that $(G,f)$ has the $2$-SIP while $(G',f)$ does not.  
    This theorem also shows that neither of these pairs has the $1$-SIP.  
    These statements can be verified using the observations in Example \ref{ex:1}.  
    
    Next, $G \cup f$ is not $2$-flattenable since it has a $K_4$ minor.  
    Therefore, Theorem \ref{thm:sitharam_willoughby} says that there exists some subset $F \subseteq E(G \cup f)$ such that $((G \cup f) \setminus F, F)$ is not $2$-convex.  
    To verify this, let $f'$ be the edge of $G$ deleted to obtain $G'$.  
    Since some atom of $G \cup f$ contains $f'$ and is $K_4$, Theorem \ref{thm:2-SIP} shows that $((G \cup f) \setminus f', f')$ does not have the $2$-SIP.  
    \qed
\end{example}

% The following provide motivation to study these properties. 
% If a graph-nonedge-set pair $(G,F)$ is $d$-convex, then the CCS $\Omega^d_F(G,\ell)$ can be efficiently traversed and sampled using well-known methods.  
% Additionally, if $G$ belongs to certain large and natural classes of graphs, then there are efficiently computable maps from $\Omega^d_F(G,\ell)$ to $\mathcal{C}^d(G,\ell)$ \cite{SitharamWang2014Beast}.  
% Furthermore, characterizations of $d$-convexity have led to improvements in molecular and particle assembly modeling and kinematic mechanism analysis and design, respectively, via the opensource software EASAL \cite{Ozkan2018ACMTOMS,Prabhu2020JCIM} and CayMos \cite{SitharamWang2014Beast,Wang2014cayley}, and to efficient algorithms for the core problem of distance constraint graph realization \cite{Baker2015}.  

Finally, some of the concepts in this paper appear to be related to similarly long papers \cite{garamvolgyi2023partial,jackson2006globally} on \emph{generically} globally-linked pairs in graphs.  
 However,   our study of $d$-SIP, $d$-convexity, $d$-Cayley-connectedness does not require any genericity conditions, and nonedges with $d$-SIP occur in sparser graphs overall than nonedges that are globally linked, weakly globally linked or even linked.

\section{Contributions and proof obstacles}
\label{sec:results}
 
  \begin{note}
   We shorten ``$d$-flattenability forbidden minor'' to ``$d$-forbidden minor,'' and we further omit $d$ from this terminology when it is clear from context.   
  \end{note}  

Most usage of minors in the literature involves proving that a graph either has or does not have a given minor, and the specifics of how this minor is obtained are not important.  
Hence, minors are often treated as equivalence classes of graphs, as defined in Section \ref{sec:related}.  
Instead, want to keep track of vertices as edge deletions and contractions are performed, and so we introduce \emph{rooted} minors, which have been studied extensively \cite{robertson2004graph,bohme2024rooted,fabila2011rooted,wood2010thomassen,wollan2008extremal,kawarabayashi2004rooted,jorgensen2007extremal}.  
Refer to Figure \ref{fig:minors}.  
Each vertex of a  \emph{rooted} minor of $G$  corresponds to  a connected subgraph of $G$, all of whose edges have been \emph{contracted} to obtain the minor. 
These connected subgraphs partition $V(G)$, hence, given a pair of distinct vertices of the rooted minor, their corresponding two connected subgraphs of $G$  have disjoint  edge sets. 
A pair of vertices is an edge in the rooted minor only if the  graph - induced by the union of the corresponding two connected subgraphs of $G$ - is connected.  
In this case, the edge in the rooted minor corresponds to the set of edges of $G$ that have one endpoint in each of the corresponding two connected subgraphs. Note that  the  ``if'' direction of the previous statement does not hold. I.e., any  edges  straddling the two connected subgraphs could be deleted to obtain a nonedge in the rooted minor.
 A more formal definition of a rooted minor is given in Section \ref{sec:minor-notation}.  
Unless otherwise specified, all minors in this paper are rooted.   
% Sometimes it will be convenient to refer to a vertex $x$ or an edge or nonedge $f$ of $[G]$ without making use of the bracket notation, in which case the brackets are implicit, and so $x^{-1}$ and $f^{-1}$ denote the sets $[x]^{-1}$ and $[f]^{-1}$, respectively. 
A pair of distinct vertices $u,v \in V(G)$ is \emph{contracted} in a rooted minor if a vertex in the minor corresponds to  a connected subgraph of $G$ containing both $u$ and $v$, and \emph{preserved} in the minor otherwise.  
In the latter case, the minor is \emph{$uv$-preserving}.  
The main result of the paper is a proof of Conjecture \ref{conj:sitharam-gao} in Section \ref{sec:related}), restated as follows.

% However, our results concern minors obtained without contracting a given edge, and so we often want to refer to different representatives of the same minor equivalence class.  
% For this reason, and to state our results, we first define a natural \emph{minor map} $\pi$ from a graph $G$ to any of its minors $M$ obtained via a single deletion or contraction of an edge $uv$.  
% More precisely, $\pi$ is a map from $V(G) \cup E(G) \cup \emptyset$ to $V(M) \cup E(M) \cup \emptyset$.  
% If $uv$ is deleted to obtain $M$, then $\pi$ is the identity on all but $uv$, which it sends to the empty-set.  
% If $uv$ is contracted to obtain $M$, then $\pi$ is the identity on all but $u$, $v$, $uv$, and all other edges $xu$ and $xv$, and it sends the first three to the vertex $w$ added by this contraction and the last two to $xw$.  
% By composing these maps, we get a minor map between $G$ and any of its minors.  
% An edge $uv$ of a graph $G$ is \emph{contracted} in $M$ if $\pi(u) = \pi(v)$, and \emph{preserved} in $M$ otherwise.  
% In the latter case, $M$ is \emph{$uv$-preserving}.  
% The main result of the paper is the following.

\begin{theorem}[Main: $(d \leq 3)$-SIP characterization]
    \label{thm:3-sip_characterization}
    For any dimension $d \le 3$, a graph-nonedge pair $(G,f)$ has the $d$-SIP if and only if each atom of $G \cup f$ that contains $f$ has no $f$-preserving $d$-forbidden minor.  
\end{theorem}

\begin{remark}
 This theorem is equivalent to Conjecture \ref{conj:sitharam-gao} since the right hand side can be stated equivalently as: ``$f$ is contracted in each forbidden minor of each atom of $G\cup f$ that contains $f$.''
 For atoms $G\cup f$, the right hand side is even simpler: ``$f$ is contracted in each forbidden minor of $G\cup f$.''  However, inductive aspects of the   proof  - in both directions, even for the special case of atoms - appear to require the full strength of the general theorem, as discussed below.
 \end{remark}

 We note that for any dimension $d \leq 2$, this result essentially follows from Theorem \ref{thm:2-SIP} and its proof in \cite{sitharam2010characterizing}.  
Specifically, the proof of Theorem 5.2 in \cite{sitharam2010characterizing}, which is used to prove Theorem \ref{thm:2-SIP}, already incorporates the additional ingredient that if $G \cup f$ has an atom that contains $f$ and has a $d$-forbidden minor, then it has an atom (in fact the same atom) with an \emph{$f$-preserving} $d$-forbidden minor. 
Therefore our main contribution is for $d=3$, and the proof is presented in Sections \ref{sec:reverse} (converse direction) and \ref{sec:forward} (forward direction).  
In Section \ref{sec:conclusions}, we conjecture that  Theorem \ref{thm:3-sip_characterization} is true for all $d \geq 4$.

\iffalse
Let $[G]$ be a minor of a graph $G$.  
A vertex $[u]$ of $[G]$ is written as $[u]^{-1}$ when we  refer to the corresponding set of vertices in $G$.  
For any pair $uv$ of vertices in $G$, $[uv]^{-1}$ is the set $[u]^{-1}$ if $[u] = [v]$, and the set $[u]^{-1} \times [v]^{-1}$ otherwise.  
A nonedge (resp. edge) $uv$ is \emph{doubled in $[G]$} if it is preserved and $[uv]^{-1}$ contains at least one edge (resp. two edges) of $G$; and $uv$ is \emph{retained in $[G]$} if it is preserved and not doubled.  
$[G]$ is \emph{$uv$-preserving (resp. retaining)} if $uv$ is preserved (resp. retained) in $[G]$.  
\fi

 \medskip\noindent
Consider a natural approach to proving the forward ``only if'' direction of  Theorem \ref{thm:3-sip_characterization}. Assume that when  $G\cup f$ is a forbidden minor we have a certificate against $d$-SIP of $(G,f)$, i.e. a length assignment for the edges of $G$ for which there are two disjoint intervals of  lengths attained by $f$. We can then extend this minor certificate to a certificate against $d$-SIP for general $(G,f)$ provided $f$ is \emph{retained} in an   \emph{induced} forbidden minor of $G\cup f$. By \emph{induced} minor, we mean a minor obtained only using edge contractions, without any edge deletions; and by \emph{retained} in a minor, we mean that the pair $uv = f$    is  not only preserved in  a minor of $G$, but remains a nonedge.
The certificate against $d$-SIP for $(G,f)$  is obtained by setting all   lengths  of contracted edges to 0 and the length of each  non-contracted edge to that of its corresponding edge in the minor certificate. 
The following observation  formalizes this approach.

\begin{observation}
    \label{obs:3-sip-retained}
    For any dimension $d \leq 3$, a graph-nonedge pair $(G,f)$  has the $d$-SIP only if $G \cup f$  does not have an $f$-retaining induced $d$-forbidden minor.  
\end{observation}

\begin{proof}
    Let $M$ be an $f$-retaining induced $d$-forbidden minor, and let $f'$ be the  nonedge in $M$ corresponding to $f$.  
    The fact that $(M,f')$ does not have the $d$-SIP  follows relatively straightforwardly for $d\le 3$ (see Example \ref{ex:1} and the proof of Lemma \ref{lem:no_type_2_pairs_no_3-sip}, if $M$ is $K_{d+2}$; and the proof of Theorem 8 in \cite{belk2007realizability1}; if $M$ is $K_{2,2,2}$).  
    Along with Lemma \ref{lem:cm_non-sip}, stated later in this section, it follows that $(G,f)$ does not have the $d$-SIP.  
\end{proof}

To illuminate the obstacles to  straightforward extensions of the  forward ``only if''  direction of Observation \ref{obs:3-sip-retained} towards a proof of the forward direction of  Theorem \ref{thm:3-sip_characterization}, we first inspect the relationships between a nonedge $f$ being preserved (resp. retained), in a minor (resp. induced minor).
Observe that a nonedge $f$ is preserved  in a   minor $M$ of $G\cup f$ if and only if it is retained   in $M$.  One direction of this statement is immediate. For the   other direction  first delete the edges in $G$  whose endpoints correspond to $f$'s endpoints in $M$.  While being preserved in a minor is equivalent to being preserved in an induced minor, the
  equivalence between preserved and  retained  is false for induced minors  (where the minor is obtained using only contractions). 

Due to this reason, Observation \ref{obs:3-sip-retained} cannot be extended beyond induced forbidden minors.
In particular, Lemma \ref{lem:gluing} allows us to choose $(G,f)$ 
 s.t. $(G,f)$  has $d$-SIP although  $f$  is retained in a forbidden minor, i.e., preserved in an (induced) forbidden minor. This is achieved by  some atom of $G \cup f$ that does not contain $f$ having a $d$-forbidden minor. For example, take $d=2$ and $(G,f)$ to be the pair on the left in Figure \ref{fig:not_minor closed}. In other words, this weaker property of $f$ being retained in a forbidden minor is insufficient as a certificate against $d$-SIP.
  However, if $G$ is an atom, the weaker property of $f$ being retained in a forbidden minor is a sufficient certificate for $(G,f)$ to not have $d$-SIP, for $d\le 3$,  as we show  below in Corollary 
  \ref{cor:3-SIP-atom_fbm_char}, of the forward direction of the main theorem applied to atoms.  
  
%============
 From a larger perspective,  Theorem \ref{thm:3-sip_characterization} has the flavor of a finite forbidden minor characterization, 
although $d$-SIP is not minor closed, in particular, it is not closed under edge deletions. 
   Specifically, only those atoms of $G \cup f$ that \emph{contain} $f$ must avoid a certain presentation of forbidden minors; and   
 atoms containing $f$ are allowed to have these forbidden minors  as long as  $f$ is contracted in them. 
 For a more formal exploration, we introduce \emph{minor pairs} of a graph-nonedge pair, analogous to minors of graphs.  
$(M,f')$ is a \emph{(resp. induced) minor pair} of a graph-nonedge pair $(G,f)$ if $M$ is an $f$-retaining (resp. induced) minor of $G$ and $f'$ is the nonedge of $M$ corresponding to $f$.  
A property of graph-nonedge pairs is \emph{(resp. induced) minor pair closed} if whenever a pair $(G,f)$ has the property, so do all its (resp. induced) minor pairs.  
% Let $[G]$ be a minor of a graph $G$.  
% A pair of distinct vertices $[u][v]$ of $[G]$ is written as $[uv]$.  
% If $[u] = [v]$, then $[uv]$ is treated as a vertex.  
% For any subset $X$ of vertices, edges, and nonedges of a graph $G$, define the set $[X] = \{[x] : x \in X\}$.  
% For any set of nonedges $F$ of $G$, a \emph{minor pair} $([G],F')$ of $(G,F)$ consists of a minor $[G]$ and a subset $F' \subseteq [F]$ of nonedges.  
% If $[G]$ is obtained via only contractions, no edge deletions, then it is an \emph{induced} minor of $G$ and $([G],F')$ is an \emph{induced} minor pair of $(G,F)$.  
% A property is \emph{(resp. induced) minor pair closed} if whenever a pair $(G,F)$ has the property, so do all its (resp. induced) minor pairs. 

In Example \ref{ex:2}, it was shown that the $2$-SIP is not closed under edge deletions, and hence this property is not minor pair closed.  
Using a similar argument, we see that the $d$-SIP (the property on the left-hand side of the theorem statement) is not minor pair closed for any $d \geq 2$.  
(Aside: later in this section,  a repeatedly used Lemma \ref{lem:tool} will show that the $d$-SIP is closed under a restricted type of edge deletion).
However, Lemma \ref{lem:cm_non-sip} establishes that the $d$-SIP is induced minor pair closed.
Not surprisingly, the right hand side of Theorem \ref{thm:3-sip_characterization}, namely the property ``each atom of $G \cup f$ that contains $f$ has no $f$-preserving $d$-forbidden minor," is closed under contractions but not under edge deletions, and therefore is induced minor pair closed, but not minor pair closed.
However, if we strengthen this property to  ``$G \cup f$  does not have an $f$-preserving forbidden minor {\bf (*)}'' then it  is closed under both edge deletions and contractions and therefore is minor pair closed. 
Moreover, when $G \cup f$ \emph{is} an atom, this strengthened property {\bf (*)} is precisely  the right hand side of Theorem \ref{thm:3-sip_characterization}, and is  necessary for $d$-SIP for $d\le 3$. 
Hence, we obtain the following corollary for atoms. 

We note that   we are  unaware of a direct proof of this corollary  except as a consequence of the forward direction of  Theorem \ref{thm:3-sip_characterization}: even for the special case of atoms, the induction  in the forward direction appears to  involve   non-atoms.
% This is a non-trivial statement that requires the full power of Theorem \ref{thm:3-sip_characterization}, and which we state as a corollary.  

%  The existence of a minor pair $([G],F')$ of $(G,F)$ with $F'$ nonempty,  implies at least one nonedge in $F$  remains a nonedge in $[F]$, and if $F=f$  is a single nonedge, then   it is  retained as a nonedge  $F'=f'$, in $[G]$, i.e. not only do its endpoints remain distinct,  no edge in $G$ maps to $f'$ in $[G]$. This   motivates variants of the concept of  preserved nonedges.  
 
% Recall that a pair $uv$ of vertices in $G$ is preserved in $[G]$ if $[u] \neq [v]$.  We say a nonedge (resp. edge) $uv$ is \emph{doubled in $[G]$} if it is preserved and $[uv]^{-1}$ contains at least one edge (resp. two edges) of $G$; and $uv$ is \emph{retained in $[G]$} if it is preserved and not doubled.  
%$[G]$ is \emph{$uv$-preserving (resp. retaining)} if $uv$ is preserved (resp. retained) in $[G]$.  
%XXXXXXXXXXX

\begin{corollary}[$(d \leq 3)$-SIP is minor pair closed for atoms]
    \label{cor:3-SIP-atom_fbm_char}
    For any dimension $d \leq 3$, the $d$-SIP is minor pair closed for atoms: for any graph $G$ with nonedge $f$, if $G \cup f$ is an atom and $(G,f)$ has the $d$-SIP, then all its minor pairs also have the $d$-SIP.
\end{corollary}

\medskip\noindent
  Next observe that the converse of  Observation \ref{obs:3-sip-retained} does not hold even   when $G\cup f$ is an atom. I.e. there are atoms $G \cup f$ where  $(G,f)$ does not have $3$-SIP  although $f$ is not retained in  any induced forbidden minor.  For example,  let $G \cup f$ to be the graph in Figure \ref{fig:min_2} with $f$ being the solid green edge.  
  But  $f$   not being retained in \emph{any} forbidden minor   is precisely the strengthened property {\bf (*)}. And 
  the converse of Theorem \ref{thm:3-sip_characterization}  shows sufficiency for $(G,f)$  to have  $d$-SIP for $d \le 3$ even for non-atoms (a core  idea underlying the converse was illustrated in Example 1, Figures \ref{fig:k3-e} \ref{fig:k3-e_realization},   \ref{fig:k4-e}, and   \ref{fig:converse_4}). 
 However, as in the case of Corollary \ref{cor:3-SIP-atom_fbm_char} - i.e. the necessity of  the property {\bf (*)} for ($\le 3$)-SIP for atoms -  we are  unaware of a direct proof of  the sufficiency of (*) for ($\le 3$)-SIP  for   atoms,    except as a consequence of the general converse direction of  Theorem \ref{thm:3-sip_characterization}:   even for the special case of atoms, the induction  in the  converse direction appears to  involve   non-atoms.

  Furthermore, as pointed out earlier, when $G\cup f$ is not an atom,   property {\bf (*)} is  not necessary for $(G,f)$ to have $(\le 3)$-SIP.   In particular, Lemma \ref{lem:gluing} allows us to choose $(G,f)$ so that it has the $d$-SIP although some atom of $G \cup f$ that \emph{does not} contain $f$ has a $d$-forbidden minor, e.g. the left Figure \ref{fig:not_minor closed} for $d=2$,.

 \medskip\noindent
Next, Corollary \ref{cor:flat-iff-sip}, below,  adds an apparently much weaker, but in fact, equivalent statement to Theorem \ref{thm:sitharam_willoughby} in Section \ref{sec:related} for dimensions $d \leq 3$.  
\begin{corollary}[Adding equivalent statement to Theorem \ref{thm:sitharam_willoughby} for $d \leq 3$]
    \label{cor:flat-iff-sip}
    For any dimension $d \le 3$ and graph $G$ the following are equivalent.
     \begin{enumerate}
         \item  $G$ is $d$-flattenable
         \item for any subset $F \subseteq E(G)$, the pair $(G \setminus F,F)$ is $d$-convex
         \item for each edge $f$ of $G$, the pair $(G \setminus f, f)$ has the $d$-SIP  
     \end{enumerate}     
\end{corollary}
\begin{proof}
    (1) $\implies$ (2) is Theorem \ref{thm:sitharam_willoughby}.  
    (2) $\implies$ (3) is  immediate from definition of convexity.  
    For (3) $\implies$ (1), assume to the contrary that (3) is true while (1) is false.  
    Then, $G$ has a $d$-forbidden minor.  
    % , for each edge $f$ of $G$, $(G \setminus f, f)$ has the $d$-SIP but $G$ is not $d$-flattenable.  
    Since each $d$-forbidden minor is an atom, some atom of $G$ must have a $d$-forbidden minor.  
    Combining this with Corollary \ref{cor:atom_gluing}, below, allows us to assume wlog that $G$ is an atom.  
    However, (3) and Theorem \ref{thm:3-sip_characterization} imply that no edge of $G$ is preserved in any $d$-forbidden minor of $G$, which is a contradiction.    
\end{proof}

% \begin{proof}
%     If $(G,f)$ has the $d$-SIP, then Corollary \ref{thm:dleq3-sip_characterization} states that $G \cup f$ has no $f$-preserving forbidden minor, since $G \cup f$ is an atom.  
%     Hence, $[G] \cup [f]$ has no $[f]$-preserving forbidden minor.  
%     Thus, the same theorem shows that $([G],[f])$ has the $d$-SIP.  
% \end{proof}

\begin{figure}[htb]
    \centering
    \begin{subfigure}{0.19\linewidth}
        \centering
        \includegraphics[width=0.5\linewidth]{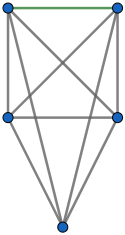}
        \subcaption{}
        \label{fig:min_1}
    \end{subfigure}
    \begin{subfigure}{0.19\linewidth}
        \centering
        \includegraphics[width=0.7\linewidth]{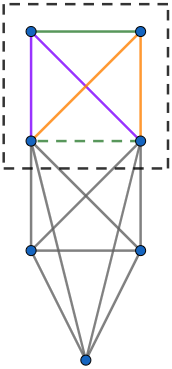}
        \subcaption{}
        \label{fig:min_2}
    \end{subfigure}
    \begin{subfigure}{0.19\linewidth}
        \centering
        \includegraphics[width=0.7\linewidth]{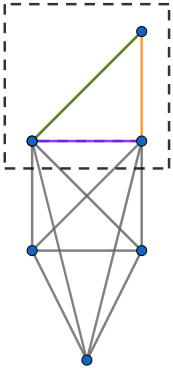}
        \subcaption{}
        \label{fig:min_3}
    \end{subfigure}
    \begin{subfigure}{0.19\linewidth}
        \centering
        \includegraphics[width=0.7\linewidth]{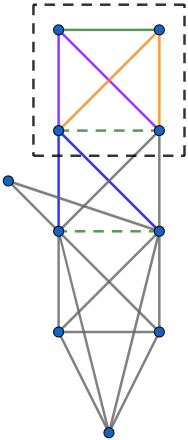}
        \subcaption{}
        \label{fig:min_4}
    \end{subfigure}
    \begin{subfigure}{0.19\linewidth}
        \centering
        \includegraphics[width=0.7\linewidth]{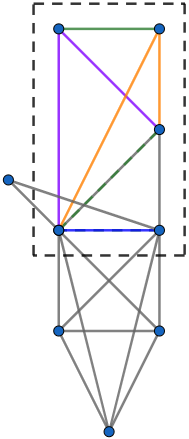}
        \subcaption{}
        \label{fig:min_5}
    \end{subfigure}
    \begin{subfigure}{0.19\linewidth}
        \centering
        \includegraphics[width=0.65\linewidth]{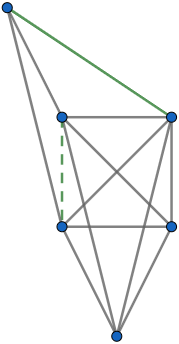}
        \subcaption{}
        \label{fig:min_6}
    \end{subfigure}
    \begin{subfigure}{0.19\linewidth}
        \centering
        \includegraphics[width=\linewidth]{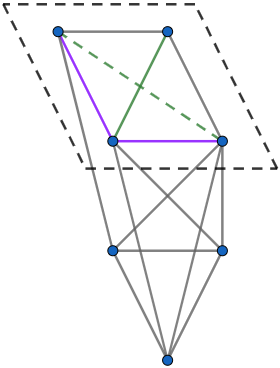}
        \subcaption{}
        \label{fig:min_7}
    \end{subfigure}
    \begin{subfigure}{0.19\linewidth}
        \centering
        \includegraphics[width=0.63\linewidth]{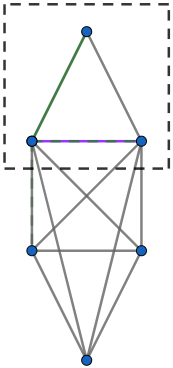}
        \subcaption{}
        \label{fig:min_8}
    \end{subfigure}
    \begin{subfigure}{0.19\linewidth}
        \centering
        \includegraphics[width=0.91\linewidth]{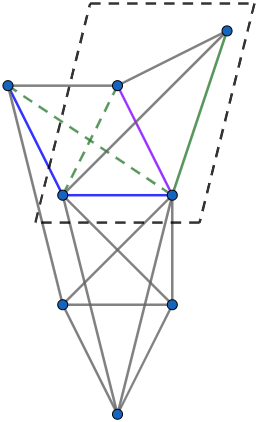}
        \subcaption{}
        \label{fig:min_9}
    \end{subfigure}
    \begin{subfigure}{0.19\linewidth}
        \centering
        \includegraphics[width=0.8\linewidth]{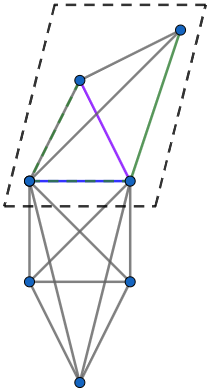}
        \subcaption{}
        \label{fig:min_10}
    \end{subfigure}
    \caption{See Example \ref{ex:3} in this section, the discussion in Section \ref{sec:proof_obstacles},   Examples \ref{ex:minimal}-\ref{ex:cms} and the discussion in Section \ref{sec:forward}. Subfigures (a), (b), (d), (f), (g), and (i) show graphs $G \cup f$ where $f$ is the green edge and $(G,f)$ is a minimal graph-nonedge pair (Definition \ref{def:minimal_pair} in Section \ref{sec:forward}).  
    Dashed green line-segments are nonedges.  
    Top-level expanded $f$-separating CMSs of $G \cup f$ (Definition \ref{def:expanded_e-separating_pair} in Section \ref{sec:forward}) are shown as the purple and orange subgraphs in (B) and (D), the purple subgraph in (G), and the connected subgraph containing one purple edge and one blue edge in (I).  
    The expanded $f$-components of these CMSs are the subgraphs in the dashed black boxes.  
    Contracting one purple edge in (b) and (d) yield (c) and (e), respectively.  
    Contracting one bluw edge in (g) and (i) yield (h) and (j), respectively.  
    The purple subgraphs in (c) and (h) and the blue subgraphs in (e) and (j) are $f$-separating CMSs whose $f$-components are the subgraphs in the dashed black boxes.  
    The purple $f$-separating CMSs are in the top-level while the blue $f$-separating CMSs are not.  
    }
    \label{fig:forward}
\end{figure}

\begin{example}
    \label{ex:3}
    Theorem \ref{thm:3-sip_characterization} can be used to verify the following statements.  
    A graph-nonedge pair $(G,f)$ has the $3$-SIP if it is any of the pairs in Figures \ref{fig:k3-e}-\ref{fig:k4-e}, the top-left or bottom-left of Figure \ref{fig:2_and_3-convexity_ex} such that $f$ is any of the dashed nonedges, Figure \ref{fig:d-convexity_ex} such that $f$ is any of the dashed nonedges, Figure \ref{fig:not_minor closed}, Figure \ref{fig:k5_wing} or \ref{fig:k222_wing} such that $f$ is $w_1w_2$, and Figure \ref{fig:partial_3-tree} such that $f$ is any of the dashed nonedges.  
    $(G,f)$ has the $3$-SIP if $G \cup f$ is any of the graphs in Figures \ref{fig:min_3}, \ref{fig:min_5}, \ref{fig:min_8}, or \ref{fig:min_10}, where $f$ is the top-most green edge.  
    $(G,f)$ does not have the $3$-SIP if it is the pair in Figure \ref{fig:k5-e} or \ref{fig:clique_sep} or if $G \cup f$ is any of the graphs in Figures \ref{fig:min_1}, \ref{fig:min_2}, \ref{fig:min_4}, \ref{fig:min_6}, \ref{fig:min_7}, or \ref{fig:min_9}, where $f$ is the green edge.  
    
    The fact that the pairs in Figure \ref{fig:k5_k222_e_contracted} have the $3$-SIP demonstrates that Theorem \ref{thm:2-SIP} does not extend to $d = 3$ even after replacing ``$K_{d+2}$ minor'' with ``$d$-forbidden minor'' on the right-hand side.  
    \qed
\end{example}

\begin{figure}[htb]
    \centering
    \begin{subfigure}{0.49\textwidth}
        \centering
        \includegraphics[width = 0.45\textwidth]{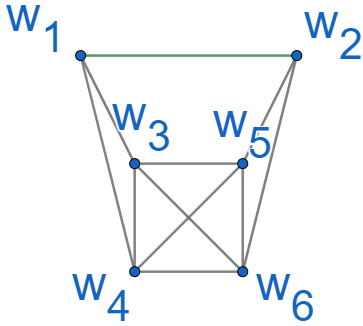}
\subcaption{}
        \label{fig:k5_wing}
    \end{subfigure}
    \begin{subfigure}{0.49\textwidth}
        \centering
        \includegraphics[width = 0.4\textwidth]{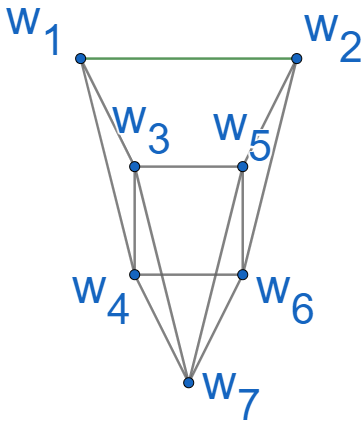}
\subcaption{}
        \label{fig:k222_wing}
    \end{subfigure}
    \caption{Graphs with $K_5$ and $K_{2,2,2}$ minors reachable only by contracting the edges labeled $w_1w_2$.  
    See Example \ref{ex:3} in this section and the discussions in Sections \ref{sec:proof_obstacles} and \ref{sec:reverse}.  
    These figures are also referenced throughout the paper.  }
    \label{fig:k5_k222_e_contracted}
\end{figure}

\begin{figure}[htb]
   \centering
   \begin{subfigure}[t]{0.49\linewidth}
       \centering
       \includegraphics[width=0.5\linewidth]{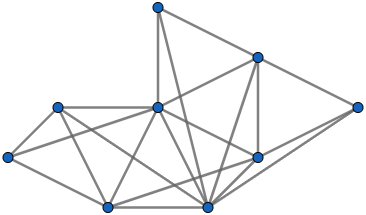}
       \caption{}
       \label{fig:3-tree}
   \end{subfigure}
   \begin{subfigure}[t]{0.49\linewidth}
       \centering
       \includegraphics[width=0.5\linewidth]{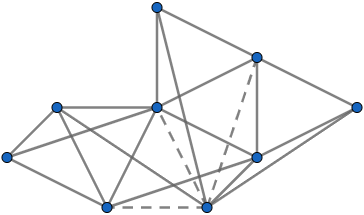}
       \caption{}
       \label{fig:partial_3-tree}
   \end{subfigure}
   \caption{(a) a $3$-tree $T$ and (b) a $3$-connected partial $3$-tree obtained from $T$ by deleting the dashed edges.  
   See Example \ref{ex:4} and the discussion in this section.  }
   \label{fig:3-tree_partial_3-tree}
\end{figure}

Next, Theorems \ref{prop:3-connected_partial_3-tree_star_lemma} and \ref{prop:partial_3-tree_3-reflection}, below, are  independently interesting results concerning properties of partial 3-trees and covering maps. They are tools used in, and byproducts of, proving Theorem \ref{thm:3-sip_characterization}.  
They are proved in Section \ref{sec:tool-thms}.  
A \emph{$3$-tree} is either $K_4$ or a graph obtained from a $3$-tree $T$ by adding a single vertex connected to each vertex in a clique of size $3$ in $T$.  
Equivalently, a $3$-tree is an edge-maximal tree-width $3$ graph.  
A \emph{partial $3$-tree} is any subgraph of a $3$-tree.  
Equivalently, a partial $3$-tree is any graph that avoids having $K_5$, $K_{2,2,2}$, $V_8$, and $C_5 \times C_2$ as minors \cite{arnborg1990forbidden}; see Figures \ref{fig:d-convexity_ex} and \ref{fig:fbms}.  
See Figure \ref{fig:3-tree_partial_3-tree} for examples.  
Considerable effort has gone into determining if a given set $F$ of nonedges of a partial $3$-tree $G$ is such that $G \cup F$ is a partial $3$-tree \cite{bodlaender2006safe}.  
Theorem \ref{prop:3-connected_partial_3-tree_star_lemma} provides a new tool for this task.  
See Example \ref{ex:4}, below.  
For any integer $k \geq 1$, a graph is $k$-connected if it is connected and any of its separators has size at least $k$.  
We write a countable indexed set $\{X_1,X_2,\dots\}$ as $\{X_i\}$.  
\begin{theorem}[$3$-connected partial $3$-tree star theorem]
    \label{prop:3-connected_partial_3-tree_star_lemma}
    For any $3$-connected partial $3$-tree $G$ and any subset $\{wu_i\}$ of its nonedges, if $G \cup wu_i$ is a partial $3$-tree for each $wu_i$, then $G \cup \{wu_i\}$ is a partial $3$-tree.  
\end{theorem}

The above-mentioned forbidden minor characterization for partial $3$-trees along with Theorem \ref{thm:sitharam_willoughby} shows that if a  pair $(G,F)$ is such that $G \cup F$ is a partial $3$-tree, then $(G,F)$ is $3$-convex.  
Theorem \ref{prop:partial_3-tree_3-reflection} shows that if $G \cup F$ is a partial $3$-tree, then $(G,F)$ has an even stronger property, which we now define.
$(G,F)$ has the \emph{$d$-covering map property} if, for any squared edge-length map $\ell$, the image of each connected component of the CS $\mathcal{C}^d(G,\ell)$ under the Cayley map $\phi_F$ is the entire CCS $\Omega^d_F(G,\ell)$.   Since $G\cup F$ is a partial $3$-tree,   the result implies that the CS $\mathcal{C}^d(G,\ell)$ is a so-called branched covering space of the base space $\Omega^d_F(G,\ell)$   given by the Cayley map $\phi_F$.   The $d$-covering map property further excludes the existence of lower dimensional connected components of $\mathcal{C}^d(G,\ell)$ whose image under $\phi_F$ does not cover the base space.  

\begin{theorem}[Partial $3$-tree nonedge set pairs have the $3$-covering map property]
    \label{prop:partial_3-tree_3-reflection}
     For any graph $G$ and nonedge set $F$ where $G \cup F$ is a partial $3$-tree, $(G,F)$ has the $3$-covering map property.  
\end{theorem}

\begin{example}
    \label{ex:4}
    Consider the partial $3$-tree $P$ in Figure \ref{fig:partial_3-tree} and let $F$ be the set of dashed nonedges.  
    Since $P \cup f$ is a partial $3$-tree for any $f \in F$, Theorem \ref{prop:3-connected_partial_3-tree_star_lemma} states that $P \cup F$ is a partial $3$-tree.  
    Note that $P \cup F$ is the $3$-tree in Figure \ref{fig:3-tree}.  
    Consider any squared edge-length map $\ell$ that assigns each edge of $P$ a non-zero value so that the CS $\mathcal{C}^3(P,\ell)$ is non-empty.  
    Also, let $u$ be the left-most vertex in Figure \ref{fig:partial_3-tree}.  
    For any realization $p$ in this CS, observe that the realization obtained by reflecting $p(u)$ across the plane defined by its neighbors is not in the same connected component of this CS as $p$.  
    Even so, Theorem \ref{prop:partial_3-tree_3-reflection} states that the image of each connected component of this CS is the CCS $\Omega^3_F(P,\ell)$.  
    \qed
\end{example}

Finally, we present several repeatedly used tool lemmas that  capture $d$-convexity, $d$-Cayley-connectedness, and   $d$-SIP.  
In particular, Lemmas \ref{lem:cm_non-sip}-\ref{lem:gluing} and Corollary \ref{cor:atom_gluing}, stated below, are used throughout the proof of Theorem \ref{thm:3-sip_characterization}.  
See Example \ref{ex:5}, below, for example applications of these results.  
% A minor is \emph{induced} if it can be obtained using only contractions.  
% Note that induced minors and topological minors are not equivalent, i.e., while every  topological minor is an induced minor, the converse is not true in general.  
% A property of graph-nonedge-set pairs is \emph{induced} minor pair closed if it is closed under contractions.  
To state Lemma \ref{lem:cm_non-sip}, we generalize the notion of a minor pair to sets of nonedges $F$ and $F'$ of a graph $G$ and its minor $M$.  
We say  $(M,F')$ is a minor pair of  $(G,F)$ if $M$ is a minor of $G$ and $F'$ is a, possibly empty, set of nonedges of $M$  whose corresponding set of nonedges $F$ in $G$  are retained in $M$.  

\begin{lemma}[$d$-Convexity/Cayley-connectedness/SIP is induced minor pair closed]
    \label{lem:cm_non-sip}
    The $d$-convexity (resp. $d$-Cayley-connectivity) property is induced minor pair closed.  
\end{lemma}

\begin{proof}
    % The proof uses  arguments  similar to \cite{sitharam2010characterizing, sitharam2014flattenability}. 
     For a graph $G$ with nonedge set $F$, let $(G,F)$  be $d$-convex (resp. $d$-Cayley-connected)   with induced minor pair $(M,F')$.  
    Observe that, for any linkage $(M,\ell')$, there exists a linkage $(G,\ell)$ such that the CCS $\Omega^d_{F'}(M,\ell')$ is isomorphic to the intersection of the CCS $\Omega^d_F(G,\ell)$ with a coordinate subspace.  
    Therefore, since the CCS $\Omega^d_F(G,\ell)$ is convex (resp. connected), by assumption, so is the CCS $\Omega^d_{F'}(M,\ell')$.  
\end{proof}

Lemmas \ref{lem:tool} and \ref{lem:gluing} along with Corollary \ref{cor:atom_gluing}, below, are proved in Section \ref{sec:tool_lemmas}.  

\begin{lemma}[Edge deletions that preserve $d$-Cayley-connectivity]
    \label{lem:tool}
    Let $G$ be a graph and $F$ and $F'$ be disjoint sets of its nonedges.  
    For any dimension $d \geq 1$, if $(G,F)$ and $(G \cup F, F')$ are $d$-Cayley-connected, then $(G, F \cup F')$ is $d$-Cayley-connected.  
    Consequently, $(G,F')$ is $d$-Cayley-connected.  
\end{lemma}

% Corollary \ref{cor:atom_edge_deletion_closed}, below, follows from Corollary \ref{thm:dleq3-sip_characterization}.  

% \begin{corollary}[Edge deletions in atoms preserve $(d \leq 2)$-SIP]
%     \label{cor:atom_edge_deletion_closed}
%    Let $G$ be a graph, $f$ be one of its edges or nonedges, and $G'$ be a graph obtained from $G$ by deleting some edges.  
%    For any dimension $d \leq 2$, if $(G,f)$ has the $d$-SIP and $G \cup f$ is an atom, then $(G',f)$ has the $d$-SIP.  
% \end{corollary}

% \begin{proof}
%     Assume that $(G,f)$ has the $d$-SIP and $G \cup f$ is an atom.  
%     By Corollary \ref{cor:3-SIP-atom_fbm_char}, $G \cup f$ has no $f$-preserving $d$-forbidden minor.  
%     Hence, no atom of $G' \cup f$ that contains $f$ has an $f$-preserving $d$-forbidden minor.  
%     Thus, by Corollary \ref{thm:dleq3-sip_characterization}, $(G',f)$ has the $d$-SIP.  
% \end{proof}

Consider a graph $G$, any of its induced subgraphs $H_1$ and $H_2$, and any subset $U \subseteq V(G)$.  
$H_1 \cup U$ is the graph induced by $V(H_1) \cup U$.  
If $U$ contains exactly one vertex $u$, then we replace $U$ with $u$ in this notation.  
The \emph{$U$-components} of $G$ are the vertex-maximal induced subgraphs of $G$ that are not separated by any subset of $U$.  
% Let $\{J_i\}$ be the connected components of $G \setminus U$.  
% The \emph{$U$-components} of $G$ are the graphs $J_i \cup U_i$, where $U_i$ is the minimum cardinality subset of $U$ such that $J_i$ is a connected component of $G \setminus U_i$.  
The subgraph $H$ induced by $U$ is endowed with the terminology associated with $U$ - e.g., the \emph{$H$-components} of $G$ are the $U$-components are $G$.  

\begin{lemma}[Clique-gluing]
\label{lem:gluing}
    Let $(G,f)$ be a graph-nonedge pair, let $C$ be a clique separator of $G \cup f$, and let $\{G_i\}$ be the set of $C$-components of $G \cup f$ that contain $f$.
    For any dimension $d \geq 1$, $(G,f)$ has the $d$-SIP if and only if, for each $G_i$, $(G_i \setminus f,f)$ has the $d$-SIP.  
\end{lemma}

\begin{corollary}[Atom-gluing]
    \label{cor:atom_gluing}
    For any dimension $d \geq 1$, a graph-nonedge pair $(G,f)$ has the $d$-SIP if and only if, for each atom $G_i$ of $G \cup f$ that contains $f$, $(G_i \setminus f, f)$ has the $d$-SIP.  
\end{corollary}

\begin{example}
    \label{ex:5}
    Consider the graph-nonedge pairs $(G,f)$ (left) and $(G',f)$ (right) in Figure \ref{fig:not_minor closed}, and let $f'$ be the edge deleted from $G$ to obtain $G'$.  
    We saw in Example \ref{ex:2} that $(G,f)$ has the $2$-SIP while $(G',f)$ does not.  
    Since the endpoints of $f'$ form a clique separator of $G \cup f$ and $G' \cup f'$, these facts can also be verified using Lemma \ref{lem:gluing}.  
    Lemma \ref{lem:cm_non-sip} says that any induced minor pair of $(G,f)$ has the $2$-SIP, which can be verified using Theorem \ref{thm:2-SIP}.  
    The contrapositive of Lemma \ref{lem:tool} states that since $(G',f)$ does not have the $2$-SIP, either $(G' \cup f',f)$ or $(G',f')$ does not have the $2$-SIP.  
    Since some atom of $G' \cup f'$ contains $f'$ and is $K_4$, Theorem \ref{thm:2-SIP} shows that $(G',f')$ does not have the $2$-SIP, which verifies this statement.  
    Lastly, observe that presence of $f'$ separates $f$ from any atom with a $K_4$ minor, while the absence of $f'$ causes $(G',f)$ to inherit the non-$2$-SIP from $(G',f')$.  
    Corollary \ref{cor:3-SIP-atom_fbm_char} states that this situation can only occur if $G \cup f$ is not an atom.  
    \qed
\end{example}

Other independently interesting tools and byproducts concern the structure of graphs with and without clique separators, which are collected in \ref{sec:atoms}.  

\subsection{Proof obstacle details}
\label{sec:proof_obstacles}

In the  first subsection, we  explore how the characterization of graph-nonedge pairs with the $d$-SIP changes from $d = 1$ to $d = 3$.  
 In subsequent subsections, we examine specific obstacles to  each direction of the proof of Theorem \ref{thm:3-sip_characterization}.  

\subsubsection{$d$-SIP characterization differences between dimensions $d \leq 3$}

% XXXXXXX do we need this here? Move a part of it to the obstacles int he forward direction whee the bracket notation is used? Or better just to the next section altogether XXXX
% For the following discussion, let $[G]$ be a minor of a graph $G$.  
% For any vertex $[u]$ of $[G]$, we write $[u]^{-1}$ to refer to the set of vertices in $G$ that map to $[u]$.  
% We write a pair of vertices $[u][v]$ of $[G]$ as $[uv]$ and define the set $[uv]^{-1}$ to be $[u]^{-1} \times [v]^{-1}$ if $[u] \neq [v]$, and $[u]^{-1}$ otherwise.  
% In the latter case, we treat $[uv]$ as a vertex of $[G]$.  
% An edge $f$ of a graph $G$ is \emph{doubled} in a minor $[G]$ if it is preserved and $[f]^{-1}$ contains at least two edges of $G$.  
% The edge $f$ is \emph{retained} in $[G]$ if it is preserved and not doubled in $[G]$, in which case $[G]$ is \emph{$f$-retaining}. 
% XXXXXXX

The $d = 1$ case, where $K_3$ is the only $1$-forbidden minor, is well-behaved in the sense that it can be checked without decomposing the graph into atoms.  
In particular, all of the following statements are equivalent for $d = 1$: 
\begin{enumerate}[(i)]
    \item $(G,f)$ does not have the $d$-SIP.
    \label{s0}

    %\item Some cycle in $G$ contains $f$ (a cycle has a 1-forbidden (induced) minor)  
   % \label{sd1}

    \item Some atom of $G \cup f$ that contains $f$ has a $d$-forbidden minor.  
    \label{s1}
    
    \item Some atom of $G \cup f$ that contains $f$ has an $f$-preserving $d$-forbidden minor.  
    \label{s2}
    
    \item Some atom of $G \cup f$ that contains $f$ has an $f$-retaining $d$-forbidden minor.  
    \label{s3}
    
    % There exists an $f$-retaining minor $[G \cup f]$ with an atom that contains $[f]$ and is some $d$-forbidden minor.  
    \item $G \cup f$ has some $f$-preserving $d$-forbidden minor.  
    \label{s4}
\end{enumerate}

The smallest graph-nonedge pair $(G,f)$ that does not have the $1$-SIP is the one where $G \cup f$ is $K_3$, which evidences the equivalence of (\ref{s0}) and (\ref{s4}).  
The equivalence of (\ref{s1})-(\ref{s4}) can be seen through the following cycle of implications: (\ref{s1}) $\implies$ (\ref{s2}) $\implies$ (\ref{s3}) $\implies$ (\ref{s4}) $\implies$ (\ref{s1}).  
The last implication follows easily from Lemma \ref{lem:path_in_graph_path_in_atom} in \ref{sec:atoms}, and the other implications have straightforward proofs.  
Using the equivalence of (\ref{s0}) and (\ref{s4}), it is easy to see that the $1$-SIP is closed under both edge deletions and contractions.  

Once we move to $d = 2$, where $K_4$ is the only $2$-forbidden minor, (\ref{s4}) is no longer equivalent to (\ref{s0}).  
For example, the pair on the left in Figure \ref{fig:not_minor closed} has the $2$-SIP even though it has an $f$-preserving $K_4$ minor.  
%Specifically, all 
%Figure \ref{fig:k4-e} shows the smallest graph-nonedge pair that does not have the $2$-SIP.  
%The minimality of this pair shows that (\ref{s0}) and (\ref{sd1}) are not equivalent.  
Theorem \ref{thm:2-SIP} and Theorem 5.2 in \cite{sitharam2010characterizing} imply (\ref{s0}) is equivalent to (\ref{s1}) - (\ref{s3}).  
While (\ref{s1}), (\ref{s2}), or (\ref{s3}) $\implies$ (\ref{s4}), (\ref{s4}) does not imply any of these. 
%For example, see the graph-nonedge pair on the left in Figure \ref{fig:not_minor closed}.  
Also, unlike when $d = 1$, the above equivalences imply that the $2$-SIP is not closed under edge deletions.  
See Example \ref{ex:2}.  

Moving to the $d = 3$ case, where $K_5$ and $K_{2,2,2}$ are the only $3$-forbidden minors, even more of the above statements are no longer equivalent to (\ref{s0}).  
Figure \ref{fig:k5-e} shows the smallest graph-nonedge pair that does not have the $3$-SIP.  
Theorem \ref{thm:3-sip_characterization} shows that (\ref{s0}) and (\ref{s2}) are equivalent.  
No other above statement is equivalent to (\ref{s0}), as we now demonstrate.  
It is easy to see that (\ref{s2}) $\implies$ both (\ref{s1}) and (\ref{s4}).  
Also, (\ref{s3}) $\implies$ (\ref{s2}) is immediate.  
However, the corresponding converses  are false.  
The graph-nonedge pairs in Figure \ref{fig:k5_k222_e_contracted}, where the $f=w_1w_2$, demonstrate that (\ref{s1}) $\notimplies$ (\ref{s2}).  
Additionally, an example similar to the $d=2$ case shows that (\ref{s4}) $\notimplies$ (\ref{s2}).  
Lastly, to see that (\ref{s2}) $\notimplies$ (\ref{s3}), note that the graph-nonedge pair $(G,f)$ in Figure \ref{fig:min_6} is such that $G \cup f$ is an atom with an $f$-preserving $K_5$ minor, but none of its $f$-retaining minors have any $K_5$ or $K_{2,2,2}$ minor.  
This pair can also be used to show that the $3$-SIP is not closed under edge deletions, using the same argument as in the $d=2$ case.  
%Let $u$ be the degree two endpoint of $f$ and let $G'$ be the graph obtained from $G$ by adding the edge such that the neighborhood of $u$ becomes a clique.  
%Lemma \ref{lem:gluing} and Theorem \ref{thm:sitharam_willoughby} can be used to show that $(G',f)$ has the $3$-SIP.  
%However, by Proposition \ref{lem:no_type_2_no_3-sip} in Section \ref{app:forward}, $(G,f)$ does not have the $3$-SIP.  

 This is a concrete illustration of what was pointed out earlier: $(G,f)$ may have the $3$-SIP although some atom of $G \cup f$ that contains $f$ has a $3$-forbidden minor, so long as $f$ is contracted in any such minor.  
On the other hand, $(G,f)$ may not have the $3$-SIP even if $f$ is not retained in any $3$-forbidden minor. 

More significantly, the latter statement results in infinitely many \emph{minimal} graph-nonedge pairs (Definition \ref{def:minimal_pair} in Section \ref{sec:forward}), with respect to contractions, that do not have the $3$-SIP.  
Figure \ref{fig:forward} illustrates minimal pairs that can be used to construct infinite families of such pairs.  
For any dimension $d \leq 2$, retaining and preserving turn out to be equivalent, hence there are finitely many minimal pairs that do not have the $d$-SIP.  
Furthermore, the $3$-SIP behaves badly when edges are deleted.  
For any graph-nonedge pair $(G',f')$  where $G'$ is obtained from $G$ by deleting some edge, any of the following four possibilities hold: $(G,f)$ has the $3$-SIP while $(G',f')$ does not; both $(G,f)$ and $(G',f')$ have the $3$-SIP; $(G,f)$ does not have the $3$-SIP while $(G',f')$ does; or both $(G,f)$ and $(G',f')$ do not have the $3$-SIP.

\subsubsection{Obstacles for the converse direction of Theorem \ref{thm:3-sip_characterization}}

We illustrate with an example. Observe that the graph-nonedge pairs $(G,w_1w_2)$ in Figure \ref{fig:k5_k222_e_contracted} are such that $G \cup w_1w_2$ is an atom with no $w_1w_2$-preserving forbidden minor.  
The converse direction of Theorem \ref{thm:3-sip_characterization} states that $(G,w_1w_2)$ has the $3$-SIP despite $G \cup w_1w_2$ having a forbidden minor.  
To see this, observe that $\{w_3,w_4\}$ and $\{w_4,w_5\}$ are clique separators of $G$ that separate $w_1$ and $w_2$.  
The above properties of $G \cup w_1w_2$ ensure the existence of such separators and force the subgraphs between them to be $3$-connected partial $3$-trees.  
Figure \ref{fig:k4-e} in Example 1, and Figure \ref{fig:converse_4} illustrate a more general example of a graph-nonedge pair with these properties.  
We must prove that, for any squared edge-length map $\ell$, the CCS $\Omega^3_{w_1w_2}(G,\ell)$ is connected.  
The only apparent obstacle to this is determining how the configuration space of the graph between $w_3w_4$ and $w_5w_6$ influences $\Omega^3_{w_1w_2}(G,\ell)$.  
In general, this configuration space may be disconnected and any two of its connected components may have differing sets of attainable length-tuples for $(w_2w_6,w_4w_5)$, which could possibly force $\Omega^3_{w_1w_2}(G,\ell)$ to be disconnected.  
We circumvent this obstacle by demonstrating that each connected component of the CS $\mathcal{C}^3(G,\ell)$ has the same set of attainable lengths for $w_1w_5$ and also for $w_1w_6$ (Lemma \ref{lem:ux2_uy2_3-covering} in Section \ref{sec:reverse}). 
See Example 1,  Figures \ref{fig:k3-e}-\ref{fig:k4-e}.
This is enough to show that the CCS $\Omega^3_{\{w_1w_5,w_1w_6\}}(G,\ell)$ is connected, i.e., $(G,\{w_1w_5,w_1w_6\})$ is $3$-Cayley-connected (Proposition \ref{prop:x2y2_3-Cayley-connected} in Section \ref{sec:reverse}).  
Combining this with the observation that $(G \cup \{w_1w_5,w_1w_6\}, w_1w_2)$ has the $3$-SIP allows us to apply Lemma \ref{lem:tool} to show that $(G, w_1w_2)$ has the $3$-SIP.  
The above arguments are detailed in Section \ref{sec:reverse}.

\subsubsection{Obstacles for the forward direction of Theorem \ref{thm:3-sip_characterization}}

We illustrate with an example. Let $(G,f)$ be the graph-nonedge pair such that $G \cup f$ is the graph in Figure \ref{fig:min_4} with $f$ as the top-most green edge, and observe that $G \cup f$ is an atom with an $f$-preserving $K_5$ minor $M$, which we can ensure is an induced minor.  
The forward direction of Theorem \ref{thm:3-sip_characterization} states that $(G,f)$ does not have the $3$-SIP despite $f$ not being retained in $M$, but if $f$ had been retained this would follow from Observation \ref{obs:3-sip-retained} in Section \ref{sec:results}.  
% To see this, first note that Lemma \ref{lem:no_type_2_pairs_no_3-sip} in Section \ref{sec:forward} shows that, for any edge $f'$ of $M$, $(M \setminus f', f')$ does not have the $3$-SIP.  
% Hence, had $f$ been retained in $M$, Lemma \ref{lem:cm_non-sip} would show that $(G,f)$ does not have the $3$-SIP.  
Hence, the primary obstacle is arguing that $(G,f)$ does not have the $3$-SIP without contracting down to $M$.  
To achieve this,  it is not difficult to see that any minor pair of $(G,f)$ in Figure \ref{fig:min_4}, other than $(G,f)$ itself, has the $3$-SIP, and so $(G,f)$ is minimal in this sense.  
In general, we can always reduce to such a minimal pair.  
Then, we propagate the property of not having the $3$-SIP up through the green dashed nonedges.  
Another major obstacle is enumerating all possible subgraphs above the top-most green dashed nonedge, i.e., the subgraph in the black dashed box.  
This step is necessary to complete the propagation, and is completed by checking numerous cases in Lemma \ref{prop:expC_3or4} in \ref{sec:proof_top-level_H'_cases}.

\section{Formal concepts and notation for analyzing graph minors.}
\label{sec:minor-notation}

Here we present concepts and notation for working with minors, and  redefine previous concepts  more formally.  
Refer to Figure \ref{fig:minors}.  
A graph $M$ is a \emph{minor} of $G$ if it is isomorphic to some graph obtained from $G$ via some sequence of edge deletions and contractions.  
$M$ is \emph{induced} if it can be obtained using only contractions.  
A \emph{minor map} $\pi: V(G) \rightarrow V(M)$ is a surjection that satisfies $\pi(u) = \pi(v)$ only if $u$ and $v$ are connected by a path in $G$ and, for any edge $\pi(u)\pi(v)$ of $M$, there is an edge in $G$ with one endpoint in  the set $\pi(u)^{-1}$ and the other endpoint in the set $\pi(v)^{-1}$.   
This map induces a map $\pi': V(G) \times V(G) \rightarrow V(M) \cup (V(M) \times V(M))$ that satisfies $\pi'(uv) = \pi(u)$ if $\pi(u) = \pi(v)$, and $\pi'(uv) = \pi(u)\pi(v)$ otherwise.  
To reach $M$ from $G$ using $\pi$, we contract an edge $uv$ of $G$ if $\pi(u) = \pi(v)$, and we delete $uv$ if $\pi(u)\pi(v)$ is a nonedge of $M$.  
A few simple but important observations:  two minor maps that induce the same edge deletions and contractions on $G$ yield the same minor $M$; $\pi$ defines a unique induced minor of $G$ containing $M$ as a spanning subgraph; and  $M$ can be recovered given $\pi$ and the edges of $G$ to be deleted.
% Throughout this paper, the edge deletions used to obtain non-induced minors are either clear from context or not important, and so we never explicitly give them along with a minor map.  

\begin{figure}[htb]
   \centering
   \begin{subfigure}[t]{0.49\linewidth}
       \centering
       \scalebox{0.65}{
            \begin{tikzpicture}
                % [u_1]^{-1}
                \node[circle, draw=violet, fill=violet, inner sep=0pt, minimum size=4pt, label = left:{{\color{violet}$u_1$}}] (u1) at (0,0) {};
    
                \node[circle, draw=violet, fill=violet, inner sep=0pt, minimum size=4pt] (u12) at (0,-0.5) {};
    
                \node[circle, draw=violet, fill=violet, inner sep=0pt, minimum size=4pt] (u13) at (0.5,-1) {};
    
                \node[circle, draw=violet, fill=violet, inner sep=0pt, minimum size=4pt] (u14) at (0.5,0.5) {};
    
                \draw[-,color=violet] (u1) to (u12);
                \draw[-,color=violet] (u12) to (u13);
                \draw[-,color=violet] (u13) to (u14);
                \draw[-,color=violet] (u14) to (u1);
    
                \node[draw,minimum width=1.5cm,minimum height=2cm, label = left:{$[u_1]^{-1}$}] (u1-1) at (0.1,-0.25) {};
    
                % [u_2]^{-1}
                \node[circle, draw=violet, fill=violet, inner sep=0pt, minimum size=4pt, label = right:{{\color{violet}$u_2$}}] (u2) at (4,-1) {};
    
                \node[circle, draw=violet, fill=violet, inner sep=0pt, minimum size=4pt] (u22) at (3.5,-0.5) {};
    
                \draw[-,color=violet] (u2) to (u22);
    
                \node[draw,minimum width=1.5cm,minimum height=1cm, label = right:{$[u_2]^{-1}$}] (u2-1) at (4,-0.75) {};
    
                % [u_5]^{-1}
                \node[circle, draw=red, fill=red, inner sep=0pt, minimum size=4pt, label = right:{\color{red} $u_5$}] (u5) at (4,-2) {};
    
                \node[circle, draw, fill, inner sep=0pt, minimum size=4pt] (u52) at (3.5,-2) {};
    
                \node[circle, draw, fill, inner sep=0pt, minimum size=4pt] (u53) at (3.5,-2.5) {};
    
                \draw[-] (u5) to (u52);
                \draw[-,] (u52) to (u53);
                \draw[-] (u53) to (u5);
    
                \node[draw,minimum width=1.5cm,minimum height=1cm, label = right:{$[u_5]^{-1}$}] (u5-1) at (4,-2.25) {};
    
                % [u_6]^{-1}
                \node[circle, draw=red, fill=red, inner sep=0pt, minimum size=4pt, label = right:{\color{red} $u_6$}] (u6) at (3,-4) {};
    
                \node[draw,minimum width=1cm,minimum height=1cm, label = right:{$[u_6]^{-1}$}] (u6-1) at (3.25,-4) {};
    
                % [u_4]^{-1}
                \node[circle, draw=red, fill=red, inner sep=0pt, minimum size=4pt, label = left:{\color{red} $u_4$}] (u4) at (1.5,-4) {};
    
                \node[draw,minimum width=1cm,minimum height=1cm, label = left:{$[u_4]^{-1}$}] (u4-1) at (1.25,-4) {};
    
                % [u_3]^{-1}
                \node[circle, draw=red, fill=red, inner sep=0pt, minimum size=4pt, label = left:{\color{red} $u_3$}] (u3) at (0,-2) {};
    
                \node[circle, draw, fill, inner sep=0pt, minimum size=4pt] (u32) at (0.5,-2.25) {};
    
                \draw[-] (u3) to (u32);
    
                \node[draw,minimum width=1.5cm,minimum height=1cm, label = left:{$[u_3]^{-1}$}] (u3-1) at (0.1,-2.25) {};
    
                % between edges
                % from [u_1]^{-1}
                \draw[-,color=violet] (u14) to node[label = above:{\color{violet} $J^{-1}$}] {} (u22);
                \draw[-] (u12) to (u3);
                \draw[-] (u13) to (u52);
    
                % from [u_2]^{-1}
                \draw[-] (u2) to (u52);
                \draw[-] (u22) to (u52);
    
                % from [u_3]^{-1}
                \draw[-,color=red] (u3) to (u4);
                \draw[-] (u32) to (u4);
                \draw[-] (u3) to (u52);
                \draw[-] (u32) to (u6);
    
                % from [u_4]^{-1}
                \draw[-] (u4) to (u53);
                \draw[-,color=red] (u4) to node[label = below:{\color{red} $H$}] {} (u6);
    
                % from [u_5]^{-1}
                \draw[-,color=red] (u5) to (u6);
           \end{tikzpicture}
       }
       \caption{}
       \label{fig:minor_1}
   \end{subfigure}
   \begin{subfigure}[t]{0.49\linewidth}
       \centering
       \pgfdeclarelayer{0}
        \pgfsetlayers{0,main}
        \begin{tikzpicture}
            \node[circle, draw=violet, fill=violet, inner sep=0pt, minimum size=4pt, label = left:{{\color{violet}$[u_1]$}}] (u1) at (1,0) {};

            \node[circle, draw=violet, fill=violet, inner sep=0pt, minimum size=4pt, label = right:{{\color{violet}$[u_2]$}}] (u2) at (3.5,0) {};

            \node[circle, draw=red, fill=red, inner sep=0pt, minimum size=4pt, label = left:{{\color{red}$[u_3]$}}] (u3) at (1.5,-1) {};

            \node[circle, draw=red, fill=red, inner sep=0pt, minimum size=4pt, label = left:{{\color{red}$[u_4]$}}] (u4) at (1.5,-2) {};

            \node[circle, draw=red, fill=red, inner sep=0pt, minimum size=4pt, label = right:{{\color{red}$[u_5]$}}] (u5) at (3,-1) {};

            \node[circle, draw=red, fill=red, inner sep=0pt, minimum size=4pt, label = right:{{\color{red}$[u_6]$}}] (u6) at (3,-2) {};

            \begin{pgfonlayer}{0}
                \draw[-,color=violet] (u1) to node[label = above:{\color{violet} $J$}] {} (u2);
                \draw[-] (u1) to (u3);
                \draw[-] (u1) to (u4);
    
                \draw[-] (u2) to (u5);
                \draw[-] (u2) to (u6);
    
                \draw[-,color=red] (u3) to (u4);
                \draw[-] (u3) to (u5);
                \draw[-] (u3) to (u6);
    
                \draw[-] (u4) to (u5);
                \draw[-,color=red] (u4) to node[label = below:{\color{red} $[H]$}] {} (u6);
    
                \draw[-,color=red] (u5) to (u6);
            \end{pgfonlayer}
       \end{tikzpicture}
       \caption{}
       \label{fig:minor_2}
   \end{subfigure}
   \caption{Illustrations of the concepts in this section and in Sections \ref{sec:results} and \ref{sec:forward}.  
   (a) a graph $G$ and a rooted minor $[G]$ defined using the minor map between $G$ and its minor in (b).  
   The vertices of $[G]$ are shown as black boxes.  
   Note that $H$ is preserved but neither retained nor weakly retained, and $H$ is a proper subgraph of $[H]^{-1}$.  
   % $[G]$ is obtained from $G$ by contracting each edge whose endpoints are contained in some set $[u_i]^{-1}$ and deleting the edge of $G$ connecting the sets $[u_1]^{-1}$ and $[u_5]^{-1}$.  
   % The edges and nonedges within each set $[u_i]^{-1}$ are contracted in $[G]$.  
   % The nonedge $u_1u_2$ is doubled in $[G]$ since $[u_1] \neq [u_2]$ and $[u_1u_2]^{-1}$ contains an edge of $G$.  
   % The nonedge $u_2u_3$ is retained in $[G]$ since $[u_2] \neq [u_3]$ and $[u_2u_3]^{-1}$ does not contain an edge.
   % Similarly, the edges $u_3u_4$ and $u_5u_6$ are doubled and retained in $[G]$, respectively.  
   % $H$ is the induced subgraph of $G$ shown in red and $[H]$ is the subgraph of $[G]$ induced by $[V(H)]$ shown in red.  
   % $J$ is the induced subgraph of $[G]$ shown in purple and $J^{-1}$ is the subgraph of $G$ induced by $\bigcup_{v \in V(J)} v^{-1}$ shown in purple.  
   % Note that $H$ is a proper subgraph of $[H]^{-1}$.  
   % $H$ is preserved in $[G]$ since each of its edges is preserved in $[G]$, but $H$ is not weakly retained in $[G]$ since one of its nonedges is doubled in $[G]$.  
   }
   \label{fig:minors}
\end{figure}

The graph $[G]$ whose vertex-set is $\{\pi^{-1}(u)\}_{u \in V(M)}$ and whose edge-set is $\{\pi^{-1}(u)\pi^{-1}(v) | uv \in E(M)\}$ is a \emph{rooted} minor of $G$.  
This definition is equivalent to the one given at the beginning of Section \ref{sec:results}, except here vertices correspond to vertex sets and not the connected subgraphs they induce.  
Unless otherwise specified, all minors are rooted.  
We use subscripts to distinguish between distinct   minors, e.g., $[G]_1$ and $[G]_2$ are two minors of $G$.  
The vertex of $[G]$ containing a vertex $u$ of $G$ is written as $[u]$, or as $[u]^{-1}$ when we want to treat it as a set.  
A pair of vertices $[u][v]$ of $[G]$ is written as $[uv]$, and $[uv]^{-1}$ is the set $[u]^{-1} \times [v]^{-1}$ if $[u] \neq [v]$, and $[u]^{-1}$ otherwise.  
In the latter case, $[uv]$ is treated as a vertex of $[G]$.  
 When the context is clear, it is convenient to refer to a vertex, edge, or nonedge  of $[G]$ without the bracket notation, in which case the brackets are implicit, and so $x^{-1}$ denotes the set $[x]^{-1}$.  
For any distinct vertices $u,v \in V(G)$, $uv$ is \emph{contracted in $[G]$} if $[u]=[v]$, $uv$ is \emph{preserved in $[G]$} if $[u] \neq [v]$, a nonedge (resp. edge) $uv$ is \emph{doubled in $[G]$} if it is preserved and $[uv]^{-1}$ contains at least one edge (resp. two edges) of $G$, and $uv$ is \emph{retained in $[G]$} if it is preserved and not doubled.  
$[G]$ is \emph{$uv$-preserving (resp. retaining)} if $uv$ is preserved (resp. retained) in $[G]$.  

For any subset $X$ of vertices, edges, and nonedges of $G$, define the set $[X] = \{[x] : x \in X\}$.  
For any induced subgraph $H$ of $G$, $[H]$ is the subgraph of $[G]$ induced by $[V(H)]$.  
For any induced subgraph $J$ of $[G]$, $J^{-1}$ is the subgraph of $G$ induced by $\bigcup_{u \in J} u^{-1}$.  
Note that $H$ is always a subgraph of $[H]^{-1}$, and this containment can be proper.  
$H$ is \emph{preserved} (resp. retained) in $[G]$ if each edge and nonedge of $H$ is preserved (resp. retained) in $[G]$, and $H$ is \emph{weakly retained} in $[G]$ if $H$ is preserved in $[G]$ and each nonedge of $H$ is retained in $[G]$.  

% Let $X,Y \subseteq E(G)$ be the sets of all contracted and deleted edges in $M$, respectively.  
% It is a well-known fact that any minor $M'$ of $G$ for which $X$ and $Y$ are the sets of all contracted and deleted edges in $M'$, respectively, is isomorphic to $M$.  
% We will work with such a minor $[G]$ (of $G$), defined as follows.  
% Let two vertices $u,v \in V(G)$ be equivalent if each edge along some path in $G$ between $u$ and $v$ is contained in $X$.  
% Also, let $[u]^{-1}$ denote the equivalence class containing $u$.  
% The vertex set $V([G])$ is the set of all such equivalence classes.  
% For clarity, if we are referring to $[u]^{-1}$ as a vertex, then we will write it as $[u]$.  
% Additionally, we write the pair $[u][v]$ as $[uv]$ and define the set $[uv]^{-1}$ to be $[u]^{-1} \times [v]^{-1}$ if $[u] \neq [v]$, and $[u]^{-1}$ otherwise.  
% In the latter case, we call $[uv]$ a vertex of $[G]$ and treat it as such.  
% The edge set $E([G])$ contains $[uv]$ if and only if $[u] \neq [v]$ and $[uv]^{-1}$ contains some edge in $E(G) \setminus Y$.  
% Sometimes it will be convenient to refer to a vertex $x$ or an edge or nonedge $f$ of $[G]$ without making use of the bracket notation, in which case the brackets are implicit, and so $x^{-1}$ and $f^{-1}$ denote the sets $[x]^{-1}$ and $[f]^{-1}$, respectively.  
% We use subscripts to differentiate minors denoted with brackets, e.g., $[G]_1$ and $[G]_2$ are two minors of $G$.  

Finally, if $[G]$ is isomorphic to some graph $M$, we say $[G]$ \emph{is} $M$.  
Let $[x]$ and $w$ be either vertices, edges, or nonedges of $[G]$ and $M$, respectively.  
Also, let $[H]$ and $K$ be subgraphs of $[G]$ and $M$, respectively.  
We say $[G]$ is $M$ \emph{with $[x]$ as $w$} and \emph{with $[H]$ as $K$} if some isomorphism from $[G]$ to $M$ sends $[x]$ to $w$ and sends $[H]$ to $K$.  
Sometimes we write $[w]$ and $[K]$ when referring to $[x]$ and $[H]$, respectively.

\section{Proof of the converse direction of the main Theorem \ref{thm:3-sip_characterization}}
\label{sec:reverse}

We outline the argument using Figure \ref{fig:converse_sketch} before filling in the details.  
See also Figure \ref{fig:converse-outline} for a proof diagram.  
From here on,  we  shorten ``$3$-forbidden minor'' to ``forbidden minor.''  
Since the  notation $(G,f)$ is used to refer to general graph-nonedge   pairs in the inductive and other steps of the proof, we designate  the notation
$(G_{\star},f_{\star})$ throughout the proof for the particular pair in Theorem \ref{thm:3-sip_characterization}.
Assume no atom of $G_{\star} \cup f_{\star}$ that contains $f_{\star}$ has an $f_{\star}$-preserving forbidden minor.  
We prove by induction on the number $n$ of vertices in $G_{\star}$ that $(G_{\star},f_{\star})$ has the $3$-SIP.  
For $n \leq 5$, all graphs with this property are $3$-flattenable, and thus applying Theorem \ref{thm:sitharam_willoughby} proves the claim for these base cases. 
In the inductive step,   Corollary \ref{cor:atom_gluing} (a repeatedly used tool of independent interest proved in Section \ref{sec:tool_lemmas}), along with the same theorem allow us to assume that $G_{\star} \cup f_{\star}$ is an atom that has a forbidden minor.  
It is important to note that, as pointed out in Section \ref{sec:proof_obstacles}, the CS is typically disconnected for almost any linkage of $G_{\star}$.  
However, we show that $f_{\star}$ attains the same lengths for each of these connected components.  

\begin{figure}[htb]
    \centering
    \begin{tikzpicture}
        \node[draw] (converse) at (0,0) {Converse direction of Theorem \ref{thm:3-sip_characterization}};
        %-----------------------------------------------------
        \node[draw] (p3) at (0,-4.5) {Proposition \ref{lem:winged-graph_cut_vertices}};

        \node[draw] (l9) at (0,-6) {Lemma \ref{lem:K5_2-separator_case_1}};
        %-----------------------------------------------------
        \node[draw] (p2) at (-3,-4.5) {Proposition \ref{lem:e_not_st}};

        \node[draw] (l6) at (-3,-6) {Lemma \ref{lem:wing-tips}};
        %-----------------------------------------------------
        \node[draw] (p1) at (-6,-6) {Proposition \ref{prop:G_is_winged_graph}};

        \node[draw] (l4) at (-7.5,-7.5) {Lemma \ref{lem:G_not_k5_k222_u4v0}};

        \node[draw] (l5) at (-4.5,-7.5) {Lemma \ref{lem:k5_k222_not_wing_fm_e_not_contracted}};
        %-----------------------------------------------------
        \node[draw] (p4) at (3,-1.5) {Proposition \ref{prop:chained_minor}};

        \node[draw] (l10) at (1.5,-3) {Lemma \ref{prop:hinge_edge}};

        \node[draw] (l11) at (1.5,-7.5) {Lemma \ref{lem:pins_have_3-sip}};
        %-----------------------------------------------------
        \node[draw] (p5) at (6,-1.5) {Proposition \ref{prop:x2y2_3-Cayley-connected}};
        %=====================================================
        \node[draw,] (l12) at (4.5,-4.5) {Lemma \ref{lem:links_3-connected}};

        \node[draw] (l46) at (4.5,-7.5) {Lemma \ref{lem:pin_pairs_contracted}};
        %=====================================================
        \node[draw, align = left] (l13) at (7.5,-7.5) {
        Lemma \ref{lem:ux2_uy2_3-covering}\\
        {\small \textit{Lemma \ref{lem:pin_pairs_partial_3-tree}}}\\
        {\small \textit{Theorems \ref{prop:3-connected_partial_3-tree_star_lemma}, \ref{prop:partial_3-tree_3-reflection}}}
        };
         %-----------------------------------------------------
         % dependencies
         % converse 
         \draw[->] (p1)--(converse);
         \draw[->] (p2)--(converse);
         \draw[->] (p3)--(converse);
         \draw[->] (p4)--(converse);
         \draw[->] (p5)--(converse);

         %p1
         \draw[->] (l4)--(p1);
         \draw[->] (l5)--(p1);

         %p2
         \draw[->] (l6)--(p2);
         \draw[->] (p1)--(p2);

         %l6
         \draw[->] (l5)--(l6);

         %p3 
         \draw[->] (l9)--(p3);

         %l9
         \draw[->] (l5)--(l9);

         %p4
         \draw[->] (l10)--(p4);
         \draw[->] (l11)--(p4);

         %l10
         \draw[->] (l6) to [out=40,in=-170] (l10);
         \draw[->] (p3)--(l10);

         %p5
         \draw[->] (l12)--(p5);
         \draw[->] (l13)--(p5);

         %l12
         \draw[->] (p1) to [out=8,in=-175] (l12);
         \draw[->] (l4) to [out=-15,in=-150] (l12);
         \draw[->] (l6)--(l12);
         \draw[->] (l46)--(l12);
    \end{tikzpicture}
    \caption{Proof outline for the converse direction of Theorem \ref{thm:3-sip_characterization}.  
    See the discussion in this section.  }
    \label{fig:converse-outline}
\end{figure}

The remainder of the proof can be split into two parts. 
In the first part, a series of steps illustrated in Figure \ref{fig:converse_sketch}, Propositions \ref{prop:G_is_winged_graph}, \ref{lem:e_not_st}, and \ref{prop:chained_minor} reduce $G_{\star} \cup f_{\star}$ to a graph that has a special structure forced by the strong property that $f_{\star}$ is contracted in every forbidden minor in every atom of $G_{\star} \cup f_{\star}$  that contains $f_{\star}$.  
Briefly, the structure can be described as follows, and is depicted in Figure \ref{fig:converse_4}.  
There is a maximal sequence of size-two clique separators $\{x_i,y_i\}$ of $G_{\star}$ that each separate the endpoints of $f_{\star}$.  
Furthermore, there are no other size-two separators of $G_{\star}$ and the graph between any consecutive pair of these separators is a $3$-connected partial $3$-tree.  

\begin{figure}[htb]
    \centering
    \begin{subfigure}{0.24\linewidth}
        \centering
        \includegraphics[width=0.8\linewidth]{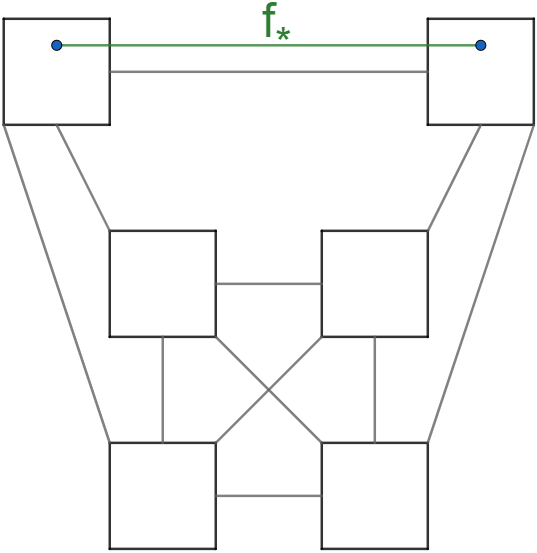}
        \subcaption{}
        \label{fig:converse_1}
    \end{subfigure}
    \begin{subfigure}{0.24\linewidth}
        \centering
        \includegraphics[width=0.7\linewidth]{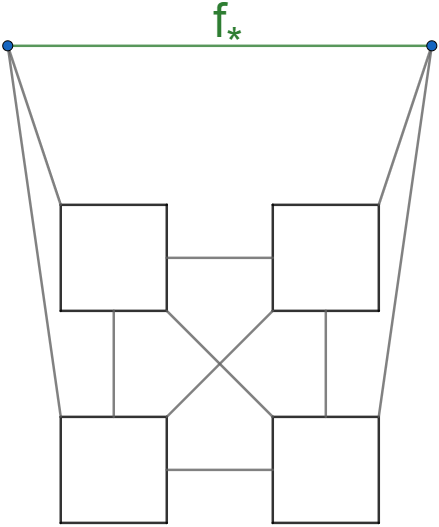}
        \subcaption{}
        \label{fig:converse_2}
    \end{subfigure}
    \begin{subfigure}{0.24\linewidth}
        \centering
        \includegraphics[width=0.7\linewidth]{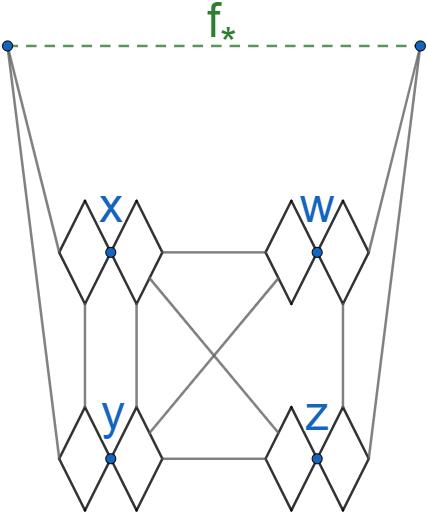}
        \subcaption{}
        \label{fig:converse_3}
    \end{subfigure}
    \begin{subfigure}{0.24\linewidth}
        \centering
        \includegraphics[width=\linewidth]{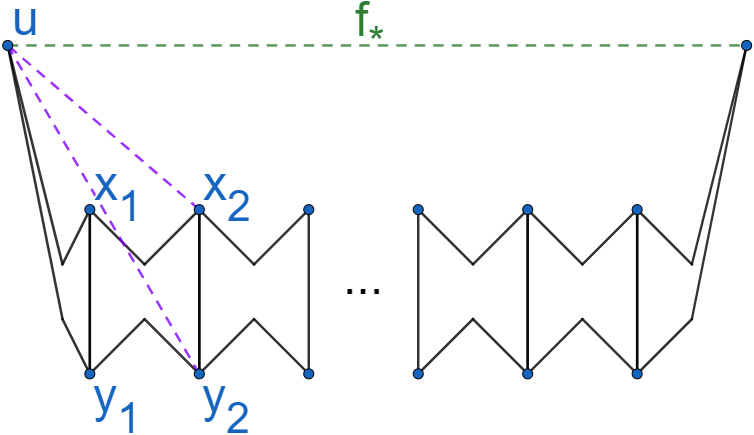}
        \subcaption{}
        \label{fig:converse_4}
    \end{subfigure}
    \caption{(a) An $f_{\star}$-winged graph minor $[G_{\star} \cup f_{\star}]$.  
    (b) $[G_{\star} \cup f_{\star}]$ with singleton wing tips.  
    (c) Pins $xy$ and $wz$ of $[G_{\star} \cup f_{\star}]$.  
    (d) $[G_{\star} \cup f_{\star}]$ is chained and $(G_{\star}, \{ux_2,uy_2\})$ is $3$-Cayley-connected.  
    See the discussions in this section and Section \ref{sec:proof_obstacles}.  }
    \label{fig:converse_sketch}
\end{figure}

The second part of the proof exploits the special structure shown in the first part of the proof. 
% This second part of the proof uses a further induction on the number of separators. 
Specifically, consider the nonedges $ux_2$ and $uy_2$ of $G_{\star}$ shown in Figure \ref{fig:converse_4}.  
Observe that $\{u,x_2,y_2\}$ is a clique separator of $G_{\star} \cup \{ux_2,uy_2\}$.  
Let $H$ be the $\{u,x_2,y_2\}$-component of $G_{\star} \cup \{ux_2,uy_2\}$ that contains $f_{\star}$.  
It is not hard to argue that $H \cup f_{\star}$ has no $f_{\star}$-preserving forbidden minor since $G_{\star} \cup f_{\star}$ has no such minor.  
Hence, the inductive hypothesis tells us that $(H,f_{\star})$ has the $3$-SIP, which,   
 when combined with the repeatedly used tool Lemma \ref{lem:gluing}, implies that $(G_{\star} \cup \{ux_2,uy_2\},f_{\star})$ has the $3$-SIP.  
 If $(G_{\star},\{ux_2,uy_2\})$ is $3$-Cayley-connected, then Lemma \ref{lem:tool}, another repeatedly used tool, applies to show that $(G_{\star},f_{\star})$ has the $3$-SIP.  

The proof is completed via Proposition \ref{prop:x2y2_3-Cayley-connected}, which shows that $(G_{\star},\{ux_2,uy_2\})$ is $3$-Cayley-connected.  
The key ingredient, Lemma \ref{lem:ux2_uy2_3-covering} overcomes the obstacle sketched in Section \ref{sec:proof_obstacles}: although each $3$-connected subgraph of $G_{\star}$ between the consecutive clique separators (see Figure \ref{fig:converse_4}) could force disconnectedness of the overall CS, the connected components can be viewed as reflections of the realizations of these  sublinkages, each of which attains the same  set of lengths for the nonedge $ux_i$, and similarly for  $uy_i$.  
See Example 1, Figures \ref{fig:k3-e}-\ref{fig:k4-e} for illustration.
This is sufficient to show connectedness of the set of length-tuples for ${ux_i, uy_i}$.
The same statement holds for the pair of nonedges $\{vx_i, vy_i\}$ for any $i$, resulting in the same interval of lengths for $f_{\star}$.  
This completes the proof sketch of the converse direction of the main Theorem \ref{thm:3-sip_characterization}.   

\medskip\noindent
 Before diving into the first part of the proof, which reduces $G_{\star} \cup f_{\star}$ to a graph 
 that has a special structure, we first outline the flow of the 3 constituent Propositions \ref{prop:G_is_winged_graph}, \ref{lem:e_not_st}, and \ref{prop:chained_minor}. 
 Given a graph $G$ and one of its edges $f$, a minor $[G]_1$ is an \emph{$f$-winged graph minor} if it is
% isomorphism between some minor $[G]_1$ and 
one of the graphs in Figure \ref{fig:k5_k222_e_contracted} with $[f]$ as $w_1w_2$.  
% maps $[e]$ to $w_1w_2$.  
Since atoms are $2$-connected, Proposition \ref{prop:G_is_winged_graph}, below,
% stated below and proved in \ref{sec:proofs_G_is}, 
shows that there exists an $f_{\star}$-winged graph minor   
  $[G_{\star} \cup f_{\star}]$, e.g.\ as in  Figure \ref{fig:converse_1}.  
Let $f_{\star} = uv$.  
Subsequently, Proposition \ref{lem:e_not_st} allows us to assume that $[u]^{-1} = \{u\}$ and $[v]^{-1} = \{v\}$, as in Figure \ref{fig:converse_2}.  
Then, with this assumption, Proposition \ref{lem:winged-graph_cut_vertices} shows the existence of a minimal-sized set of vertices $\{x,y\}$ whose deletion from $G_{\star}$ separates $u$ from $v$, as in Figure Figure \ref{fig:converse_3}.  
Symmetry implies the existence of $\{w,z\}$ with the same properties.  
Thereafter, Proposition \ref{prop:chained_minor} allows us to assume that $xy$ and $wz$ are edges of $G_{\star}$.  
Hence, some $\{x,y\}$-component of $G_{\star}$ must contain $w$ and $z$, and so $G_{\star}$ can be decomposed by a maximal set $\{x_iy_i\}$ of these separators, as in Figure \ref{fig:converse_4}.  
\begin{figure}[htb]
    \centering
    \includegraphics[width=0.3\linewidth]{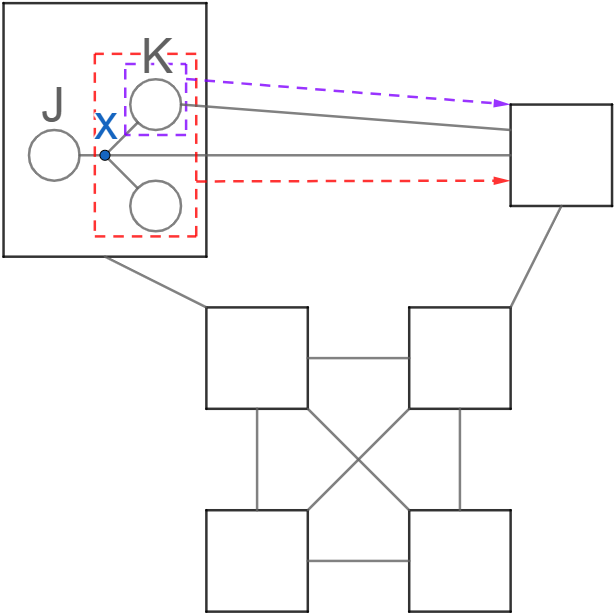}
    \caption{In red, a vertex exchange that exchanges $x$ and fixes $J$.  
    In purple, a component exchange that exchanges $K$.  
    See the discussion in this section.}
    \label{fig:exchange}
\end{figure}

Proposition \ref{prop:G_is_winged_graph} is proved via Lemmas \ref{lem:G_not_k5_k222_u4v0} and \ref{lem:k5_k222_not_wing_fm_e_not_contracted}, below.  
The proof of Lemma \ref{lem:G_not_k5_k222_u4v0} requires the following operation that transforms one minor $[G]_1$ of a graph $G$ into another minor $[G]_2$.  
Refer to Figure \ref{fig:exchange}.  
Consider an edge $[ab]_1$ of $[G]_1$, the subgraph $H$ of $G$ induced by the set $[a]_1^{-1}$, a vertex $x \in [a]_1^{-1}$ with a neighbor in the set $[b]_1^{-1}$, and a connected component $J$ of $H \setminus x$. 
The minor $[G]_2$ obtained from $[G]_1$ by replacing $[a]_1^{-1}$ with $V(J)$ and $[b]_1^{-1}$ with $[b]_1^{-1} \cup (V(H) \setminus V(J))$ is said to be obtained via the \emph{vertex exchange from $[a]_1^{-1}$ to $[b]_1^{-1}$ that exchanges $x$ and fixes $J$}.  
\begin{figure}[htb]
    \centering
    \begin{subfigure}{0.49\textwidth}
        \centering
        \includegraphics[width = 0.4\textwidth]{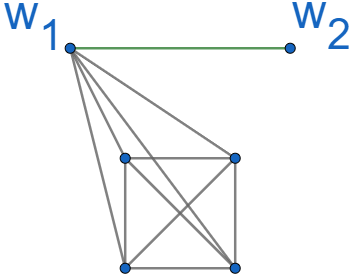}
\subcaption{}
        \label{fig:k5_not_wing_u4_v0}
    \end{subfigure}
    \begin{subfigure}{0.49\textwidth}
        \centering
        \includegraphics[width =0.3\textwidth]{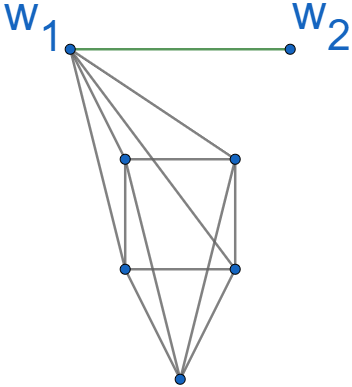}
\subcaption{}
        \label{fig:k222_not_wing_u4_v0}
    \end{subfigure}
    \caption{Graphs in Lemma \ref{lem:G_not_k5_k222_u4v0}.  
    See also the proofs of Proposition \ref{prop:G_is_winged_graph} in this section, Lemma \ref{lem:C'_5_H'_wing} in Section \ref{sec:prop-8}, and Lemma \ref{lem:C'_degree_1_reducible_0} in \ref{sec:proof_final_top_graphs}.  }
    \label{fig:k5_k222_e_contracted_degree_1}
\end{figure}

\begin{lemma}[Transforming a minor in Figure \ref{fig:k5_k222_e_contracted_degree_1} into one in Figures \ref{fig:k5_k222_e_contracted} or \ref{fig:k5_not_wing}]
%isomorphic to the graphs in Figures \ref{fig:k5_k222_e_contracted} or \ref{fig:k5_not_wing}]
\label{lem:G_not_k5_k222_u4v0}
    Let $G$ be a $2$-connected graph and $e$ be one of its edges.  
    If some minor $[G]$ is the graph in Figure \ref{fig:k5_not_wing_u4_v0} (resp. Figure  \ref{fig:k222_not_wing_u4_v0}) with $[e]$ as $w_1w_2$, then some minor $[G]_1$ is one of the graphs in Figures \ref{fig:k5_wing}, \ref{fig:k5_not_wing_u4_vgeq1}, or \ref{fig:k5_not_wing_u3} 
    (resp. Figures \ref{fig:k222_wing} or \ref{fig:k222_not_wing_u4_vgeq1}-\ref{fig:k222_not_wing_u2_diagonals}) with $[e]_1$ as $w_1w_2$.  
    % isomorphism between some minor $[G]_1$ and the graph in Figure \ref{fig:k5_not_wing_u4_v0} (resp. Figure  \ref{fig:k222_not_wing_u4_v0}) maps $[e]_1$ to $w_1w_2$, then some isomorphism between some minor $[G]_2$ and one of the graphs in Figures \ref{fig:k5_wing} (resp. Figure \ref{fig:k222_wing}) or \ref{fig:k5_not_wing} maps $[e]_2$ to $w_1w_2$.  
    Furthermore, let $\{[w_i]\}$ be the vertex set of $[G]$.  
    Then, $[G]$ and $[G]_1$ can be chosen so that, for any $i \notin \{1,2\}$, $[w_i]^{-1} = [w_i]_1^{-1}$.  
\end{lemma}
\begin{proof}
    Let $e=uv$ and $[G]$ be a minor that is the graph in Figure \ref{fig:k5_not_wing_u4_v0} (resp. Figure  \ref{fig:k222_not_wing_u4_v0}) with $[u]$ as $w_1$ and $[v]$ as $w_2$.  
    Additionally, let $[u]^{-1}$ have minimum size over all such minors.  
    %such that (i) some isomorphism between it and the graph in Figure \ref{fig:k5_not_wing_u4_v0} (resp. Figure  \ref{fig:k222_not_wing_u4_v0}) maps $[u]$ to $w_1$ and $[v]$ to $w_2$ and (ii) $[u]^{-1}$ has minimum size over all minors satisfying (i).  
    Let $H$ be the subgraph of $G$ induced by $[u]^{-1}$ and $x$ be any vertex in $[u]^{-1} \setminus \{u\}$ with a neighbor in $[v]^{-1}$, which exists since $G$ is $2$-connected.  
    Also, let $J$ be the connected component of $H \setminus x$ that contains $u$.  
    First, we show how to choose $x$ so that $J$ contains some vertex with a neighbor in $G \setminus ([u]^{-1} \cup [v]^{-1})$.  
    If $J$ contains such a vertex, then our current choice for $x$ works.  
    Otherwise, $G$ being $2$-connected gives a path in $G \setminus x$ between $u$ and some vertex in $G \setminus ([u]^{-1} \cup [v]^{-1})$.  
    Hence, some connected component $K$ of $H \setminus x$ contains both a vertex $y$ with a neighbor in $[v]^{-1}$ and a vertex with a neighbor in $G \setminus ([u]^{-1} \cup [v]^{-1})$.  
    Note that the connected component of $H \setminus y$ that contains $u$ has size strictly greater than $J$, since it contains both $J$ and $x$.  
    Hence, we can set $x$ to be $y$ and repeat this argument until $J$ has the desired property.  
    Consider the minor $[G]_1$ obtained via the vertex exchange from $[u]^{-1}$ to $[v]^{-1}$ that exchanges $x$ and fixes $J$.  
    This operation along with our minimality assumption on $[u]^{-1}$ and our choice for $x$ ensure that $[G]_1$ has the desired properties.  
    % If some vertex in $[v]_1^{-1}$ does not have a neighbor in $G \setminus ([u]_1^{-1} \cup [v]_1^{-1})$, then our above choice for $x$ ensures that $[G]_1$ is the graph in Figure \ref{fig:k5_not_wing_u4_v0} (resp. Figure  \ref{fig:k222_not_wing_u4_v0}) such that $[u]_1$ is $w_1$ and $[v]_1$ is $w_2$.  
    % % some isomorphism between $[G]_1$ and the graph in Figure \ref{fig:k5_not_wing_u4_v0} (resp. Figure  \ref{fig:k222_not_wing_u4_v0}) maps $[u]_1$ to $w_1$ and $[v]_1$ to $w_2$.  
    % However, note that $[u]_1^{-1}$ is a proper subset of $[u]^{-1}$, which contradicts the minimality of $[u]^{-1}$.  
    % Therefore, $x$ has some neighbor in $G \setminus ([u]^{-1} \cup [v]^{-1})$, and thus $[G]_1$ clearly has the desired properties.  
\end{proof}
\begin{figure}[htb!]
    \centering
    \begin{subfigure}{0.19\textwidth}
        \centering
        \includegraphics[width = \textwidth]{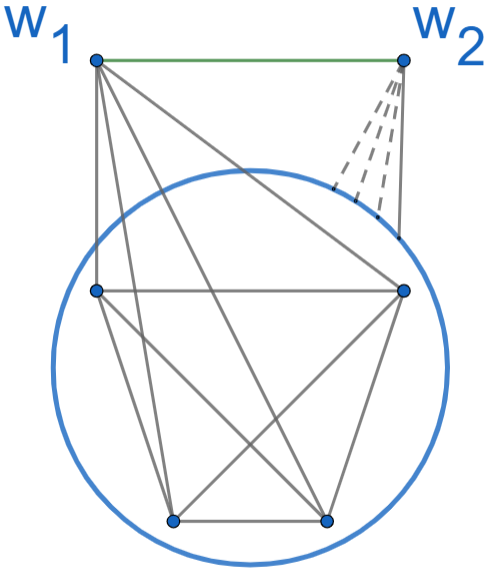}
\subcaption{}
        \label{fig:k5_not_wing_u4_vgeq1}
    \end{subfigure}
    \begin{subfigure}{0.19\textwidth}
        \centering
        \includegraphics[width = \textwidth]{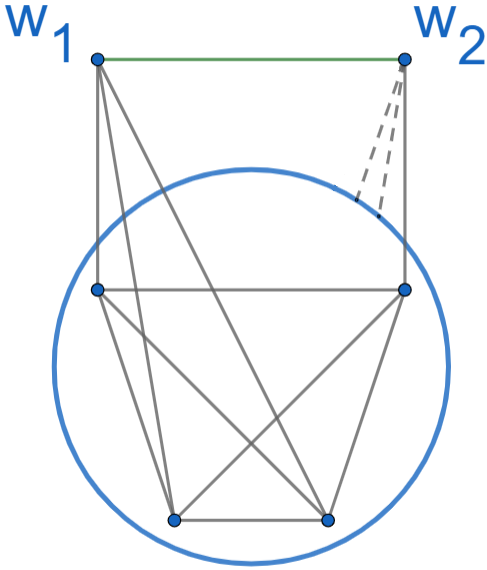}
\subcaption{}
        \label{fig:k5_not_wing_u3}
    \end{subfigure}
%     \begin{subfigure}{0.19\textwidth}
%         \centering
%         \includegraphics[width = \textwidth]{k5_u2_v3-4.png}
% \subcaption{}
%         \label{fig:k5_not_wing_u2}
%     \end{subfigure}
%     \begin{subfigure}{0.19\textwidth}
%         \centering
%         \includegraphics[width = \textwidth]{k5_u1_v3-4.png}
% \subcaption{}
%         \label{fig:k5_not_wing_u1}
%     \end{subfigure}
    \begin{subfigure}{0.19\textwidth}
        \centering
        \includegraphics[width = \textwidth]{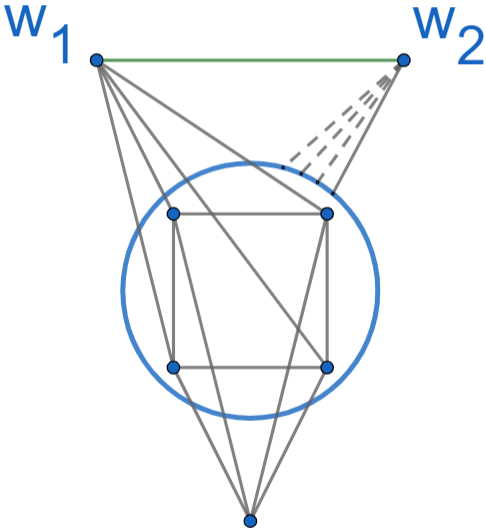}
\subcaption{}
        \label{fig:k222_not_wing_u4_vgeq1}
    \end{subfigure}
    \begin{subfigure}{0.19\textwidth}
        \centering
        \includegraphics[width = \textwidth]{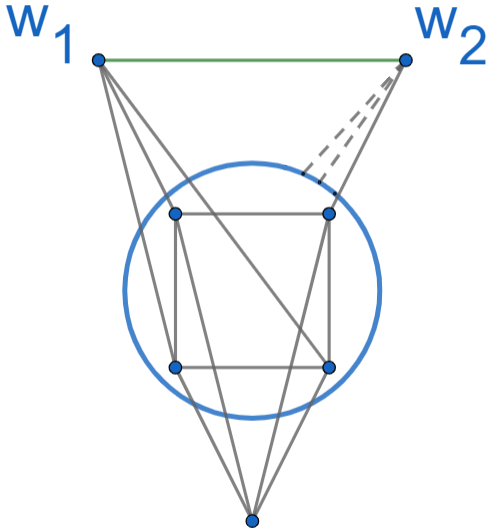}
\subcaption{}
        \label{fig:k222_not_wing_u3}
    \end{subfigure}
%     \begin{subfigure}{0.19\textwidth}
%         \centering
%         \includegraphics[width = \textwidth]{k222_u2_v3-4.png}
% \subcaption{}
%         \label{fig:k222_not_wing_u2_edges}
%     \end{subfigure}
        \begin{subfigure}{0.19\textwidth}
        \centering
        \includegraphics[width = \textwidth]{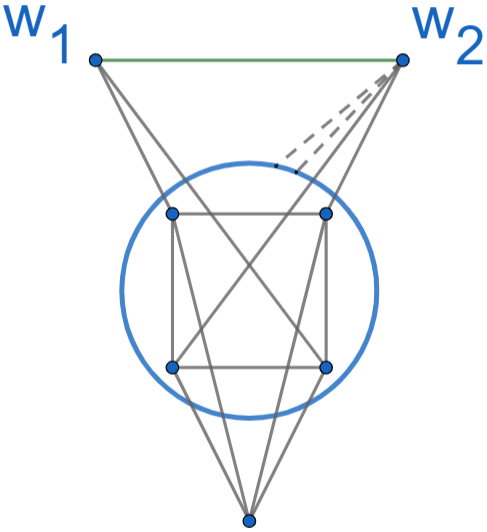}
\subcaption{}
        \label{fig:k222_not_wing_u2_diagonals}
    \end{subfigure}
%     \begin{subfigure}{0.19\textwidth}
%         \centering
%         \includegraphics[width = \textwidth]{k222_u1_v3-4.png}
% \subcaption{}
%         \label{fig:k222_not_wing_u1}
%     \end{subfigure}
    \caption{Graphs in Lemmas \ref{lem:G_not_k5_k222_u4v0} and \ref{lem:k5_k222_not_wing_fm_e_not_contracted}.  
    % Solid line-segments are edges while dashed line-segments are either edges or nonedges.  
    Line-segments between $w_2$ and the blue circle represent an edge between $w_2$ and some vertex in the circle, but note that dashed line-segments may be nonedges.  
    See also the proofs of Proposition \ref{prop:G_is_winged_graph} in this section, Lemma \ref{lem:K5_2-separator_case_1} in Section \ref{sec:proof_winged-graph_cut_vertices}, Lemma \ref{lem:C'_5_H'_wing} in Section \ref{sec:prop-8}, and Lemma \ref{lem:C'_degree_1_reducible_0} in Section \ref{sec:proof_final_top_graphs}.  }
    \label{fig:k5_not_wing}
\end{figure}

\begin{lemma}[Graphs in Figure \ref{fig:k5_not_wing} have $w_1w_2$-preserving forbidden minors]
    \label{lem:k5_k222_not_wing_fm_e_not_contracted}
    If some minor $[G]$ of a graph $G$ is one of the graphs in Figure \ref{fig:k5_not_wing}, then, for any edge $e$ in $[w_1w_2]^{-1}$, $G$ has an $e$-preserving forbidden minor.  
    % some isomorphism between some minor $[G]$ and one of the graphs in Figure \ref{fig:k5_not_wing} maps $[e]$ to $w_1w_2$, then $G$ has an $e$-preserving forbidden minor.  
\end{lemma}
\begin{proof}
    For each graph in Figure \ref{fig:k5_not_wing}, a solid line-segment represents an edge and a dashed line-segment represents a nonedge.  
    If a line-segment is between $w_2$ and the blue circle, then the endpoints of this edge or nonedge is $w_2$ and some vertex in the circle.  
    The lemma follows from the observation that each of these graphs has a $w_1w_2$-preserving forbidden minor.  
\end{proof}
\begin{proposition}[Graphs with winged graph minors]
    \label{prop:G_is_winged_graph}
    If $G$ is a $2$-connected graph with an edge $e$ and a forbidden minor but no $e$-preserving forbidden minor, then $G$ has an $e$-winged graph minor.  
\end{proposition}
\begin{proof}
    Assume $G$ is $2$-connected and has a forbidden minor but no $e$-preserving forbidden minor.  
    Then, $[e]$ is an edge of some minor $[G]$ whose contraction yields a forbidden minor.  
    Hence, $[G]$ is one of the graphs in Figures \ref{fig:k5_k222_e_contracted}, \ref{fig:k5_not_wing}, or \ref{fig:k5_k222_e_contracted_degree_1} with $[e]$ as $w_1w_2$, and thus applying Lemmas \ref{lem:G_not_k5_k222_u4v0} and \ref{lem:k5_k222_not_wing_fm_e_not_contracted} completes the proof.  
\end{proof}
A graph-nonedge pair $(G,f)$ is a \emph{winged graph pair} if $G \cup f$ has no $f$-preserving forbidden minor and has an \emph{$f$-winged graph minor}.  
The above discussion demonstrates that $(G_{\star},f_{\star})$ is a winged graph pair.  
% Hence, it suffices to assume that $(G_{\star},f_{\star})$ is an arbitrary winged graph pair and show that it has the $3$-SIP.  
Let $f_{\star} = uv$.  
The vertex set of any $f_{\star}$-winged graph minor $[G_{\star} \cup f_{\star}]$ is written as an indexed set $\{[a_i]\}$, where $u \in [a_1]^{-1}$, $v \in [a_2]^{-1}$, $[a_3]$ and $[a_4]$ are neighbors of $[a_1]$, and $[a_5]$ and $[a_6]$ are neighbors of $[a_2]$.  
% is mapped to $w_i$ under the isomorphism defining $[G_{\star} \cup f_{\star}]$.  
Proposition \ref{lem:e_not_st}, below,
% stated below and proved in \ref{sec:proof_e_not_st}, 
allows us to assume that $[a_1]^{-1} = \{u\}$ and $[a_2]^{-1} = \{v\}$, in which case we say that $[G_{\star} \cup f_{\star}]$ has \emph{singleton wing tips}.  
See Figure \ref{fig:converse_2}.  

Proposition \ref{lem:e_not_st} is proved using Lemma \ref{lem:wing-tips}, below.  
For any $f$-winged graph minor $[G \cup f]$ of a winged graph pair $(G,f)$, with $f=uv$, we call vertices $s \in [u]^{-1}$ and $t \in [v]^{-1}$ \emph{wing tips} of $[G \cup f]$ if either they are the only vertices in $[u]^{-1} \cup [v]^{-1}$ or they separate every other vertex in $[u]^{-1} \cup [v]^{-1}$ from every vertex not in $[u]^{-1} \cup [v]^{-1}$.  
Consider the following operation, similar to a vertex exchange, that transforms one minor $[G]_1$ of a graph $G$ into another minor $[G]_2$.  
Refer to Figure \ref{fig:exchange}.  
Consider an edge $[ab]_1$ of $[G]_1$, the subgraph $H$ of $G$ induced by the set $[a]_1^{-1}$, a connected subgraph $J$ of $H$, and a connected component $K$ of $H \setminus J$ that contains a vertex with a neighbor in the set $[b]_1^{-1}$.  
The minor $[G]_2$ obtained by replacing $[a]_1^{-1}$ with $[a]_1^{-1} \setminus V(K)$ and $[b]_1^{-1}$ with $[b]_1^{-1} \cup V(K)$ is said to be obtained via the \emph{component exchange from $[a]_1^{-1}$ to $[b]_1^{-1}$ that exchanges $K$}.  
\begin{lemma}[Wing tips of winged graph minors]
    \label{lem:wing-tips}
    For any winged graph pair $(G,f)$, with $f = uv$, and any $f$-winged graph minor $[G \cup f]$, there exists an $f$-winged graph minor $[G \cup f]_1$ with wing tips such that $[u]_1^{-1} \cup [v]_1^{-1}$ is a subset of $[u]^{-1} \cup [v]^{-1}$ and both $[f]_1^{-1}$ and $[f]^{-1}$ contain the same edges of $G$.  
\end{lemma}
\begin{proof}
    Let $[G \cup f]_1$ be any $f$-winged graph minor whose vertex set is $\{[a_i]_1\}$ such that (i) $[u]_1^{-1} \cup [v]_1^{-1}$ is a subset of $[u]^{-1} \cup [v]^{-1}$ and both $[f]_1^{-1}$ and $[f]^{-1}$ contain the same edges of $G$ and (ii) $[u]_1^{-1} \cup [v]_1^{-1}$ has minimum size over all such minors satisfying (i).  
    Assume there exist at least two vertices in $[u]_1^{-1}$ with neighbors in $[a_3]_1^{-1} \cup [a_4]_1^{-1}$.  
    Then, we can choose $x$ to be one of these vertices other than $u$.  
    Let $H$ be the subgraph of $G \cup f$ induced by $[u]_1^{-1}$ and let $J$ be the connected component of $H \setminus x$ that contains $u$.  
    Assume wlog that $x$ has a neighbor in $[a_3]_1^{-1}$.  
    If $J$ contains some vertex with a neighbor in $[a_4]_1^{-1}$, then let $[G \cup f]_2$ be the minor obtained via the vertex exchange from $[u]_1^{-1}$ to $[a_3]_1^{-1}$ that exchanges $x$ and fixes $J$.  
    Else, if $J$ contains some vertex with a neighbor in $[a_3]_1^{-1}$ and $x$ has a neighbor in $[a_4]_1^{-1}$, then let $[G \cup f]_2$ be the minor obtained via the vertex exchange from $[u]_1^{-1}$ to $[a_4]_1^{-1}$ that exchanges $x$ and fixes $J$.  
    Lastly, if no vertex in $[a_4]_1^{-1}$ neighbors either $x$ or any vertex in $J$, then some connected component $K$ of $H \setminus x$ other than $J$ contains some vertex with a neighbor in $[a_4]_1^{-1}$.  
    In this case, let $[G \cup f]_2$ be the minor obtained via the component exchange from $[u]_1^{-1}$ to $[a_4]_1^{-1}$ that exchanges $K$.  
    
    Next, in each case above, the neighborhoods of vertices in $[G \cup f]_1$ along with Lemma \ref{lem:k5_k222_not_wing_fm_e_not_contracted} show that $[G \cup f]_2$ is an $f$-winged graph minor such that $[u]_2^{-1} \cup [v]_2^{-1}$ is a proper subset of $[u]_1^{-1} \cup [v]_1^{-1}$ and both $[f]_2$ and $[f]_1$ contain the same edges of $G$.  
    However, this contradicts Property (ii).  
    Therefore, $[u]_1^{-1}$ contains exactly one vertex with a neighbor in $[a_3]_1^{-1} \cup [a_4]_1^{-1}$.  
    Symmetrically, $[v]_1^{-1}$ contains exactly one vertex with a neighbor in $[a_5]_1^{-1} \cup [a_6]_1^{-1}$.  
    Thus, $[G \cup f]_1$ satisfies Property (i) and has wing tips, as desired.  
\end{proof}

 Proposition \ref{lem:e_not_st} additionally requires tool Lemmas \ref{lem:tool} and \ref{lem:gluing} from Section \ref{sec:results} that are of independent interest and  are used several times throughout this paper.  
 We prove them now in the following subsection.  

\subsection{Proofs of Tool Lemmas \ref{lem:tool} and \ref{lem:gluing}}
\label{sec:tool_lemmas}

Lemma \ref{lem:tool}, about edge deletions which preserve $d$-SIP, is proved using Lemma \ref{lem:projection}, below.  
For any dimension $n > 0$, consider a subset $S \subset \mathbb{R}^n$ and let $S'$ be one of its coordinate projections given by the projection map $\pi$.  
For any point $x \in S'$, we denote the fiber $\pi^{-1}(x)$ by $S(x)$.  
Recall that in this paper, all sets are closed and compact, all maps are closed, and connected implies path-connected, unless otherwise specified.

\begin{lemma}[Connected projections of disconnected sets have disconnected fibers]
    \label{lem:projection}
    For any dimension $n > 0$, consider a disconnected compact subset $S \subset \mathbb{R}^n$ and let $S'$ be one of its coordinate projections.  
    If $S'$ is connected, then there exists a point $x \in S'$ whose fiber $S(x)$ is disconnected.  
\end{lemma}

\begin{proof}
    Assume to the contrary that $S'$ is connected and the fiber $S(x)$ of any point $x \in S'$ is connected.  
    We can assume wlog that $S$ has exactly two connected components $S_1$ and $S_2$.  
    Consider any two points $p \in S_1$ and $q \in S_2$ and let $p'$ and $q'$ be their projections in $S'$.  
    By assumption, there exists a path $P$ in $S'$ between $p'$ and $q'$.  
    Since the fiber $S(x)$ of any point $x$ along $P$ is connected and the sets $S_1$ and $S_2$ are compact, $S(x)$ is a subset of either $S_1$ or $S_2$.  
    Hence, define $C(x)$ to be $1$ if $S(x) \subseteq S_1$, and $2$ otherwise.  
    Let $Y$ be the set of all points $y$ along $P$ such that $C(y) = C(p')$ and let $Z$ be the set of all points $z$ along $P$ such that $C(z) = C(q')$.  
    By our assumptions, $C(p') \neq C(q')$ and $Y \cap Z = \emptyset$.  
    However, this implies that $P$ is disconnected, which is a contradiction.  
    Thus, the lemma is proved.  
\end{proof}
\begin{remark}
    Lemma \ref{lem:projection} is true if we replace $S$ with any subset of a topological vector space.
\end{remark}
\begin{proof}[Proof of Lemma \ref{lem:tool}]
    Assume to the contrary that $(G,F)$ and $(G \cup F, F')$ are $d$-Cayley-connected but $(G, F \cup F')$ is not.  
    By definition, there exists a squared edge-length map $\ell$ such that the CCS $\Omega^d_{F \cup F'}(G,\ell)$ is disconnected.  
    The CCS $\Omega^d_{F}(G,\ell)$, which is connected by assumption, is a coordinate projection of $\Omega^d_{F \cup F'}(G,\ell)$.  
    Hence, letting $S = \Omega^d_{F \cup F'}(G,\ell)$, $S' = \Omega^d_{F}(G,\ell)$, and applying Lemma \ref{lem:projection} shows that there exists a point $x \in S'$ such that the fiber $S(x)$ is disconnected.  
    However, this implies that there exists a squared edge-length map $\ell'$ such that the CCS $\Omega^d_{F'}(G \cup F, \ell')$ is disconnected, which contradicts the fact that $(G \cup F, F')$ is $d$-Cayley-connected.  
    Therefore, $(G, F \cup F')$ is $d$-Cayley-connected.  
    The fact that $(G, F')$ is also $d$-Cayley-connected follows from the observation that the CCS $\Omega^d_{F'}(G,\ell)$ is a coordinate projection of $\Omega^d_{F \cup F'}(G,\ell)$.  
    This completes the proof.  
\end{proof}

Lemma \ref{lem:gluing}, about how $d$-SIP interacts with clique separators, is proved using Lemma \ref{lem:ccs_intersection}, below.  
For any two induced subgraphs $H_1$ and $H_2$ of a graph $G$, $H_1 \cap H_2$ is the subgraph of $G$ induced by $V(H_1) \cap V(H_2)$.  

\begin{lemma}[Intersection of CCSs]
    \label{lem:ccs_intersection}
    Consider a graph-nonedge pair $(G,f)$, a clique separator $C$ of $G \cup f$, and the set $\{G_i\}$ of all $C$-components of $G \cup f$ that contain $f$.  
    Also, consider any linkage $(G,\ell)$ and let $\ell_i$ be the restriction of $\ell$ to the edges of $G_i \setminus f$.  
    For any dimension $d \geq 0$, the following equation holds:
    \begin{equation*}
        \Omega^d_f(G,\ell) = \bigcap_i \Omega^d_f(G_i \setminus f,\ell_i).  
    \end{equation*}
\end{lemma}
\begin{proof}
    Let $I = \Omega^d_f(G,\ell)$ and $I_i = \Omega^d_f(G_i \setminus f,\ell_i)$.  
    The fact that $I \subseteq \bigcap_i I_i$ follows immediately from the definitions of these sets.   
    % For any point $x \in I$, there exists a $d$-realization of $(G,\ell)$ in which $f$ attains the squared length $x$.  
    % From this realization, we get a $d$-realization of each linkage $(G_i \setminus f, \ell_i)$ in which $f$ attains the squared length $x$.  
    % Hence, $x$ is contained in $\bigcap_{i=1}^n I_i$.  
    To see that $\bigcap_i I_i \subseteq I$, let $\{H_i\}$ be the set of all $C$-components of $G \cup f$ and let $\ell'_i$ be the restriction of $\ell$ to the edges of $H_i \setminus f$.  
    Since $C$ is a clique, for any point in $\bigcap_i I_i$, corresponding $d$-realizations of each $(H_i \setminus f,\ell'_i)$ can be glued together on $C$ via translations, rotations, and overall reflection to obtain a $d$-realization of $(G,\ell)$.  
    % Also, let $y$ be any point in $\bigcap_i I_i$ and let $p_i$ be any $d$-realization of $(H_i \setminus f,\ell'_i)\}$ such that if $H_i$ contains $f$, then $f$ attains the squared length $y$ in $p_i$.  
    % For any two $C$-components $H_i$ and $H_j$, observe that the restrictions of $p_i$ and $p_j$ to the vertex set of $H_i \cap H_j$ are equivalent up to translations, rotations, and overall reflection.  
    % Hence, we can construct a $d$-realization of $(G,\ell)$ in which $f$ attains the squared length $y$ by gluing all the realizations in $\{p_i\}$ together at their common intersections.  
    % Thus, $y$ is contained in $I$.  
\end{proof}

Recall from Section \ref{sec:minor-notation} that induced subgraph $H$ of a graph $G$ is \emph{preserved} in a minor $[G]$ if each edge and nonedge of $H$ is preserved in $[G]$, and \emph{weakly retained} in $[G]$ if $H$ is preserved in $[G]$ and each nonedge of $H$ is retained in $[G]$.  

\begin{proof}[Proof of Lemma \ref{lem:gluing}]
     For one direction, assume that each pair $(G_i \setminus f,f)$ has the $d$-SIP.  
     Consider any linkage $(G,\ell)$ and let $\ell_i$ be the restriction of $\ell$ to the edges of $G_i \setminus f$.  
     Also, let $I = \Omega^d_f(G,\ell)$ and $I_i = \Omega^d_f(G_i \setminus f,\ell_i)$.  
     Since each set $I_i$ is a single interval, by assumption, Lemma \ref{lem:ccs_intersection} shows that $I$ is a single interval, and so $(G,f)$ has the $d$-SIP.  
     
     For the other direction, let $\{H_i\}$ be the set of all $C$-components of $G \cup f$, $H'_i = H_i \setminus f$, and $C' = C \setminus f$.  
     Assume wlog that $H_1$ contains $f$ and $(H_1 \setminus, f)$ does not have the $d$-SIP.  
     We will show that there exists an induced minor $[G]$ such that $([G],[f])$ does not have the $d$-SIP, and then apply the contrapositive of Lemma \ref{lem:cm_non-sip} to complete the proof.  
     Let $[G]_1$ be the minor obtained by contracting all edges in $H'_i \setminus C'$ for each $H'_i$ with $i \neq 1$.  
     If $C'$ does not contain $f$, then let the minor $[G]_2$ be obtained from $[G]_1$ by contracting any edge incident to each vertex $[H'_i \setminus C']_1$ with $i \neq 1$.  
     Note that $[f]_2$ is a nonedge of $[G]_2$, $H'_1$ is weakly retained in $[G]_2$, and $[G]_2 = [H'_1]_2$.  
     Hence, since $(H'_1, f)$ does not have the $d$-SIP, neither does $([G]_2,[f]_2)$, as desired.  
     
     Finally, if $C'$ contains $f$, then consider the minor $[G]_2$ obtained from $[G]_1$ as follows.  
     For any vertex $[H'_i \setminus C']_1$ adjacent to some vertex that is not an endpoint of $[f]_1$, contract the edge between these two vertices.  
     For any vertex $[H'_i \setminus C']_1$ that has degree $1$, contract the edge incident to it.  
     Note that $[f]_2$ is a nonedge in $[G]_2$, $H'_1$ is weakly retained in $[G]_2$, and $[G]_2$ can be obtained from $[H'_1]_2$ by adding vertices of degree $2$ that neighbor the endpoints of $[f]_2$.  
     Using these facts along with the triangle-inequality and the assumption that $(H_1, f)$ does not have the $d$-SIP, it is easy to see that $([G]_2,[f]_2)$ does not have the $d$-SIP, as desired.  
     This completes the proof.  
\end{proof}

A clique separator $C$ of a graph $G$ is \emph{minimal}, and referred to as a \emph{clique minimal separator (CMS)} of $G$, if for any of its induced proper subgraphs $C'$, the set of $C$-components of $G$ is not equal to the set of $C'$-components of $G$.  

\begin{proof}[Proof of Corollary \ref{cor:atom_gluing}]
    Let $\mathcal{G}_0 = \{G \cup f\}$ and $\mathcal{G}_0,\dots,\mathcal{G}_n$ be a maximal sequence such that $\mathcal{G}_{i+1}$ is obtained from $\mathcal{G}_i$ by replacing one of its graphs $H$ that has a CMS $C$ with the $C$-components of $H$.  
    It was shown in \cite{fudos1997graph} that, for any such sequence, $\mathcal{G}_n$ is the set of all atoms of $G \cup f$.  
    Lemma \ref{lem:gluing} shows that $(G,f)$ has the $d$-SIP if and only if, for each set $\mathcal{G}_i$ and each of its graphs $H$ that contain $f$, $(H \setminus f, f)$ has the $3$-SIP.  
    This completes the proof.  
\end{proof}
We are now ready to prove Proposition \ref{lem:e_not_st}. 

\begin{proposition}[Winged graph minors with singleton wing tips]
    \label{lem:e_not_st}
    Consider a winged graph pair $(G,f)$ where $G \cup f$ is $2$-connected.  
    If the converse direction of Theorem \ref{thm:3-sip_characterization} is true for any graph-nonedge pair whose graph has strictly fewer vertices than $G$, then either $(G,f)$ has the $3$-SIP or some $f$-winged graph minor of $G \cup f$ has singleton wing tips.  
\end{proposition}
\begin{proof}
    Let $f = uv$ and let $[G \cup f]$ be any $f$-winged graph minor with wing tips $s$ and $t$, as given by Lemma \ref{lem:wing-tips}, which additionally minimizes the size of $[u]^{-1} \cup [v]^{-1}$.  
    If $[u]^{-1}$ and $[v]^{-1}$ each contain exactly one vertex, then the lemma is proved.  
    Otherwise, if $st$ is an edge of $G \cup f$, then $\{s,t\}$ is a CMS of $G \cup f$.  
    Since $G \cup f$ has no $f$-preserving forbidden minor, by assumption, neither does any $\{s,t\}$-component $H$ of $G \cup f$ that contains $f$.  
    Hence, since $|V(H)| < |V(G)|$, the converse direction of Theorem \ref{thm:3-sip_characterization} along with Lemma \ref{lem:gluing} show that $(G,f)$ has the $3$-SIP.  
    Therefore, assume that $st$ is a nonedge of $G \cup f$.  
    We will show that either Lemma \ref{lem:tool} can be used to prove that $(G,f)$ has the $3$-SIP or there exists an $f$-winged graph minor $[G \cup f]_1$ with wing tips such that $[u]_1^{-1} \cup [v]_1^{-1}$ is a proper subset of $[u]^{-1} \cup [v]^{-1}$, a contradiction.  
    This will complete the proof.  
    
    Since $f$ is an edge of $G \cup f$, it is not $st$.  
    A similar argument using the facts that $G \cup f$ has no $f$-preserving forbidden minor and $|V(H)| < |V(G)|$ along with the converse direction of Theorem \ref{thm:3-sip_characterization} and Lemma \ref{lem:gluing} show that $(G \cup st, f)$ has the $3$-SIP.  
    If $(G, st)$ has the $3$-SIP, then Lemma \ref{lem:tool} states that $(G,f)$ has the $3$-SIP.  
    Otherwise, note that $\{s,t\}$ is a clique separator of $G \cup st$, and so, by Lemma \ref{lem:gluing}, there exists an $\{s,t\}$-component $H$ of $G \cup st$ that contains $st$ and such that $(H \setminus st, st)$ does not have the $3$-SIP.  
    The fact that $|V(H)| < |V(G)|$ along with the contrapositive of the converse direction of Theorem \ref{thm:3-sip_characterization} gives an $st$-preserving forbidden minor of $H$.  
    Note that $H$ must contain $f$, or else $G \cup f$ clearly has an $f$-preserving forbidden minor, which is not the case.  
    For the same reason, $H \cup f$ has no $f$-preserving forbidden minor.
    Also, since $G \cup f$ is $2$-connected, so is $H \cup f$.  
    Therefore, Proposition \ref{prop:G_is_winged_graph} shows that $(H,f)$ is a winged graph pair with an $f$-winged graph minor $[H \cup f]_2$ whose vertex set is $\{[b_i]_2\}$.  
    
    Finally, since $H \cup f$ has an $st$-preserving forbidden minor, we can choose $[H \cup f]_2$ such that, wlog, $s$ is not contained in $[u]_2^{-1} \cup [v]_2^{-1}$.  
    Assume that $s \in [b_3]_2^{-1}$.  
    Using Lemma \ref{lem:wing-tips}, we can additionally assume that $[H \cup f]_2$ has wing tips $s'$ and $t'$.  
    Let $X = V(G \setminus H)$ and note that every vertex in $X$ is separated in $G$ from $V(H) \setminus \{s,t\}$ by $\{s,t\}$.  
    Hence, since $st$ is an edge of $H \cup f$ and $s \in [b_3]_2^{-1}$, there exists an $f$-winged graph minor $[G \cup f]_1$ with wing tips $s'$ and $t'$ and whose vertex set is $\{[a_i]_1\}$, where $[a_3]_1^{-1} = [b_3]_2^{-1} \cup X$ and $[a_i]_1^{-1} =[b_i]_2^{-1}$ for any $i \neq 3$.  
    Thus, since $[u]_1^{-1} \cup [v]_1^{-1} = [u]_2^{-1} \cup [v]_2^{-1}$ and $[u]_2^{-1} \cup [v]_2^{-1} \subsetneq V(H) \subseteq [u]^{-1} \cup [v]^{-1}$, we arrive at the desired contradiction.  
\end{proof}
Next we move to Proposition \ref{prop:chained_minor}, the last one in the first part of the proof of the converse direction of the main Theorem \ref{thm:3-sip_characterization}.  
% This proposition is proved using Proposition \ref{lem:winged-graph_cut_vertices} and Lemmas \ref{prop:hinge_edge} and \ref{lem:pins_have_3-sip}, below.  
For any graph $G$, subset $U \subseteq V(G)$, and vertices $s,t \in V(G) \setminus U$, $U$ is an \emph{$st$-separator} of $G$ if no connected component of $G \setminus U$ contains both $s$ and $t$, and additionally \emph{minimal} if no proper subset of $U$ is an $st$-separator of $G$.  
Let $f_{\star} = uv$.  
A pair of vertices $xy$ of $G_{\star}$ is a \emph{pin} of $[G_{\star} \cup f_{\star}]$ if $\{x,y\}$ is a minimal $uv$-separator of $G_{\star} \setminus [f_{\star}]^{-1}$ such that either $x \in [a_3]^{-1}$ and $y \in [a_4]^{-1}$ or $x \in [a_5]^{-1}$ and $y \in [a_6]^{-1}$.  
Note that the paths in $G_{\star} \setminus [f_{\star}]^{-1}$ given by the existence of $[G_{\star} \cup f_{\star}]$ show that any minimal $uv$-separator of $G_{\star} \setminus [f_{\star}]^{-1}$ that is contained in $V(G_{\star}) \setminus \{[u]^{-1},[v]^{-1}\}$ and has size at most two is a pin of $[G_{\star} \cup f_{\star}]$.  
Proposition \ref{lem:winged-graph_cut_vertices}, stated below and used to prove Proposition \ref{prop:chained_minor}, 
% stated below and proved in \ref{sec:proof_winged-graph_cut_vertices}, 
shows that $[G_{\star} \cup f_{\star}]$ has pins $xy$ and $wz$ with $x \in [a_3]^{-1}$, $y \in [a_4]^{-1}$, $w \in [a_5]^{-1}$, and $z \in [a_6]^{-1}$.  
See Figure \ref{fig:converse_3}.  
% stated below and proved in \ref{sec:proof_chained_minor}, 
Proposition \ref{prop:chained_minor} then allows us to assume that each pin of $[G_{\star} \cup f_{\star}]$ is an edge of $G_{\star} \cup f_{\star}$, in which case we call $[G_{\star} \cup f_{\star}]$ \emph{chained}.  
See Figure \ref{fig:converse_4}.  

Proposition \ref{lem:winged-graph_cut_vertices} is proved using Lemma \ref{lem:K5_2-separator_case_1}, below.  
The proof of this lemma involves inspecting many simple cases, and so it is given in \ref{sec:proof_winged-graph_cut_vertices}.  
% we omit the proofs of the latter two lemmas, but the proofs are available upon request.  
%
\begin{lemma}[Inductive structure of winged graph minors]
\label{lem:K5_2-separator_case_1}
Let $(G,f)$ be a winged graph pair, $[G \cup f]$ be an $f$-winged graph minor whose vertex set is $\{[c_i]\}$, and $X \subseteq [c_3]^{-1}$ and $Y \subseteq [c_4]^{-1}$ be the subsets of all vertices with neighbors in $[c_1]^{-1}$.  
Also, let $A \subseteq [c_3]^{-1}$ and $B \subseteq [c_4]^{-1}$ be the subsets of all vertices with neighbors in $[c_5]^{-1} \cup [c_6]^{-1}$ if $[G \cup f]$ is the graph in Figure \ref{fig:k5_wing}, otherwise; let $A \subseteq [c_3]^{-1}$ and $B \subseteq [c_4]^{-1}$ be the subsets of all vertices with neighbors in $[c_5]^{-1} \cup [c_7]^{-1}$ and $[c_6]^{-1} \cup [c_7]^{-1}$ if $[G \cup f]$ is the graph in Figure \ref{fig:k222_wing}.  
If some set in $\{X,Y\}$ and some set in $\{A,B\}$ both have size at least two, then there exists an $f$-winged graph minor $[G \cup f]_1$ whose vertex set $\{[d_i]_1\}$ is such that $[d_3]_1^{-1} \cup [d_4]_1^{-1}$ is a proper subset of $[c_3]^{-1} \cup [c_4]^{-1}$ and both $[f]_1^{-1}$ and $[f]^{-1}$ contain the same set of edges of $G$.  
\end{lemma}

\begin{proposition}[Pins of a winged graph minor]
\label{lem:winged-graph_cut_vertices}
For any winged graph pair $(G,f)$, any $f$-winged graph minor $[G \cup f]$, whose vertex set is $\{[a_i]\}$, has pins $xy$ and $wz$ with $x \in [a_3]^{-1}$, $y \in [a_4]^{-1}$, $w \in [a_5]^{-1}$, and $z \in [a_6]^{-1}$.  
\end{proposition}
\begin{proof}
    We prove the existence of $xy$ by induction on the size $n$ of $[a_3]^{-1} \cup [a_4]^{-1}$.  
    The existence of $xy$ is immediate when $n = 2$.  
    Assume $xy$ exists when $n \leq k$ for any $k \geq 2$, and we will prove it when $n=k+1$.  
    Let $X$, $Y$, $A$, and $B$ be sets as defined in Lemma \ref{lem:K5_2-separator_case_1}.  
    % Let $X \subseteq [a_3]^{-1}$ and $Y \subseteq [a_4]^{-1}$ be the subsets of all vertices with neighbors in $[a_1]^{-1}$.  
    % If $[G \cup f]$ is the graph in Figure \ref{fig:k5_wing}, then let $A \subseteq [a_3]^{-1}$ and $B \subseteq [a_4]^{-1}$ be the subsets of all vertices with neighbors in either $[a_5]^{-1}$ or $[a_6]^{-1}$.  
    % % some isomorphism between $[G \cup f]$ and the graph in Figure \ref{fig:k5_wing} maps $[f]$ to $w_1w_2$, then let $A \subseteq [a_3]^{-1}$ and $B \subseteq [a_4]^{-1}$ be the subsets of all vertices with neighbors in either $[a_5]^{-1}$ or $[a_6]^{-1}$.  
    % Otherwise, if $[G \cup f]$ is the graph in Figure \ref{fig:k222_wing}, then let $A \subseteq [a_3]^{-1}$ and $B \subseteq [a_4]^{-1}$ be the subsets of all vertices with neighbors in $[a_5]^{-1} \cup [a_7]^{-1}$ and $[a_6]^{-1} \cup [a_7]^{-1}$, respectively.  
    % some isomorphism between $[G \cup f]$ and the graph in Figure \ref{fig:k222_wing} maps $[f]$ to $w_1w_2$, then let $A \subseteq [a_3]^{-1}$ and $B \subseteq [a_4]^{-1}$ be the subsets of all vertices with neighbors in $[a_5]^{-1} \cup [a_7]^{-1}$ and $[a_6]^{-1} \cup [a_7]^{-1}$, respectively.  
    If either $|X| = |Y| = 1$ or $|A| = |B| = 1$ is true, then $xy$ clearly exists.  
    Otherwise, by applying Lemma \ref{lem:K5_2-separator_case_1}, we obtain an $f$-winged graph minor $[G \cup f]_1$, whose vertex set is $\{[b_i]_1\}$, such that $[b_3]_1^{-1} \cup [b_4]_1^{-1}$ is a proper subset of $[a_3]^{-1} \cup [a_4]^{-1}$ and both $[f]_1^{-1}$ and $[f]^{-1}$ contain the same set of edges of $G$.  
    Therefore, the inductive hypothesis states that there exists a pin $x'y'$ of $[G \cup f]_1$ with $x' \in [b_3]_1^{-1}$ and $y' \in [b_4]_1^{-1}$.  
    Observe that $x',y' \in [a_3]^{-1} \cup [a_4]^{-1}$ and $\{x',y'\}$ is a $uv$-separator of $G \setminus [f]^{-1}$.  
    Along with the paths in $G \setminus [f]^{-1}$ given by the existence of $[G \cup f]$, this shows wlog that $x' \in [a_3]^{-1}$ and $y' \in [a_4]^{-1}$.  
    Thus, setting $xy = x'y'$ completes the inductive proof.  
    The proof that $wz$ exists is symmetric.  
\end{proof}

The proof of Proposition \ref{prop:chained_minor} also requires Lemmas \ref{prop:hinge_edge} and \ref{lem:pins_have_3-sip}, below.  

\begin{lemma}[No pin-preserving forbidden minor]
\label{prop:hinge_edge}
    For any winged-graph pair $(G,f)$ and any pin $xy$ of an $f$-winged graph minor of $G \cup f$, $G \cup \{f,xy\}$ has no $f$-preserving forbidden minor.  
\end{lemma}
\begin{proof}
    The lemma is immediate if $xy$ is an edge of $G \cup f$.  
    Otherwise, assume to the contrary that $f$ is preserved in some forbidden minor $[G \cup \{f,xy\}]$.  
    Let $[G \cup f]_1$ be an $f$-winged graph minor such that $xy$ is one of its pins.  
    By Lemma \ref{lem:wing-tips}, we can assume wlog that $[G \cup f]_1$ has wing tips $s$ and $t$.  
    Since $[G \cup f]_1$ has at least two pins, by Proposition \ref{lem:winged-graph_cut_vertices}, $C = \{x,y,t\}$ is a separator of $G \cup \{f,xy\}$.  
    Furthermore, since every forbidden minor is $4$-connected and has at least five vertices, there is a unique $C$-component $H$ of $G \cup \{f,xy\}$ that contains some vertex $w$ such that $[w]$ is not contained in $[C]$.  
    If $H$ contains $f$, then the paths in $G \cup f$ given by the existence of $[G \cup f]_1$ imply the existence of some minor $[G \cup f]_2$ in which $H$ is preserved and $[C]_2$ is a clique.  
    Combining these facts shows that $[G \cup f]_2$ has an $[f]_2$-preserving forbidden minor, and thus $G \cup f$ has an $f$-preserving forbidden minor, which is a contradiction.  

    Next, if $H$ does not contain $f$, then the existence of $[G \cup \{f,xy\}]$ implies the existence of a minor $[G \cup \{f,xy\}]_4$ that preserves $H$, has an $[f]_4$-preserving forbidden minor, and such that, for every vertex $w \in V(G) \setminus V(H)$, $[w]_4$ is contained in $[C]_4$.  
    Hence, both endpoints of $[f]_4$ are contained $[C]_4$.  
    Since $s$ and $t$ are wing tips of $[G \cup f]_1$, it cannot be that $[f]_4 = [xy]_4$.  
    % If $[f]_4 = [xy]_4$, then $xy$ is doubled in $[G \cup \{f,xy\}]_4$.  
    Hence, assume wlog that $[f]_4 = [xt]_4$.  
    Then, we must have $[s]_4 = [x]_4$.  
    Since $C$ is a separator of $G \cup f$, $[y]_1^{-1} \setminus V(H)$ must be a subset of $[x]_4^{-1} \cup [y]_4^{-1}$.  
    Therefore, it is easy to see that $xy$ is doubled in $[G \cup \{f,xy\}]_4$, and so $G \cup f$ clearly has an $f$-preserving forbidden minor, which is a contradiction.  
    Thus, $G \cup \{f,xy\}$ has no $f$-preserving forbidden minor.  
\end{proof}
\begin{lemma}[Graph-pin pairs have the $3$-SIP]
    \label{lem:pins_have_3-sip}
    Consider any winged graph pair $(G,f)$ and any pin $xy$ of an $f$-winged graph minor $[G \cup f]$.  
    If the converse direction of Theorem \ref{thm:3-sip_characterization} is true for any graph-nonedge pair whose graph has strictly fewer vertices than $G$, then $(G \setminus \{[f]^{-1},xy\}, xy)$ has the $3$-SIP.  
\end{lemma}
\begin{proof}
    Note that $\{x,y\}$ is a clique separator of $(G \setminus [f]^{-1}) \cup xy$.  
    Hence, by Lemma \ref{lem:gluing}, it suffices to show that show that $(H \setminus xy, xy)$ has the $3$-SIP for each $\{x,y\}$-component $H$ of $(G \setminus [f]^{-1}) \cup xy$ that contains $xy$.  
    If $(H \setminus xy, xy)$ does not have the $3$-SIP, then, since $|V(H)| < |V(G)|$, the contrapositive of the converse direction of Theorem \ref{thm:3-sip_characterization} shows that $H$ has an $xy$-preserving forbidden minor.  
    However, this fact along with the paths in $G \setminus [f]^{-1}$ given by the existence of $[G \cup f]$ imply that $G \cup f$ has an $f$-preserving forbidden minor, which is a contradiction.  
    Thus, $(H \setminus xy, xy)$ has the $3$-SIP, and the proof is complete.  
\end{proof}
Now we are ready to state and prove Proposition \ref{prop:chained_minor}.

\begin{proposition}[Chained winged-graph minors]
    \label{prop:chained_minor}
    For any winged graph pair $(G,f)$, if the converse direction of Theorem \ref{thm:3-sip_characterization} is true for any graph-nonedge pair whose graph has strictly fewer vertices than $G$, then there exists a winged graph pair $(G',f)$ with a chained $f$-winged graph minor $[G' \cup f]$ such that $G$ is a spanning subgraph of $G'$ and $(G,f)$ has the $3$-SIP if $(G' \setminus [f]^{-1},[f]^{-1})$ is $3$-Cayley-connected.  
\end{proposition}
\begin{proof}
    Let $(G',f)$ be any $f$-winged graph pair with an $f$-winged graph minor $[G' \cup f]$ such that (i) $G$ is a spanning subgraph of $G'$ and $(G,f)$ has the $3$-SIP if $(G' \setminus [f]^{-1},[f]^{-1})$ is $3$-Cayley-connected and (ii) the number $k$ of pins of $[G' \cup f]$ that are edges of $G' \cup f$ is maximized.  
    Note that $[G' \cup f]$ has at least two pins, by Proposition \ref{lem:winged-graph_cut_vertices}.  
    If $[G' \cup f]$ is chained, then we are done.  
    Otherwise, some pin $xy$ of $[G' \cup f]$ is a nonedge of $G' \cup f$.  
    Using Lemma \ref{prop:hinge_edge}, we see that $(G' \cup xy,f)$ is a winged graph pair.  
    Also, the minor $[G' \cup \{f,xy\}]_1$ such that, for any two vertices $u,v \in V(G')$, $[u]_1 = [v]_1$ if and only if $[u] = [v]$ is an $f$-winged graph minor.  
    Note that $G$ is a spanning subgraph of $G' \cup xy$ and the number of pins of $[G' \cup \{f,xy\}]_1$ that are edges of $G' \cup \{f,xy\}$ is strictly greater than $k$.  
    Furthermore, Lemma \ref{lem:pins_have_3-sip} shows that $(G' \setminus [f]^{-1},xy)$ has the $3$-SIP.  
    Therefore, if $((G' \setminus [f]^{-1}) \cup xy,[f]^{-1})$ is $3$-Cayley connected, then Lemma \ref{lem:tool} shows that $(G' \setminus [f]^{-1},[f]^{-1})$ is $3$-Cayley connected.  
    Consequently, $(G,f)$ has the $3$-SIP if $((G' \setminus [f]^{-1}) \cup xy,[f]^{-1})$ is $3$-Cayley connected.  
    Combining the above facts contradicts the maximality of $(G',f)$.  
    Thus, every pin of $[G' \cup f]$ is an edge of $G' \cup f$, and the proof is complete.  
\end{proof}

Now that a special structure has been established for $G_{\star} \cup f_{\star}$ in the first part of the proof of the converse direction of the main Theorem \ref{thm:3-sip_characterization}, we move to the second part of the proof.  
Despite this special structure, almost any linkage of $G_{\star}$ has a disconnected CS.  
Hence the main work of the second part of the proof is to establish that each connected component of the CS attains the same interval of lengths for $f_{\star}$.  

For any pin $xy$ of the chained $f_{\star}$-winged graph minor $[G_{\star} \cup f_{\star}]$, there exist unique $\{x,y\}$-components $H_u(xy)$ and $H_v(xy)$ of $G_{\star} \setminus [f_{\star}]^{-1}$ that contain $u$ and $v$, respectively.  
Since every pin of $[G_{\star} \cup f_{\star}]$ is an edge of $G_{\star} \cup f_{\star}$, each pin of $[G_{\star} \cup f_{\star}]$ other than $xy$ is contained in either $H_u(xy)$ or $H_v(xy)$, but not both.  
Hence, we define the \emph{ordered sets of pins} $\{x_iy_i\}$ of $[G_{\star} \cup f_{\star}]$ to be the set of all pins of $[G_{\star} \cup f_{\star}]$ indexed such that $H_v(x_iy_i)$ contains $x_{i+1}y_{i+1}$.  
The \emph{links} of $[G_{\star} \cup f_{\star}]$ are the graphs $H_u(x_1y_1)$ and $H_v(x_ny_n)$ and, for each $i < n$, the $\{x_{i+1},y_{i+1}\}$-component of $H_v(x_iy_i)$ that contains $x_iy_i$.  
Proposition \ref{prop:x2y2_3-Cayley-connected}  
% stated below and proved in \ref{sec:proof_x2y2_3-Cayley-connected}, 
shows that $(G_{\star}, \{ux_2,uy_2\})$ is $3$-Cayley-connected.  
See Figure \ref{fig:converse_4}.  
This proposition establishes that although the link of $G_{\star}$ between any consecutive pins $x_{i-1}y_{i-1}$ and $x_iy_i$ could force disconnectedness of the overall CS, the connected components of the CS can be viewed as partial reflections each of which attains the same connected set of lengths for the pair of nonedges $\{ux_i, uy_i\}$.  
The same statement holds for the pair of nonedges $\{vx_i, vy_i\}$ for any $i$, resulting in the same interval of lengths for $f_{\star}$.

Proposition \ref{prop:x2y2_3-Cayley-connected} is proved using Lemmas \ref{lem:links_3-connected} and \ref{lem:ux2_uy2_3-covering}, stated below and proved in \ref{sec:proof_x2y2_3-Cayley-connected}.  
\begin{lemma}[Links are $3$-connected]
    \label{lem:links_3-connected}
    Consider any winged graph pair $(G,f)$, where $G \cup f$ is $2$-connected, and any chained $f$-winged graph minor $[G \cup f]$ with singleton wing tips.  
    If the converse direction of Theorem \ref{thm:3-sip_characterization} is true for all graph-nonedge pairs whose graphs have strictly fewer vertices than $G$, then either $(G,f)$ has the $3$-SIP or each link of $[G \cup f]$ is $3$-connected.  
\end{lemma}

Recall the $d$-covering map property defined in Section \ref{sec:results}.  

\begin{lemma}[$(G,ux_2)$ and $(G,uy_2)$ both have the $3$-covering map property]
    \label{lem:ux2_uy2_3-covering}
    Let $(G,f)$ be a winged graph pair, with $f=uv$, that has a chained $f$-winged graph minor with singleton wing tips and whose ordered set of pins is $\{x_i,y_i\}$.  
    Also, consider the links $L_0$, containing $u$ and $x_1y_1$, and $L_1$, containing $x_1y_1$ and $x_2y_2$, of this minor.  
    If $L_0$ and $L_1$ are $3$-connected, then $(G,ux_2)$ and $(G,uy_2)$ have the $3$-covering map property.  
\end{lemma}
Lemma \ref{lem:ux2_uy2_3-covering} is proved using the tool Theorems \ref{prop:3-connected_partial_3-tree_star_lemma} and \ref{prop:partial_3-tree_3-reflection}, which were stated in the Contributions Section \ref{sec:results} as independently interesting byproducts of the main result Theorem \ref{thm:3-sip_characterization}.  
We devote the next subsection to their proofs before stating and proving Proposition \ref{prop:x2y2_3-Cayley-connected} and completing the proof of the converse direction of the main Theorem \ref{thm:3-sip_characterization}.  
\subsection{Proofs of Tool Theorems \ref{prop:3-connected_partial_3-tree_star_lemma} and \ref{prop:partial_3-tree_3-reflection} concerning partial 3-trees and covering maps}
\label{sec:tool-thms}

The proof of Theorem \ref{prop:3-connected_partial_3-tree_star_lemma} requires Lemmas \ref{lem:3-flat_and_not_minimally_3-connected_implies_partial_3-tree} and \ref{lem:not_3-or-2_implies_not_1}, below, which may themselves be independently interesting results about partial $3$-trees.

\begin{lemma}[$3$-flattenable and not minimally $3$-connected implies partial $3$-tree]
    \label{lem:3-flat_and_not_minimally_3-connected_implies_partial_3-tree}
    If a graph $G$ is $3$-flattenable and $3$-connected and contains a $3$-connected proper spanning subgraph, then $G$ is a partial $3$-tree.
\end{lemma}

\begin{proof}
    Assume that $G$ is $3$-connected and $3$-flattenable and contains a $3$-connected proper spanning subgraph.  
    In \cite{belk2007realizability1} it was shown that a graph is $3$-flattenable it and only if it is a subgraph of a graph constructed from $K_4$s, $V_8$s, and $C_5 \times C_2$s via $1$-sums, $2$-sums, and $3$-sums.  
    Since neither $V_8$ nor $C_5 \times C_2$ contains a $K_3$ subgraph, these graphs can only be involved in $1$-sums and $2$-sums.  
    Hence, $G$ is a subgraph of some graph $H$ constructed as described above.  
    Since $G$ is $3$-connected, $H$ is either a $3$-tree, $V_8$, or $C_5 \times C_2$.  
    Recall from Section \ref{sec:results} that the latter two graphs are forbidden minors for partial $3$-trees.  
    Therefore, $G$ is a partial $3$-tree unless it is once of these two graphs.  
    However, observe that no proper spanning subgraph of these two graphs is $3$-connected, and so $G$ is neither of these graphs.  
    This completes the proof.  
\end{proof}

\begin{lemma}[If $G \cup \{wu_i\}$ with $|\{wu_i\}| \in \{2,3\}$ has a forbidden minor, then so does $G \cup wu_i$ for some $wu_i$]
\label{lem:not_3-or-2_implies_not_1}
    Let $G$ be a $3$-connected graph and $F = \{wu_i\}$ be some size two or three set of its nonedges.  
    If there exists a forbidden minor $[G \cup F]$ in which $[w]$ and each $[u_i]$ are pairwise distinct vertices and each $wu_i$ is retained, then $G \cup wu_i$ has a forbidden minor for some $wu_i$.
\end{lemma}

The proof of Lemma \ref{lem:not_3-or-2_implies_not_1} is given in \ref{sec:star_thm}.  

\begin{proof}[Proof of Theorem \ref{prop:3-connected_partial_3-tree_star_lemma}]
    Let $G$ be a $3$-connected partial $3$-tree with a nonedge set $\{wu_i\}$ for which $G \cup wu_i$ is a partial $3$-tree for each $wu_i$.  
    If $G \cup \{wu_i\}$ is $3$-flattenable, then the theorem follows from Lemma \ref{lem:3-flat_and_not_minimally_3-connected_implies_partial_3-tree}.  
    Otherwise, let $F \subseteq \{wu_i\}$ be a minimum-sized set such that there exists a forbidden minor $[G \cup F]$.  
    By our assumptions, we have $|F| \geq 2$.  
    We first reduce the problem to the case where some nonedge in $F$ is retained in $[G \cup F]$.  
    If some $f \in F$ is contracted in every forbidden minor of $G \cup F$, then Proposition \ref{prop:G_is_winged_graph} shows that $(G \cup (F \setminus \{f\}), f)$ is a winged graph pair with an $f$-winged graph minor.  
    Hence, either Proposition \ref{lem:winged-graph_cut_vertices} or Lemma \ref{lem:wing-tips} shows that $G \cup (F \setminus \{f\})$ has a separator of size two, a contradiction.  
    Therefore, every nonedge in $F$ is preserved in some forbidden minor of $G \cup F$.  
    This allows us to assume wlog that the set $X$ of all nonedges in $F$ that are preserved in $[G \cup F]$ is non-empty.  
    The minimality condition on $F$ implies that each nonedge in $X$ is retained in $[G \cup F]$ and any two of these nonedges $wu_i$ and $wu_j$ are such that $[u_i] \neq [u_j]$.  
    
    Next, since $G$ is connected, each vertex of any forbidden minor has degree four, and $w$ is an endpoint of every nonedge in $F$, we have $|X| \leq 3$.  
    If $|X| > 1$, then Lemma \ref{lem:not_3-or-2_implies_not_1} can be used to show the existence of the proper subset $F' \subsetneq F$ such that $G \cup F'$ has a forbidden minor, a contradiction.  
    Therefore, wlog we have $X = \{wu_1\}$.  
    Let $H$ be the subgraph of $G$ induced by $[w]^{-1}$ and let $H_w$ be its connected component that contains $w$.  
    Also, let $\{[v_i]\}$ be the neighborhood of $[w]$ in $[G \cup F] \setminus [u_1]$.  
    There are several cases.  
    
    If each set $[v_i]^{-1}$ contains some vertex whose neighborhood in $G$ contains some vertex in $H_w$, then $G \cup wu_1$ clearly has a forbidden minor, a contradiction.  
    Otherwise, let $\{J_i\}$ be the set of connected components of $H \setminus H_w$.  
    The minimality of $F$ implies that each $J_i$ contains both exactly one vertex $u_j$ such that $wu_j \in F$ and some vertex whose neighborhood in $G$ contains some vertex in at least one of the sets $[v_i]^{-1}$ or $[u_1]^{-1}$.  
    If exactly two sets in $\{[v_i]^{-1}\}$, say $[v_1]^{-1}$ and $[v_2]^{-1}$, each contain some vertex that neighbora some vertex in $H_w$, then wlog $J_1$ contains some vertex whose neighborhood in $G$ contains some vertex in $[v_3]^{-1}$.  
    The minor $[G \cup F]_1$ obtained via the component exchange from $[w]^{-1}$ to $[v_3]^{-1}$ that exchanges $J_1$ either is or has a forbidden minor in which at least two pairs in $F$ are retained, which is a previously handled case.  
    
    Consequently, assume that exactly one set in $\{[v_i]^{-1}\}$, say $[v_1]^{-1}$, contains some vertex whose neighborhood in $G$ contains some vertex in $H_w$.  
    Then, $[v_2]^{-1}$ and $[v_3]^{-1}$ each contain some vertex that neighbors some vertex in some $J_i$.  
    Assume the set $\{J_i\}$ contains more than one component.  
    If only one component, say $J_1$, contains vertices with neighbors in $[v_2]^{-1} \cup [v_3]^{-1}$, then we arrive at a similar contradiction by considering the minor obtained via the component exchange from $[w]^{-1}$ to either $[u_1]^{-1}$ or $[v_3]^{-1}$, whichever contains a vertex with a neighbor in $J_2$, that exchanges $J_2$.  
    Otherwise, wlog $J_1$ contains some vertex that neighbors some vertex in $[v_2]^{-1}$ and $J_2$ contains some vertex that neighbors some vertex in $[v_3]^{-1}$, and we arrive at a similar contradiction by considering the minor obtained via the component exchange from $[w]^{-1}$ to $[v_3]^{-1}$ that exchanges $J_2$.  
    
    Lastly, assume $J_1$ is the only component in $\{J_i\}$.  
    First, we argue that we can assume $w$ is the only vertex in $H_w$.  
    Since $G$ is $3$-connected, each connected component $H'$ of $H_w \setminus w$ must contain some vertex that neighbors some vertex outside of $H$.  
    Our above assumptions imply that this vertex outside $H$ is contained in $[v_1]^{-1}$.  
    Hence, we can instead work with the minor obtained from $[G \cup F]$ via the component exchange from $[w]^{-1}$ to $[v_1]^{-1}$ that exchanges $H'$.  
    Repeating this argument for the new minor allows us to make the desired assumption on $H_w$.  
    
    Next, since $G$ is $3$-connected, there exist three vertex-disjoint paths $P_a$, $P_b$, and $P_c$ in $G$ between $w$ and $u_1$, where $a$, $b$, and $c$ are the first vertices along these paths, respectively, not in $[w]^{-1}$.  
    Let $K$ be the subgraph of $G$ induced by $[v_1]^{-1}$.  
    By our above assumptions, $K$ contains $a$, $b$, and $c$.  
    Let $a'$ be the first vertex in $P_a$ after $a$ that is not contained in $K$, and let $R_a$ be the path in $P_a$ between $a$ and the vertex immediately preceding $a'$.  
    Define $R_b$ and $R_c$ similarly.  
    Wlog, we can assume that some connected component $K_{bc}$ of $K \setminus R_a$ contains both $R_b$ and $R_c$.  
    Let $V_a$ be the set of vertices of $G \setminus ([w]^{-1} \cup [v_1]^{-1})$ that are adjacent to some vertex in $K \setminus K_{bc}$ and let $V_{bc}$ be the set of vertices of $G \setminus ([w]^{-1} \cup [v_1]^{-1})$ that are adjacent to some vertex in $K_{bc}$.  
    There are several cases.  
    % Since $[v_1]$ neighbors $[u_1]$, we can wlog assume that $V_a$ contains some vertex in $[v_1]^{-1}$.  
    
    For any neighbor $[x]$ of $[v_1]$ other than $[w]$, if $V_a$ is a subset of $[x]^{-1}$, then the minor obtained via the component exchange from $[v_1]^{-1}$ to $[x]^{-1}$ that exchanges $K \setminus K_{bc}$ either demonstrates that $G \cup (F \setminus \{wu_1\})$ has a forbidden minor, which contradicts the minimality of $F$, or places us in the above handled case where exactly two sets in $\{[v_i]^{-1}\}$ each contain some vertex that neighbors some vertex in $H_w$.  
    Otherwise, let $[x]$ be a neighbor of $[v_1]$ such that $V_{bc} \cap [x]^{-1}$ is non-empty.  
    If no such neighbor $[x]$ exists, then some vertex in $K_{bc}$ neighbors some vertex in $J_1$.  
    In this case, the minor obtained via the component exchange from $[v_1]^{-1}$ to $[w]^{-1}$ that exchanges $K_{bc}$ demonstrates that $G \cup (F \setminus \{wu_2\})$ has a forbidden minor, which contradicts the minimality of $F$.  
    Therefore, assume $[x]$ exists and, if possible, chooseit so that $V_a \cap [x]^{-1}$ is empty.  
    If $[G \cup F]$ is $K_5$, then the minor $[G \cup F]_1$ obtained via the component exchange from $[v_1]^{-1}$ to $[x]^{-1}$ that exchanges $K_{bc}$ either violates the minimality of $F$ or places us in the above handled case where exactly two sets in $\{[v_i]^{-1}\}$ each contain some vertex that neighbors some vertex in $H_w$.  
    If $[G \cup F]$ is $K_{2,2,2}$, then there is one additional possibility for $[G \cup F]_1$: it has a $K_5$ minor in which $wu_1$ is preserved.  
    However, in this case we could have instead considered $[G \cup F]_1$ to begin with.  
    Thus, all cases have been checked, and so the theorem is proved.  
\end{proof}

The proof of Theorem \ref{prop:partial_3-tree_3-reflection} requires Lemmas \ref{lem:preimage_closure} and \ref{lem:A_empty_config_space_connected}, below, which concern the following objects.  
Consider a graph $G$ with $F$ being a set of its nonedges such that $G \cup F$ is a partial $3$-tree, a squared edge-length map $\ell$, and the Cayley map $\phi_F$ from the CS $\mathcal{C}^3(G,\ell)$ to the CCS $\Omega^3_F(G,\ell)$.  
Let $Y$ be the relative interior of the CCS.  
For any realization $p$ in the CS and any $K_4$ subgraph $H$ of $T$, the orientation of the set of points to which $V(H)$ maps under $p$ is given by a determinant on these points - this determinant is $0$ if and only if these points are coplanar, and either positive or negative otherwise.  
The orientation vector $\tau(p) = (\tau_i(p))_{i\in \{1,\dots,m\}}$ is such that $\tau_i(p) = 0$ if the orientation determinant for $V(A_i)$ is $0$ under $p$, $\tau_i(p) = -1$ if this determinant is negative, and $\tau_i(p) = 1$ if this determinant is positive.  
Let $\Delta$ be the set of all vectors $\delta = (\delta_i)_{i\in \{1,\dots,m\}}$ with $\delta_i \in \{-1,1\}$.  
The $\delta$-oriented CS $\mathcal{C}^3_{\delta}(G,\ell)$ is the open subset of $\phi^{-1}(Y)$ containing every realization whose orientation vector is $\delta$.  

\begin{lemma}[Realizations in the CCS relative interior preimage]
    \label{lem:preimage_closure}
    Let $G$ be a graph with $F$ being a subset of its nonedges such that $G \cup F$ is a $3$-tree, $\ell$ be a squared edge-length map, $\phi_F$ be the Cayley map from the CS $\mathcal{C}^3(G,\ell)$ to the CCS $\Omega^3_F(G,\ell)$, and $Y$ be the the relative interior of the CCS.  
    Then:
    \begin{enumerate}[(i)]
        \item For each $\delta \in \Delta$, the image of the $\delta$-oriented CS under $\phi_F$ is $Y$.  
        \item Each connected component of the CS contains some realization in a $\delta$-oriented CS for some $\delta \in \Delta$.  
    \end{enumerate}
\end{lemma}

\begin{proof}
    Since $G \cup F$ is $3$-flattenable, Theorem \ref{thm:sitharam_willoughby} shows that the CCS is convex.  
    Hence, each point $y \in Y$ is not on any boundary of the CCS.  
    This implies that each realization in $\phi_F^{-1}(y)$ is contained in a $\delta$-oriented CS for some $\delta \in \Delta$.  
    Since $G \cup F$ is a $3$-tree, it follows that, for each $\delta \in \Delta$, the $\delta$-oriented CS contains some realization in $\phi_F^{-1}(y)$, which proves (i).  
    Next, (ii) follows easily from the fact that the CCS is the semi-algebraic set defined by a set of tetrahedral inequalities: one for each $K_4$ subgraph of $G \cup F$ and subject to the edge-length map $\ell$.  
\end{proof}

\begin{lemma}[Oriented CSs are connected]
    \label{lem:A_empty_config_space_connected}
    Let $G$ be a graph with $F$ being a subset of its nonedges such that $G \cup F$ is a $3$-tree, let $\ell$ be a squared edge-length map, and let $\phi_F$ be the Cayley map from the CS $\mathcal{C}^3(G,\ell)$ to the CCS $\Omega^3_F(G,\ell)$.  
    Also, consider the set $A = \{A_1,\dots,A_m\}$ containing each $K_4$ subgraph of $G \cup F$ that maps to a non-coplanar set of points under some realization in the CS.  
    If the set $A$ is empty, then the CS is connected; otherwise, for each $\delta \in \Delta$, the $\delta$-oriented CS is connected.  
\end{lemma}

The proof of Lemma \ref{lem:A_empty_config_space_connected} is given in \ref{sec:proof_partial_3-tree_3-reflection}.  

\begin{proof}[Proof of Theorem \ref{prop:partial_3-tree_3-reflection}]
    We must show that $(G,F)$ has the $3$-covering map property whenever $G \cup F$ is a partial $3$-tree.  
    It suffices to consider the case where $G \cup F$ is a $3$-tree.  
    Let $\ell$ be a squared edge-length map and let $\phi_F$ be the Cayley map from the CS $\mathcal{C}^3(G,\ell)$ to the CCS $\Omega^3_F(G,\ell)$.  
    We must show that the image of each connected component of the CS under $\phi_F$ is the CCS.  
    Let $A = \{A_1,\dots,A_m\}$ be the set containing each $K_4$ subgraph $A_i$ of $G \cup F$ that maps to a non-coplanar set of points under some realization in the CS.  
    If $A$ is empty, then Lemma \ref{lem:A_empty_config_space_connected}(i) shows that the CS is connected, and so the theorem is immediate.  
    Otherwise, Lemma \ref{lem:preimage_closure}(ii) and Lemma \ref{lem:A_empty_config_space_connected}(ii) show that each connected component in the CS contains a $\delta$-oriented CS for some $\delta \in \Delta$.  
    Lemma \ref{lem:preimage_closure}(i) then states that the image of this $\delta$-oriented CS under $\phi_F$ is $Y$.  
    The theorem now follows since $\phi_F$ is a closed map.  
\end{proof}
Now we are ready to state and prove Proposition \ref{prop:x2y2_3-Cayley-connected}.  
\begin{proposition}[$(G, \{ux_2,uy_2\})$ is $3$-Cayley-connected]
    \label{prop:x2y2_3-Cayley-connected}
    Consider a winged graph pair $(G,f)$, with $f=uv$, such that $G \cup f$ is $2$-connected.  
    If $(G,f)$ has a chained $f$-winged graph minor with singleton wing tips and an ordered set of pins $\{x_i,y_i\}$, then either $(G,f)$ has the $3$-SIP or $(G, \{ux_2,uy_2\})$ is $3$-Cayley-connected.  
\end{proposition}
\begin{proof}
    Consider the links $L_0$, containing $u$ and $x_1y_1$, and $L_1$, containing $x_1y_1$ and $x_2y_2$, of $[G \cup f]$, and let $L = L_0 \cup L_1$.  
    By Lemma \ref{lem:links_3-connected}, we can assume wlog that $L_0$ and $L_1$ are both $3$-connected.  
    In this case, we will show that $(G, \{ux_2,uy_2\})$ is $3$-Cayley-connected.  
    Since $[G \cup f]$ has singleton wing tips and its pins are edges of $G$, it suffices to show that $(L, \{ux_2,uy_2\})$ is $3$-Cayley-connected.  
    By definition, we must show that the CCS $\Omega^3_{\{ux_2,uy_2\}}(L,\ell)$ is connected for any squared edge-length map $\ell$.  
    Consider the CS $\mathcal{C}^3(L,\ell)$ and its set of connected components $\{X_i\}$.  
    Let $Y_i$ be the image of $X_i$ under the Cayley map $\phi_{ux_2,uy_2}$.  
    Observe that $\Omega^3_{\{ux_2,uy_2\}}(L,\ell)$ is the union of the sets $Y_i$.  
    Furthermore, each $Y_i$ is connected since $X_i$ is connected and $\phi_{ux_2,uy_2}$ is continuous.  
    Therefore, we complete the proof by showing that $Y_i \cap Y_j \neq \emptyset$ for any $i$ and $j$.  
    
    Let $Y_i(ux_2)$ and $Y_j(ux_2)$ be the coordinate projections of $Y_i$ and $Y_j$, respectively, onto $ux_2$.  
    Also, let $Y_i(uy_2)$ and $Y_j(uy_2)$ be the coordinate projections of $Y_i$ and $Y_j$, respectively, onto $uy_2$.  
    Since $Y_i$ is connected, it contains a continuous curve $\tau_i$ whose coordinate projections are $Y_i(ux_2)$ and $Y_i(uy_2)$.  
    Similarly, $Y_j$ contains a continuous curve $\tau_j$ whose coordinate projections are $Y_j(ux_2)$ and $Y_j(uy_2)$.  
    Additionally, Lemma \ref{lem:ux2_uy2_3-covering} can be used to show that $Y_i(ux_2) = Y_j(ux_2)$ and $Y_i(uy_2) = Y_j(uy_2)$.  
    Therefore, $\tau_i$ and $\tau_j$ are contained in $Y_i(ux_2) \times Y_i(uy_2)$.  
    Thus, the Intermediate Value Theorem shows that these curves share some point, and so $Y_i \cap Y_j \neq \emptyset$.  
    This completes the proof.  
\end{proof}
 Finally, we  can complete the proof of the converse direction of Theorem \ref{thm:3-sip_characterization}.  
\begin{proof}[Proof of the converse direction of Theorem \ref{thm:3-sip_characterization}]
    Let $(G,f)$ be a graph-nonedge pair such that no atom of $G \cup f$ that contains $f$ has an $f$-preserving forbidden minor.  
    We will show that $(G,f)$ has the $3$-SIP by induction of the number $n$ of vertices in $G$.  
    When $n \leq 5$, our assumption implies that $G \cup f$ has no forbidden minor, and so Theorem \ref{thm:sitharam_willoughby} shows that $(G,f)$ has the $3$-SIP.  
    Next, assume that the converse direction of Theorem \ref{thm:3-sip_characterization} is true when $n \leq k$, for any integer $k \geq 5$, and we will prove it when $n = k + 1$.  
    If $G \cup f$ has no forbidden minor, then Theorem \ref{thm:sitharam_willoughby} shows that $(G,f)$ has the $3$-SIP.  
    Otherwise, Corollary \ref{cor:atom_gluing} allows us to assume wlog that $G \cup f$ is an atom.  
    Since atoms are $2$-connected, Proposition \ref{prop:G_is_winged_graph} shows that $(G,f)$ is a winged graph pair.  
    Consequently, Propositions \ref{lem:e_not_st}, \ref{lem:winged-graph_cut_vertices}, and \ref{prop:chained_minor} allow us to assume that $(G,f)$ has a chained $f$-winged graph minor $[G \cup f]$ with singleton wing tips.  
    Hence, let $\{x_iy_i\}$ be the ordered set of pins of $[G \cup f]$.  
    Proposition \ref{prop:x2y2_3-Cayley-connected} shows that $(G,\{ux_2,uy_2\})$ is $3$-Cayley-connected, where $f = uv$.  
    Therefore, by Lemma \ref{lem:tool}, it suffices to show that $(G \cup \{ux_2,uy_2\}, f)$ has the $3$-SIP, which wedo next.  
    
    Observe that for some vertex $w \in \{x_1,y_1\}$, $\{u,x_2,y_2\}$ is a clique $wv$-separator of $G \cup \{f,ux_2,uy_2\}$.  
    Hence, by Lemma \ref{lem:gluing}, it suffices to consider the $\{u,x_2,y_2\}$-component $H$ of $G \cup \{f,ux_2,uy_2\}$ that contains $f$ and show that $(H \setminus f, f)$ has the $3$-SIP.  
    If this is not the case, then, since $|V(H)| < |V(G)|$, the inductive hypothesis states that $H$ has an $f$-preserving forbidden minor.  
    However, since both $\{x_1,y_1\}$ and $\{x_2,y_2\}$ are minimal $uv$-separators of $G$, there exist paths in the $\{u,x_2,y_2\}$-component of $G \cup f$ that contains $w$ between $w$ and each vertex in the set $\{u,x_2,y_2\}$, and each path contains exactly one vertex in $\{u,x_2,y_2\}$.  
    This implies that some minor $[G \cup f]_1$ is $H$ with $[f]_1$ as $f$.  
    % the existence of some isomorphism between some minor $[G \cup f]_1$ and $H$ that maps $[f]_1$ to $f$.  
    Consequently, $G \cup f$ has an $f$-preserving forbidden minor, which is a contradiction.  
    Thus, $(H \setminus f, f)$ has the $3$-SIP, and so the proof is complete.  
\end{proof}

\section{Proof of the forward direction of main Theorem \ref{thm:3-sip_characterization}}
\label{sec:forward}

We first outline the argument before filling in the details.  
See Figure \ref{fig:forward-outline} for a proof diagram.  
We use the same notational convention as in the converse direction of the proof of Theorem \ref{thm:3-sip_characterization}: since the notation $(G,f)$ is used to refer to general graph-nonedge pairs in the inductive and other steps of the proof, we designate  the notation $(G_{\star},f_{\star})$ throughout the proof for the particular pair in Theorem \ref{thm:3-sip_characterization}.  
Assume some atom of $G_{\star} \cup f_{\star}$ contains $f_{\star}$ and has an $f_{\star}$-preserving forbidden minor $[G_{\star} \cup f_{\star}]$, which we can ensure is an induced minor.  
We show that $(G_{\star},f_{\star})$ does not have the $3$-SIP.  
Proposition \ref{lem:no_type_2_no_3-sip}   handles the two cases where $G_{\star} \cup f_{\star}$ is one of the forbidden minors $K_5$ or $K_{2,2,2}$.  
Otherwise, if $f_{\star}$ is retained in the induced forbidden minor $[(G_{\star},f_{\star})]$, then  combining this proposition   with Lemma \ref{lem:cm_non-sip} completes the proof.  
Thus most of the proof focuses on the case where $f_{\star}$ is preserved but not retained, i.e., doubled, in a forbidden minor.  
In this case, Lemma \ref{lem:cm_non-sip} does not apply, and so the main obstacle is propagating the non-$3$-SIP certificate from the forbidden minor edges to $f_{\star}$.  

\begin{figure}[htbp]
    \centering
    \begin{subfigure}{0.19\linewidth}
        \centering
        \includegraphics[width=0.5\linewidth]{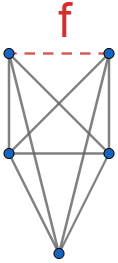}
        \subcaption{}
        \label{fig:k5_f}
    \end{subfigure}
    \begin{subfigure}{0.19\linewidth}
        \centering
        \includegraphics[width=0.8\linewidth]{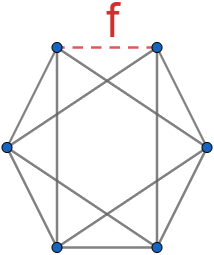}
        \subcaption{}
        \label{fig:k222_f}
    \end{subfigure}
    \begin{subfigure}{0.19\linewidth}
        \centering
        \includegraphics[width=0.6\linewidth]{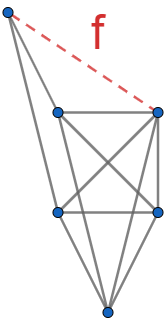}
        \subcaption{}
        \label{fig:k5_f_expanded}
    \end{subfigure}
    \begin{subfigure}{0.19\linewidth}
        \centering
        \includegraphics[width=0.8\linewidth]{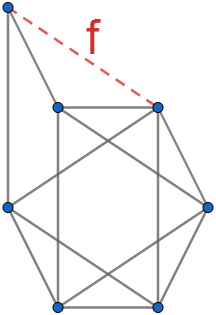}
        \subcaption{}
        \label{fig:k222_f_expanded}
    \end{subfigure}
    \caption{All minimal graph-nonedge pairs with no Type (2) edge.  
    See the discussion, Example \ref{ex:minimal}, Lemma \ref{lem:no_type_2_pairs_no_3-sip}, and Proposition \ref{lem:no_type_2_no_3-sip} in this section.  }
    \label{fig:no_type_2}
\end{figure}

Take for example $(G_{\star},f_{\star})$ to be the pair in either Figure \ref{fig:k5_f_expanded} or \ref{fig:k222_f_expanded}, where exactly one edge $e$ is contracted to reach $[G_{\star} \cup f_{\star}]$ and $[f_{\star}]^{-1}$ contains exactly one edge $f'$ other than $f_{\star}$.   
Proposition \ref{lem:no_type_2_no_3-sip} provides a squared edge-length map $\ell_1$ for $G_{\star} \setminus f'$ with $\ell_1(e) = 0$ and for which the CCS $\Omega^3_{\{f_{\star},f'\}}(G_{\star} \setminus f', \ell_1)$ is two disjoint intervals.  
See the plane where $\ell(e)=0$ in the schematic illustration in Figure \ref{fig:geom_idea}.  
Observe that any section of the CCS that fixes the length of $f'$ contains exactly one point, and so $\ell_1$ cannot witness the non-$3$-SIP property of $f_{\star}$.  
This property on sections holds precisely because $f_{\star}$ is doubled in $[G_{\star} \cup f_{\star}]$.  
We argue that there exists a squared edge-length map $\ell_2$ for $G_{\star} \setminus f'$ that agrees with $\ell_1$, except on $e$ which it maps to some real number $\epsilon > 0$, such that some section of the CCS $\Omega^3_{\{f_{\star},f'\}}(G_{\star} \setminus f', \ell_2)$ that fixes the length of $f'$ is the union of disjoint intervals.  
This section is schematically illustrated as  the intersection of the red fiber and the solid disks in Figure \ref{fig:geom_idea}.  

\medskip\noindent
Thus, letting $\ell'$ be the squared edge-length map for $G_{\star}$ that agrees with $\ell_2$ and sets $\ell'(f')$ to be the length of $f'$ in this section demonstrates that the CCS $\Omega^3_{f_{\star}}(G_{\star}, \ell')$ is two disjoint intervals, which  would complete the proof of the forward direction of the main Theorem \ref{thm:3-sip_characterization}. 

\begin{figure}
    \centering
    % \begin{tikzpicture}
    %      \coordinate (O) at (0,0,0);
    %       \draw[thick,->] (0,0,0) -- (4,0,0) node[below]{$f_2$};
    %       \draw[thick,->] (0,0,0) -- (0,3,0) node[left]{$e$};
    %       \draw[thick,->] (0,0,0) -- (0,0,5) node[below]{$f_1$};

    %       \draw[blue,-] (0.5,-0.5,0) -- (1,-1,0);
    %       \draw[blue,-] (1.5,-1.5,0) -- (2,-2,0);
    % \end{tikzpicture}
    \pgfdeclarelayer{0}
    \pgfdeclarelayer{1}
    \pgfdeclarelayer{2}
    \pgfdeclarelayer{3}
    \pgfdeclarelayer{4}
    \pgfsetlayers{0,1,main,2,3,4}
    \begin{tikzpicture}
          \begin{pgfonlayer}{0}
              % axes
              \draw[thick,->] (0,0,0) -- (4,0,0) node[below]{$\ell(f')$};
              \draw[thick,->] (0,0,0) -- (0,3,0) node[left]{$\ell(e)$};
              \draw[thick,->] (0,0,0) -- (0,0,5) node[below]{$\ell(f_{\star})$};
              
              % epsilon axes
              \draw[thick,->] (0,2,0) -- (4,2,0);
              \draw[thick,->] (0,2,0) -- (0,2,5);
    
              % l(e) axis marks
              \draw[thick,-] (-0.2,0) -- (0,0) node[left=0.1cm]{$\ell(e)=0$};
              \draw[thick,-] (-0.2,2) -- (0,2) node[left=0.1cm]{$\ell(e)=\epsilon$};
          \end{pgfonlayer}

          \begin{pgfonlayer}{1}
              % cone lines
              \draw[blue!20,-] (1.5,-1.5,0) -- (0.68,0.93);
              \draw[blue!20,-] (2,-2,0) -- (2.8,1);
              \draw[color=blue!20] (1.73,1,0) ellipse (1.065 and 0.5);

              \draw[blue!20,-] (0.5,-0.5,0) -- (-0.32,1.93);
              \draw[blue!20,-] (1,-1,0) -- (1.8,2);
              \draw[color=blue!20] (0.73,2,0) ellipse (1.065 and 0.5);
          \end{pgfonlayer}

          \begin{pgfonlayer}{4}
              % continued
              \draw[color=blue!20,dashed] (1.73,1,0) ellipse (1.065 and 0.5);
              \draw[color=blue!20,dashed] (0.73,2,0) ellipse (1.065 and 0.5);
          \end{pgfonlayer}

          % ccs l(e) = 0
          \draw[blue,-] (0.5,-0.5,0) -- (1,-1,0);
          \draw[blue,-] (1.5,-1.5,0) -- (2,-2,0);

          \begin{pgfonlayer}{3}
              % ccs l(e) = epsilon
              \filldraw[color=blue,fill=blue] (0.75,1.5,0) ellipse (0.9 and 0.4);
          \end{pgfonlayer}

          %continued
          \filldraw[color=blue,fill=blue] (1.75,0.5,0) ellipse (0.9 and 0.4);

          % ccs l(e) = 0 fiber
          \filldraw[color=red,fill=red] (1.2,0) circle (0.04);
          \draw[red,-] (1.2,0) -- (-0.35,-2);
          \filldraw[color=red,fill=red] (0.68,-0.68) circle (0.04);

          \begin{pgfonlayer}{3}
              % ccs l(e) = epsilon fiber
              \filldraw[color=red,fill=red] (2,2) circle (0.04);
              \draw[red,-] (2,2) -- (0.6,0.1);
              \draw[red,very thick,-] (1.65,1.52) -- (1.45,1.25);
              \draw[red,very thick,-] (1.10,0.78) -- (0.86,0.45);
          \end{pgfonlayer}
    \end{tikzpicture}
    % \begin{tikzpicture}
    %      \coordinate (O) at (0,0,0);
    %       \draw[thick,->] (0,0,0) -- (4,0,0) node[below]{$\ell(f')$};
    %       \draw[thick,->] (0,0,0) -- (0,3,0) node[left]{$\ell(e)$};
    %       \draw[thick,->] (0,0,0) -- (0,0,5) node[below]{$\ell(f_{\star})$};

    %       \filldraw[color=blue,fill=blue!20] (0.75,-0.75,0) ellipse (0.9 and 0.4);
    %       \filldraw[color=blue,fill=blue!20] (1.75,-1.75,0) ellipse (0.9 and 0.4);

    %       \filldraw[color=red,fill=red] (2.2,0) circle (0.04);
    %       \draw[red,-] (2.2,0) -- (0.7,-2);

    %       \draw[thick,-] (-0.2,0) -- (0,0) node[left=0.1cm]{$\ell(e)=\epsilon$};
    % \end{tikzpicture}
    \caption{See the proof sketch at the beginning of this section.  
    Shown are two solid cones (drawn hollow for clarity) that schematically illustrate the CCS $\Omega^3_{\{f_{\star},f',e\}}(G_{\star} \setminus \{f',e\}, \ell)$.  
    Two sections of this CCS that fix the length $\ell(e)$ for $e$ to be $0$ and $\epsilon > 0$, respectively, are shown in dark-blue: the former consists of two line-segments and the latter consists of two ellipses.  
    These sections  schematically illustrate the CCSs $\Omega^3_{\{f_{\star},f'\}}(G_{\star} \setminus f', \ell_1)$ and $\Omega^3_{\{f_{\star},f'\}}(G_{\star} \setminus f', \ell_2)$, respectively.  
    All sections of $\Omega^3_{\{f_{\star},f'\}}(G_{\star} \setminus f', \ell_1)$ that fix a length for $f'$ consist of a single point, as illustrated by the intersection of the red fiber with the CCS.  
    Some section of $\Omega^3_{\{f_{\star},f'\}}(G_{\star} \setminus f', \ell_2)$ that fixes a length for $f'$ consists of two disjoint intervals, as illustrated by the intersection of the red fiber with the CCS.  
    This section is the CCS $\Omega^3_{f_{\star}}(G_{\star},\ell')$, which demonstrates that $(G_{\star},f_{\star})$ does not have the $3$-SIP.  }
    \label{fig:geom_idea}
\end{figure}
While  this proof sketch is inherently geometric, we lack the tools to prove it directly using properties of the Euclidean distance cone.  
Instead, we  devise a lengthy combinatorial argument that reduces the problem to solving a few elementary geometry problems.  
First, Lemmas \ref{lem:cm_non-sip} and \ref{lem:gluing} (repeatedly used tools of independent interest stated in Section \ref{sec:results}) allow us to assume $(G_{\star},f_{\star})$ is \emph{minimal}  as defined below.
 Then Theorem \ref{thm:no_pair_in_D_has_3-sip}  establishes that minimal pairs do not have $3$-SIP, which then  completes the forward direction  
 of \ref{thm:3-sip_characterization}.

\begin{definition}[Types of edges of a graph-nonedge pair]
    Given a graph-nonedge pair $(G,f)$, the edges of $G$ can be partitioned into the following four classes.  
    Let $[G \cup f]$ be the minor obtained by contracting any one edge of $G$.  
    Then, the contracted edge is of Type:
    \begin{enumerate}
        \item if $[G \cup f]$ has no $[f]$-preserving forbidden minor, 
        \item if $[G \cup f]$ has some $[f]$-preserving forbidden minor, but no atom of $[G \cup f]$ that contains $[f]$ has such a minor, 
        \item if some atom of $[G \cup f]$ that contains $[f]$ has an $[f]$-preserving forbidden minor and $f$ is doubled in $[G \cup f]$, or 
        \item if none of the above apply, in which case the edge is called \emph{reducing}.
    \end{enumerate}
\end{definition}

\begin{definition}[Minimal pairs]
    \label{def:minimal_pair}
    A graph-nonedge pair $(G,f)$ is \emph{minimal} if $G \cup f$ is an atom that has an $f$-preserving forbidden minor and $G$ has no reducing edge.  
\end{definition}

\begin{example}
    \label{ex:minimal}
    A graph-nonedge pair $(G,f)$ is not minimal if it is the pair (i) in Figure \ref{fig:clique_sep} even though $G \cup f$ has an $f$-preserving forbidden minor, since $G \cup f$ is not an atom, or (ii) on the bottom-right of Figure \ref{fig:atoms} even though $G \cup f$ is an atom that has an $f$-preserving forbidden minor, since the edge whose contraction yields an $f$-retaining $K_5$ minor is reducing.  
    $(G,f)$ is minimal if it is any of the pairs in Figure \ref{fig:no_type_2} or if $G \cup f$ is any of the graphs in Figures \ref{fig:min_1}, \ref{fig:min_2}, \ref{fig:min_4}, \ref{fig:min_6}, \ref{fig:min_7}, or \ref{fig:min_9}.  
    \qed
\end{example}

\begin{remark}
    Notice that our definition of minimal pair does not prevent one minimal pair $(G,f)$  from properly containing another minimal pair $(G',f')$ where $G'$ is a proper induced subgraph of $G$.  
    For example, this is true if $G \cup f$ is the graph in Figure \ref{fig:min_2} and $G' \cup f'$ is the graph in Figure \ref{fig:min_1}; or $G \cup f$ is the graph in Figure \ref{fig:min_4} and $G' \cup f'$ is the graph in Figure \ref{fig:min_2} with $f$ and $f'$ being solid green.  
\end{remark}

Next, Theorem \ref{thm:no_pair_in_D_has_3-sip}, below, demonstrates that no minimal pair has the $3$-SIP.  
A squared edge-length map $\ell$ is a \emph{proper non-$d$-SIP} map for $(G,f)$ if $\ell$ maps each edge of $G$ to a positive real number and the CCS $\Omega^3_f(G,\ell)$ contains exactly two distinct positive real numbers.  
For example, a squared edge-length map for the pair in Figure \ref{fig:k5-e} assigns length $1$ to each in one $K_4$ subgraph of $G$ and length $2$ to the remaining three edges of $G$.  

\begin{figure}[b!]
    \centering
    \scalebox{0.8}{
    \hspace*{-3cm}
    \begin{tikzpicture}
        \node[draw] (forward) at (0,0) {Forward direction of Theorem \ref{thm:3-sip_characterization}};
        %-----------------------------------------------------
        \node[draw] (min) at (0,-1.5) {Theorem \ref{thm:no_pair_in_D_has_3-sip}};
        %-----------------------------------------------------
        \node[draw] (p8) at (-1,-4.5) {Proposition \ref{prop:IH}};

        \node[draw] (l39) at (-9,-24.5) {Lemma \ref{lem:G'_is_J+decorations}};
        %=====================================================
        \node[draw, align = left] (l38) at (0,-6.5) {
        Lemma \ref{lem:G'_f'_no_reducible_edge}\\
        {\small \textit{Lemmas \ref{lem:J+H'_atom},}}\\
        {\small \textit{\ref{lem:k_f-preserving_fbm},\ref{lem:I'_f'_cms},\ref{lem:gluing_min_k-clique-sum_graphs},\ref{lem:gluing_top_part_to_atom}}}
        };
        %=====================================================
        \node[draw, align = left] (l36) at (0,-12.5) {
        Lemma \ref{lem:f'_and_J}\\
        {\small \textit{Lemmas \ref{lem:I'_fbm_f_and_C'_nonedges_not_contracted},}}\\
        {\small \textit{\ref{lem:I'_f'_cms},\ref{lem:I'_sep_components_contain_C'-E},\ref{lem:I'_path},\ref{lem:path_in_graph_path_in_atom}}}
        };

        % \node[draw, below right = 1cm and 0cm of l36] (l35) {Lemma \ref{lem:I'_fbm_f_and_C'_nonedges_not_contracted}};
        %=====================================================
        \node[draw, align = left] (l33) at (0,-14.5) {
        Lemma \ref{prop:C'_leq_4vert}\\
        {\small \textit{Lemmas \ref{lem:I'_sep_components_contain_C'-E},}}\\
        {\small \textit{\ref{lem:gluing_min_k-clique-sum_graphs},\ref{lem:path_in_graph_path_in_atom}}}
        };
        %=====================================================
        \node[draw, align = left] (l32) at (0,-16.5) {
        Lemma \ref{lem:C'_5_H'_wing}\\
        {\small \textit{Lemmas \ref{lem:G_not_k5_k222_u4v0},}}\\
        {\small \textit{\ref{lem:k5_k222_not_wing_fm_e_not_contracted},\ref{lem:clq_sep_props}}}
        };
        %=====================================================
        \node[draw, align = left] (l30) at (0,-20) {
        Lemma \ref{cor:top-level_structure}
        };
        %=====================================================
        \node[draw] (l29) at (0,-21.5) {Lemma \ref{lem:type_2_yields_smaller_expanded_f-component}};
        
        \node[draw, align = left] (l27) at (0,-24.5) {
        Lemma \ref{lem:I'_fbm_after_f-component_contraction}\\
        {\small \textit{Lemma \ref{lem:pair_e'_exists}}}
        };

        \node[draw] (l28) at (3,-23) {Lemma \ref{lem:G'_connected}};
        %=====================================================
        \node[draw] (l31) at (-1.5,-18.5) {Lemma \ref{lem:4-clique separator_of_H'_contains_f}};
        %=====================================================
        \node[draw] (l25) at (3,-24.5) {Lemma \ref{lem:expanded_clique_graphs_and_connections}};
        %=====================================================
        \node[draw] (l37) at (0,-8) {Lemma \ref{lem:C'_a-n}};
        %=====================================================
        \node[draw, align = left] (l34) at (3,-9.5) {
        Lemma \ref{prop:expC_3or4}\\
        {\small \textit{Lemmas \ref{lem:3top_uin}-\ref{lem:top_remaining}}}
        };

        \node[draw, align = left] (l51) at (3,-11) {
        Lemma \ref{lem:top-level_contraction_not_type_2}
        };

        \node[draw, align = left] (l50) at (3,-12.5) {
        Lemma \ref{lem:cliques_in_top}
        };

        \node[draw] (l49) at (3,-20) {
        Lemma \ref{lem:nonedge_C'_no_fbm_H'+f'}
        };

        \node[draw, align = left] (l48) at (3,-21.5) {
        Lemma \ref{lem:I'_one_vertex}\\
        {\small \textit{Lemma \ref{lem:path_in_graph_path_in_atom}}}
        };
        %=====================================================
        %-----------------------------------------------------
        \node[draw] (p7) at (-3,-20) {Proposition \ref{lem:expanded_top-level_e-separating_clique_exists}};
        %=====================================================
        \node[draw] (l23) at (-3,-21.5) {Lemma \ref{lem:separating_minor_has_separating_pair}};

        \node[draw, align = left] (l22) at (-3,-23) {Lemma \ref{cor:minimal_k-clique-sum_component_containing_fbm_and_f}\\
        {\small \textit{Lemma \ref{lem:3-flat_fm_containing_e_pass_through_minimal_3-clique-sum_components}}}};

        \node[draw] (l21) at (-3,-24.5) {Lemma \ref{lem:preserving_forbidden_minor_implies_contraction_minor}};
        %=====================================================
        %-----------------------------------------------------
        \node[draw, align = left] (p6) at (-6,-21.5) {
        Proposition \ref{lem:no_type_2_no_3-sip}\\
        {\small \textit{Lemma \ref{lem:no_type_2_neighborhoods}}}
        };
        %=====================================================
        \node[draw] (l19) at (-6,-23) {Lemma \ref{lem:no_type_2_pairs_no_3-sip}};

        \node[draw] (l18) at (-6,-24.5) {Lemma \ref{lem:4-cycle_non-sip}};
        %=====================================================
        % \node[draw, below right = 1cm and -1cm of p6] (l20) {Lemma \ref{lem:no_type_2_neighborhoods}};
        %-----------------------------------------------------
        \node[draw] (p9) at (3,-3) {Proposition \ref{prop:final_top_graphs}};
        %=====================================================
        \node[draw] (l41) at (1.5,-4.5) {Lemma \ref{lem:C'_a-i}};

        \node[draw] (l40) at (-9,-23) {Lemma \ref{lem:J_contains_C'_degree_1}};
        %-----------------------------------------------------
        \node[draw, align = left] (p10) at (6,-23) {
        Proposition \ref{prop:IS}\\
        {\small \textit{Lemma \ref{lem:degree_3_decorations}}}
        };
        %=====================================================
        \node[draw, align = left] (l42) at (4,-4.5) {
        Lemma \ref{lem:C'_degree_1_reducible_3}\\
        {\small \textit{Lemmas \ref{lem:G_not_k5_k222_u4v0},\ref{lem:k5_k222_not_wing_fm_e_not_contracted},}}\\
        {\small \textit{\ref{lem:C'_degree_1_reducible_0}-\ref{lem:C'_degree_1_reducible_2},\ref{lem:gluing_min_k-clique-sum_graphs},\ref{lem:gluing_top_part_to_atom}}}
        };

        % \node[draw, below right = 1cm and -1.8cm of l42, align = left] (l579) {
        % Lemmas \ref{lem:C'_degree_1_reducible_0}-\ref{lem:C'_degree_1_reducible_2}\\
        % {\small \textit{Lemmas \ref{lem:G_not_k5_k222_u4v0},\ref{lem:k5_k222_not_wing_fm_e_not_contracted},\ref{lem:gluing_min_k-clique-sum_graphs}}}\\
        % {\small \textit{\ref{lem:gluing_top_part_to_atom}}}
        % };
        %=====================================================
        % \node[draw, below = of p10] (l24) {Lemma \ref{lem:degree_3_decorations}};
         %-----------------------------------------------------
         % dependencies
         
         %forward
         \draw[->] (min)--(forward);
         
         % min 
         \draw[->] (p6) to [out=90,in=-150] (min);
         \draw[->] (p7) to [out=90,in=-140] (min);
         \draw[->] (p8)--(min);
         \draw[->] (p9)--(min);
         \draw[->] (p10) to [out=70,in=-5] (min);

         %p6
         \draw[->] (l19)--(p6);
         % \draw[->] (l20)--(p6);

         %l19
         \draw[->] (l18)--(l19);

         %p7
         \draw[->] (l23)--(p7);

         %l23
         \draw[->] (l22)--(l23);

          %l22
         \draw[->] (l21)--(l22);

         %p8 
         \draw[->] (l38)--(p8);
         \draw[->] (l39) to [out=135,in=-140] (p8);

         %l38
         \draw[->] (l36) to [out=110,in=-130] (l38);
         \draw[->] (l37)--(l38);
         \draw[->] (l22) to [out=139,in=-130] (l38);

         %l37
         \draw[->] (l34)--(l37);
         \draw[->] (l36)--(l37);

         %l36
         % \draw[->] (l35)--(l36);
         \draw[->] (l22) to [out=139,in=-135] (l36);
         \draw[->] (l33)--(l36);
         \draw[->] (l25) to [out=140,in=-45] (l36);

         %l34
         \draw[->] (l33) to [out=48,in=-150] (l34);
         \draw[->] (l22) to [out=139,in=180] (l34);
         \draw[->] (l30) to [out=55,in=-135] (l34);
         \draw[->] (l51)--(l34);
         \draw[->] (l49) to [out=35,in=-45] (l34);

         %l51
         \draw[->] (l23) to [out=150,in=170] (l51);
         \draw[->] (l33) to [out=40,in=-140] (l51);
         \draw[->] (l50)--(l51);

         %l50
         \draw[->] (l22) to [out=-10,in=-100] (l50);
         \draw[->] (l25) to [out=48,in=-35] (l50);
         \draw[->] (l28) to [out=40,in=-40] (l50);
         \draw[->] (l33)--(l50);
         \draw[->] (l49)--(l50);

         %l49
         \draw[->] (l25) to [out=140,in=-145] (l49);
         \draw[->] (l48)--(l49);

          %l48
         \draw[->] (l28)--(l48);

         %l33
         \draw[->] (l25) to [out=145,in=-50] (l33);
         \draw[->] (l32)--(l33);

         %l32
         \draw[->] (l25) to [out=150,in=-60] (l32);
         \draw[->] (l31)--(l32);
         \draw[->] (l22) to [out=139,in=-180] (l32);
         \draw[->] (l30)--(l32);

         %l31
         \draw[->] (l30)--(l31);

         %l30
         \draw[->] (l27) to [out=135,in=-140] (l30);
         \draw[->] (l29)--(l30);

         %l29
         \draw[->] (l25) to [out=155,in=-55] (l29);
         \draw[->] (l27)--(l29);
         \draw[->] (l28)--(l29);

         %l28
         \draw[->] (l25) -- (l28);

         %p9
         \draw[->] (l41)--(p9);
         \draw[->] (l42)--(p9);

         %l42
         % \draw[->] (l579)--(l42);
         \draw[->] (l22) to [out=139,in=-120] (l42);
         \draw[->] (l38)--(l42);
         \draw[->] (l40)--(l42);

         %l41
         \draw[->] (l37) to [out=30,in=-90] (l41);
         \draw[->] (l38)--(l41);
         \draw[->] (l40) to [out=90,in=-160] (l41);

         %l40
         \draw[->] (l39)--(l40);
         \draw[->] (l25) to [out=157,in=-20] (l40);

         %p10
         % \draw[->] (l24)--(p10);
         \draw[->] (l18) to [out=10,in=-170] (p10);
    \end{tikzpicture}
    }
    % \vspace{-50pt}
    \caption{Proof outline for the forward direction of Theorem \ref{thm:3-sip_characterization}.  
    See the discussion in this section.  }
    \label{fig:forward-outline}
\end{figure}

\begin{theorem}[Every minimal pair has a proper non-$3$-SIP map]
\label{thm:no_pair_in_D_has_3-sip}
Every minimal graph-nonedge pair has a proper non-$3$-SIP map, and therefore does not have the $3$-SIP.  
\end{theorem}

Next, we outline the proof of Theorem \ref{thm:no_pair_in_D_has_3-sip} using Propositions \ref{lem:no_type_2_no_3-sip}-\ref{prop:IS} below, which is by induction on the number $n$ of vertices in $G_{\star}$.  
The base case, where $G _{\star} \cup f_{\star}$ is $K_5$, is handled by Proposition \ref{lem:no_type_2_no_3-sip}.  
In the inductive step, this proposition lists all minimal graph-nonedge pairs with no Type (2) edge, shown in Figure \ref{fig:no_type_2}, and shows that they have proper non-$3$-SIP maps.  
A minimal pair whose graph has a Type (2) has the property that $f_{\star}$ is doubled in any forbidden minor in which it is preserved.  
This is the same predicament as in the above geometric proof sketch, which we now handle by exploiting the structure of minimal pairs with Type (2) edges.  
In particular, we show that such a pair contains nested smaller minimal pairs (containing different nonedges), the smallest being one of the pairs in Figure \ref{fig:no_type_2}.  
Hence, a natural first approach is to consider the set of all minimal pairs that properly contain exactly one minimal pair, and then proceed bottom-up.  
A subset of this set is shown in Figure \ref{fig:level1}; however, it is not clear how to produce the entire set.  
Instead, we proceed top-down as detailed in following example.  
% Examples of this nested structure are illustrated in Figure \ref{fig:forward}.  

\begin{figure}[htbp]
    \centering
    \begin{subfigure}[t]{0.19\linewidth}
        \centering
        \includegraphics[width=0.8\linewidth]{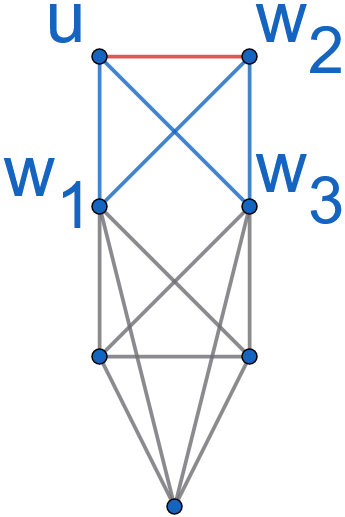}
        \label{fig:level1-1}
    \end{subfigure}
    \begin{subfigure}[t]{0.19\linewidth}
        \centering
        \includegraphics[width=0.8\linewidth]{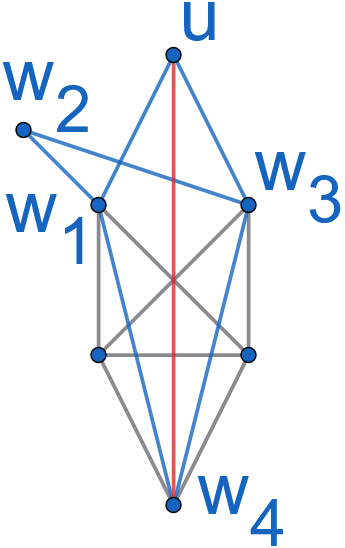}
        \label{fig:level1-2}
    \end{subfigure}
    \begin{subfigure}[t]{0.19\linewidth}
        \centering
        \includegraphics[width=\linewidth]{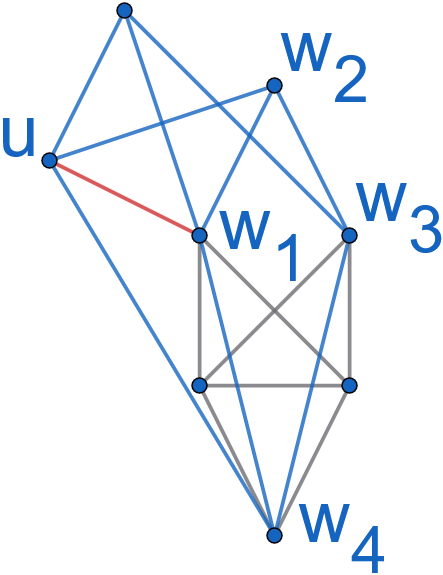}
        \label{fig:level1-3}
    \end{subfigure}
    \begin{subfigure}[t]{0.19\linewidth}
        \centering
        \includegraphics[width=\linewidth]{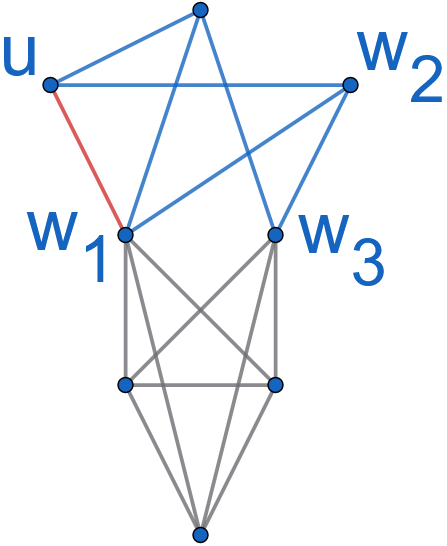}
        \label{fig:level1-4}
    \end{subfigure}
    \caption{Minimal pairs (nonedge in red) that properly contain exactly one minimal pair.  
    From left to right, the subgraphs of red and blue edges are the graphs in Figures \ref{fig:3top_vin_uonly_2}, \ref{fig:cycle_vin_uonly_2}, \ref{fig:cycle_vin_uplus_3}, and \ref{fig:3top_vin_uplus_2}, respectively.  
    See the proof sketch at the beginning of this section.  }
    \label{fig:level1}
\end{figure}

Let $G _{\star} \cup f_{\star}$ be the graph in Figure \ref{fig:min_2} with $f_{\star}$ being the green edge, let $[G _{\star} \cup f_{\star}]$ be the minor shown in Figure \ref{fig:min_3}, which is obtained by contracting a Type (2) edge; and let $C$ and $M$ be the purple and $K_5$ subgraphs of $[G _{\star} \cup f_{\star}]$, respectively.  
Note that $C$ is a CMS that separates $f_{\star}$ from some the forbidden minor $M$.  
Proposition \ref{lem:expanded_top-level_e-separating_clique_exists} guarantees the existence of such a CMS.  
Let $H$ be the $C$-component of $[G _{\star} \cup f_{\star}]$ that contains $f^{\star}$, shown in the dashed black box.  
Consider the subgraphs $C^{-1}$ and $H^{-1}$ of $G _{\star} \cup f_{\star}$, and let $G' = (G _{\star} \cup f_{\star}) \setminus (H^{-1} \setminus C^{-1})$.  
Note that $G'$ contains the smaller minimal pair in Figure \ref{fig:min_1}.  
Choosing the CMS $C$ carefully, as we have, ensures that $H^{-1}$ is not a subgraph of any minimal pair properly contained in $(G_{\star},f_{\star})$.  
Figures \ref{fig:min_4} and \ref{fig:min_5} illustrate CMSs with and without this property.  
This very strong property forces a special structure on $H^{-1}$ described in Proposition \ref{prop:final_top_graphs}: it is one of the graphs in Figures \ref{fig:3top_vin_uonly_2}-\ref{fig:3top_vin_uplus_2}.  
The proof of this proposition is highly technical and is given in Section \ref{prop:final_top_graphs}.  
We also get Proposition \ref{prop:IH}, which states that $(G',f')$ has a proper non-$3$-SIP map $\ell'$, where $f'$ is the nonedge of $C^{-1}$ in our example.  
This argument has reduced our initial problem of showing $(G_{\star}, f_{\star})$ does not have the $3$-SIP to solving elementary geometry problems about the graphs in Figures \ref{fig:3top_vin_uonly_2}-\ref{fig:3top_vin_uplus_2}.  
These problems are detailed and solved in Proposition \ref{prop:IS}, which extends $\ell'$ to a proper non-$3$-SIP map $\ell$ for $(G_{\star}, f_{\star})$. This completes the proof outline of  Theorem \ref{thm:no_pair_in_D_has_3-sip}.

\begin{remark}
    Figure \ref{fig:level1} shows minimal pairs that properly contain exactly one minimal pair.  
    Observe that all but one of the graphs in Figures \ref{fig:3top_vin_uonly_2}-\ref{fig:3top_vin_uplus_2} appear as $H^{-1}$ in these pairs.  
    We were not able to find an example of a minimal pair with $H^{-1}$ being the graph in Figure \ref{fig:chord_vin_uonly_2}, which possibly means that this graph cannot be $H^{-1}$ for any minimal pair or that it only appears as $H^{-1}$ for sufficiently large minimal pairs.  
    Our top-down proof covers both cases.  
\end{remark}

We now prove the above-mentioned Propositions \ref{lem:no_type_2_no_3-sip}-\ref{prop:IS}, starting with Proposition \ref{lem:no_type_2_no_3-sip}, which lists all minimal pairs with no Type (2) edge and shows they each have a proper non-$3$-SIP map.  
The proof requires Lemmas \ref{lem:no_type_2_pairs_no_3-sip} and \ref{lem:no_type_2_neighborhoods}, below.  
Lemma \ref{lem:no_type_2_pairs_no_3-sip} is proved using Lemma \ref{lem:4-cycle_non-sip}, also below.  

\begin{lemma}[Transferring non-$3$-SIP through $K_4$]
    \label{lem:4-cycle_non-sip}
    Let $G$ be obtained from $K_4$ be deleting one of its edges $f$, let $H$ be one of its $K_3$ subgraphs, and let $f'$ be the edge of $H$ that does not share an endpoint with $f$.  
    Also, let $\ell_1$ and $\ell_2$ be squared edge-length maps for $H$ that assign positive real numbers to each edge and differ only in the numbers they assign to $f'$.  
    If the linkages $(H,\ell_1)$ and $(H,\ell_2)$ both have some $3$-realization, then there exist squared edge-length maps $\ell'_1$ and $\ell'_2$ for $G$ such that 
    \begin{enumerate}[(i)]
        \item they agree with each other on the edges not in $H$ and assign these edges positive real numbers, 
        \item $\ell'_i$ agrees with $\ell_i$ on the edges of $H$, and 
        \item the CCSs $\Omega^3_f(G,\ell'_1)$ and $\Omega^3_f(G,\ell'_2)$ each contain exactly one positive value, and these values are distinct.
    \end{enumerate}
\end{lemma}

\begin{proof}
    Let the vertex set of $G$ be $\{v_i\}$, where $H$ does not contain $v_1$, $f = v_1v_2$, and $f' = v_3v_4$.  
    % Also, let $a = \ell_1(v_2v_3)$, $b = \ell_1(v_2v_4)$, $c = \ell_1(v_3v_4)$, and $d = \ell_2(v_3v_4)$.  
    % Wlog, assume that $c < d$.  
    Set $\ell_1'(v_1v_3)$ and $\ell_1'(v_1v_4)$ to be positive numbers such that $\ell_1(v_3v_4) = |\ell_1'(v_1v_3) - \ell_1'(v_1v_4)|$ and $\ell_2(v_3v_4) = \ell_1'(v_1v_3) + \ell_1'(v_1v_4)$.  
    By swapping the values of $\ell_1'(v_1v_3)$ and $\ell_1'(v_1v_4)$, we can additionally ensure that $\ell_1'(v_1v_3) \neq \ell_1(v_2v_3)$.  
    It is easy to see that Conditions (i)-(iii) are satisfied, and so the lemma is proved.  
\end{proof}

\begin{lemma}[Every pair in Figure \ref{fig:no_type_2} has a proper non-$3$-SIP map]
    \label{lem:no_type_2_pairs_no_3-sip}
    Every graph-nonedge pair in Figure \ref{fig:no_type_2} has a proper non-$3$-SIP map.  
\end{lemma}

\begin{proof}
    For the graph-nonedge pair $(G,f)$ in Figure \ref{fig:k5_f}, let $\ell$ be the squared edge-length map that assigns $1$ to each edge of $G$.  
    Then, $\Omega^3_f(G,\ell)$ clearly contains exactly two points.  
    A proper non-$3$-SIP map $\ell'$ for the graph-nonedge pair $(G',f')$ in Figure \ref{fig:k222_f} is constructed in the proof of Theorem 8 in \cite{belk2007realizability1}.  
    Lemma \ref{lem:4-cycle_non-sip}, $\ell$, and the two values in $\Omega^3_f(G,\ell)$ can be used to construct a non-$3$-SIP squared edge-length map for the graph-nonedge pair in Figure \ref{fig:k5_f_expanded}.  
    Similarly, Lemma \ref{lem:4-cycle_non-sip}, $\ell'$, and the two points in $\Omega^3_{f'}(G',\ell')$ can be used to construct a non-$3$-SIP squared edge-length map for the graph-nonedge pair in Figure \ref{fig:k222_f_expanded}.  
\end{proof}

\begin{lemma}[No Type (2) edge properties]
    \label{lem:no_type_2_neighborhoods}
    Let $(G,f)$ be a minimal graph-nonedge pair, where $f=uv$, let $[G \cup f]$ be an $f$-preserving forbidden minor, and let $H = [G \cup f] \setminus \{[u],[v]\}$.  
    If $G$ has no Type (2) edge, then the following statements are true:
    \begin{enumerate}[(i)]
        \item any edge of $G$ that does not share an endpoint with $f$ is preserved in $[G \cup f]$, 
        \item for any vertex $[w] \in V(H)$, $|[w]^{-1}|=1$, 
        \item any vertex $x \in ([u]^{-1} \cup [v]^{-1}) \setminus \{u,v\}$ neighbors $u$, $v$, and at least two vertices in the subgraph $H^{-1}$ of $G \cup f$, but no vertex in $[x]^{-1} \setminus \{u,v\}$, and 
        \item the sets $[u]^{-1}$ and $[v]^{-1}$ each have cardinality at most $2$.
    \end{enumerate}
\end{lemma}

The proof of Lemma \ref{lem:no_type_2_neighborhoods} is given in \ref{sec:prop_6}.  

\begin{proposition}[No Type (2) edge implies non-$3$-SIP map]
    \label{lem:no_type_2_no_3-sip}
    Figure \ref{fig:no_type_2} shows all minimal graph-nonedge pairs that have no Type (2) edge, each of which has a proper non-$3$-SIP map.  
\end{proposition}

The proof of Proposition \ref{lem:no_type_2_no_3-sip}  involves a case analysis detailed in \ref{sec:prop_6}.  We give the essential ideas and structure here.
It is easy to check that the pairs in Figure \ref{fig:no_type_2} are minimal, and Lemma \ref{lem:no_type_2_pairs_no_3-sip} show that they each have a proper non-$3$-SIP map.  
It remains to show that these are the only minimal pairs with no Type (2) edge.  
Lemma \ref{lem:no_type_2_neighborhoods} provides a means to partition all candidate minimal pairs, after which the proof is completed by exhaustive case checking.  In the inductive step, Proposition  \ref{lem:no_type_2_no_3-sip} allows us to assume that $G_{\star}$ has a Type (2) edge.  
We now present several definitions that will allow us to exploit the presence of such an edge.  

For any induced subgraphs $H$ and $H'$ of a graph $G$, $H \setminus H'$ is the subgraph of $G$ induced by $V(H) \setminus V(H')$.  
For any edge $e$ of $G$ and minor $[G]_1$, $[G]_1$ is \emph{$e$-separating} if it has an $[e]_1$-preserving forbidden minor but none of its atoms that contain $[e]_1$ has such a minor.  
% For the following, refer to Figure \ref{fig:minors}.  
% $[H]_1$ is the subgraph of $[G]_1$ induced by $[V(H)]$.  
% For any induced subgraph $J$ of $[G]_1$, $J^{-1}$ is the subgraph of $G$ induced by $\bigcup_{u \in J} u^{-1}$.  
% Note that $[H]_1^{-1}$ is not necessarily $H$.  
For any minor $[H]_2$, $[G]_1$ is an \emph{extension} of $[H]_2$ if $|V([G]_1)| = |V([H]_2)|$ and, for any two vertices $u,v \in V(H)$, $[u]_1 = [v]_1$ if and only if $[u]_2 = [v]_2$.  
If $H$ contains $e$, then $e$ being contracted, preserved, or retained in $[H]_2$ of $H$ are well-defined concepts (see Section \ref{sec:minor-notation}). 
We extend this terminology also for $e \not\in H$: $e$ is \emph{strongly contracted in $[H]_2$} if it is contracted in every extension of $[H]_2$ and \emph{weakly preserved (resp. retained) in $[H]_2$} if it is preserved (resp. retained) in at least one extension of $[H]_2$.
 
 \begin{note}
      When the context is clear, we will drop ``strongly/weakly" from the above terminology; for example, we will simply refer to an $e$-preserving minor of $H$ even when $e$ is not in $H$.
 \end{note}

\begin{example}
    \label{ex:extension}
    Let $G \cup f$ be the graph in Figure \ref{fig:min_4}, where $f$ is the green edge, and let $[G \cup f]$ be its minor shown in Figure \ref{fig:min_5}.  
    Observe that the subgraph $H$ of $G \cup f$ induced by the bottom five vertices and the vertex incident to the two blue edges has an $f$-preserving $K_5$ minor.  
    $[G \cup f]$ is an $f$-separating minor since its $K_5$ subgraph is its only atom that has an $f$-preserving forbidden minor.  
    Lastly, let $J$ be the subgraph of $G \cup f$ induced by the bottom five vertices.  
    Note that $[J]^{-1} = H$.  
    \qed
\end{example}

\begin{definition}[Expanded edge-separating pairs and CMSs]
    \label{def:expanded_e-separating_pair}
    Let $G$ be a graph and $e$ be any of its edges.  
    An \emph{$e$-separating pair $(C,M)$ of an $e$-separating minor $[G]$} contains an atom $M$ of $[G]$ with an $[e]$-preserving forbidden minor and a CMS $C$ of $[G]$ that is, for some endpoint $[u]$ of $[e]$ and any vertex $[x]$ of $M \setminus C$, a minimal clique $[ux]$-separator.  
    Consider the $C$-components $I$ and $H$ of $[G]$ that contain $M$ and $[e]$, respectively.  
    $C$, $I$, and $H$ are the \emph{CMS}, \emph{minor-component}, and \emph{$e$-component} of $(C,M)$, respectively.  
    
    Next, the pair $(C^{-1},M^{-1})$ is an \emph{expanded $e$-separating pair of $G$}, and $C^{-1}$, $I^{-1}$, and $H^{-1}$ are the \emph{expanded CMS}, \emph{expanded minor-component}, and \emph{expanded $e$-component} of $(C^{-1},M^{-1})$, respectively.   
    $(C^{-1},M^{-1})$ is a \emph{top-level} expanded $e$-separating pair of $G$, and is said to be \emph{in the top-level of $G$}, if no proper subgraph of $H^{-1}$ is the expanded $e$-component of any expanded $e$-separating pair of $G$.  
    $(C,M)$ is a \emph{top-level} $e$-separating pair of $[G]$ if $(C^{-1},M^{-1})$ is in the top-level of $G$.  

    Finally, an \emph{$e$-separating CMS $E$ of $[G]$} is a CMS in some $e$-separating pair $(E,M)$ of $[G]$.  
    We treat the minor-component and $e$-component of $(E,M)$ as belonging to $E$.  
    $E^{-1}$ is an \emph{expanded $e$-separating CMS of $G$}, and is in the \emph{top-level} if $(E^{-1},M^{-1})$ is in the top-level.  
\end{definition}

\begin{example}
    \label{ex:cms}
    Refer to the above proof sketch of Theorem \ref{thm:no_pair_in_D_has_3-sip} where the graph $G \cup f$ in Figure \ref{fig:min_2} and its minor $[G \cup f]$ in Figure \ref{fig:min_3} are discussed.  
    Observe that $(C,M)$ is an $f$-separating pair of $[G \cup f]$, and hence $C$ is an $f$-separating CMS.  
    The $f$-component of $C$ is $H$ and the minor component is $M$.  
    The subgraph $C^{-1}$ of $G \cup f$ is a top-level expanded $f$-separating CMS whose expanded $f$-component and expanded minor component are $H^{-1}$ and $M^{-1}$, respectively.  
    Note that the orange subgraph of $G \cup f$ is also a top-level expanded $f$-separating CMS whose expanded $f$-component and expanded minor component are $H^{-1}$ and $M^{-1}$, respectiveely.  

    Next, the above statements are true if we let $G \cup f$ be the graph in Figure \ref{fig:min_4}, where $f$ is the green edge, $[G \cup f]$ be the minor shown in Figure \ref{fig:min_5}, and $C$ and $M$ be the blue and $K_5$ subgraphs of $[G \cup f]$, respectively, except that $C^{-1}$ is not in the top-level of $G \cup f$.  
    To verify this, note that the purple subgraph $E^{-1}$ of $G \cup f$ is an expanded $f$-separating CMS whose expanded $f$-component $J^{-1}$ is the subgraph in the dashed black box, which is a proper subgraph of $H^{-1}$.  
    It is easy to check that $E^{-1}$ is in the top-level of $G \cup f$.  
    Observe that the degree two vertex whose neighbors are the vertices of $C$ is not contained in $H$ or $M$.  

    Finally, the above statements are true for the pairs of graphs in Figures \ref{fig:min_7} and \ref{fig:min_8} and in Figures \ref{fig:min_9} and \ref{fig:min_10}, except, in the latter case, $E^{-1}$ is the connected subgraph consisting of one purple edge and one blue edge.  
    \qed
\end{example}

Contracting any Type (2) edge of $G_{\star} \cup f_{\star}$ yields an $f_{\star}$-separating minor by definition.  
Hence, Proposition \ref{lem:expanded_top-level_e-separating_clique_exists}, below, 
% stated below and proved in \ref{sec:expanded_top-level_e-separating_clique_exists}, 
shows that $G_{\star} \cup f_{\star}$ has a top-level expanded $f_{\star}$-separating CMS.  
We extract three repeatedly used tools from the proof of this proposition, stated as Lemmas \ref{lem:preserving_forbidden_minor_implies_contraction_minor} - \ref{lem:separating_minor_has_separating_pair}, below.  

\begin{lemma}[Edge-preserving forbidden minor implies edge-preserving induced forbidden minor]
    \label{lem:preserving_forbidden_minor_implies_contraction_minor}
    If an edge of a graph is preserved in some forbidden minor, then some such minor is induced.  
\end{lemma}

\begin{proof}
    If an edge $e$ of a graph $G$ is preserved in some $K_5$ minor, then it is clearly preserved in some induced $K_5$ minor.  
    Otherwise, if $e$ is preserved in some $K_{2,2,2}$ minor, then consider the minor $[G]$ obtained by performing all contractions between $G$ and this $K_{2,2,2}$ minor.  
    If $[G]$ is $K_{2,2,2}$, then we are done.  
    Otherwise, $[G]$ has $K_{2,2,2}$ as a proper spanning subgraph, in which case $[G]$ has an $[e]$-preserving $K_5$ minor.  
\end{proof}

\begin{lemma}[Edge-preserving forbidden minors extend minors of atoms]
    \label{cor:minimal_k-clique-sum_component_containing_fbm_and_f}
    If an edge of a graph $G$ is preserved in some forbidden minor, then some such minor is an extension of a minor of some atom of $G$.  
\end{lemma}

\begin{proof}
    Observe that $K_5$ and $K_{2,2,2}$ are both atoms.  
    Hence, the lemma follows by applying Lemma \ref{lem:preserving_forbidden_minor_implies_contraction_minor} followed by Lemma \ref{lem:3-flat_fm_containing_e_pass_through_minimal_3-clique-sum_components} in \ref{sec:atoms}.  
\end{proof}

\begin{lemma}[Graphs with edge-separating forbidden minors have an edge-separating pair]
    \label{lem:separating_minor_has_separating_pair}
    Given a graph $G$ and one of its edges $e$, if $G$ has an $e$-separating forbidden minor, then that minor has an $e$-separating pair.  
\end{lemma}

\begin{proof}
    By definition, any $e$-separating forbidden minor $[G]$ has an $[e]$-preserving forbidden minor but none of its atoms that contain $[e]$ have such a minor.  
    Hence, Lemma \ref{cor:minimal_k-clique-sum_component_containing_fbm_and_f} shows that some atom $M$ of $[G]$ does not contain some endpoint $[u]$ of $[e]$ and has an $[e]$-preserving forbidden minor.  
    Therefore, there must exist some CMS $C$ of $[G]$ that is, for any vertex $[x]$ of $M \setminus C$, a clique minimal $[ux]$-separator and such that $M \setminus C$ is non-empty.  
    Thus, $(C,M)$ is an $e$-separating pair of $[G]$.  
\end{proof}

\begin{proposition}[Graphs with edge-separating minors have an expanded edge-separating CMS]
    \label{lem:expanded_top-level_e-separating_clique_exists}
    Given a graph $G$ and one of its edges $e$, if $G$ has an $e$-separating minor, then $G$ has a top-level expanded $e$-separating CMS.  
\end{proposition}

\begin{proof}
    Lemma \ref{lem:separating_minor_has_separating_pair} states that any $e$-separating minor has some $e$-separating pair.  
    Thus, $G$ clearly has a top-level expanded $e$-separating pair, and hence a top-level expanded $e$-separating CMS.  
\end{proof}

Next, let $C^{-1}$ be any top-level expanded $f_{\star}$-separating CMS of $G_{\star} \cup f_{\star}$, let $H^{-1}$ be its expanded $f_{\star}$-component, and let $G' = G_{\star} \setminus (H^{-1} \setminus C^{-1})$.  
Proposition \ref{prop:IH}, below, demonstrates the nested structure of minimal pairs: $G'$ contains a minimal pair whose nonedge $f'$ is contained in $C^{-1}$, and $G'$ can be constructed from this pair by repeatedly adding degree two vertices.  
The inductive hypothesis then gives a proper non-$3$-SIP map $\ell'$ for $(G',f')$.  
Combining this with Proposition \ref{prop:final_top_graphs}, which lists all possible graphs for $H^{-1}$, allows us to extend $\ell'$ to a proper non-$3$-SIP map $\ell$ for $(G_{\star}, f_{\star})$.  
This extension process is detailed in Proposition \ref{prop:IS}, also below.  

Before stating these propositions and discussing their proofs, we give examples.  
If $G_{\star} \cup f_{\star}$ is the graph in Figure \ref{fig:min_2} (resp. \ref{fig:min_7}), where $f_{\star}$ is the green edge, then we can take $G' \cup f'$ to be the graph in Figure \ref{fig:min_1} (resp. \ref{fig:min_6}), where $f'$ is the green edge.  
Observe that $\ell$ can be chosen to agree with $\ell'$ on edges of $G'$ and to assign squared lengths to the two remaining edges of $G_{\star}$ whose sum and difference are the two values in the CCS $\Omega^3_{f'}(G',\ell')$.  
This ensures that $\ell$ is a proper non-$3$-SIP map for $(G_{\star}, f_{\star})$.  
Similarly, $\ell$ can be extended to obtain a non-$3$-SIP map for the pair in Figure \ref{fig:min_4} (resp. \ref{fig:min_9}).  

\begin{figure}[htb]
    \centering
    \begin{subfigure}{0.19\linewidth}
        \centering
        \includegraphics[width=0.5\linewidth]{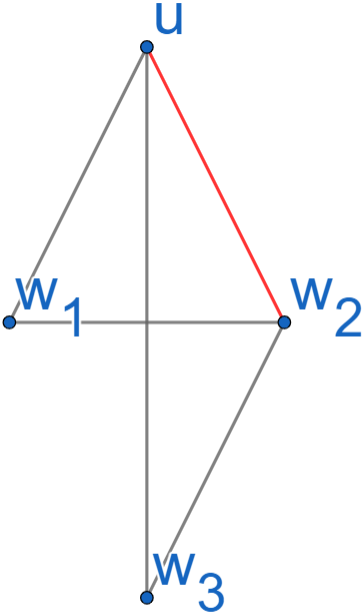}
        \subcaption{}
        \label{fig:3top_vin_uonly_2}
    \end{subfigure}
    \begin{subfigure}{0.19\linewidth}
        \centering
        \includegraphics[width=0.8\linewidth]{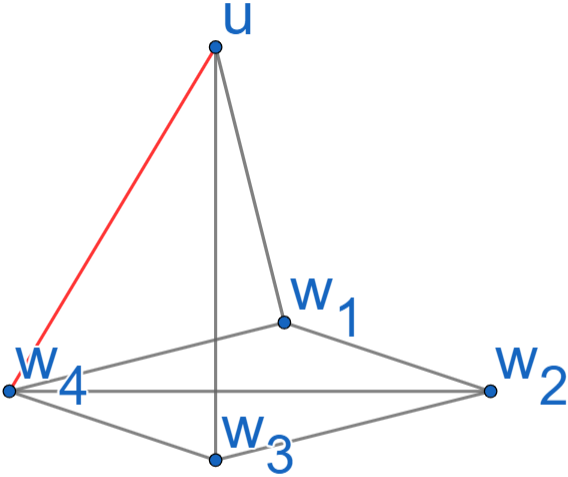}
        \subcaption{}
        \label{fig:chord_vin_uonly_2}
    \end{subfigure}
    \begin{subfigure}{0.19\linewidth}
        \centering
        \includegraphics[width=0.8\linewidth]{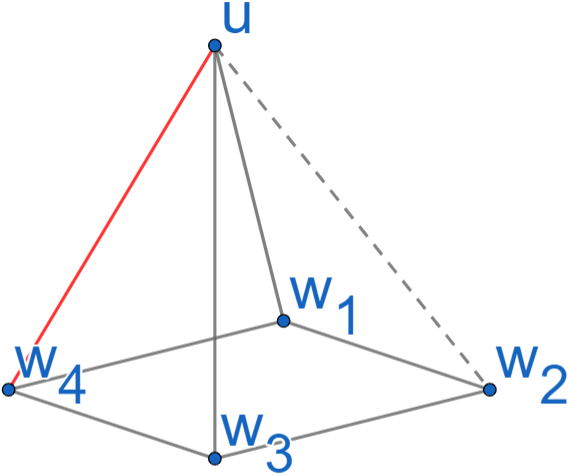}
        \subcaption{}
        \label{fig:cycle_vin_uonly_2}
    \end{subfigure}
    \begin{subfigure}{0.19\linewidth}
        \centering
       \includegraphics[width=0.8\linewidth]{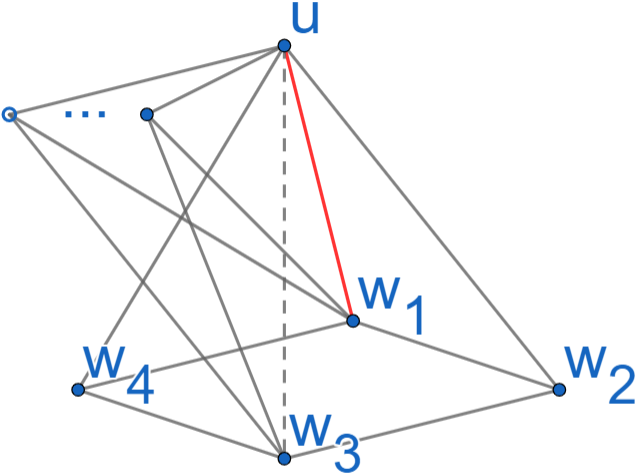}
        \subcaption{}
        \label{fig:cycle_vin_uplus_3}
    \end{subfigure}
    \begin{subfigure}{0.19\linewidth}
        \centering
        \includegraphics[width=0.7\linewidth]{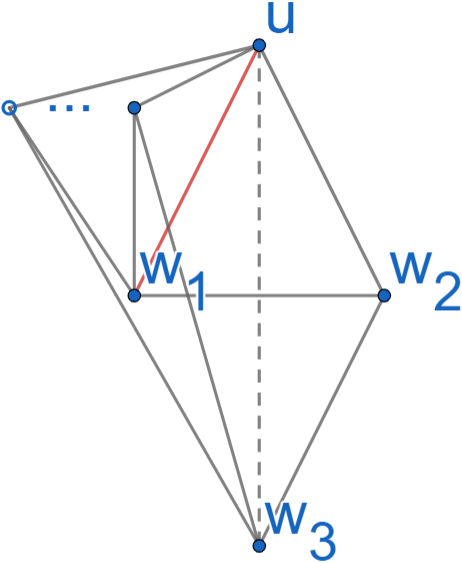}
        \subcaption{}
        \label{fig:3top_vin_uplus_2}
    \end{subfigure}
    \begin{subfigure}{0.19\linewidth}
        \centering
        \includegraphics[width=0.8\linewidth]{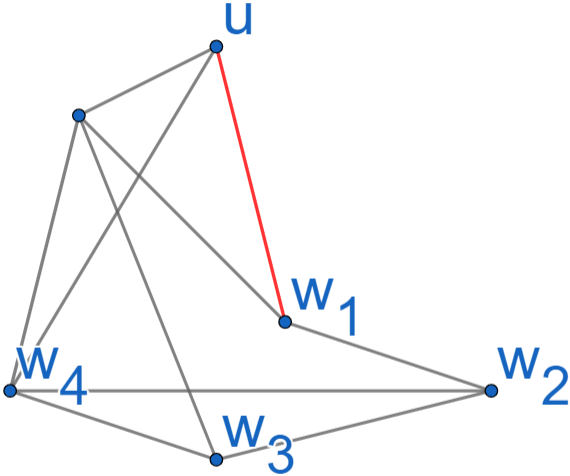}
        \subcaption{}
        \label{fig:deg1_vin_uplus_1}
    \end{subfigure}
    \begin{subfigure}{0.19\linewidth}
        \centering
        \includegraphics[width=0.8\linewidth]{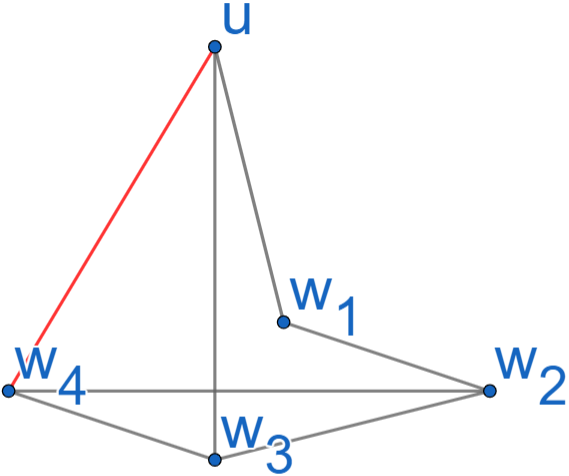}
        \subcaption{}
        \label{fig:deg1_vin_uonly_2}
    \end{subfigure}
    \begin{subfigure}{0.19\linewidth}
        \centering
        \includegraphics[width=0.8\linewidth]{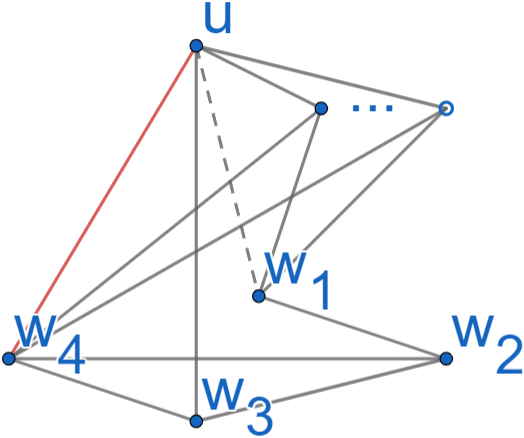}
        \subcaption{}
        \label{fig:deg1_vin_uplus_5}
    \end{subfigure}
    % \begin{subfigure}{0.19\linewidth}
    %     \centering
    %     \includegraphics[width=0.8\linewidth]{deg1_vin_uplus_6.png}
    %     \subcaption{}
    %     \label{fig:deg1_vin_uplus_6}
    % \end{subfigure}
    \begin{subfigure}{0.19\linewidth}
        \centering
        \includegraphics[width=0.9\linewidth]{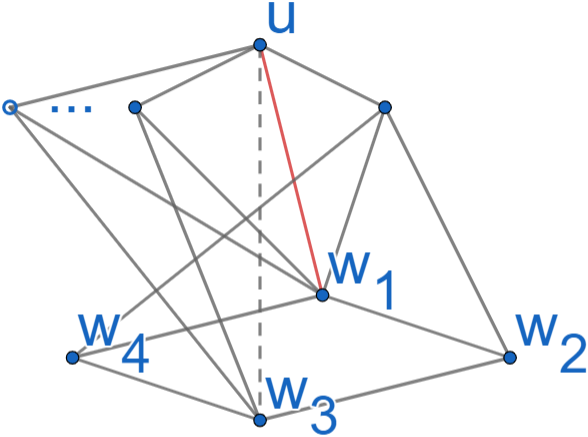}
        \subcaption{}
        \label{fig:cycle_vin_uplus_2}
    \end{subfigure}
    \begin{subfigure}{0.19\linewidth}
        \centering
        \includegraphics[width=0.8\linewidth]{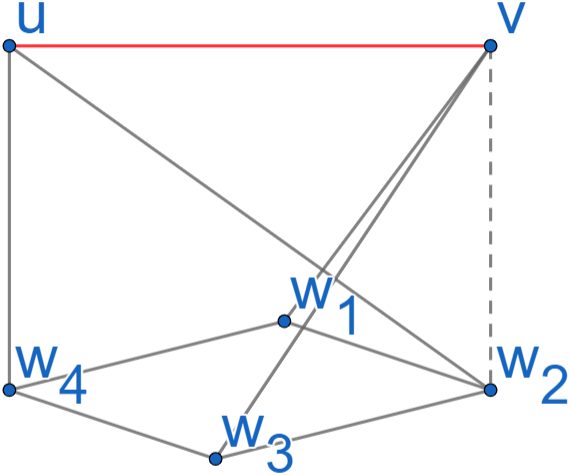}
        \subcaption{}
        \label{fig:cycle_vout_uvonly}
    \end{subfigure}
    \begin{subfigure}{0.19\linewidth}
        \centering
        \includegraphics[width=\linewidth]{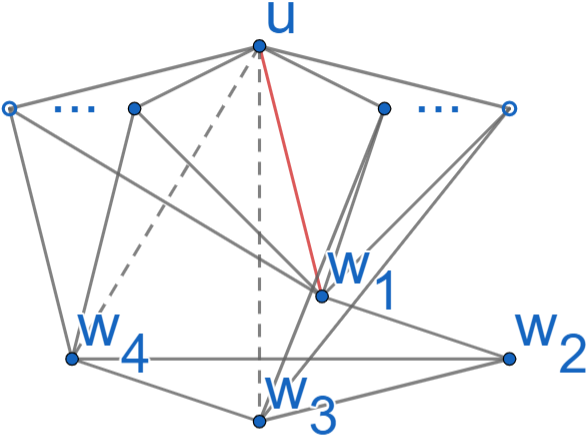}
        \subcaption{}
        \label{fig:deg1_vin_uplus_3}
    \end{subfigure}
    \begin{subfigure}{0.19\linewidth}
        \centering
        \includegraphics[width=0.9\linewidth]{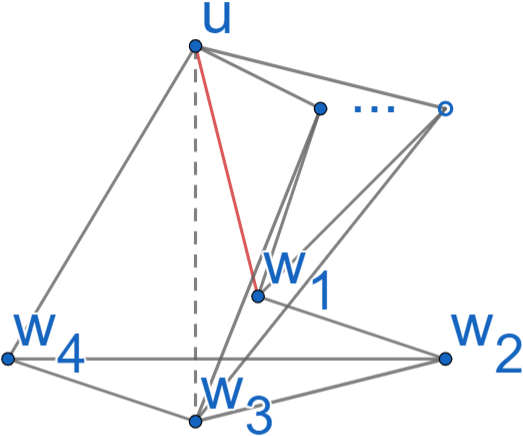}
        \subcaption{}
        \label{fig:deg1_vin_uplus_4}
    \end{subfigure}
    \begin{subfigure}{0.19\linewidth}
        \centering
        \includegraphics[width=0.8\linewidth]{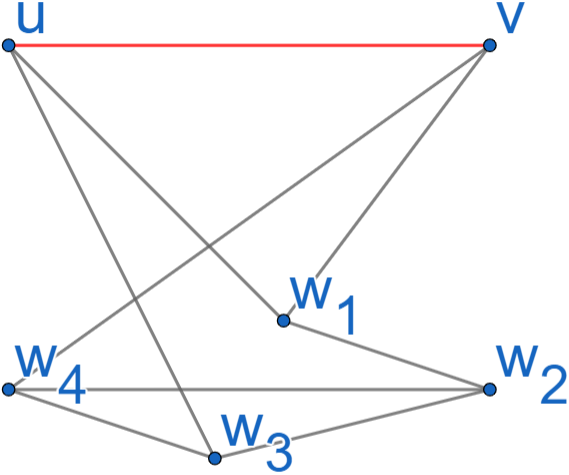}
        \subcaption{}
        \label{fig:deg1_vout_uvonly}
    \end{subfigure}
    \begin{subfigure}{0.19\linewidth}
        \centering
        \includegraphics[width=0.6\linewidth]{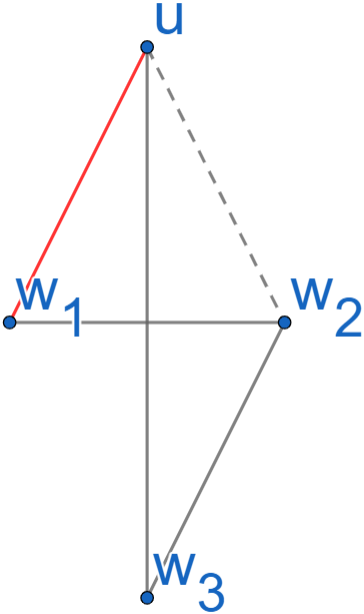}
        \subcaption{}
        \label{fig:3top_vin_uonly_1}
    \end{subfigure}
    \begin{subfigure}{0.19\linewidth}
        \centering
        \includegraphics[width=0.8\linewidth]{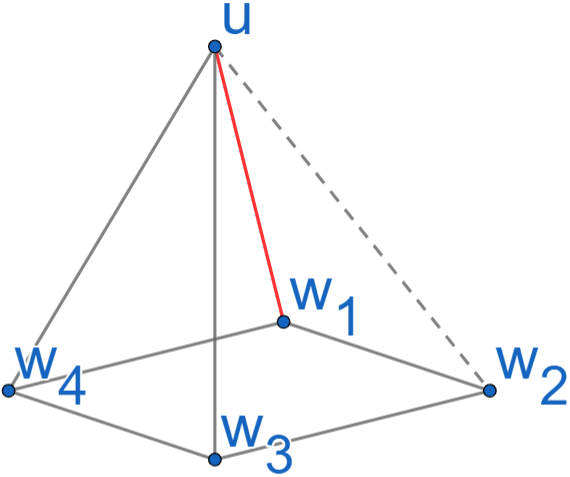}
        \subcaption{}
        \label{fig:cycle_vin_uonly_1}
    \end{subfigure}
    \begin{subfigure}{0.19\linewidth}
        \centering
        \includegraphics[width=0.8\linewidth]{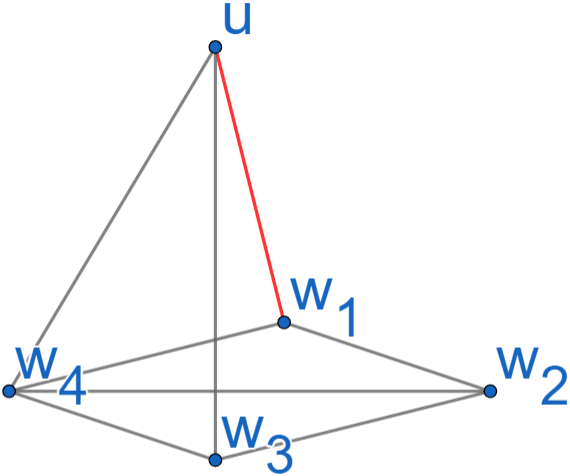}
        \subcaption{}
        \label{fig:chord_vin_uonly_1}
    \end{subfigure}
    \begin{subfigure}{0.19\linewidth}
        \centering
        \includegraphics[width=0.8\linewidth]{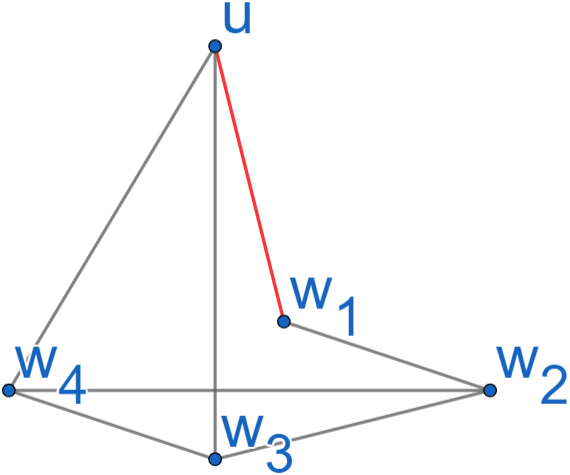}
        \subcaption{}
        \label{fig:deg1_vin_uonly_1}
    \end{subfigure}
    % \begin{subfigure}{0.19\linewidth}
    %     \centering
    %     \includegraphics[width=0.8\linewidth]{3top_vin_uplus_1.png}
    %     \subcaption{}
    %     \label{fig:3top_vin_uplus_1}
    % \end{subfigure}
    % \begin{subfigure}{0.19\linewidth}
    %     \centering
    %     \includegraphics[width=0.8\linewidth]{deg1_vin_uplus_2.png}
    %     \subcaption{obsolete}
    %     \label{fig:deg1_vin_uplus_2}
    % \end{subfigure}
    \begin{subfigure}{0.19\linewidth}
        \centering
        \includegraphics[width=0.8\linewidth]{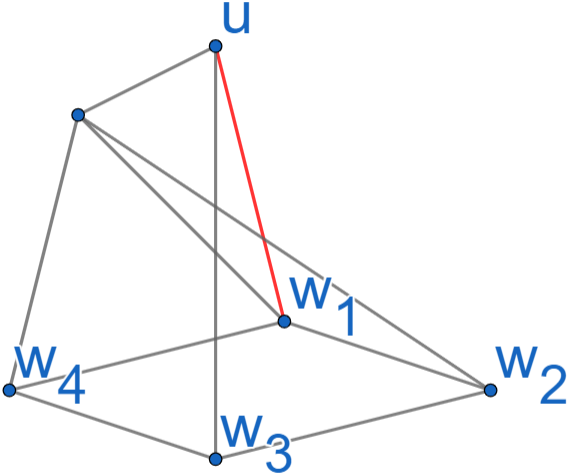}
        \subcaption{}
        \label{fig:cycle_vin_uplus_4}
    \end{subfigure}
    % \begin{subfigure}{0.19\linewidth}
    %     \centering
    %     \includegraphics[width=0.8\linewidth]{cycle_vin_uplus_1.png}
    %     \subcaption{missing edge}
    %     \label{fig:cycle_vin_uplus_1}
    % \end{subfigure}
    \caption{See the discussion, Proposition \ref{prop:IS}, and Lemma \ref{prop:expC_3or4} in this section. All candidate expanded $f$-components of a top-level expanded $f$-separating CMS that have fewer than five vertices.  
    We argue that all except (a)-(e) can be eliminated, and Figure \ref{fig:level1} shows instantiations of (a), (c), (d), and (e).  
    Dashed line-segments may be edges or nonedges and hollow vertices may or may not be present.  
    In (e), (h), (i), (k), or (l) either one of the dashed line-segments is an edge or at least one of the hollow vertices is present.  
   }
    \label{fig:top-level_H'_cases}
\end{figure}

\begin{proposition}[Minimal pair inductive structure]
    \label{prop:IH}
    Let $(G,f)$ be a minimal graph-nonedge pair with a top-level expanded $f$-separating CMS $C^{-1}$ whose expanded $f$-component is $H^{-1}$, and let $G' = G \setminus (H^{-1} \setminus C^{-1})$.  
    There exists a nonedge $f' = uv$ of $C^{-1}$ and an induced subgraph $J$ of $G'$ that contains $f'$ such that $(J,f')$ is minimal and $G'$ can be obtained from $J$ by adding some vertices whose neighborhoods are the set $\{u,v\}$.  
    Furthermore, these properties are true when (i) $f' = w_1w_3$ and $H^{-1}$ is the graph in Figure \ref{fig:cycle_vin_uonly_2}, $f$ is the red edge, and $C^{-1}$ is the subgraph induced by $\{w_i\}$, or (ii) $f'$ is any nonedge of $C^{-1}$ and $H^{-1}$ is the graph in Figure \ref{fig:cycle_vin_uplus_3}, $f$ is the red edge, and $C^{-1}$ is the subgraph induced by $\{w_i\}$.  
\end{proposition}

\begin{proposition}[Pruned possible graphs for top-level $H^{-1}$]
    \label{prop:final_top_graphs}
    Let $H^{-1}$ be the expanded $f$-component of any top-level expanded $f$-separating CMS $C^{-1}$ of a minimal graph-nonedge pair.  
    Then, $H^{-1}$ is one of the graphs in Figures \ref{fig:3top_vin_uonly_2} - \ref{fig:3top_vin_uplus_2} such that the dashed line-segment is a nonedge, $f$ is the red edge, and $C^{-1}$ is the subgraph induced by $\{w_i\}$.  
\end{proposition}

The proofs of Propositions \ref{prop:IH} and \ref{prop:final_top_graphs} are very technical and involve extensive case analyses, and so we devote a separate Section \ref{sec:prop-8-9} to them.  

Proposition \ref{prop:IS} is proved using Lemma \ref{lem:degree_3_decorations}, below.  

\begin{figure}[htb]
    \centering
    \begin{subfigure}{0.49\textwidth}
        \centering
        \includegraphics[width = 0.7\textwidth]{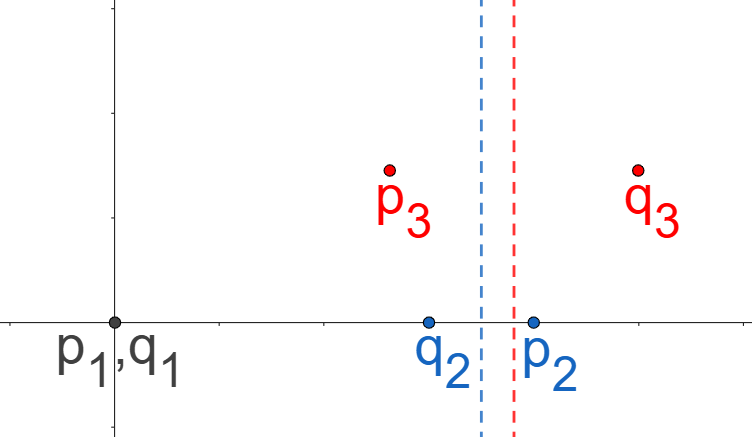}
        \subcaption{}
        \label{fig:triangle_1}
    \end{subfigure}
    \begin{subfigure}{0.49\textwidth}
        \centering
        \includegraphics[width = 0.7\textwidth]{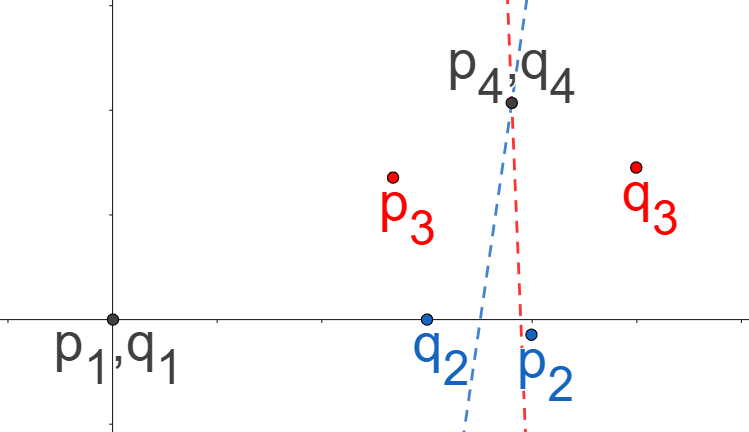}
        \subcaption{}
        \label{fig:triangle_2}
    \end{subfigure}
    \caption{Realizations in the proof of Lemma \ref{lem:degree_3_decorations}.  }
    \label{fig:triangle}
\end{figure}

\begin{lemma}[Degree $3$ decorations preserve non-$3$-SIP]
    \label{lem:degree_3_decorations}
    Let $G$ be $K_4$ and $H$ be one of its $K_3$ subgraphs.  
 Let $\ell_1$ and $\ell_2$ be any squared edge-length maps for $H$ such that some edge is assigned distinct positive real numbers by both maps.  
    If the linkages $(H,\ell_1)$ and $(H,\ell_2)$ both have $2$-realizations, then there exist squared edge-length maps $\ell'_1$ and $\ell'_2$ for $G$ such that 
    \begin{enumerate}[(i)]
        \item they agree with each other on the edges not in $H$ and assign these edges positive real numbers, 
        \item $\ell'_i$ agrees with $\ell_i$ on the edges of $H$, 
        \item the linkages $(G,\ell'_1)$ and $(G,\ell'_2)$ both have $2$-realizations.  
    \end{enumerate}
\end{lemma}

Lemma \ref{lem:degree_3_decorations} is proved in \ref{sec:prop-10}.  
The crux of the proof is illustrated in Figure \ref{fig:triangle}: consider realizations $p$ and $q$ of $(H,\ell_1)$ and $(H,\ell_2)$, respectively, with $p_1$ and $q_1$ at the origin and with $p_2$ and $q_2$ on the positive $x$-axis.  
All points equidistant from $p_2$ and $q_2$ lie on the dashed blue line, and all points equidistant from $p_3$ and $q_3$ lie on the dashed red line.  
We show that the $G$ linkages can be chosen so that the realization $p$  is rotated to a position where these lines intersect, and then the lemma easily follows.  

% \begin{lemma}
%     \label{lem:f-retaining_fbm_after_adding_C'_nonedges}
%     Let $(G,f)$ be a minimal graph-nonedge pair with a top-level expanded $f$-separating CMS $C^{-1}$ such that the expanded $f$-component $H^{-1}$ of $C^{-1}$ is one of the graphs in Figure \ref{fig:top-level_H'_cases} such that $f$ is $uv$ and $C^{-1}$ is the subgraph induced by the set $\{w_1,w_2,w_3,w_4\}$.  
%     % some isomorphism between the expanded $f$-component $H^{-1}$ of $C^{-1}$ and one of the graphs in Figure \ref{fig:top-level_H'_cases} maps $f$ to $uv$ and $C^{-1}$ to the subgraph induced by the set $\{w_1,w_2,w_3,w_4\}$.  
%     Also, let $G' = (G \cup f) \setminus (H^{-1} \setminus C^{-1})$ and let $F$ be the set of nonedges of $C^{-1}$.  
%     If, for some nonedge $f'$ of $C^{-1}$, $(G',f')$ does not have the $3$-SIP and $H^{-1} \cup F$ has an $f$-retaining minor that is the graph in either Figure \ref{fig:k5_f} or Figure \ref{fig:k5_f_expanded}, then neither does $(G,f)$.  
%     % isomorphic to the graph in either Figure \ref{fig:k5_f} or Figure \ref{fig:k5_f_expanded}, then neither does $(G,f)$.  
% \end{lemma}

\begin{proposition}[Extending non-$3$-SIP maps for minimal pairs]
    \label{prop:IS}
    Let $G'$ be any of the graphs in Figures \ref{fig:3top_vin_uonly_2} - \ref{fig:3top_vin_uplus_2} such that the dashed line-segment is a nonedge, let $H$ be the subgraph of $G$ induced by the set $\{w_i\}$, let $f$ be the red edge, and let $f'$ be the nonedge of $H$ that does not share an endpoint with $f$ if possible, or any other nonedge of $H$ otherwise.  
    Also, let $G = (G' \setminus f) \cup f'$.  
    Lastly, let $\ell_1$ and $\ell_2$ be squared edge-length maps for $H$ that assign each edge a positive real number and differ only in the numbers assigned to $f'$.  
    If the linkages $(H,\ell_1)$ and $(H,\ell_2)$ both have at least one $3$-realization, then there exist squared edge-length maps $\ell'_1$ and $\ell'_2$ for $G$ such that 
    \begin{enumerate}[(i)]
        \item they agree with each other on the edges not in $H$ and assign these edges positive real numbers, 
        \item $\ell'_i$ agrees with $\ell_i$ on the edges of $H$, and 
        \item the CCSs $\Omega^3_f(G,\ell'_1)$ and $\Omega^3_f(G,\ell'_2)$ each contain exactly one positive value, and these values are distinct.
    \end{enumerate}
\end{proposition}

\begin{proof}
    If $G'$ is one of the graphs in Figures \ref{fig:3top_vin_uonly_2} - \ref{fig:cycle_vin_uonly_2}, then Lemma \ref{lem:4-cycle_non-sip} applies and proves the proposition.  
    Else, if $G'$ is the graph in Figure \ref{fig:cycle_vin_uplus_3}, then we can first apply Lemma \ref{lem:4-cycle_non-sip} to the subgraph of $G'$ obtained by deleting the unlabeled vertices, and then add these vertices back and apply Lemma \ref{lem:degree_3_decorations} to prove the proposition.  
    Lastly, if $G'$ is the graph in Figure \ref{fig:3top_vin_uplus_2}, then consider the subgraph of $G'$ obtained by deleting all but one unlabeled vertex, which we call $x$.  
    We first apply Lemma \ref{lem:4-cycle_non-sip} to this subgraph while treating $xw_2$ as $f$, and then we apply this lemma again while treating $xw_2$ as $f'$ and $uw_1$ as $f$.  
    Finally, we add the deleted vertices and apply Lemma \ref{lem:degree_3_decorations} to prove the proposition.  
\end{proof}

We are now ready to prove Theorem \ref{thm:no_pair_in_D_has_3-sip} and the forward direction of Theorem \ref{thm:3-sip_characterization}.  

\begin{proof}[Proof of Theorem \ref{thm:no_pair_in_D_has_3-sip}]
    We show that a minimal graph-nonedge pair $(G,f)$ has a proper non-$3$-SIP map by induction on the number $n$ of vertices in $G$.  
    Since $G \cup f$ has a forbidden minor, we have $n \geq 5$.  
    When $n = 5$, $G \cup f$ is $K_5$, and so the theorem follows from Proposition \ref{lem:no_type_2_no_3-sip}.  
    Assume that the theorem is true when $n \leq k$, for any integer $k \geq 5$, and we will show that it is true when $n = k + 1$.  
    The same proposition allows us to assume that $(G,f)$ has a Type (2) edge.  
    Hence, $(G,f)$ has a top-level expanded $f$-separating CMS $C^{-1}$, by Proposition \ref{lem:expanded_top-level_e-separating_clique_exists}.  
    Let $H^{-1}$ be the expanded $f$-component of $C^{-1}$ and $G' = G \setminus (H^{-1} \setminus C^{-1})$.  
    Consider any minimal graph-nonedge pair $(J,f')$ obtained by applying Proposition \ref{prop:IH} to $(G,f)$ and $C^{-1}$.  
    Proposition \ref{prop:final_top_graphs} tells us that $H^{-1}$ is one of the graphs in Figures \ref{fig:3top_vin_uonly_2} - \ref{fig:3top_vin_uplus_2} such that the dashed line-segment is a nonedge, $f$ is the red edge, and $C^{-1}$ is the subgraph induced by $\{w_i\}$.  
    Proposition \ref{prop:IH} allows us to choose $f' = w_1w_3$ if $H^{-1}$ is the graph in Figure \ref{fig:cycle_vin_uonly_2}, and $f' = w_2w_4$ if $H^{-1}$ is the graph in Figure \ref{fig:cycle_vin_uplus_3}.  

    Next, since $J$ has strictly fewer vertices than $G$, the inductive hypothesis applies and states that $(J,f')$ has a proper non-$3$-SIP map.  
    Using the relationship between $G'$ and $J$ described in Proposition \ref{prop:IH} along with the triangle-inequality, we can obtain a proper non-$3$-SIP map $\ell_{G'}$ for $(G',f')$.  
    By definition, this means that the CCS $\Omega^3_{f'}(G',\ell_{G'})$ contains exactly two distinct positive real numbers.  
    Let $\ell_1$ and $\ell_2$ be squared edge-length maps for the edges of $C^{-1} \cup f'$ that agree with $\ell_{G'}$ on the edges of $C^{-1}$ and assign different squared lengths from $\Omega^3_{f'}(G',\ell_{G'})$ to $f'$.  
    Note that the linkages $(C^{-1} \cup f', \ell_1)$ and $(C^{-1} \cup f', \ell_2)$ each have at least one $3$-realization.  
    Therefore, consider the edge-length maps $\ell'_1$ and $\ell'_2$ for the edges of $(H^{-1} \setminus f) \cup f'$ given by Proposition \ref{prop:IS}.  
    Define $\ell_G$ to be the edge-length map for the edges of $G$ that agrees with $\ell_{G'}$ on the edges of $G'$ and $\ell'_1$ on the edges of $H^{-1}$.  
    Clearly, $\ell_G$ is a proper non-$3$-SIP map for $(G,f)$, which completes the proof by induction.  
\end{proof}

\begin{proof}[Proof of the forward direction of Theorem \ref{thm:3-sip_characterization}]
    Let $(G,f)$ be a graph-nonedge pair such that some atom of $G \cup f$ that contains $f$ has an $f$-preserving forbidden minor.  
    % First, we show that we can assume wlog that $(G,f)$ is minimal.  
    Corollary \ref{cor:atom_gluing} allows us to assume wlog that $G \cup f$ is an atom.  
    If $(G,f)$ has a reducing edge, then let $[G \cup f]$ be the minor obtained by contracting any such edge.  
    By definition, $[G \cup f]$ is $f$-retaining has an atom that contains $[f]$ and an $[f]$-preserving forbidden minor.  
    Hence, by Lemma \ref{lem:cm_non-sip}, if $([G \cup f] \setminus [f], [f])$ does not have the $3$-SIP, then neither does $(G,f)$.  
    Therefore, we can assume wlog that $(G,f)$ has no reducing edge.  
    Combining this with the fact that $G \cup f$ is an atom shows that $(G,f)$ is minimal, and thus the theorem follows from Theorem \ref{thm:no_pair_in_D_has_3-sip}.  
\end{proof}

\subsection{Proofs of Propositions \ref{prop:IH} and \ref{prop:final_top_graphs}}
\label{sec:prop-8-9}

Consider the graph-nonedge pair $(G_{\star},f_{\star})$ from Section \ref{sec:forward}, let $C^{-1}$ be any top-level expanded $f_{\star}$-separating CMS of $G_{\star} \cup f_{\star}$, and let $H^{-1}$ be its expanded $f_{\star}$-component.  
We present the highly technical proofs of Proposition \ref{prop:IH}, which demonstrates the above-mentioned nested structure of minimal pairs, and Proposition \ref{prop:final_top_graphs}, which lists all possible graphs for $H^{-1}$.  
The proofs are completed in two parts: the first part enumerates a finite set of candidate graphs for $H^{-1}$ and the second part gradually eliminates all but a few of these candidates.  
The first part begins with Lemma \ref{lem:expanded_clique_graphs_and_connections}, below, which enumerates all possible graphs for $C^{-1}$ and describes the resulting structure on $G_{\star} \cup f_{\star}$.  
Recall from Section \ref{sec:minor-notation} that an induced subgraph $H$ of a graph $G$ is \emph{retained} in a minor $[G]$ if edge and nonedge of $H$ is retained in $[G]$.  
For any induced subgraphs $H_1$ and $H_2$ of $G$, $H_1 \cup H_2$ is the subgraph of $G$ induced by $V(H_1) \cup V(H_2)$.  

\begin{figure}[htb]
    \centering
    \begin{subfigure}{0.24\textwidth}
        \centering
        \includegraphics[width = 0.5\textwidth]{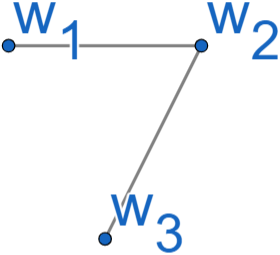}
        \subcaption{}
        \label{fig:exp_clique_sep_k2}
    \end{subfigure}
    \begin{subfigure}{0.24\textwidth}
        \centering
        \includegraphics[width = 0.5\textwidth]{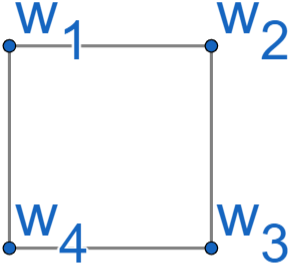}
        \subcaption{}
        \label{fig:exp_clique_sep_k3_cycle}
    \end{subfigure}
    \begin{subfigure}{0.24\textwidth}
        \centering
        \includegraphics[width = 0.5\textwidth]{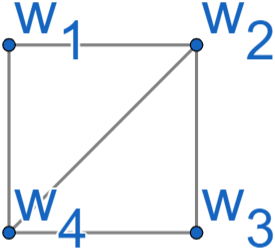}
\subcaption{}
        \label{fig:exp_clique_sep_k3_chord}
    \end{subfigure}
    \begin{subfigure}{0.24\textwidth}
        \centering
        \includegraphics[width = 0.5\textwidth]{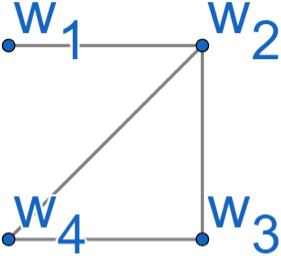}
\subcaption{}
        \label{fig:exp_clique_sep_k3_degree_1}
    \end{subfigure}
    \begin{subfigure}{0.19\textwidth}
        \centering
        \includegraphics[width = 0.6\textwidth]{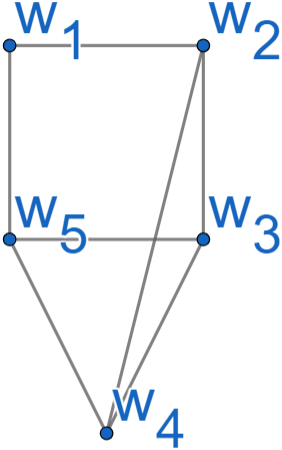}
\subcaption{}
        \label{fig:exp_clique_sep_k4_cycle}
    \end{subfigure}
    \begin{subfigure}{0.19\textwidth}
        \centering
        \includegraphics[width = 0.6\textwidth]{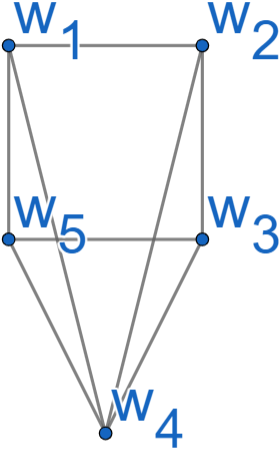}
\subcaption{}
        \label{fig:exp_clique_sep_k4_cycle_2}
    \end{subfigure}
    \begin{subfigure}{0.19\textwidth}
        \centering
        \includegraphics[width = 0.6\textwidth]{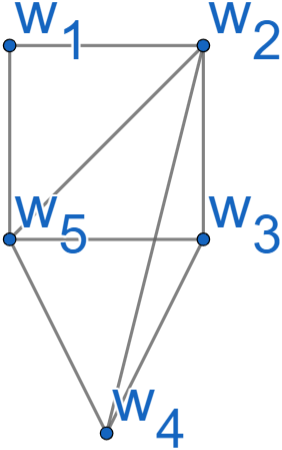}
\subcaption{}
        \label{fig:exp_clique_sep_k4_chord}
    \end{subfigure}
    \begin{subfigure}{0.19\textwidth}
        \centering
        \includegraphics[width = 0.6\textwidth]{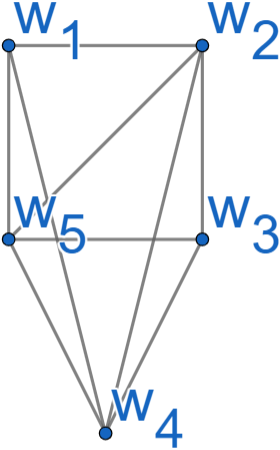}
\subcaption{}
        \label{fig:exp_clique_sep_k4_chord_2}
    \end{subfigure}
    \begin{subfigure}{0.19\textwidth}
        \centering
        \includegraphics[width = 0.6\textwidth]{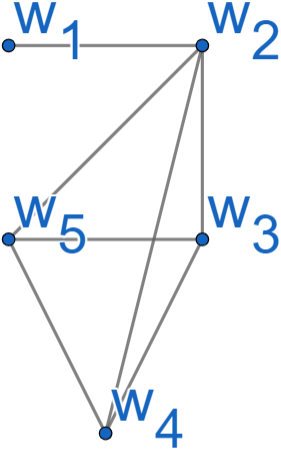}
\subcaption{}
        \label{fig:exp_clique_sep_k4_degree_1}
    \end{subfigure}
    \caption{All possible graphs for top-level expanded $f$-separating of minimal graph-nonedge pairs (Definition \ref{def:expanded_e-separating_pair} in this section).  
    See the discussion and Lemmas \ref{lem:expanded_clique_graphs_and_connections} and \ref{prop:C'_leq_4vert} in this section.  }
    \label{fig:expanded_e-separating_cliques}
\end{figure}

\begin{figure}[htb]
    \centering
    \begin{subfigure}{0.19\textwidth}
        \centering
        \includegraphics[width = 0.5\textwidth]{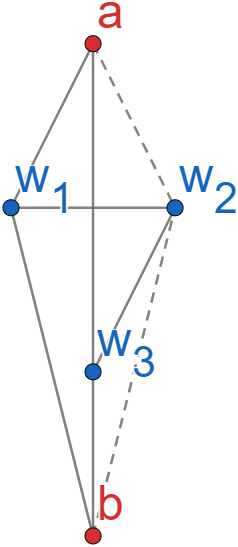}
        \subcaption{}
        \label{fig:exp_clique_sep_k2_conn}
    \end{subfigure}
    \begin{subfigure}{0.19\textwidth}
        \centering
        \includegraphics[width = 0.5\textwidth]{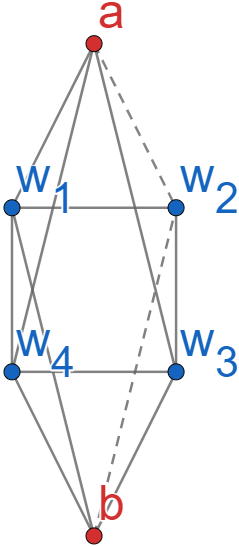}
        \subcaption{}
        \label{fig:exp_clique_sep_k3_cycle_conn_2}
    \end{subfigure}
    \begin{subfigure}{0.19\textwidth}
        \centering
        \includegraphics[width = 0.5\textwidth]{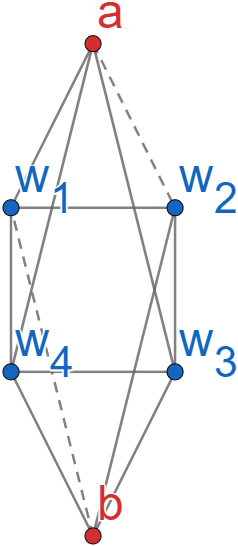}
        \subcaption{}
        \label{fig:exp_clique_sep_k3_cycle_conn_1}
    \end{subfigure}
    \begin{subfigure}{0.19\textwidth}
        \centering
        \includegraphics[width = 0.5\textwidth]{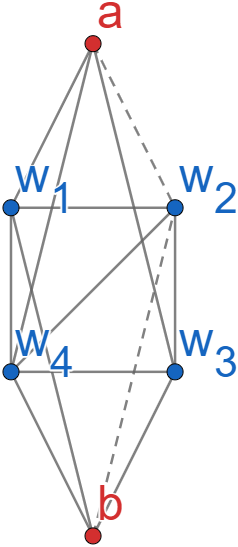}
\subcaption{}
        \label{fig:exp_clique_sep_k3_chord_conn}
    \end{subfigure}
    \begin{subfigure}{0.19\textwidth}
        \centering
        \includegraphics[width = 0.5\textwidth]{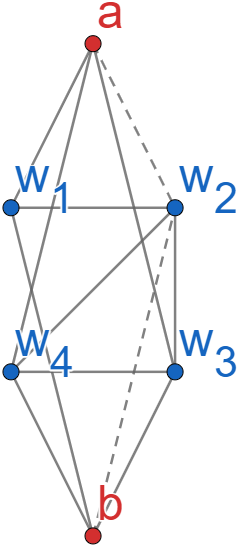}
\subcaption{}
        \label{fig:exp_clique_sep_k3_degree_1_conn}
    \end{subfigure}
        \begin{subfigure}{0.19\textwidth}
        \centering
        \includegraphics[width = 0.5\textwidth]{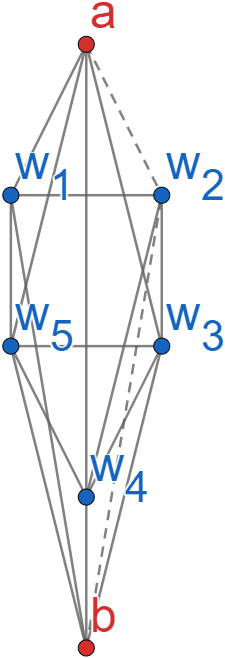}
\subcaption{}
        \label{fig:exp_clique_sep_k4_cycle_conn_1}
    \end{subfigure}
    \begin{subfigure}{0.19\textwidth}
        \centering
        \includegraphics[width = 0.5\textwidth]{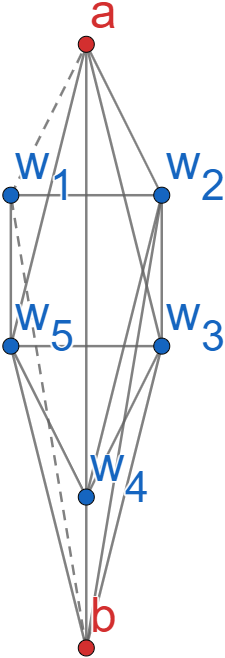}
\subcaption{}
        \label{fig:exp_clique_sep_k4_cycle_conn_2}
    \end{subfigure}
    \begin{subfigure}{0.19\textwidth}
        \centering
        \includegraphics[width = 0.5\textwidth]{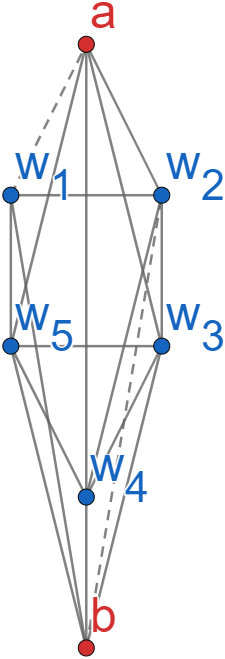}
\subcaption{}
        \label{fig:exp_clique_sep_k4_cycle_conn_3}
    \end{subfigure}
%     \begin{subfigure}{0.19\textwidth}
%         \centering
%         \includegraphics[width = 0.4\textwidth]{exp_clique_sep_k4_cycle_conn_4.png}
% \subcaption{}
%         \label{fig:exp_clique_sep_k4_cycle_conn_4}
%     \end{subfigure}
    \begin{subfigure}{0.19\textwidth}
        \centering
        \includegraphics[width = 0.5\textwidth]{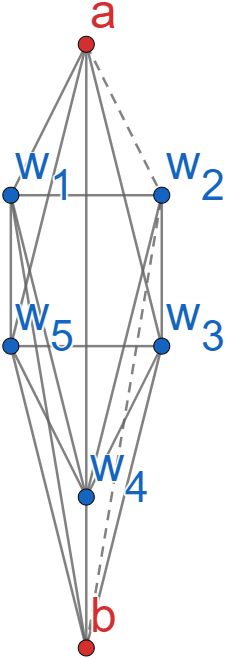}
\subcaption{}
        \label{fig:exp_clique_sep_k4_cycle_2_conn_1}
    \end{subfigure}
    \begin{subfigure}{0.19\textwidth}
        \centering
        \includegraphics[width = 0.5\textwidth]{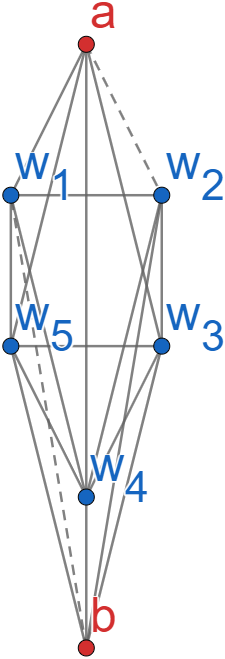}
\subcaption{}
        \label{fig:exp_clique_sep_k4_cycle_2_conn_2}
    \end{subfigure}
    \begin{subfigure}{0.19\textwidth}
        \centering
        \includegraphics[width = 0.5\textwidth]{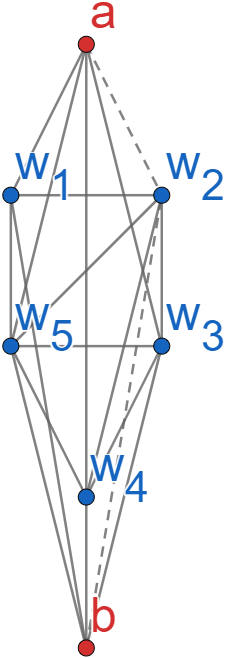}
\subcaption{}
        \label{fig:exp_clique_sep_k4_chord_conn}
    \end{subfigure}
    \begin{subfigure}{0.19\textwidth}
        \centering
        \includegraphics[width = 0.5\textwidth]{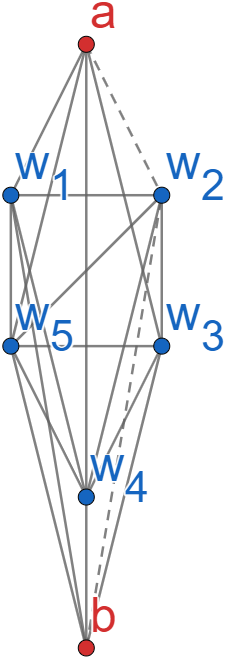}
\subcaption{}
        \label{fig:exp_clique_sep_k4_chord_2_conn}
    \end{subfigure}
        \begin{subfigure}{0.19\textwidth}
        \centering
        \includegraphics[width = 0.5\textwidth]{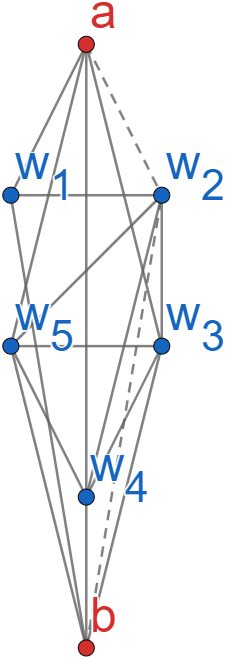}
\subcaption{}
        \label{fig:exp_clique_sep_k4_degree_1_conn}
    \end{subfigure}
    \caption{Graphs in Lemma \ref{lem:expanded_clique_graphs_and_connections}.  
    Dashed line-segments may or may not be edges.  }
    \label{fig:expanded_e-separating_cliques_conn}
\end{figure}

\begin{lemma}[Unpruned possible graphs for top-level $C^{-1}$]
    \label{lem:expanded_clique_graphs_and_connections}
    Let $C^{-1}$ be any top-level expanded $f$-separating CMS of a minimal graph-nonedge pair $(G,f)$, and let its expanded $f$-component and expanded minor component be $H^{-1}$ and $I^{-1}$, respectively.  
    Then, the following statements are true:
    \begin{enumerate}
        \item $C^{-1}$ is 
        % isomorphic to 
        one of the graphs in Figure \ref{fig:expanded_e-separating_cliques} and 
        
        \item the minor $[H^{-1} \cup I^{-1}]$ obtained by contracting every edge in both $H^{-1} \setminus C^{-1}$ and $I^{-1} \setminus C^{-1}$ is one of the graphs in  Figure \ref{fig:expanded_e-separating_cliques_conn} with $[C^{-1}]$ as the subgraph induced by the set $\{w_i\}$.  
        % some isomorphism between the minor $[H^{-1} \cup I^{-1}]$ obtained by contracting every edge in both $H^{-1} \setminus C^{-1}$ and $I^{-1} \setminus C^{-1}$ and some graph in Figure \ref{fig:expanded_e-separating_cliques_conn} maps $[C^{-1}]$ to the subgraph induced by the set $\{w_i\}$.
    \end{enumerate}
\end{lemma}

The proof of Lemma \ref{lem:expanded_clique_graphs_and_connections} is given in \ref{sec:prop-8}.  
The main obstacle to proving Statement (1) is the case where $C^{-1}$ has at least six vertices, which we show is not possible when $C^{-1}$ is in the top-level.  
Statement (2) then follows from the facts that $G_{\star} \cup f_{\star}$ is an atom and $C^{-1}$ is an expanded $f$-separating CMS.  

Next, Lemma \ref{cor:top-level_structure}, below, describes the neighborhoods of vertices in $H^{-1} \setminus C^{-1}$, showing that each edge of this subgraph is incident to some endpoint of $f_{\star}$.  
This lemma is proved using using Lemmas \ref{lem:I'_fbm_after_f-component_contraction} and \ref{lem:type_2_yields_smaller_expanded_f-component}, also below.

% \begin{proof}
%     The proof has many tedious but straightforward cases, so we omit the details and instead outline the main idea.  
%     The only obstacle to the existence of the vertices $x$ and $y$ and the paths $P_{ux}$ and $P_{vy}$ is that the edge contracted to obtain $[G \cup f]$ has an endpoint in each path.  
%     However, since $G \cup f$ is an atom, and hence $2$-connected, we can find another path that avoids the endpoints of this edge.  
%     We then consider the cases for shared vertices between this new path and $P_{ux}$ or $P_{vy}$ to obtain a pair of vertices and paths with the desired properties.  
% \end{proof}

\begin{lemma}[No Type (1) edge in $(H^{-1} \setminus f) \setminus C^{-1}$]
    \label{lem:I'_fbm_after_f-component_contraction}
    Let $C^{-1}$ be any expanded $f$-separating CMS of a minimal graph-nonedge pair $(G,f)$, with $f = uv$, and let its expanded $f$-component and expanded minor component be $H^{-1}$ and $I^{-1}$, respectively.  
    Contracting any edge of $(H^{-1} \setminus f) \setminus C^{-1}$ yields a minor $[G \cup f]$ such that $[I^{-1}]$ has an $[f]$-preserving forbidden minor.  
\end{lemma}

\begin{lemma}[No Type (2) edge in $(H^{-1} \setminus f) \setminus C^{-1}$]
    \label{lem:type_2_yields_smaller_expanded_f-component}
    Let $C^{-1}$ be an expanded $f$-separating CMS of a minimal graph-nonedge pair $(G,f)$, and let $H^{-1}$ be its expanded $f$-component.  
    If $C^{-1}$ is in the top level of $(G,f)$, then no edge of $(H^{-1} \setminus f) \setminus C^{-1}$ is of Type (2).
\end{lemma}

Lemmas \ref{lem:I'_fbm_after_f-component_contraction} and \ref{lem:type_2_yields_smaller_expanded_f-component} are proved in \ref{sec:prop-8}.  
To prove Lemma \ref{lem:I'_fbm_after_f-component_contraction}, we show that the $f_{\star}$-preserving forbidden minor of $I^{-1}$, which exists by definition, is not destroyed by the contraction leading to $[G_{\star} \cup f_{\star}]$.  
Lemma \ref{lem:type_2_yields_smaller_expanded_f-component} is proved by assuming to the contrary that $(H^{-1} \setminus f_{\star}) \setminus C^{-1}$ contains some Type (2) edge and then concluding that $C^{-1}$ is not in the top-level, a contradiction.  

\begin{lemma}[Structure of top-level $H^{-1}$]
    \label{cor:top-level_structure}
    Let $C^{-1}$ be a top-level expanded $f$-separating CMS of a minimal graph-nonedge pair $(G,f)$, with $f=uv$, and let $H^{-1}$ be its expanded $f$-component.  
    The neighborhood of any vertex $w$ in $H^{-1} \setminus (C^{-1} \cup \{u,v\})$) contains only $u$, $v$, and at least one vertex in $C^{-1} \setminus \{u,v\}$.  
    Furthermore, the vertices $u$ and $v$ may have neighbors in $C^{-1}$.  
\end{lemma}

\begin{proof}
    Since $w$ is a vertex of $H^{-1} \setminus (C^{-1} \cup \{u,v\})$, it must be adjacent to some vertex $z$ in $H^{-1} \setminus C^{-1}$.  
    Note that Lemmas \ref{lem:I'_fbm_after_f-component_contraction} and \ref{lem:type_2_yields_smaller_expanded_f-component} imply that $wz$ is of Type (3).  
    Hence, $z$ must be either $u$ or $v$, and both $wu$ and $wv$ must be edges.  
    Next, since $G \cup f$ is an atom, $w$ must be adjacent to some vertex other than $u$ or $v$.  
    The same two lemmas imply that this vertex is contained in $C^{-1}$.  
    Thus, the lemma is proved.  
\end{proof} 

Lemmas \ref{lem:expanded_clique_graphs_and_connections} and \ref{cor:top-level_structure} are used to prove Lemma \ref{prop:C'_leq_4vert}, below, which eliminates the graphs in Figures \ref{fig:expanded_e-separating_cliques}(e)-(i) as possibilities for  $C^{-1}$.  
The proof of Lemma \ref{prop:C'_leq_4vert} also requires Lemma \ref{lem:C'_5_H'_wing}, below.  

\begin{figure}[!htbp]
    \centering
    \begin{subfigure}{0.24\textwidth}
        \centering
        \includegraphics[width = 0.8\textwidth]{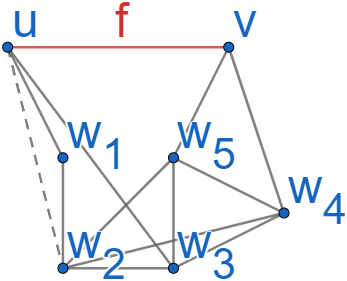}
        \subcaption{}
        \label{fig:5_degree_1}
    \end{subfigure}
    \begin{subfigure}{0.24\textwidth}
        \centering
        \includegraphics[width = 0.8\textwidth]{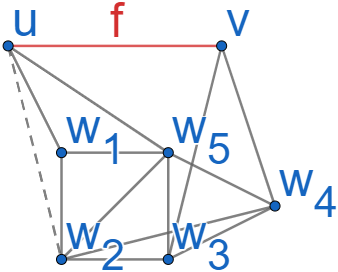}
        \subcaption{}
        \label{fig:5_chord_arg_1}
    \end{subfigure}
    \begin{subfigure}{0.24\textwidth}
        \centering
        \includegraphics[width = 0.8\textwidth]{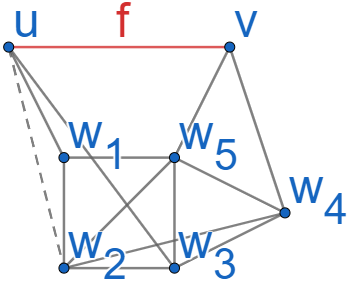}
        \subcaption{}
        \label{fig:5_k5_chord-e}
    \end{subfigure}
    \begin{subfigure}{0.24\textwidth}
        \centering
        \includegraphics[width = 0.8\textwidth]{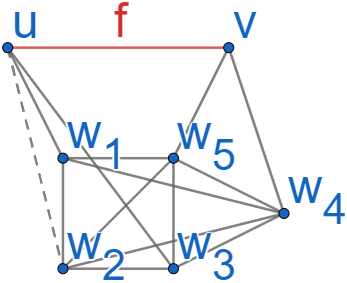}
        \subcaption{}
        \label{fig:5_k5_chord}
    \end{subfigure}
    \begin{subfigure}{0.24\textwidth}
        \centering
        \includegraphics[width = 0.8\textwidth]{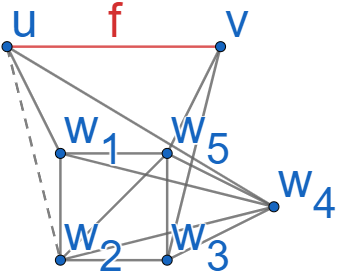}
        \subcaption{}
        \label{fig:5_k5_chord_arg_2}
    \end{subfigure}
    \begin{subfigure}{0.24\textwidth}
        \centering
        \includegraphics[width = 0.8\textwidth]{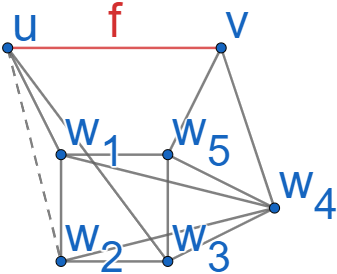}
        \subcaption{}
        \label{fig:5_cycle_4_u3}
    \end{subfigure}
    \begin{subfigure}{0.24\textwidth}
        \centering
        \includegraphics[width = 0.8\textwidth]{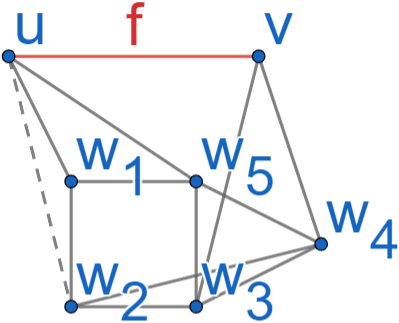}
        \subcaption{}
        \label{fig:5_cycle_arg_3_2}
    \end{subfigure}
    \begin{subfigure}{0.24\textwidth}
        \centering
        \includegraphics[width = 0.8\textwidth]{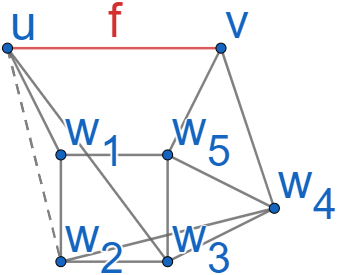}
        \subcaption{}
        \label{fig:5_cycle_3_u3}
    \end{subfigure}
    \begin{subfigure}{0.24\textwidth}
        \centering
        \includegraphics[width = 0.8\textwidth]{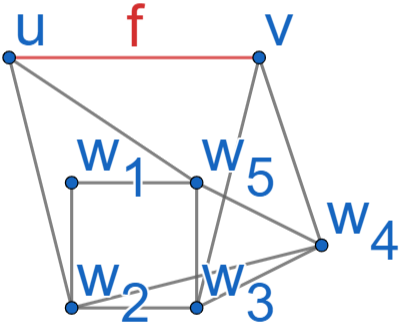}
        \subcaption{}
        \label{fig:5_cycle_3_arg_3_2}
    \end{subfigure}
    \begin{subfigure}{0.24\textwidth}
        \centering
        \includegraphics[width = 0.8\textwidth]{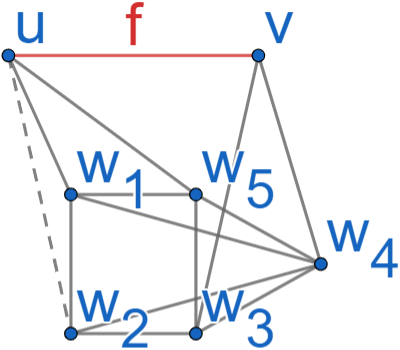}
        \subcaption{}
        \label{fig:5_cycle_4_u2_edge}
    \end{subfigure}
    \begin{subfigure}{0.24\textwidth}
        \centering
        \includegraphics[width = 0.8\textwidth]{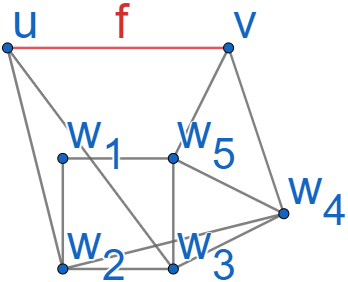}
        \subcaption{}
        \label{fig:5_cycle_3_u2_edge}
    \end{subfigure}
    \begin{subfigure}{0.24\textwidth}
        \centering
        \includegraphics[width = 0.8\textwidth]{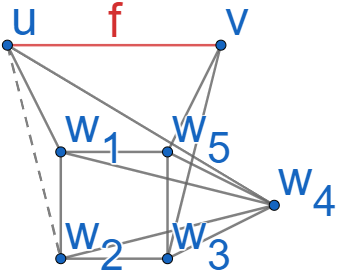}
        \subcaption{}
        \label{fig:5_cycle_4_arg_4}
    \end{subfigure}
    \caption{All possible expanded $f$-components of a top-level expanded $f$-separating CMS that has exactly five vertices.  
    Dashed line-segments may be edges or nonedges.  
    See Lemma \ref{lem:C'_5_H'_wing} in this section.  }
    \label{fig:5_top}
\end{figure}

\begin{lemma}[Cases for $H^{-1}$ when $C^{-1}$ has $5$ vertices]
\label{lem:C'_5_H'_wing}
     Let $C^{-1}$ be a top-level expanded $f$-separating CMS of a minimal graph-nonedge pair $(G,f)$.  
     If $|V(C^{-1})| = 5$, then the expanded $f$-component of $C^{-1}$ is one of the graphs in Figure \ref{fig:5_top} with $f$ as the red edge and $C^{-1}$ as the subgraph induced by the set $\{w_i\}$.  
     % some isomorphism between the expanded $f$-component of $C^{-1}$ and one of the graphs in Figure \ref{fig:5_top} maps $f$ to the red edge and $C^{-1}$ to the subgraph induced by the set $\{w_i\}$.  
\end{lemma}

\begin{lemma}[Pruned possible graphs for $C^{-1}$]
    \label{prop:C'_leq_4vert}
    If $C^{-1}$ is a top-level expanded $f$-separating CMS of a minimal graph-nonedge pair $(G,f)$, then $C^{-1}$ is one of the graphs in Figures \ref{fig:exp_clique_sep_k2}-\ref{fig:exp_clique_sep_k3_degree_1}.  
\end{lemma}

Lemmas \ref{lem:C'_5_H'_wing} and \ref{prop:C'_leq_4vert} are proved in Section \ref{sec:prop-8}.  
Lemma \ref{lem:C'_5_H'_wing} is proved by showing that, in any $f_{\star}$-separating minor $[G_{\star} \cup f_{\star}]$ where $[C^{-1}]$ is a CMS with four vertices, $[H^{-1}]$ is the graph in Figure \ref{fig:k5_wing} such that $[f_{\star}]$ is $w_1w_2$.  
The lemma then follows easily from Lemma \ref{lem:expanded_clique_graphs_and_connections} Statement (2).  
Lemma \ref{prop:C'_leq_4vert} is proved by assuming to the contrary that $C^{-1}$ is one of the graphs in Figures \ref{fig:expanded_e-separating_cliques}(e)-(i) and then concluding that $(G_{\star},f_{\star})$ is not minimal, a contradiction.  

The above lemmas have considerably reduced the possible graphs for $H^{-1}$, which are now enumerated in Lemma \ref{prop:expC_3or4}, below.  
The proof of this lemma, given in \ref{sec:proof_top-level_H'_cases}, involves case checking using techniques as in the proof of Lemma \ref{prop:C'_leq_4vert}, but for far more cases.  

\begin{lemma}[Unpruned possible graphs for top-level $H^{-1}$]
    \label{prop:expC_3or4}
    If $C^{-1}$ is a top-level expanded $f$-separating CMS of a minimal graph-nonedge pair $(G,f)$, then the expanded $f$-component of $C^{-1}$ is one of the graphs in Figure \ref{fig:top-level_H'_cases} with $f$ as the red edge and $C^{-1}$ as the subgraph induced by the vertex set $\{w_i\}$.  
    % some isomorphism between the expanded $f$-component of $C^{-1}$ and one of the graphs in Figure \ref{fig:top-level_H'_cases} maps $f$ to the red edge and $C^{-1}$ to the subgraph induced by the vertex set $\{w_i\}$.  
\end{lemma}

Next, Lemma \ref{lem:C'_a-n}, stated below and proved in \ref{sec:prop-8}, reduces the possible graphs for $H^{-1}$ to those in Figures \ref{fig:3top_vin_uonly_2} - \ref{fig:deg1_vout_uvonly}.  
The proof argues that $(G_{\star},f_{\star})$ has a reducing edge if $H^{-1}$ is not one of these graphs, a contradiction.  

\begin{lemma}[Pruning cases for $H^{-1}$]
    \label{lem:C'_a-n}
    The expanded $f$-component $H^{-1}$ of any expanded $f$-separating CMS $C^{-1}$ of a minimal graph-nonedge pair is one of the graphs in Figures \ref{fig:3top_vin_uonly_2} - \ref{fig:deg1_vout_uvonly}, where $C^{-1}$ is the subgraph induced by the vertex set $\{w_i\}$.  
    Furthermore, if $H^{-1}$ is one of the graphs in Figures \ref{fig:cycle_vin_uonly_2}, \ref{fig:cycle_vin_uplus_3}, \ref{fig:3top_vin_uplus_2}, \ref{fig:cycle_vin_uplus_2}, \ref{fig:deg1_vin_uplus_3}, or \ref{fig:deg1_vin_uplus_4}, then at least one dashed line-segment is a nonedge.  
\end{lemma}

All this reduction work leads to Lemmas \ref{lem:G'_f'_no_reducible_edge} and \ref{lem:G'_is_J+decorations}, below, from which Proposition \ref{prop:IH} immediately follows.  
The proof of Lemma \ref{lem:G'_f'_no_reducible_edge}, given in \ref{sec:proof_IH_2}, is another instance of checking many cases.  

\begin{lemma}[Additional structure necessary for induction]
    \label{lem:G'_f'_no_reducible_edge}
    Let $(G,f)$ be a minimal graph-nonedge pair with a top-level expanded $f$-separating CMS $C^{-1}$ whose expanded $f$-component is $H^{-1}$, and let $G' = G \setminus (H^{-1} \setminus C^{-1})$.  
    There exists a nonedge $f'$ of $C^{-1}$ such that some atom of $G' \cup f'$ contains $f'$ and has an $f'$-preserving forbidden minor and the pair $(G',f')$ has no reducing edge.  
    Furthermore, these statements are true when (i) $f' = w_1w_3$ and $H^{-1}$ is the graph in Figure \ref{fig:cycle_vin_uonly_2}, $f$ is the red edge, and $C^{-1}$ is the subgraph induced by $\{w_i\}$, or (ii) $f'$ is any nonedge of $C^{-1}$ and $H^{-1}$ is one of the graphs in Figures \ref{fig:cycle_vin_uplus_3}, \ref{fig:cycle_vin_uplus_2}, or \ref{fig:cycle_vout_uvonly}, $f$ is the red edge, and $C^{-1}$ is the subgraph induced by $\{w_i\}$.  
\end{lemma}

\begin{lemma}[Structure of $G' \cup f'$]
    \label{lem:G'_is_J+decorations}
    Consider a minimal graph-nonedge pair $(G,f)$ with a top-level expanded $f$-separating CMS $C^{-1}$ whose expanded $f$-component is $H^{-1}$, and let $G' = G \setminus (H^{-1} \setminus C^{-1})$.  
    If $f' = uv$ is a nonedge of $C^{-1}$ such that some atom $J$ of $G' \cup f'$ contains $f'$ and has an $f'$-preserving forbidden minor and the pair $(G',f')$ has no reducing edge, then $G' \cup f'$ can be obtained from $J$ by adding some vertices whose neighborhoods are $\{u,v\}$.  
\end{lemma}

\begin{proof}
    Assume to the contrary that $G' \cup f'$ cannot be obtained from $J$ by adding vertices whose neighborhood are $\{u,v\}$.  
    Then, there exists a CMS $E$ of $G' \cup f'$ contained in $J$ such that $E$ contains some vertex $x$ that is either the only vertex in $E$ or some vertex other than $u$ or $v$.  
    Also, there exists a vertex $y$ not contained in $J$ and such that $xy$ is an edge of $G' \cup f'$ and all paths in $G' \cup f'$ between $y$ and any vertex in $J \setminus E$ contains some vertex in $E$.  
    Hence, $J$ is weakly retained in the minor $[G' \cup f']$ obtained by contracting $xy$.  
    Furthermore, $[G' \cup f']$ is $f'$-retaining and contains some atom that contains $[f']$ and has an $[f']$-preserving forbidden minor.  
    The above facts show that $xy$ is reducing with respect to $(G',f')$, which is a contradiction.  
\end{proof}

\begin{proof}[Proof of Proposition \ref{prop:IH}]
    The proposition follows easily from Lemmas \ref{lem:G'_f'_no_reducible_edge} and \ref{lem:G'_is_J+decorations}.  
\end{proof}

Finally, Proposition \ref{prop:final_top_graphs} follows immediately from Lemmas \ref{lem:C'_a-i} and \ref{lem:C'_degree_1_reducible_3}, stated below and proved in \ref{sec:proof_final_top_graphs}, which are the final reductions of the possible graphs for $H^{-1}$.  
The proofs of both lemmas argue that $(G_{\star},f_{\star})$ has a reducing edge if either lemma is false, a contradiction.  

\begin{lemma}[Additional pruning of cases for $H^{-1}$]
    \label{lem:C'_a-i}
    For any minimal graph-nonedge pair $(G,f)$ with a top-level expanded $f$-separating CMS $C^{-1}$, the expanded $f$-component of $C^{-1}$ is one of the graphs in Figures \ref{fig:3top_vin_uonly_2} - \ref{fig:deg1_vin_uplus_5}, where $C^{-1}$ is the subgraph induced by the vertex set $\{w_i\}$.  
    Furthermore, if $H^{-1}$ is one of the graphs in Figures \ref{fig:cycle_vin_uonly_2}, \ref{fig:cycle_vin_uplus_3}, or \ref{fig:3top_vin_uplus_2}, then the dashed line-segment is a nonedge.  
\end{lemma}

\begin{lemma}[$H^{-1}$ is not graph in Figures \ref{fig:deg1_vin_uplus_1}-\ref{fig:deg1_vin_uplus_5}]
    \label{lem:C'_degree_1_reducible_3}
    None of the following graphs are the expanded $f$-component of any top-level expanded $f$-separating CMS of a minimal graph-nonedge pair:
    \begin{enumerate}[(i)]
        \item the graph in Figure \ref{fig:deg1_vin_uplus_1}, 

        \item neither of the graphs in Figures \ref{fig:deg1_vin_uonly_2} or \ref{fig:deg1_vin_uplus_5} with $uw_1$ being an edge, or 

        \item the graph in Figure \ref{fig:deg1_vin_uplus_5} with $uw_1$ being a nonedge.  
    \end{enumerate}
\end{lemma}

\begin{proof}[Proof of Proposition \ref{prop:final_top_graphs}]
    The proposition follows immediately from Lemmas \ref{lem:C'_a-i} and \ref{lem:C'_degree_1_reducible_3}.  
\end{proof}

\section{Open problems and conjectures}
\label{sec:conclusions}

\begin{op}
    Give a polynomial time algorithm to decide the characterization in Theorem \ref{thm:3-sip_characterization}.  
\end{op}

The polynomial-time algorithm in \cite{robertson2004graph} for finite forbidden minors will likely be useful for this problem.  

% We believe our characterization in Theorem \ref{thm:3-sip_characterization} extends to any dimension as follows.  

\begin{conjecture}
    For any dimension $d \geq 1$, a graph-nonedge pair $(G,f)$ has the $d$-SIP if and only if no atom of $G \cup f$ that contains $f$ has an $f$-preserving $d$-forbidden minor.  
\end{conjecture}

Since the number of forbidden minors for $d$-flattenability is conjectured to grow asymptotically with the number of forbidden minors for partial $d$-trees \cite{belk2007realizability1}, which is at least $75$ for $d=4$ \cite{sanders1993linear}, a proof of this conjecture must avoid using properties of the explicit $d$-flattenability forbidden minors.  
To this end, we posit the following conjecture.  

% \begin{op}
%     For any dimension $d \geq 0$, give a characterization of the set of $d$-flattenability forbidden minors, other than the set of minors that characterize $d$-flattenable graphs.  
% \end{op}

\begin{conjecture}
    For any dimension $d \geq 2$, a graph $G$ is a forbidden minor for $d$-flattenability if and only if both of the following are true:
    \begin{enumerate}
        \item for any edge $e$ of $G$, $(G \setminus e, e)$ does not have the $d$-SIP and 
        \item for any minor $[G]$ obtained via a single edge deletion or contraction and any edge edge $e$ of $[G]$, $([G] \setminus e, e)$ has the $d$-SIP.  
    \end{enumerate}
\end{conjecture}

Finally, an important motivation for studying the $d$-SIP is to characterize graphs $G$ and nonedge sets $F$ for which $(G,F)$ is $d$-convex.  
Recall Theorem \ref{thm:2-convexity} states that, for any dimension $d \leq 2$, $(G,F)$ is $d$-convex if and only if each atom of $G \cup F$ that contains a non-empty subset of $F$ has no $K_{d+2}$ minor.  
The following is an equivalent restatement of the right-hand-side of this theorem via Theorem \ref{thm:2-SIP}: for any $f \in F$, $(G \cup (F \setminus \{f\}, f)$ has the $d$-SIP.  

\begin{conjecture}
    For any dimension $d \geq 1$ and any graph $G$ with nonedge set $F$, $(G,F)$ is $d$-convex if and only if, for any $f \in F$, $(G \cup (F \setminus \{f\}, f)$ has the $d$-SIP.  
\end{conjecture}

\section{Acknowledgements}

The authors were partially supported by grants from the National Science Foundation Division of Mathematical Sciences [NSF DMS 1564480, NSF DMS 1563234] awarded to M.S.  

We thank the Fields Institute (thematic program Spring 2021, focus program Summer 2023), and ICERM (semester program Spring 2025) for support and facilitation, and Sean Dewar and Eleftherios Kastis for useful discussions, e.g. towards generalizing Lemma \ref{lem:tool}.  
We thank Matthias Himmelmann for recreating the top-right of Figure \ref{fig:2_and_3-convexity_ex} from \cite{vinzant2014spectrahedron}.  

\bibliographystyle{elsarticle-num.bst}
% \bibliography{references}

\appendix

\section{Proofs for the converse direction}
\label{app:converse}

\subsection{Proof of Lemma \ref{lem:K5_2-separator_case_1} for Proposition \ref{lem:winged-graph_cut_vertices}}
\label{sec:proof_winged-graph_cut_vertices}

\begin{proof}[Proof of Lemma \ref{lem:K5_2-separator_case_1}]  
    The proof is split into three cases: some set in $\{X,Y\}$ has size one, some set in $\{A,B\}$ has size one, and each set in $\{X,Y,A,B\}$ has size at least two.  
    We prove the lemma in the first case.  
    The other two cases are similar and involve inspecting many cases, so we omit them, but they are available upon request.  
    We further focus on the case where $[G \cup f]$ is the graph in Figure \ref{fig:k5_wing}.  
    The case where $[G \cup f]$ is the graph in Figure \ref{fig:k222_wing} is similar, so we omit the details but highlight the differences at the end of the proof.  
    
    Wlog, assume that $|Y| = 1$ and $|X| > 1$.  
    We will show how to perform a sequence of vertex and component exchanges that transform $[G \cup f]$ into the desired $f$-winged graph minor $[G \cup f]_1$.  
    Let $H$ be the subgraph of $G \cup f$ induced by $[c_3]^{-1}$ and consider any vertex $a \in A$ with a neighbor in $[c_5]^{-1}$.  
    Since $|X| > 1$, some connected component $H_1$ of $H \setminus a$ contains a vertex in $X$.  
    There are four cases.  

    \smallskip
    \noindent\textbf{Case 1:} $H_1$ contains some vertices with neighbors in $[c_4]^{-1}$ and $[c_6]^{-1}$.  
    \smallskip

    Let $[G \cup f]_1$ be obtained via the vertex exchange from $[c_3]^{-1}$ to $[c_5]^{-1}$ that exchanges $a$ and fixes $H_1$.  
        Clearly, $[G \cup f]_1$ is one of the graphs in Figures \ref{fig:k5_not_wing_u4_vgeq1}, \ref{fig:k5_not_wing_u3}, or \ref{fig:k5_wing} such that $[f]_1$ is $w_1w_2$, $[d_3]_1^{-1} \cup [d_4]_1^{-1}$ is a proper subset of $[c_3]^{-1} \cup [c_4]^{-1}$, and both $[f]_1^{-1}$ and $[f]^{-1}$ contain the same edges of $G$.  
        % some isomorphism between $[G \cup f]_1$ and one of the graphs in Figures \ref{fig:k5_not_wing_u4_vgeq1} - \ref{fig:k5_not_wing_u1} or \ref{fig:k5_wing} maps $[f]_1$ to $w_1w_2$, $[d_3]_1^{-1} \cup [d_4]_1^{-1}$ is a proper subset of $[c_3]^{-1} \cup [c_4]^{-1}$, and both $[f]_1^{-1}$ and $[f]^{-1}$ contain the same subset of edges of $G$.  
        Using Lemma \ref{lem:k5_k222_not_wing_fm_e_not_contracted}, we see that either $G \cup f$ has an $f$-preserving forbidden minor, which is a contradiction, or $[G \cup f]_1$ has the desired properties.  

    \begin{note}
        \label{note:1}
        In the remainder of the proof, we examine many cases for which similar arguments apply.  
        As to not overburden the proof, we will state how to obtain $[G \cup f]_1$ and make observations about neighborhoods of vertices, and we leave it to the reader to use these observations along with Lemma \ref{lem:k5_k222_not_wing_fm_e_not_contracted} to verify that either $G \cup f$ has an $f$-preserving forbidden minor, which is a contradiction, or $[G \cup f]_1$ has the desired properties.  
        % and the facts about neighborhoods stated below 
        % We omit the details of these arguments and simply identify either the minor $[G \cup f]_1$ with the desired properties or a minor $[G \cup f]_2$ whose existence leads to a contradiction, as above.  
    \end{note}

    \noindent\textbf{Case 2:} $H_1$ does not contain a vertex with a neighbor in either $[c_4]^{-1}$ or $[c_6]^{-1}$.  
    \smallskip
    
    Note that $H \setminus H_1$ contains some vertices with neighbors in $[c_4]^{-1}$ and $[c_6]^{-1}$.  
    Let $[G \cup f]_1$ be the minor obtained via the component exchange from $[c_3]^{-1}$ to $[c_1]^{-1}$ that exchanges $H_1$.  

    \smallskip
    \noindent\textbf{Case 3:} $H_1$ contains some vertex with a neighbor in $[c_6]^{-1}$ but no vertex with a neighbor in $[c_4]^{-1}$.  
    \smallskip
    
    Let $a'$ be any vertex in $H_1$ with a neighbor in $[c_6]^{-1}$ and note that $H \setminus H_1$ is a subgraph of some connected component $J_1$ of $H  \setminus a'$ and contains some vertex with a neighbor in $[c_4]^{-1}$.  
    If some connected component $J_2$ of $H \setminus a'$ other than $J_1$ contains some vertex in $X$, then let $[G \cup f]_1$ be the minor obtained via the component exchange from $[c_3]^{-1}$ to $[c_1]^{-1}$ that exchanges $J_2$.  
    Otherwise, $J_1$ contains some vertex in $X$ since $|X| \geq 2$.  
        Hence, let $[G \cup f]_1$ be the minor obtained via the vertex exchange from $[c_3]^{-1}$ to $[c_6]^{-1}$ that exchanges $a'$ and fixes $J_1$.  

    \smallskip
    \noindent\textbf{Case 4:} $H_1$ contains some vertex with a neighbor in $[c_4]^{-1}$ but no vertex with a neighbor in $[c_6]^{-1}$.  
    \smallskip

    This case contains many subcases.  
    The reader should keep in mind the edges and nonedges of $[G \cup f]$ as well as the assumptions leading to each subcase.  
    
    Note that $H \setminus H_1$ contains some vertex with a neighbor in $[c_6]^{-1}$.
    If $H \setminus H_1$ contains some vertex with a neighbor in $[c_4]^{-1}$, then let $[G \cup f]_1$ be the minor obtained via the component exchange from $[c_3]^{-1}$ to $[c_1]^{-1}$ that exchanges $H_1$.  
    Otherwise, 
        if $H_1$ contains some vertex with a neighbor in $[c_5]^{-1}$, then let $[G \cup f]_1$ be the minor obtained via the vertex exchange from $[c_3]^{-1}$ to $[c_6]^{-1}$ that exchanges $a$ and fixes $H_1$.  
        Hence, assume that $A \subseteq V(H-H_1)$.  
            If some connected component $H_2$ of $H \setminus a$ other than $H_1$ contains some vertex with a neighbor in $[c_6]^{-1}$, then let $[G \cup f]_1$ be the minor obtained via the component exchange from $[c_3]^{-1}$ to $[c_6]^{-1}$ that exchanges $H_2$.  
            Otherwise, $a$ has some neighbor in $[c_6]^{-1}$.  
                If $H_2$ contains some vertex with a neighbor in $[c_5]^{-1}$, then $[G \cup f]_1$ is the minor obtained via the component exchange from $[c_3]^{-1}$ to $[c_5]^{-1}$ that exchanges $H_2$.  
                
                Therefore, assume that $A = \{a\}$, and so $|B| \geq 2$ by assumption.  
                Hence, there exists a vertex $b \in B \setminus Y$ since $|Y|=1$.  
                Wlog, assume that $b$ neighbors some vertex in $[c_5]^{-1}$.  
                Let $K$ be the subgraph of $G \cup f$ induced by $[c_4]^{-1}$ and let $K_1$ be the connected component of $K \setminus b$ that contains the vertex $y \in Y$.
                By arguments similar to those above, we can conclude that either (a) $K_1$ contains some vertex with a neighbor in $[c_3]^{-1}$ but no vertex with a neighbor in $[c_6]^{-1}$ or (b) $K_1$ contains some vertex with a neighbor in $[c_6]^{-1}$ but no vertex with a neighbor in $[c_3]^{-1}$.  
                    
                    Assume that Statement (a) is true, and so $K \setminus K_1$ contains some vertex with a neighbor in $[c_6]^{-1}$.
                        If $K_1$ contains some vertex with a neighbor in $[c_5]^{-1}$, then let $[G \cup f]_1$ be the minor obtained via the vertex exchange from $[c_4]^{-1}$ to $[c_6]^{-1}$ that exchanges $b$ and fixes $K_1$.  
                        Otherwise, we have that $B \subseteq V(K \setminus K_1)$.
                        Hence, some connected component $K_2$ of $K \setminus b$ other than $K_1$ contains some vertex in $B$.  
                            If $K_2$ contains some vertex with a neighbor in $[c_6]^{-1}$, then let $[G \cup f]_1$ be the minor obtained via the component exchange from $[c_4]^{-1}$ to $[c_6]^{-1}$ that exchanges $K_2$.  
                            Otherwise, $b$ must have some neighbor in $[c_6]^{-1}$.
                            In this case, let $[G \cup f]_1$ be the minor obtained via the component exchange from $[c_4]^{-1}$ to $[c_5]^{-1}$ that exchanges $K_2$.  
            
                    Next, assume that Statement (b) is true, and so $K \setminus K_1$ contains some vertex with a neighbor $w \in [c_3]^{-1}$.
                    By a previous assumption, we have $w \in V(H_1)$.
                    Hence, $H \setminus H_1$ is a subgraph of some connected component $J_1$ of $H \setminus w$.
                        If some connected component $J_2$ of $H \setminus w$ other than $J_1$ contains some vertex in $X$, then let $[G \cup f]_1$ be the minor obtained via the component exchange from $[c_3]^{-1}$ to $[c_1]^{-1}$ that exchanges $J_2$.  
                         Otherwise, $J_1$ contains some vertex in $X$ since $|X| \geq 2$.  
                            If $w \in X$, then consider the minor $[G \cup f]_2$ obtained via the component exchange from $[c_4]^{-1}$ to $[c_6]^{-1}$ that exchanges $K_1$, and let it have have the vertex set $\{[b_i]_2\}$.  
                            Let $[G \cup f]_1$ be the minor obtained via the vertex exchange from $[b_3]_2^{-1}$ to $[b_4]_2^{-1}$ that exchanges $w$ and fixes $J_1$.  
                            Otherwise, 
                                if some connected component $H_2$ of $H \setminus a$ other than $H_1$ contains some vertex in $X$, then let $[G \cup f]_1$ be the minor obtained via the component exchange from $[c_3]^{-1}$ to $[c_1]^{-1}$ that exchanges $H_2$.  
                                Therefore, it must be that $J_1 \cap H_1$ contains some vertex in $X$.  
                                    If $a \in X$, then consider the minor $[G \cup f]_2$ obtained via the component exchange from $[c_4]^{-1}$ to $[c_6]^{-1}$ that exchanges $K_1$, and let it have the vertex set $\{[b_i]_2\}$.  
                                    Let $[G \cup f]_1$ be the minor obtained via the component exchange from $[b_3]_2^{-1}$ to $[b_4]_2^{-1}$ that exchanges $H_1$.  
                                    
                                    Lastly, by the above arguments, it suffices to assume that $X \subseteq V(J_1 \cap H_1)$.  
                                    Consider any vertex $x_1 \in X$ and let $R_1$ be the connected component of $H \setminus x_1$ that contains $w$.  
                                        If $R_1$ contains $a$, then let $[G \cup f]_1$ be the minor obtained via the vertex exchange from $[c_3]^{-1}$ to $[c_1]^{-1}$ that exchanges $x_1$ and fixes $R_1$.  
                                        Otherwise, let $R_2$ be the connected component of $H \setminus x_1$ that contains $a$.  
                                        Since $|X| \geq 2$, there exists a vertex $x_2 \in X$ in some connected component of $H \setminus x_1$.  
                                        By the above arguments, this connected component must be either $R_1$ or $R_2$.  
                                            If $x_2$ is contained in $R_1$, then consider the minor $[G \cup f]_2$ obtained via the component exchange from $[c_4]^{-1}$ to $[c_6]^{-1}$ that exchanges $K_1$, and let it have the vertex set $\{[b_i]_2\}$.  
                                            Let $[G \cup f]_1$ be the minor obtained via the component exchange from $[b_3]_2^{-1}$ to $[b_4]_2^{-1}$ that exchanges $R_1$.  
                                            Otherwise, let $[G \cup f]_2$ be obtained via the component exchange from $[c_4]^{-1}$ to $[c_6]^{-1}$ that exchanges $K_1$.  
                                            Let $[G \cup f]_1$ be the minor obtained via the vertex exchange from $[b_3]_2^{-1}$ to $[b_4]_2^{-1}$ that exchanges $x_1$ and fixes $R_2$.  

    \smallskip
    The cases above are exhaustive, and so the lemma is proved for the case where $[G \cup f]$ is the graph in Figure \ref{fig:k5_wing}.  

    Finally, consider the case where $[G \cup f]$ is the graph in Figure \ref{fig:k222_wing}.  
    As stated at the beginning of this proof, the argument is similar, so we omit the details but note the differences in the argument by examining the case where $H_1$ some contains vertices with neighbors in $[c_4]^{-1}$ and $[c_7]^{-1}$.  
    Let $[G \cup f]_1$ be obtained via the vertex exchange from $[c_3]^{-1}$ to $[c_5]^{-1}$ that exchanges $a$ and fixes $H_1$.  
    Clearly, $[G \cup f]_1$ is some graph $G'$ that can be obtained from one of the graphs in Figures \ref{fig:k222_not_wing_u4_vgeq1} - \ref{fig:k222_not_wing_u2_diagonals} or \ref{fig:k222_wing} by adding edges between its vertices, $[f]_1$ is $w_1w_2$, $[d_3]_1^{-1} \cup [d_4]_1^{-1}$ is a proper subset of $[c_3]^{-1} \cup [c_4]^{-1}$, and both $[f]_1^{-1}$ and $[f]^{-1}$ contain the same edges of $G$.  
    % some isomorphism between $[G \cup f]_1$ and some graph $G'$ that can be obtained from one of the graphs in Figures \ref{fig:k222_not_wing_u4_vgeq1} - \ref{fig:k222_not_wing_u1} or \ref{fig:k222_wing} by adding edges between its vertices maps $[f]_1$ to $w_1w_2$, $[d_3]_1^{-1} \cup [d_4]_1^{-1}$ is a proper subset of $[c_3]^{-1} \cup [c_4]^{-1}$, and both $[f]_1^{-1}$ and $[f]^{-1}$ contain the same subset of edges of $G$.  
    If $G'$ is the graph Figure \ref{fig:k222_wing}, then $[G \cup f]_1$ has the desired properties.  
    Otherwise, it is easy to see that $G \cup f$ has an $f$-preserving forbidden minor, which is a contradiction.  

    This completes the proof.  
\end{proof}

\subsection{Proof of Lemma \ref{lem:not_3-or-2_implies_not_1} for Theorem \ref{prop:3-connected_partial_3-tree_star_lemma}}
\label{sec:star_thm}

Lemmas \ref{lem:not_3-or-2_implies_not_1} follows immediately from Lemmas \ref{lem:not_2_implies_not_1} and \ref{lem:not_3_implies_not_2}, below.  

\begin{lemma}[If $G \cup \{wx,wy\}$ has a forbidden minor, then so does $G \cup wx$]
\label{lem:not_2_implies_not_1}
    Let $G$ be a $3$-connected graph and $F = \{wx,wy\}$ be a set of its nonedges.  
    If there exists a forbidden minor $[G \cup F]$ in which $[w]$, $[x]$, and $[y]$ are distinct vertices and $wx$ and $wy$ are retained, then some $f \in F$ is such that $G \cup f$ has a forbidden minor.
\end{lemma}

\begin{proof}
    Let $G' = G \cup F$ and assume that the described forbidden minor $[G']$ exists.  
    Since $G$ is $3$-connected, it contains at least three internally vertex-disjoint paths $P_a$, $P_b$, and $P_c$ between $w$ and $x$, where $a$, $b$, and $c$ are the first vertices along these paths not contained in $[w]^{-1}$, respectively.  
    Furthermore, since $wx$ and $wy$ are retained in $[G']$, $[x]^{-1} \cup [y]^{-1}$ does not contain $a$, $b$, or $c$.  
    Also, since $[w]$ has degree two in $[G'] \setminus [F]$, we can choose these paths such that $[a]=[b]$.  
    For any minor $[G']_{\star}$ in which $[w]_{\star}$, $[x]_{\star}$, and $[y]_{\star}$ are distinct vertices and $wx$ and $wy$ are retained, define $a([G']_{\star})$ to be the length of the longest path in $G'$ that contains $a$ and consists only of vertices in $V(P_a) \cap [a]_{\star}^{-1}$.  
    Define $b([G']_{\star})$ similarly.  
    We will show how to obtain a minor $[G']_1$ from $[G']$ via a sequence of contractions and component exchanges such that (i) $[G']_1$ is a forbidden minor and $[w]_1$, $[x]_1$, and $[y]_1$ are distinct vertices and one of the following statements is true: (ii) either $wx$ or $wy$ is doubled in $[G']_1$ or (iii) $a([G']_1) \geq a([G'])$, $b([G']_1) \geq b([G'])$, and at least one of these inequalities is strict.  
    If Statements (i) and (ii) are true, then the lemma is clearly proved.  
    Otherwise, if Statements (i) and (iii) are true, then since $P_a$ and $P_b$ are each finite and end with $x$, we can repeat this argument for $[G']_1$, and so on, until we obtain a forbidden minor of $G'$ for which Statements (i) and (ii) are true.  
    
    We complete the proof by showing how to obtain $[G']_1$.  
    Let $a'$ be the first vertex after $a$ in $P_a$ that is not contained in $[a]^{-1}$.  
    Define $b'$ similarly.  
    First, we show that we can wlog assume neither $a'$ nor $b'$ is contained in $[w]^{-1}$.  
    If this is not the case, then assume wlog that $a' \in [w]^{-1}$.  
    Let $a''$ be the first vertex in $P_a$ after $a'$ that is not contained in $[w]^{-1}$, and define $S_{a'}$ to be the path in $P_a$ between $a'$ and the vertex that immediately precedes $a''$.  
    Let $J$ be the subgraph of $G'$ induced by $[w]^{-1}$ and let $J_w$ be the connected component of $J \setminus S_{a'}$ that contains $w$.  
    Since $P_a$, $P_b$, and $P_c$ are pairwise vertex-disjoint, $J_w$ contains every vertex along $P_a$, $P_b$, and $P_c$ between $w$ and the vertices immediately preceding $a$, $b$, and $c$, respectively.  
    If some vertex in $J \setminus J_w$ neighbors some vertex in $[d]^{-1}$, where $[d]$ is the neighbor of $[w]$ other than $[a]$, $[x]$, or $[y]$, then let $[G']_2$ be the minor obtained via the component exchange from $[w]^{-1}$ to $[d]^{-1}$ that exchanges $J \setminus J_w$.  
    Otherwise, let $[G']_2$ be the minor obtained via the component exchange from $[w]^{-1}$ to $[a]^{-1}$ that exchanges $J \setminus J_w$.  
    Clearly, $[G']_2$ satisfies (i), $a([G']_2) \geq a([G'])$, and $b([G']_2) \geq b([G'])$.  
    % Let $J_{a'} = J \setminus J_w$.  
    % Note that $J_{a'}$ does not contain any vertex along $P_b$ or $P_b$ between $w$ and $b$ or $c$, respectively.  
    % If $a''$ is contained in $[a]^{-1}$ and $J_w$ neighbors $[d]^{-1}$, where $[d]$ is the neighbor of $[w]$ other than $[a]$, then the minor $[G']_1$ obtained via the component exchange from $[w]^{-1}$ to $[a]^{-1}$ that exchanges $J_{a'}$ satisfies Statements (i) and (iii).  
    % Otherwise, if either $a''$ is contained in $[a]^{-1}$ and $J_w$ does not neighbor $[d]^{-1}$ or $a''$ is contained in $[d]^{-1}$, then the minor $[G']_2$ obtained via the component exchange from $[w]^{-1}$ to $[d]^{-1}$ that exchanges $J_{a'}$ satisfies (i) and $a([G']_2) = a([G'])$ and $b([G']_2) = b([G'])$.  
    Repeating this argument for $[G']_2$, and so on, yields a minor with the desired property, which we can choose to be $[G']$ to begin with.  
    
    Next, let $H$ be the subgraph of $G'$ induced by $[a]^{-1}$.  
    Also, let $R_a$ be the path in $P_a$ between $a$ and the vertex that immediately precedes $a'$.  
    Define $R_b$ similarly.  
    Consider the connected component $H_b$ of $H \setminus R_a$ that contains $b$, and let $H_a = H \setminus H_b$.  
    Note that $H_a$ contains $R_a$ and $H_b$ contains $R_b$.  
    Let $V_a$ be the subset of all vertices of $G' \setminus ([a]^{-1} \cup [w]^{-1})$ that the vertices in $H_a$ neighbor.  
    Define $V_b$ similarly.  
    Note that $a' \in V_a$ and $b' \in V_b$ since neither $a'$ nor $b'$ is contained in $[w]^{-1}$.  
    There are two cases.  

    \smallskip
    \noindent\textbf{Case 1:} $[G']$ is $K_5$.  
    \smallskip
    
    If $V_a$ is a subset of either $[x]^{-1}$ or $[y]^{-1}$, then the minor $[G']_1$ obtained via the component exchange from $[a]^{-1}$ to $[V_a]^{-1}$ that exchanges $H_a$ satisfies Statements (i) and (ii).  
    Else, if $V_a$ is a subset of $[d]^{-1}$, then the minor $[G']_1$ obtained via the component exchange from $[a]^{-1}$ to $[d]^{-1}$ that exchanges $H_a$ satisfies Statements (i) and (iii).  
    Otherwise, if $V_a$ is a subset of $[d]^{-1} \cup [x]^{-1}$ but not a subset $[d]^{-1}$ or $[x]^{-1}$, then the minor $[G']_1$ obtained via the component exchange from $[a]^{-1}$ to $[y]^{-1}$ that exchanges $H_b$ satisfies Statements (i) and (ii).  
    A similar argument works if $V_a$ is a subset of $[d]^{-1} \cup [y]^{-1}$ but not a subset $[d]^{-1}$ or $[y]^{-1}$.  
    If $V_a$ is a subset of $[x]^{-1} \cup [y]^{-1}$ but not a subset $[x]^{-1}$ or $[y]^{-1}$, then $V_b$ contains some vertex in $[d]^{-1}$.  
    Furthermore, using arguments similar to those above, we can assume $V_b$ contains some vertex in $[x]^{-1} \cup [y]^{-1}$.  
    Wlog, if this vertex is contained in $[y]^{-1}$, then the minor $[G']_1$ obtained via the component exchange from $[a]^{-1}$ to $[x]^{-1}$ that exchanges $H_a$ satisfies Statements (i) and (ii).  
    Lastly, the above arguments allow us to assume that $V_a$ and $V_b$ each contain vertices in $[d]^{-1}$, $[x]^{-1}$, and $[y]^{-1}$, in which case the minor $[G']_1$ obtained via the component exchange from $[a]^{-1}$ to $[x]^{-1}$ that exchanges $H_a$ satisfies Statements (i) and (ii).  

    \smallskip
    \noindent\textbf{Case 2:} $[G']$ is $K_{2,2,2}$.  
    \smallskip
    
    Wlog, assume that the neighborhood of $[a]$ is $\{[d],[e],[w],[x]\}$, where $[e]$ is the vertex in $V([G']) \setminus \{[a],[d],[w],[x],[y]\}$.  
    If $V_a$ is a subset of $[x]^{-1}$, then the minor $[G']_1$ obtained via the component exchange from $[a]^{-1}$ to $[x]^{-1}$ that exchanges $H_a$ satisfies Statements (i) and (ii).  
    Else, if $V_a$ is a subset of $[d]^{-1}$, then the minor $[G']_1$ obtained via the component exchange from $[a]^{-1}$ to $[d]^{-1}$ that exchanges $H_a$ satisfies Statements (i) and (iii).  
    Otherwise, if $V_a$ is a subset of $[e]^{-1}$, then consider the minor $[G']_2$ obtained via the component exchange from $[a]^{-1}$ to $[e]^{-1}$ that exchanges $H_a$.  
    Observe that contracting $[bd]_2$ yields a $K_5$ minor $[G']_3$ in which $[w]_3$, $[x]_3$, and $[y]_3$ are distinct vertices and $wx$ and $wy$ are retained.  
    Hence, the existence of $[G']_1$ follows from Case 1.  
    Similar arguments apply for all remaining cases of $V_a$.  

    \smallskip
    The above cases are exhaustive, and so the lemma is proved.  
\end{proof}

\begin{lemma}[If $G \cup \{wx,wy,wz\}$ has a forbidden minor, then so does $G \cup \{wx,wy\}$]
\label{lem:not_3_implies_not_2}
    Let $G$ be a $3$-connected graph and $F = \{wx,wy,wz\}$ be some subset of its nonedges.  
    If there exists a forbidden minor $[G \cup F]$ in which $[w]$, $[x]$, $[y]$, and $[z]$ are distinct vertices and $wx$, $wy$, and $wz$ are retained, then some size two subset $F'$ of $F$ is such that $G \cup F'$ has a forbidden minor.
\end{lemma}

\begin{proof}
    The proof is similar to the proof of Lemma \ref{lem:not_2_implies_not_1}.  
    Let $G' = G \cup F$ and assume that $[G']$ exists.  
    Since $G$ is $3$-connected, it contains at least two internally vertex-disjoint paths $P_a$ and $P_b$ between $w$ and $x$, where $a$ and $b$ are the first vertices along these paths that are not contained in $[w]^{-1}$, respectively.  
    Since $wx$, $wy$, and $wz$ are retained in $[G']$, $[x]^{-1} \cup [y]^{-1} \cup [z]^{-1}$ does not contain $a$ or $b$.  
    Also, since $[w]$ has degree one in $[G'] \setminus [F]$, we have $[a]=[b]$.  
    We will show how to obtain a minor $[G']_1$ from $[G']$ via a sequence of contractions and component exchanges such that $[G']_1$ is a forbidden minor in which $[w]_1$, $[x]_1$, $[y]_1$, and $[z]_1$ are distinct vertices and either $wx$, $wy$, or $wz$ is doubled.  
    If such a minor exists, then the lemma is clearly proved.  
    
    We complete the proof by showing how to obtain $[G']_1$.  
    Let $a'$ be the first vertex after $a$ in $P_a$ that is not contained in $[a]^{-1}$.  
    Define $b'$ similarly.  
    First, we show that we can wlog assume neither $a'$ nor $b'$ is contained in $[w]^{-1}$.  
    If this is not the case, then assume that $a' \in [w]^{-1}$.  
    Since $[w]$ has degree one in $[G'] \setminus [F]$, the first vertex $a''$ after $a'$ in $P_a$ that is not contained in $[w]^{-1}$ is contained in $[a]^{-1}$.  
    Define $S_{a'}$ to be the path in $P_a$ between $a'$ and the vertex that immediately precedes $a''$.  
    Let $J$ be the subgraph of $G'$ induced by $[w]^{-1}$ and let $J_w$ be the connected component of $J \setminus S_{a'}$ that contains $w$.  
    % and let $J_{a'} = J \setminus J_w$.  
    Since $P_a$ and $P_b$ are vertex-disjoint, $J_w$ contains every vertex along $P_b$ between $w$ and the vertex immediately preceding $b$.  
    Hence, the minor $[G']_2$ obtained via the component exchange from $[w]^{-1}$ to $[a]^{-1}$ that exchanges $J \setminus J_w$ is a forbidden minor in which $[w]_2$, $[x]_2$, $[y]_2$, and $[z]_2$ are distinct vertices and $wx$, $wy$, and $wz$ are retained.  
    By repeating this argument for $[G']_2$, and so on, we eventually obtain a minor with the desired property, which we can choose $[G']$ to be to begin with.  
    
    Next, let $H$ be the subgraph of $G'$ induced by $[a]^{-1}$.  
    Also, let $R_a$ be the path in $P_a$ between $a$ and the vertex that immediately precedes $a'$.  
    Define $R_b$ similarly.  
    Consider the connected component $H_b$ of $H \setminus R_a$ that contains $b$, and let $H_a = H \setminus H_b$.  
    Note that $H_a$ contains $R_a$ and $H_b$ contains $R_b$.  
    Let $V_a$ be the subset of all vertices of $G \setminus ([a]^{-1} \cup [w]^{-1})$ that the vertices in $H_a$ neighbor.  
    Define $V_b$ similarly.  
    Note that $a' \in V_a$ and $b' \in V_b$ since neither $a'$ nor $b'$ is contained in $[w]^{-1}$.  
    There are two cases.  

    \smallskip
    \noindent\textbf{Case 1:} $[G']$ is $K_5$.  
    \smallskip
    
    If $V_a$ is a subset of either $[x]^{-1}$, $[y]^{-1}$, or $[z]^{-1}$, then the minor $[G']_1$ obtained via the component exchange from $[a]^{-1}$ to $[V_a]^{-1}$ that exchanges $H_a$ has the desired properties.  
    Else, if $V_a$ is a subset of the union of any two of these sets, wlog say $[x]^{-1} \cup [y]^{-1}$, but not a subset of any one of these sets, then the minor $[G']_1$ obtained via the component exchange from $[a]^{-1}$ to $[z]^{-1}$ that exchanges $H_b$ has the desired properties.  
    Otherwise, if $V_a$ and $V_b$ each contain vertices in $[x]^{-1}$, $[y]^{-1}$, and $[z]^{-1}$, then the minor $[G']_1$ obtained via the component exchange from $[a]^{-1}$ to $[x]^{-1}$ that exchanges $H_a$ has the desired properties.  

    \smallskip
    \noindent\textbf{Case 2:} $[G']$ is $K_{2,2,2}$.  
    \smallskip
    
    Wlog, assume that the neighborhood of $[a]$ is $\{[d],[w],[x],[y]\}$, where $[d]$ is the vertex in $V([G']) \setminus \{[a],[w],[x],[y],[z]\}$.  
    If $V_a$ is a subset of either $[x]^{-1}$ or $[y]^{-1}$, then the minor obtained via the component exchange from $[a]^{-1}$ to $[V_a]^{-1}$ that exchanges $H_a$ has the desired properties.  
    Else, if $V_a$ is a subset of $[d]^{-1}$, then consider the minor $[G']_2$ obtained via the component exchange from $[a]^{-1}$ to $[d]^{-1}$ that exchanges $H_a$.  
    Contracting $[bx]_2$ yields a minor with the desired properties.  
    Otherwise, if $V_a$ is a subset of $[x]^{-1} \cup [y]^{-1}$ but not a subset of $[x]^{-1}$ or $[y]^{-1}$, then the minor $[G']_2$ obtained via the component exchange from $[a]^{-1}$ to $[d]^{-1}$ that exchanges $H_b$ has the desired properties.  
    Similar arguments apply for all remaining cases of $V_a$.  

    \smallskip
    The above cases are exhaustive, and so the lemma is proved.  
\end{proof}

\begin{proof}[Proof of Lemma \ref{lem:not_3-or-2_implies_not_1}]
    The lemma follows immediately from Lemmas \ref{lem:not_2_implies_not_1} and \ref{lem:not_3_implies_not_2}.  
\end{proof}

\subsection{Proofs of Lemma \ref{lem:A_empty_config_space_connected} for Theorem \ref{prop:partial_3-tree_3-reflection}}
\label{sec:proof_partial_3-tree_3-reflection}

Lemma \ref{lem:A_empty_config_space_connected} is proved using Lemma \ref{lem:fixed_lengths_connected}, below.  
Recall the set $\Delta$ of vectors, defined in Section \ref{sec:tool-thms}.  

\begin{lemma}[Fibers are connected]
    \label{lem:fixed_lengths_connected}
    Let $G$ be a graph with $F$ being a subset of its nonedges such that $G \cup F$ is a $3$-tree, let $\ell$ be a squared edge-length map, and let $\phi_F$ be the Cayley map from the CS $\mathcal{C}^3(G,\ell)$ to the CCS $\Omega^3_F(G,\ell)$.  
    Also, consider the set $A = \{A_1,\dots,A_m\}$ containing each $K_4$ subgraph of $G \cup F$ that maps to a non-coplanar set of points under some realization in the CS.  
    If $A$ is empty, then $\phi^{-1}_F(y)$ is connected for any point $y$ in the CCS; otherwise, $\phi^{-1}_F(y) \cap \mathcal{C}^3_{\delta}(G,\ell)$ is connected for any point $y$ in the relative interior of the CCS and any $\delta \in \Delta$.  
\end{lemma}

\begin{proof}
    Let $T = G \cup F$ and $\ell': E(T) \rightarrow \mathbb{R}$ be the map that agrees with $\ell$ on $E(G)$ and with $y$ on $F$.  
    If $A$ is empty, then observe that $\phi^{-1}_F(y)$ is connected if and only if the CS $\mathcal{C}^3(T,\ell')$ is connected.  
    Otherwise, $\phi^{-1}_F(y) \cap \mathcal{C}^3_{\delta}(G,\ell)$ is connected if and only if the set $\mathcal{C}^3_{\delta}(T,\ell')$ is connected.  
    Let $X$ be $\mathcal{C}^3(T,\ell')$ if $A$ is empty, and $\mathcal{C}^3_{\delta}(T,\ell')$ otherwise.  
    We will show that $X$ is connected by induction on $n = |V(T)|$.  
    
    When $n = 4$, it is easy to verify that $X$ contains exactly one realization up to translations and rotations, and hence $X$ is connected.  
    Assume that $X$ is connected when $n = k$ for any integer $k \geq 4$, and we will prove it is connected when $n = k + 1$.  
    Observe that there exists a sequence $T_1,\dots,T_{k+1}$ such that $T_1$ is some $K_4$ subgraph of $T$, $T_{k+1}=T$, and $T_i$ is obtained from $T_{i-1}$ by adding a single vertex $u_i$ connected to a clique of size three. 
    Let $\ell'_k$ be the restriction of $\ell'$ to $E(T_k)$.  
    Also, let $T'$ be the unique $K_4$ subgraph of $T$ that contains $u_{k+1}$
    Set $\delta_k$ to be $\delta$ if $T'$ is not contained in $A$; otherwise, set it to be obtained from $\delta$ by deleting the entry corresponding to $T'$.  
    Lastly, let $X_k$ be the CS $\mathcal{C}^3(T_k,\ell'_k)$ if $A \subseteq \{T'\}$, and the set $\mathcal{C}^3_{\delta_k}(T_k,\ell'_k)$ otherwise.  
    By the inductive hypothesis, $X_k$ is connected.  

    Next, let $\ell'_{k+1}$ be the restriction of $\ell'$ to $E(T')$.  
    Note that any $3$-realization in $X$ can be obtained by combining some $3$-realization in $X_k$ with some $3$-realization $q$ of $(T',\ell'_{k+1})$ via translations and rotations, and possibly a reflection of $q$.  
    Combining this with the facts that $X_k$ is connected, $q$ is unique up to translations, rotations, and reflections, and $T_k \cap T'$ is a clique implies that $X$ is connected.  
\end{proof}

\begin{figure}[htb]
    \centering
    \begin{tikzpicture}
        \node[ellipse,minimum height=20pt,minimum width=80pt,draw] (cs) at (0,0) {$X$};
        \node[ellipse,minimum height=20pt,minimum width=80pt,draw] (q) at (4,-1) {$\Tilde{X}$};
        \node[ellipse,minimum height=20pt,minimum width=80pt,draw] (ccs) at (0,-2) {$Y$};

        \path[->,>=stealth] (cs) edge node [left] {$\phi$} (ccs)
        (cs) edge node [above] {$\pi$} (q)
        (q) edge node [below] {$\psi$} (ccs);
    \end{tikzpicture}
    \caption{Maps in the proof of Lemma \ref{lem:A_empty_config_space_connected}.  }
    \label{fig:diag}
\end{figure}

\begin{proof}[Proof of Lemma \ref{lem:A_empty_config_space_connected}]
    Let $X$ be the CS and $Y$ be the CCS if is $A$ is empty; otherwise, let $X$ be the $\delta$-oriented CS and $Y$ be the relative interior of the CCS.  
    Since $\phi_F$ is continuous and closed, so is its restriction $\phi$ to $X$.  
    The fact that $\phi(X)=Y$ is immediate if $A$ is empty, and follows from Lemma \ref{lem:preimage_closure}(i) otherwise.  
    Define $\Tilde{X}$ to be the quotient space that contains for each realization $p \in X$ the equivalence class $[p] = \{q \in X | \text{ } \phi(p) = \phi(q)\}$.  
    Let $\pi: X \rightarrow \Tilde{X}$ be the quotient map that sends $p$ to $[p]$ and let $\psi: \Tilde{X} \rightarrow Y$ be the map that sends $[p]$ to $\phi(p)$.  
    See Figure \ref{fig:diag}.  
    To prove that $X$ is connected, it suffices to show that both $\Tilde{X}$ and $\pi^{-1}([p])$ are connected, for any $[p] \in \Tilde{X}$.  
    Note that $\pi^{-1}([p]) = (\phi^{-1} \circ \psi)([p])$, which is connected by Lemma \ref{lem:fixed_lengths_connected}.  
    We complete the proof by showing that $\Tilde{X}$ is connected.  
    
    First, we prove that $\psi$ is a homeomorphism.  
    It suffices to show that $\psi$ is bijective, closed, and continuous.  
    Bijectivity follows from the definitions of $\Tilde{X}$ and $\psi$.  
    To prove that $\psi$ is closed, we show that $\psi(U)$ is closed for any closed set $U \subseteq \Tilde{X}$.  
    Since $\pi$ is continuous and $\phi$ is closed, we have that $(\phi \circ \pi^{-1})(U)$ is closed.  
    The fact that $\psi$ is closed now follows from the observation that $\psi = \phi \circ \pi^{-1}$.  
    To prove that $\psi$ is continuous, we first show that $\pi$ is closed.  
    It suffices to show that $(\pi^{-1} \circ \pi)(U)$ is closed for any closed set $U \subseteq X$.  
    Since $\phi$ is continuous and closed, we have that $(\phi^{-1} \circ \phi)(U)$ is closed.  
    The fact that $(\pi^{-1} \circ \pi)(U)$ is closed now follows from the observation that $\pi^{-1} \circ \pi = \phi^{-1} \circ \phi$.  
    Finally, we show that $\psi^{-1}(U)$ is closed for any closed set $U \subseteq Y$.  
    The above properties of $\phi$ and $\pi$ imply that $(\pi \circ \phi^{-1})(U)$ is closed.  
    The fact that $\psi^{-1}(U)$ is closed now follows from the observation that $\psi^{-1} = \pi \circ \phi^{-1}$.  
    Therefore, $\psi$ is a homeomorphism.  

    Finally, since $G \cup F$ is $3$-flattenable, Theorem \ref{thm:sitharam_willoughby} shows that $Y$ is convex, and hence connected.  
    Thus, $\Tilde{X}$ is connected since $\psi$ is a homeomorphism.  
    This completes the proof.  
\end{proof}

\subsection{Proofs of Lemmas \ref{lem:links_3-connected} and \ref{lem:ux2_uy2_3-covering} for Proposition \ref{prop:x2y2_3-Cayley-connected}} 
\label{sec:proof_x2y2_3-Cayley-connected}

Lemma \ref{lem:links_3-connected} is proved using Lemma \ref{lem:pin_pairs_contracted}, below.  

\begin{lemma}[No forbidden minor preserves endpoints of pins or $f$]
    \label{lem:pin_pairs_contracted}
    Consider any winged graph pair $(G,f)$ with a chained $f$-winged graph minor $[G \cup f]$.  
    For any two distinct vertices $x$ and $w$ of $G$ such that $x$ is an endpoint of some pin of $[G \cup f]$ and $w$ is an endpoint of either some pin of $[G \cup f]$ or $f$, $(G \setminus [f]^{-1}) \cup xw$ has no $xw$-preserving forbidden minor.  
\end{lemma}

\begin{proof}
    Assume to the contrary that $xw$ is preserved in some forbidden minor of $(G \setminus [f]^{-1}) \cup xw$.  
    Let $xy$ be a pin of $[G \cup f]$.  
    Since $xy$ is an edge of $G$ and no forbidden minor has a separator of size two, some $\{x,y\}$-component of $(G \setminus [f]^{-1}) \cup xw$ has an $xw$-preserving forbidden minor.  
    However, combining this with the existence of $[G \cup f]$ implies that $G \cup f$ has some $xw$-preserving forbidden minor that is also $f$-preserving, which is a contradiction.  
    Thus, the lemma is proved.  
\end{proof}

\begin{proof}[Proof of Lemma \ref{lem:links_3-connected}]
    Assume that some link $L$ of $[G \cup f]$ is not $3$-connected.  
    We will show that $(G,f)$ has the $3$-SIP.  
    First, we show that $L$ is $2$-connected.  
    If this is not the case, then $L$ has some separator $\{w\}$.  
    Since $[G \cup f]$ is chained and has singleton wing tips, $\{w\}$ must be a separator of $G$.  
    Also, since $G \cup f$ is $2$-connected, there are exactly two $\{w\}$-components of $G$: the unique ones containing the endpoints $u$ and $v$ of $f$, respectively.  
    Hence, we have $w \notin \{u,v\}$.  
    Let the vertex set of $[G \cup f]$ be $\{[a_i]\}$.  
    The above facts above show that $w \notin [a_1]^{-1} \cup [a_2]^{-1}$.  
    However, the paths in $G$ given by the existence of $[G \cup f]$ show that $u$ and $v$ are contained in the same $\{w\}$-component of $G$, which is a contradiction.  
    Therefore, $L$ is $2$-connected.  

    Next, since $L$ is $2$-connected but not $3$-connected, it has some minimal separator of size two.  
    We prove that any such separator $\{x,y\}$ has the following four properties.  
    Observe that the definition of a link of $[G \cup f]$ prevents $xy$ from being a pin of $[G \cup f]$.  
    Let $wz$ be any pin of $[G \cup f]$ contained in $L$ and let $w' \notin \{w,z\}$ be a vertex in $L$ that is an endpoint of either $f$ or a pin of $[G \cup f]$.   
    
    \noindent\textbf{Property (i):} $\{x,y\}$ is a separator of $G$.  
    This follows from the fact that each pin of $[G \cup f]$ is an edge of $G$.  

    \noindent\textbf{Property (ii):} $w$, $z$, and $w'$ are all contained in the same $\{x,y\}$-component of $G$.  
    Since each pin of $[G \cup f]$ is an edge of $G$, $w$ and $z$ are contained in the same $\{x,y\}$-component $H$ of $G$.  
    Hence, if $w'$ is not contained in $H$, then $\{x,y\}$ must be a $uv$-separator of $G$.  
    Along with the fact that $[G \cup f]$ has singleton wing tips and the paths in $G$ given by the existence of $[G \cup f]$, we get that $xy$ is a pin of $[G \cup f]$, a contradiction.  
    Therefore, $w'$ is contained in $H$.  
    
    \noindent\textbf{Property (iii):} if $xy$ is a nonedge of $G$, then $(G \cup xy,f)$ has the $3$-SIP.  
    Using Property (ii) and the paths in $G$ given by the existence of $[G \cup f]$, we see that $u$ and $v$ are contained in the same $\{x,y\}$-component of $G$.  
    Hence, $\{x,y\}$ is a separator of $G \cup f$.  
    Consequently, by Lemma \ref{lem:gluing}, it suffices to show that $(H \setminus f,f)$ has the $3$-SIP for each $\{x,y\}$-component $H$ of $G \cup \{xy,f\}$ that contains $f$.  
    Since $G \cup f$ has no $f$-preserving forbidden minor and $\{x,y\}$ is one of its minimal separators, $H$ has no $f$-preserving forbidden minor.  
    Therefore, since $|V(H)| < |V(G)|$, the converse direction of Theorem \ref{thm:3-sip_characterization} shows that $(H \setminus f, f)$ has the $3$-SIP.  

    \noindent\textbf{Property (iv):} if $xy$ is an edge of $G$, then $(G,f)$ has the $3$-SIP.  
    The proof of this property is similar to the proof of Property (iii), so we omit the details.  

    Finally, we use Properties (i) - (iv) to show that $(G,f)$ has the $3$-SIP.  
    If $xy$ is an edge of $G$, then this follows from Property (iv).  
    Hence, assume that every separator of $L$ that has size two is a nonedge of $G$.  
    Let $\{x,y\}$ now be a separator of $L$ such that the $\{x,y\}$-component $H_{wz}$ of $L$ containing $w$, $z$, and $w'$, which exists by Property (ii), contains the maximum number of vertices.  
    We will show that $(G,xy)$ has the $3$-SIP and then apply Property (iii) and Lemma \ref{lem:tool} to complete the proof.  
    By Property (i) and Lemma \ref{lem:gluing}, it suffices to show that $(H \setminus xy, xy)$ has the $3$-SIP for each $\{x,y\}$-component $H$ of $L \cup xy$ that contains $xy$.  
    To do this, we show that $H$ has no $xy$-preserving forbidden minor and then apply the converse direction of Theorem \ref{thm:3-sip_characterization}.  
    
    Assume to the contrary that some $\{x,y\}$-component $H$ of $L \cup xy$ has an $xy$-preserving forbidden minor.  
    Note that $H_{wz}$ is $2$-connected since $L$ is $2$-connected.  
    Hence, if $H$ is not $H_{wz}$, then $G$ has a $wz$-preserving forbidden minor, contradicting Lemma \ref{lem:pin_pairs_contracted}.  
    Otherwise, for the same reason, $H_{wz}$ does not have a $wz$-preserving forbidden minor.  
    Proposition \ref{prop:G_is_winged_graph} shows that $(H_{wz} \setminus wz,wz)$ is a winged graph pair.  
    Using Lemma \ref{lem:G_not_k5_k222_u4v0}, we see that there exists a $wz$-winged graph minor $[H_{wz}]_1$ such that $x$ is not contained in $[x']_1^{-1} \cup [y']_1^{-1}$.  
    Additionally, Lemma \ref{lem:wing-tips} allows us to assume that $[H_{wz}]_1$ has wing tips $x'$ and $y'$.  
    Using Lemma \ref{lem:pin_pairs_contracted}, we see that $w$, $z$, and $w'$ must be contained in $[x']_1^{-1} \cup [y']_1^{-1}$.  
    Consequently, $\{x',y'\}$ is a separator of $H_{wz}$, and hence of $L$.  
    Wlog, assume that $[x'x]_1$ is an edge.  
    Since $xy$ is an edge of $H_{wz}$, $y$ is either contained in $[x']_1^{-1}$ or not contained in $[x']_1^{-1} \cup [y']_1^{-1}$.  
    In the former case, it must be that $y = x'$.  
    Therefore, in either case, all paths from any vertex in $\{w,z,w'\}$ to any vertex in $\{x,y\}$ contain some vertex in $\{x',y'\}$.  
    However, this implies that the the $\{x',y'\}$-component of $L$ that contains $w$, $z$, and $w'$ has strictly more vertices than $H_{wz}$, which contradicts our maximality assumption on $\{x,y\}$.  
    Thus, $H_{wz}$ has no $xy$-preserving forbidden minor, and so the proof is complete.  
\end{proof}

Next, the proof of Lemma \ref{lem:ux2_uy2_3-covering} requires Lemma \ref{lem:pin_pairs_partial_3-tree}, below.  
% Lemma \ref{lem:pin_pairs_partial_3-tree} is proved using Lemma \ref{lem:3-flat_and_not_minimally_3-connected_implies_partial_3-tree}, also below.  

\begin{lemma}[Edge additions to links that preserve partial $3$-tree]
    \label{lem:pin_pairs_partial_3-tree}
    Consider any winged graph pair $(G,f)$ with a chained $f$-winged graph minor $[G \cup f]$, and any of its links $L$ that is $3$-connected.  
    For any two distinct vertices $x$ and $w$ of $L$ such that $x$ is an endpoint of some pin of $[G \cup f]$ and $w$ is an endpoint of either some pin of $[G \cup f]$ or $f$, $L \cup xw$ is a partial $3$-tree.  
\end{lemma}

\begin{proof}
    Note that $L \cup xw$ is $3$-connected and $L$ is one of its $3$-connected spanning proper subgraphs.  
    If we can show that $L \cup xw$ is $3$-flattenable, then the lemma follows from Lemma \ref{lem:3-flat_and_not_minimally_3-connected_implies_partial_3-tree}.  
    Consequently, assume that $L \cup xw$ is not $3$-flattenable, and hence has a forbidden minor.  
    By Lemma \ref{lem:pin_pairs_contracted}, $xw$ is contracted in every such minor.  
    Hence, Proposition \ref{prop:G_is_winged_graph} shows that $(L,xw)$ is a winged graph pair.  
    Therefore, either Proposition \ref{lem:winged-graph_cut_vertices}, if some $f$-winged graph minor of $L \cup xw$ has singleton wing tips, or Lemma \ref{lem:wing-tips}, otherwise, shows that $L$ has a separator of size two, a contradiction.  
    Thus, $L \cup xw$ has no forbidden minor, and so the proof is complete.
\end{proof}

\begin{proof}[Proof of Lemma \ref{lem:ux2_uy2_3-covering}]
    Wlog, we will prove that $(G,ux_2)$ has the $3$-covering map property.  
    Since $\{x_2,y_2\}$ is a clique separator of $G$, it suffices to show that $(L = L_0 \cup L_1,ux_2)$ has this property.  
    Consider any squared edge-length map $\ell$ and the set $\{X_i\}$ of connected components of the CS $\mathcal{C}^3(L,\ell)$.  
    Given a set of nonedges $F$ of $L$, let $\phi_F$ be the Cayley map from $\mathcal{C}^3(L,\ell)$ to the CCS $\Omega^3_F(L,\ell)$.  
    We must show that $\phi_{ux_2}(X_i) = \phi_{ux_2}(X_j)$ for any $i$ and $j$.  
    To do this, we show that $\phi_{F_0 \cup F_1}(X_i) = \phi_{F_0 \cup F_1}(X_j)$, where $F_0$ is the set of nonedges in $\{ux_1,uy_1\}$ and $F_1$ is the set of nonedges in $\{x_1x_2,y_1x_2\}$.  
    Then, for each point $p \in \phi_{F_0 \cup F_1}(X_i)$, the tetrahedral-inequality applied to the squared lengths $p$ and $\ell(x_1y_1)$ solely determines the interval $\phi_{ux_2}(\phi^{-1}_{F_0 \cup F_1}(p))$, because $\{x_1,y_1\}$ is a clique separator of $L$.  
    Combining the above facts proves the lemma.  

    We complete the proof by showing that $\phi_{F_0 \cup F_1}(X_i) = \phi_{F_0 \cup F_1}(X_j)$.  
    Let $\ell_0$ and $\ell_1$ be the restrictions of $\ell$ to $E(L_0)$ and $E(L_1)$, respectively.  
    Also, let $\{Y_i\}$ be the collection containing each connected component $A$ of $\mathcal{C}^3(L_0,\ell_0)$ and each set $A^r$ obtained from $A$ be performing an overall reflection on each of its realizations.  
    Define the collection $\{Z_i\}$ using $\mathcal{C}^3(L_1,\ell_1)$ analogously.  
    Since $\{x_1, y_1\}$ is a clique separator of $L$, there exist components $Y_j$ and $Z_k$ such that each realization in $X_i$ can be obtained from some pair of realizations $(p \in Y_j, q \in Z_k)$ via rigid-body motions.  
    Furthermore, since $Y_j$ and $Z_k$ are connected, it is easy to see that every pair of realizations $(p \in Y_j, q \in Z_k)$ can be combined via rigid-body motions to obtain a realization in $X_i$.  
    Therefore, it suffices to show that $(L_0, F_0)$ and $(L_1, F_1)$ have the $3$-covering map property.  
    Since $L_0$ is $3$-connected, Lemma \ref{lem:pin_pairs_partial_3-tree} shows that $L_0 \cup ux_1$ and $L_0 \cup uy_1$ are partial $3$-trees.  
    Thus, by Theorems \ref{prop:3-connected_partial_3-tree_star_lemma} and \ref{prop:partial_3-tree_3-reflection}, $(L_0, F_0)$ has the desired property.  
    An identical argument shows that $(L_1, F_1)$ has the $3$-covering map property.  
    This completes the proof.  
\end{proof}

\section{Proofs for the forward direction}
\label{app:forward}

\subsection{Proof of Proposition \ref{lem:no_type_2_no_3-sip}}
\label{sec:prop_6}

We present the proof of Lemma \ref{lem:no_type_2_neighborhoods} followed by the proof of Proposition \ref{lem:no_type_2_no_3-sip}.  

\begin{proof}[Proof of Lemma \ref{lem:no_type_2_neighborhoods}]
    Assume that $G$ has no Type (2) edge.  
    For Statement (i), consider any edge $e$ of $G$ that does not share an endpoint with $f$.  
    Since $f$ is clearly retained in the minor of $G \cup f$ obtained by contracting $e$, $e$ is not of Type (3).  
    Hence, $e$ must be of Type (1).  
    Since $[G \cup f]$ is $f$-preserving, this implies it is also $e$-preserving.  
    
    For Statement (ii), consider any vertex $[w] \in V(H)$.  
    If $|[w]^{-1}| > 1$, then there exists an edge of $G$ that does not share an endpoint with $f$ and that is contracted in $[G \cup f]$, contradicting Statement (i).  
    Hence, we have $|[w]^{-1}| = 1$.  

    For Statement (iii), if $([u]^{-1} \cup [v]^{-1}) \setminus \{u,v\}$ is empty, then we are done.  
    Otherwise, wlog let $x \in [u]^{-1}$ be any vertex in this set.  
    Then some edge of $G$ incident on $x$ is contracted in $[G \cup f]$.  
    Statement (i) tells us that $u$ is the only neighbor of $x$ in $[x]^{-1} \setminus \{u,v\}$.  
    Furthermore, $xu$ is clearly not of Type (1), and so it must be of Type (3).  
    Hence, $xv$ is an edge of $G \cup f$.  
    
    Next, assume that $x$ does not neighbor any vertex in $H^{-1}$.  
    Since $G \cup f$ is an atom, by assumption, $x$ neighbors some vertex $y$ other than $u$ or $v$.  
    By assumption and the discussion above, we have $y \in [v]^{-1} \setminus \{v\}$.  
    Let $J$ be the subgraph of $G \cup f$ induced by $[u]^{-1}$, which we have shown to be a star rooted at $u$.  
    The minor obtained via the vertex exchange from $[u]^{-1}$ to $[v]^{-1}$ that exchanges $x$ and fixes $J \setminus x$ is an $f$-preserving forbidden minor in which $xy$ is contracted.  
    However, since neither $x$ nor $y$ is an endpoint of $f$, this contradicts Statement (i), and so $x$ must neighbor some vertex $w$ in $H^{-1}$.  
    If $w$ is the only vertex in $H^{-1}$ that $x$ neighbors, then $xw$ is contracted in the minor obtained via the vertex exchange from $[u]^{-1}$ to $[w]^{-1}$ that exchanges $x$ and fixes $J \setminus x$, and this minor has an $f$-preserving forbidden minor.  
    Therefore, we arrive at a similar contradiction, and so $x$ neighbors at least two vertices in $H^{-1}$.  

    Finally, for Statement (iv), assume wlog that $|[u]^{-1}| > 2$.  
    Consider any two distinct vertices $x,y \in ([u]^{-1} \setminus \{u\})$ and observe that $[u]$ has exactly three neighbors in $H$, say $[w_1]$, $[w_2]$, and $[w_3]$.  
    By Statements (ii) and (iii), we can assume wlog that $x$ neighbors both $w_1$ and $w_2$.  
    Furthermore, $x$ and $y$ share at least one neighbor in $H^{-1}$, say $w_1$.  
    Wlog, we can assume that some vertex in $[u]^{-1} \setminus x$ neighbors $w_3$.  
    Therefore, we arrive at a similar contradiction by considering the minor obtained via the vertex exchange from $[u]^{-1}$ to $[w_2]^{-1}$ that exchanges $x$ and fixes $J \setminus x$.  
    Thus, the lemma is proved.  
\end{proof}

\begin{proof}[Proof of Proposition \ref{lem:no_type_2_no_3-sip}]
    By Lemma \ref{lem:no_type_2_pairs_no_3-sip}, every graph-nonedge pair in Figure \ref{fig:no_type_2} has a proper non-$3$-SIP map.  
    It remains to show that these are the only minimal graph-nonedge pairs with no Type (2) edge.  
    Consider any graph-nonedge pair $(G,f)$, where $f=uv$.  
    If $(G,f)$ is a pair in Figure \ref{fig:no_type_2}, then it is easy to verify that $(G,f)$ is minimal and has no Type (2) edge.  
    Next, we show that if $(G,f)$ is minimal and has no Type (2) edge, then it is one of the pairs in Figure \ref{fig:no_type_2}.  
    Since $(G,f)$ is minimal, $G \cup f$ is an atom and has an $f$-preserving forbidden minor $[G \cup f]$.  
    Furthermore, Statements (i)-(iv) of Lemma \ref{lem:no_type_2_neighborhoods} are true and show that $5 \leq |V(G)| \leq 8$.  
    Let $H = [G \cup f] \setminus \{[u],[v]\}$ and consider the subgraph $H^{-1}$ of $G \cup f$.  
    If $|[u]^{-1}|=|[v]^{-1}|=1$, then Statement (ii) shows that $(G,f)$ is the pair in either Figure \ref{fig:k5_f} or \ref{fig:k222_f}.  
    Hence, assume wlog that $|[u]^{-1}| > 1$.  
    By Statement (iv), we have $[u]^{-1} = \{u,x\}$.  
    There are three cases.  

    \smallskip
    \noindent\textbf{Case 1:} $u$ has no neighbor in $H^{-1}$.
    \smallskip

    If $|[v]^{-1}| = 1$, then Statement (iii) shows that $\{x,v\}$ is a clique separator of $G \cup f$, contradicting the fact that $G \cup f$ is an atom.  
    Otherwise, Statement (iv) shows that $[v]^{-1} = \{v,y\}$, and so all paths in $G \cup f$ from $u$ to any vertex in $H^{-1}$ contains some vertex in the set $\{v,x,y\}$.  
    Furthermore, Statement (iii) shows that $xy$ is a nonedge and $vx$ and $vy$ are edges of $G \cup f$.  
    Hence, $vy$ is clearly of Type (2), contradicting our assumption.  

    \smallskip
    \noindent\textbf{Case 2:} $u$ has more than one neighbor in $H^{-1}$.
    \smallskip

    Let $w$ be the vertex in $H^{-1}$ that neighbors $x$ but not $u$ if such a vertex exists, or any vertex in $H^{-1}$ that neighbors $x$ otherwise.  
    Also, let $J$ be the subgraph of $G \cup f$ induced by $[u]^{-1}$.  
    By Statement (iii), $J \setminus x$ is connected and contains $u$.  
    Hence, $xw$ is contracted in the minor obtained via the vertex exchange from $[u]^{-1}$ to $[w]^{-1}$ that exchanges $x$ and fixes $J \setminus x$, and this minor has an $f$-preserving forbidden minor, contradicting Statement (i).  

    \smallskip
    \noindent\textbf{Case 3:} $u$ has exactly one neighbor in $H^{-1}$.
    \smallskip

    Let $w$ be the vertex in $H^{-1}$ that neighbors $u$.  
    Note that $x$ neighbors each vertex in $H^{-1} \setminus w$.  
    If $vw$ is not an edge of $G \cup f$, then $f$ is retained in the minor obtained by contracting $uw$, and this minor clearly has an $f$-preserving forbidden minor.  
    Hence, $uw$ is not of Type (1) or (3).  
    However, since $uw$ is not of Type (2) by assumption, it must be reducing.  
    This contradicts the fact that $(G,f)$ is minimal, and so $vw$ must be an edge of $G \cup f$.  
    
    Next, consider the case where $|[v]^{-1}|=1$.  
    If $xw$ is a nonedge of $G \cup f$, then $(G,f)$ is the pair in either Figure \ref{fig:k5_f_expanded} or \ref{fig:k222_f_expanded}.  
    Otherwise, $\{x,v,w\}$ is a clique separator of $G \cup f$, contradicting the fact that $G \cup f$ is an atom.  
    Lastly, consider the case where $|[v]^{-1}| > 1$.  
    The discussion above shows that $w$ is the only neighbor of $v$ in $H^{-1}$.  
    Let $w'$ be any vertex in $H^{-1} \setminus w$ that neighbors $x$.  
    Also, let $J$ be defined as in Case 2.  
    Then, $xw'$ is contracted in the minor obtained via the vertex exchange from $[u]^{-1}$ to $[w']^{-1}$ that exchanges $x$ and fixes $J \setminus x$, and this minor has an $f$-preserving forbidden minor, contradicting Statement (i).      

    \smallskip
    The cases above are exhaustive, and so the lemma is proved.  
\end{proof}

\subsection{Proofs of lemmas for Proposition \ref{prop:IH}}
\label{sec:prop-8}

We present proofs of the lemmas in Section \ref{sec:prop-8-9} up to Lemma \ref{lem:G'_is_J+decorations}, the last one before the proof of Proposition \ref{prop:IH}.  
We start with the proof of Lemma \ref{lem:expanded_clique_graphs_and_connections}.  

\begin{proof}[Proof of Lemma \ref{lem:expanded_clique_graphs_and_connections}]
    For Statement (1), since $(G \cup f) \setminus C^{-1}$ is disconnected and $G \cup f$ is an atom, $C^{-1}$ is not a clique.  
    Hence, Figure \ref{fig:expanded_e-separating_cliques} shows all cases for $C^{-1}$ if $|V(C^{-1})| \leq 5$.  
    We complete the proof by showing that $|V(C^{-1})| \leq 5$.  
    Assume to the contrary that this is not the case.  
    Consider any $f$-separating minor $[G \cup f]_1$ such that $(C,M)$ is one of its $f$-separating pairs.  
    Then, $C$ is a clique on at least five vertices.  
    Consequently, since any atom $J$ of $[G \cup f]_1$ that contains $[f]_1$ has no $[f]_1$-preserving forbidden minor, by definition, $J$ cannot contain $C$.  
    This implies that there exists a clique minimal $[ux]_1$-separator $E$ of $[G \cup f]_1$, for some endpoint $[u]_1$ of $[f]_1$ and any vertex $[x]_1$ in $C \setminus E$, such that $C \setminus E$ is non-empty.  
    Clearly, $(E^{-1},M^{-1})$ is an expanded $f$-separating pair of $G \cup f$ whose expanded $f$-component is a proper subgraph of $H^{-1}$.  
    However, this shows that $C^{-1}$ is not in the top-level of $(G,f)$, which is a contradiction.  
    Therefore, we have $|V(C^{-1})| \leq 5$, and so Statement (1) is proved.  

    Next, for Statement (2), observe that $C^{-1}$ is retained in $[H^{-1} \cup I^{-1}]$, and $[H^{-1}]$ and $[I^{-1}]$ are vertices of $[H^{-1} \cup I^{-1}]$.  
    Hence, Statement (1) shows that $[C^{-1}]$ is one the graphs in Figure \ref{fig:expanded_e-separating_cliques}.  
    Consequently, $[H^{-1} \cup I^{-1}]$ is some graph $G'$ that can be obtained from one of the graphs in Figure \ref{fig:expanded_e-separating_cliques_conn} by deleting some edges, each of which has one endpoint in $\{a,b\}$ and endpoint in $\{w_i\}$.  
    % there is some isomorphism between $[H^{-1} \cup I^{-1}]$ and some graph $G'$ that can be obtained from one of the graphs in Figure \ref{fig:expanded_e-separating_cliques_conn} by deleting some edges, each of which has one endpoint in $\{a,b\}$ and endpoint in $\{w_i\}$.  
    Furthermore, $[C^{-1}]$ is the subgraph of $G'$ induced by the set $\{w_i\}$.  
    % this isomorphism maps $[C^{-1}]$ to the subgraph of $G'$ induced by the set $\{w_i\}$.  
    If $G'$ can be obtained without deleting any edges, then we are done.  
    Otherwise, assume wlog that the above-mentioned minor $[G \cup f]_1$ is obtained by contracting $w_1w_2$.  
    It is easy to see that either $G \cup f$ is not an atom or, for any endpoint $[u]_1$ of $[f]_1$ and any vertex $[x]_1$ of $M \setminus C$, $C$ is not a minimal $[ux]_1$-separator of $[G \cup f]_1$.  
    In either case, we get a contradiction, and thus Statement (2) is proved.  
\end{proof}

Lemma \ref{cor:top-level_structure} is proved using Lemmas \ref{lem:I'_fbm_after_f-component_contraction} and \ref{lem:type_2_yields_smaller_expanded_f-component} in Section \ref{sec:prop-8-9}.  
The proof of Lemma \ref{lem:I'_fbm_after_f-component_contraction} requires Lemma \ref{lem:pair_e'_exists}, below.  

\begin{lemma}[Vertex-disjoint paths after contraction in $(H^{-1} \setminus f) \setminus C^{-1}$]
    \label{lem:pair_e'_exists}
    Let $C^{-1}$ be any expanded $f$-separating CMS of a minimal graph-nonedge pair $(G,f)$, with $f = uv$, and let its expanded $f$-component and expanded minor component be $H^{-1}$ and $I^{-1}$, respectively.  
    Also, consider the minor $[G \cup f]$ obtained by contracting any edge of $(H^{-1} \setminus f) \setminus C^{-1}$ and any $f$-preserving forbidden minor $[I^{-1}]_1$.  
    There exists a pair $xy$ of distinct vertices in $C^{-1}$ that is preserved in $[I^{-1}]_1$ and some paths $P_{ux}$ and $P_{vy}$ in $G \cup f$ between $u$ and $x$ and between $v$ and $y$, respectively, that each contain exactly one vertex in $I^{-1}$ and such that $[P_{ux}]$ and $[P_{vy}]$ are vertex-disjoint.  
\end{lemma}

\begin{proof}
    Since $f$ is preserved in $[I^{-1}]_1$ and $(G \cup f) \setminus C^{-1}$ is disconnected, there exists a pair of distinct vertices $x'y'$ in $C^{-1}$ and an extension $[G \cup f]_2$ of $[I^{-1}]_1$ such that $[f]_2 = [x'y']_2$.  
    This implies that $x'y'$ is preserved in $[I^{-1}]_1$ and some vertex-disjoint paths $P_{ux'}$ and $P_{vy'}$ in $G \cup f$ between $u$ and $x'$ and between $v$ and $y'$, respectively, each contain exactly one vertex in $I^{-1}$.  
    If $[P_{ux'}]$ and $[P_{vy'}]$ are vertex-disjoint, then setting $xy = x'y'$ proves the lemma.  
    
    Otherwise, letting $ab$ be the edge of $G \cup f$ that is contracted in $[G \cup f]$, it must be wlog that $P_{ux'}$ contains $a$ and $P_{vy'}$ contains $b$.  
    Since $G \cup f$ is an atom, there exists a path $P_{vz}$ in $G \cup f$ between $v$ and some vertex $z$ in $C^{-1}$ that does not contain either $a$ or $b$.  
    Wlog, we can choose $P_{vz}$ such that it contains exactly one vertex in $C^{-1}$ and does not contain $u$.  
    There are two cases, each with several subcases.  
    In each case, it will be clear that our choices for $P_{ux}$ and $P_{vy}$ each contain exactly one vertex in $I^{-1}$.  
    We will show that these paths are vertex-disjoint and at least one of them does not contain $a$ or $b$.  
    These facts imply that $[P_{ux}]$ and $[P_{vy}]$ are vertex-disjoint, which completes the proof.  

    \smallskip
    \noindent\textbf{Case 1:} $[z]_1=[x']_1$.  
    \smallskip

    There are several subcases.  
    Given a path $R$ containing two vertices $s$ and $t$, $R(st)$ denotes the path in $R$ between $s$ and $t$.  
    Additionally, we treat a path as a sequence of vertices and write the union of two paths $R_1$ and $R_2$ as $R_1R_2$.  

    \smallskip
    \noindent\textbf{Subcase 1:} $P_{vz}$ does not contain any vertex in either $P_{ux'}(ua)$ or $P_{vy'}(by')$.  
    \smallskip

    Setting $x = y'$, $y = z$, $P_{ux}=P_{ux'}(ua)P_{vy'}(by')$, and $P_{vy}=P_{vz}$ proves the lemma.  

    \smallskip
    \noindent\textbf{Subcase 2:} $P_{vz}$ contains some vertex in $P_{ux'}(ua)$, and the last of its vertices $w$ contained in either $P_{ux'}(ua)$ or $P_{vy'}(vb)$ is contained in $P_{ux'}(ua)$.  
    \smallskip
    
    Since $P_{vz}$ does not contain either $a$ or $b$, by assumption, we have $w \neq a$.  
    If $P_{vz}(wz)$ does not contain any vertex in $P_{vy'}(by')$, then set $x = z$, $y = y'$, $P_{ux}=P_{ux'}(uw)P_{vz}(wz)$, and $P_{vy} = P_{vy'}$.  
    By the properties of $w$ and $P_{vz}$, $P_{ux}$ does not contain either $a$ or $b$.  
    Also, our assumptions imply that neither $P_{ux'}(uw)$ nor $P_{vz}(wz)$ contains any vertex in $P_{vy'}$, and hence $P_{ux}$ and $P_{vy}$ are vertex-disjoint.  

    Next, assume that $P_{vz}(wz)$ contains some vertex in $P_{vy'}(by')$, and let $s$ be the first of its vertices in either $P_{ux'}(ax')$ or $P_{vy'}(by')$.  
    If $s$ is contained in $P_{ux'}(ax')$, then set $x = x'$, $y = y'$, $P_{ux}=P_{ux'}(uw)P_{vz}(ws)P_{ux'}(sx')$, and $P_{vy}=P_{vy'}$.  
    Otherwise, set $x = y'$, $y = x'$, $P_{ux}=P_{ux'}(uw)P_{vz}(ws)P_{vy'}(sy')$, and $P_{vy}=P_{vy'}(vb)P_{ux'}(ax')$.  
    In either case, a similar argument applies.  

    \smallskip
    \noindent\textbf{Subcase 3:} $P_{vz}$ contains some vertex in $P_{ux'}(ua)$, and the last of its vertices $w$ contained in either contained in either $P_{ux'}(ua)$ or $P_{vy'}(vb)$ is contained in $P_{vy'}(vb)$.  
    \smallskip

    This case is similar to Subcase 2, so we omit the details.  

    \smallskip
    \noindent\textbf{Subcase 4:} $P_{vz}$ contains some vertex in $P_{vy'}(by')$ but no vertex in $P_{ux'}(ua)$, and the first of its vertices $s$ contained in either $P_{ux'}(ax')$ or $P_{vy'}(by')$ is contained in $P_{ux'}(ax')$.  
    \smallskip

    In this case, set $x = y'$, $y = x'$, $P_{ux}=P_{ux'}(ua)P_{vy'}(by')$, and $P_{vy}=P_{vz}(vs)P_{ux'}(sx')$.  
    An argument similar to those in Subcase 2 applies.  

    \smallskip
    \noindent\textbf{Subcase 5:} $P_{vz}$ contains some vertex in $P_{vy'}(by')$ but no vertex in $P_{ux'}(ua)$, and the first of its vertices $s$ contained in either $P_{ux'}(ax')$ or $P_{vy'}(by')$ is contained in $P_{vy'}(by')$.  
    \smallskip

    In this case, set $x = x'$, $y = y'$, $P_{ux}=P_{ux'}$, and $P_{vy}=P_{vz}(vs)P_{uy'}(sy')$.  
    An argument similar to those in Subcase 2 applies.  

    \smallskip
    \noindent\textbf{Case 2:} $[z]_1 \neq [x']_1$.  
    \smallskip

    This case is similar to Case 1, so we omit the details.  

    \smallskip
    The cases above are exhaustive, and so the lemma is proved.  
\end{proof}

\begin{proof}[Proof of Lemma \ref{lem:I'_fbm_after_f-component_contraction}]
    By definition, $I^{-1}$ has an $f$-preserving forbidden minor $[I^{-1}]_1$.  
    Hence, since $I^{-1}$ is preserved in $[G \cup f]$, it suffices to show that the edge contracted in $[G \cup f]$ is contracted in some extension of $[I^{-1}]_1$. 
    This follows easily from Lemma \ref{lem:pair_e'_exists}.  
\end{proof}

The proof of Lemma \ref{lem:type_2_yields_smaller_expanded_f-component} requires Lemma \ref{lem:G'_connected}, below.

\begin{lemma}[Weakly retained subgraphs that are contained in atoms]
\label{lem:G'_connected}
    Let $C^{-1}$ be an expanded $f$-separating CMS of a minimal graph-nonedge pair $(G,f)$, and $J^{-1}$ be either its expanded $f$-component or expanded minor component.  
    Also, let $G' = (G \cup f) \setminus (J^{-1} \setminus C^{-1})$ and $[G \cup f]$ be the minor obtained by contracting any edge of $G \cup f$.  
    If $G'$ is weakly retained in $[G \cup f]$, then $[G'] \setminus E$ is connected for any CMS $E$ of $[G \cup f]$ contained in $[J^{-1}]$.  
\end{lemma}

\begin{proof}
    Assume that $G'$ is weakly retained in $[G \cup f]$.  
    Then, the endpoints of the edge contracted in $[G \cup f]$ must be contained in $J^{-1}$, and at least one of these endpoints is not contained in $C^{-1}$.  
    Hence, we have $[J^{-1}]^{-1} = J^{-1}$.  
    Since $E$ is contained in $[J^{-1}]$, this implies that $E^{-1}$ is contained in $J^{-1}$.  
    Therefore, $[G'] \setminus E$ is connected if and only if $G' \setminus A$ is connected, where $A = E^{-1} \cap C^{-1}$.  
    We complete the proof by showing that $G' \setminus A$ is connected.  
    
    If $A$ is empty, then we are done.  
    Otherwise, by definition, $G'$ contains either the expanded $f$-component or expanded minor component of $C^{-1}$, say $K^{-1}$.  
    Hence, using Lemma \ref{lem:expanded_clique_graphs_and_connections} Statement (2) and the observation that $A$ is a clique, it is easy to see that $K^{-1} \setminus A$ is connected.  
    Furthermore, since $C^{-1}$ is not a clique, $C^{-1} \setminus A$ is non-empty.  
    Therefore, since $G \cup f$ is an atom, for any vertex $x$ in $G' \setminus K^{-1}$, there exists a path in $G \cup f$ between $x$ and some vertex in $C^{-1} \setminus A$ that does not contain any vertex of $A$.  
    Note that we can choose this path such that it is contained in $G'$.  
    Combining this with the fact that $K^{-1} \setminus A$ is connected shows that $G' \setminus A$ is connected.  
\end{proof}

\begin{proof}[Proof of Lemma \ref{lem:type_2_yields_smaller_expanded_f-component}]
    Assume to the contrary that $C^{-1}$ is in the top level of $(G,f)$ but some edge $ab$ of $(H^{-1} \setminus f) \setminus C^{-1}$ is of Type (2).  
    Let $[G \cup f]$ be the $f$-separating minor obtained by contracting $ab$ and let $G' = (G \cup f) \setminus (H^{-1} \setminus C^{-1})$.  
    We will show that (i) $[G']$ is contained in some atom $N$ of $[G \cup f]$, (ii) there exists an $f$-separating pair $(E,N)$ of $[G \cup f]$, and (iii) $[C^{-1}] \setminus E$ is non-empty.  
    Since $[C^{-1}] \setminus E$ is contained in both $[G']$ and $[H^{-1}]$, (i) - (iii) show that the $f$-component $J$ of $(E,N)$ is a proper subgraph of $[H^{-1}]$.  
    Therefore, $J^{-1}$ is a proper subgraph of $H^{-1}$,  which contradicts the fact that $C^{-1}$ is in the top-level of $(G,f)$.  
    Thus, once we show (i) - (iii) are true, the proof will be complete.  
    
    To prove (i), consider any CMS $E$ of $[G \cup f]$.  
    Observe that $E$ must contain $[a]$, or else $E^{-1}$ is clearly a CMS of $G \cup f$, which is a contradiction.  
    Hence, since $E$ is a clique and $[C^{-1}]$ is a separator of $[G \cup f]$ that does not contain $[a]$, $E$ must be contained in $[H^{-1}]$.  
    Consequently, since $G'$ is retained in $[G \cup f]$, Lemma \ref{lem:G'_connected} shows that $[G'] \setminus E$ is connected.  
    Since $E$ was arbitrary, some atom $N$ of $G \cup f$ contains $G'$, as desired.  
    
    Next, to prove (ii), note that $G'$ contains the expanded minor component $I^{-1}$ of $C^{-1}$, and $[I^{-1}]$ has an $[f]$-preserving forbidden minor, by Lemma \ref{lem:I'_fbm_after_f-component_contraction}.  
    This implies that $N$ has an $[f]$-preserving forbidden minor.  
    Hence, since $[G \cup f]$ is $f$-separating, $N$ does not contain some endpoint $[u]$ of $[f]$, and so there exists a clique minimal $[ux]$-separator $E$ of $[G \cup f]$ for any vertex $[x]$ of $N \setminus E$.  
    Therefore, $(E,N)$ is an $f$-separating pair of $[G \cup f]$.  

    Finally, to prove (iii), observe that $C^{-1}$ is retained in $[G \cup f]$.  
    Hence, by Lemma \ref{lem:expanded_clique_graphs_and_connections} Statement (1), $[C^{-1}]$ is not a clique.  
    Thus, since $E$ is a clique, we get that $[C^{-1}] \setminus E$ is non-empty, which completes the proof.  
\end{proof}

Lemma \ref{prop:C'_leq_4vert} is proved using Lemma \ref{lem:C'_5_H'_wing} in Section \ref{sec:prop-8-9}.  
The proof of Lemma \ref{lem:C'_5_H'_wing} requires Lemma \ref{lem:4-clique separator_of_H'_contains_f}, below.  

\begin{lemma}[Each atom of top-level $f$-component contains $f$]
    \label{lem:4-clique separator_of_H'_contains_f}
    Let $(G,f)$ be a minimal graph-nonedge pair and let $C$ be a top-level $f$-separating CMS of an $f$-separating minor $[G \cup f]$ whose $f$-component is $H$.  
    Then, every atom of $H$ contains $[f]$.  
\end{lemma}

\begin{proof}
    If $H$ is an atom, then we are done.  
    Otherwise, it suffices to show that every CMS of $H$ contains $[f]$.  
    Assume to the contrary that some CMS $E$ of $H$ does not contain some endpoint $[u]$ of $[f]$.  
    We will show that this implies $C$ is not in the top-level of $[G \cup f]$, which is a contradiction.  
    Let $H_1$ be the $E$-component of $H$ that contains $[u]$, $K$ be any other $E$-component of $H$, and $[x]$ be any vertex of $K \setminus E$.  
    Since $[u]$ and $[x]$ are contained in distinct $E$-components of $H$, they are distinct vertices and $[ux]$ is a nonedge of $[G \cup f]$.  
    Hence, Lemma \ref{cor:top-level_structure} tells us that $[x]$ must be contained in $C$.  
    
    Next, since $C$ is a clique, $K$ is the unique $E$-component of $H$ that contains $C$.
    Hence, $E$ is a CMS of $[G \cup f]$, $H_1$ is an $E$-component of $[G \cup f]$, and the $E$-component $J$ of $[G \cup f]$ that contains $K$ also contains $C$.  
    This implies that $J$ contains $G' = [G \cup f] \setminus (H \setminus C)$, which contains some atom $M$ such that $(C,M)$ is an $f$-separating pair of $[G \cup f]$.  
    Therefore, some subgraph $E'$ is such that $(E',M)$ is an $f$-separating pair of $[G \cup f]$.  

    Finally, note that $E'$ does not contain $[x]$.  
    Hence, a similar argument shows that $G'$ and $M$ are both contained in a unique $E'$-component $J'$ of $[G \cup f]$.  
    Since $E'$ is a minimal $[uy]$-separator for any vertex $[y]$ of $M \setminus E'$, $J'$ does not contain $[u]$.  
    Thus, letting $H_2$ be the $E'$-component of $[G \cup f]$ that contains $[u]$, we get that $H_2^{-1}$ is a proper subgraph of the expanded $f$-component of $C^{-1}$.  
    This shows that $C$ is not in the top-level of $[G \cup f]$, as desired, and so the proof is complete.  
\end{proof}

\begin{proof}[Proof of Lemma \ref{lem:C'_5_H'_wing}]
    Consider any $f$-separating minor $[G \cup f]$ obtained by contracting a single edge of $C^{-1}$ and such that $C$ is one of its $f$-separating CMSs.  
    If the $f$-component $H$ of $C$ is the graph in Figure \ref{fig:k5_wing} with $[f]$ as $w_1w_2$ and $C$ as the subgraph induced by $\{w_3,w_4,w_5,w_6\}$, then we can use Lemma \ref{lem:expanded_clique_graphs_and_connections} Statement (2) to show that the expanded $f$-component of $C^{-1}$ has the desired properties.  
    % some isomorphism between the $f$-component $H$ of $C$ and the graph in Figure \ref{fig:k5_wing} maps $[f]$ to $w_1w_2$ and $C$ to the subgraph induced by $\{w_3,w_4,w_5,w_6\}$, then it is easy to check using Lemma \ref{lem:expanded_clique_graphs_and_connections} Statement (2) that the desired isomorphism between the expanded $f$-component of $C^{-1}$ and some graph in Figure \ref{fig:5_top} exists.  
    Hence, we complete the proof by showing that $H$ has these properties.  
    We prove this by first showing that (i) $H$ has no $[f]$-preserving forbidden minor.  
    Second, we use (i) to show that (ii) some minor $[H]_1$ is the graph in Figure \ref{fig:k5_wing} such that $[f]_1$ is $w_1w_2$ and $|[y]_1^{-1}| = 1$ for each vertex $y \in V(C)$.  
    % some minor $[H]_1$ is such that $|[y]_1^{-1}| = 1$, for each vertex $y \in V(C)$, and some isomorphism between $[H]_1$ and the graph in Figure \ref{fig:k5_wing} maps $[f]_1$ to $w_1w_2$.  
    Third and finally, letting $f = uv$, we use (i) and (ii) to show that (iii) $|[u]_1^{-1}| = |[v]_1^{-1}| = 1$.  
    This completes the proof.  
    
    For Statement (i), since $H$ is a $C$-component of $[G \cup f]$, Lemma \ref{lem:clq_sep_props} in \ref{sec:atoms} states that any atom $J$ of $H$ is also an atom of $[G \cup f]$.  
    Additionally, by Lemma \ref{lem:4-clique separator_of_H'_contains_f}, $J$ contains $[f]$.  
    Since $[G \cup f]$ is $f$-separating, these facts imply that $J$ has no $[f]$-preserving forbidden minor.  
    Therefore, Statement (i) follows from the contrapositive of Lemma \ref{cor:minimal_k-clique-sum_component_containing_fbm_and_f}.  
    
    For Statement (ii), note that $C$ does not contain any endpoint of $[f]$, or else we can use Lemma \ref{lem:expanded_clique_graphs_and_connections} Statement (2) to show that $H$ has an $[f]$-preserving $K_5$ minor, which contradicts Statement (i).  
    Using the same lemma, we see that $H$ has a $K_5$ minor $[H]_2$ such that $|[y]_2^{-1}| = 1$, for each vertex $y \in V(C)$.  
    By Statement (i), $f$ is contracted in $[H]_2$.  
    Hence, there exists a minor $[H]_1$ from which $[H]_2$ can be obtained by contracting $[f]_1$.  
    Clearly, $[H]_1$ is one of the graphs in Figures \ref{fig:k5_wing}, \ref{fig:k5_not_wing_u4_v0}, \ref{fig:k5_not_wing_u4_vgeq1}, or \ref{fig:k5_not_wing_u3} with $[f]_1$ as $w_1w_2$ and $[C]_1$ as the subgraph induced by $\{w_3,w_4,w_5,w_6\}$.  
    % some isomorphism between $[H]_1$ and one of the graphs in Figures \ref{fig:k5_wing}, \ref{fig:k5_not_wing_u4_v0}, or \ref{fig:k5_not_wing_u4_vgeq1}-\ref{fig:k5_not_wing_u1} maps $[f]_1$ to $w_1w_2$ and $[C]_1$ to the subgraph induced by $\{w_3,w_4,w_5,w_6\}$.  
    Since $G \cup f$ is an atom, $H$ is $2$-connected.  
    Therefore, Lemma \ref{lem:G_not_k5_k222_u4v0} allows us to assume that $[H]_1$ is one of the graphs in Figures \ref{fig:k5_wing}, \ref{fig:k5_not_wing_u4_vgeq1}, or \ref{fig:k5_not_wing_u3}.  
    % the image of this isomorphism is one of the graphs in Figures \ref{fig:k5_wing} or \ref{fig:k5_not_wing_u4_vgeq1}-\ref{fig:k5_not_wing_u1}.  
    Statement (ii) now follows from Statement (i) and Lemma \ref{lem:k5_k222_not_wing_fm_e_not_contracted}.  

    Finally, for Statement (iii), wlog assume to the contrary that $|[u]_1^{-1}| > 1$.  
    If either $[u]$ has a neighbor in $C$ or $|[u]_1^{-1}| > 2$, then we can use Lemma \ref{cor:top-level_structure} to contradict Statement (i).  
    Otherwise, let $[x]$ be the unique vertex in $[u]_1^{-1} \setminus \{[u]\}$.  
    If $|[v]_1^{-1}| = 1$, then $\{[x],[v]\}\}$ is a clique separator of $H$.  
    However, since neither $[x]$ nor $[v]$ is contained in $C$, we have $1 = |[x]^{-1}| = |[v]^{-1}|$.  
    This implies that $G \cup f$ is not as atom, which is a contradiction.  
    Hence, assume that $|[v]_1^{-1}| > 1$.  
    Similar arguments show that $[v]$ has no neighbor in $C$, there exists a unique vertex $[x']$ in $[v]_1^{-1} \setminus \{[v]\}$, and $1 = |[x]^{-1}| = |[x']^{-1}|$.  
    Hence, $\{x,x'\}$ is a size two separator of $G \cup f$, and it is easy to see that in any $f$-preserving forbidden minor of $G \cup f$ some edge $e$ in the set $\{ux,ux',vx,vx'\}$ is contracted.  
    Furthermore, the minor obtained by contracting $e$ has an $f$-separating pair whose $f$-component $J$ is such that $J^{-1}$ is a proper subgraph of $H^{-1}$.  
    This contradicts the fact that $C$ is in the top-level of $[G \cup f]$, and so Statement (iii) is proved.  
    This completes the proof.  
\end{proof}

\begin{proof}[Proof of Lemma \ref{prop:C'_leq_4vert}]
    By Lemma \ref{lem:expanded_clique_graphs_and_connections} Statement (1), $C^{-1}$ is one of the graphs in Figure \ref{fig:expanded_e-separating_cliques}.  
    We will show that $C^{-1}$ is not any of the graphs in Figures \ref{fig:expanded_e-separating_cliques}(e)-(i).  
    Assume to the contrary that $C^{-1}$ is one of these graphs.  
    We will show that $(G,f)$ has a reducing edge, and hence is not minimal, which is a contradiction.  
    Let $H^{-1}$ and $I^{-1}$ be the expanded $f$-component and expanded minor component of $C^{-1}$, respectively.  
    By Lemma \ref{lem:C'_5_H'_wing}, $H^{-1}$ is one of the graphs in Figure \ref{fig:5_top} with $f$ as the red edge and $C^{-1}$ as the subgraph induced by the set $\{w_i\}$.  
    % some isomorphism between $H^{-1}$ and one of the graphs in Figure \ref{fig:5_top} maps $f$ to the red edge and $C^{-1}$ to the subgraph induced by the set $\{w_i\}$.  
    We consider each case for $H^{-1}$.  
    
    If $H^{-1}$ is the graph in either Figure \ref{fig:5_degree_1} or \ref{fig:5_chord_arg_1}, then let $[G \cup f]$ be the minor obtained by contracting $uw_1$.  
    Observe that $f$ is retained in $[G \cup f]$ and $[H^{-1}]$ is an atom that contains $[f]$ such that $[H^{-1}] \cup [w_1w_4]$ has an $f$-preserving forbidden minor.  
    Hence, some atom $J$ of $[G \cup f]$ contains $[H^{-1}]$.  
    Let $G' = G \setminus (H^{-1} \setminus C^{-1})$.  
    Since $G'$ is preserved in $[G \cup f]$, Lemma \ref{lem:expanded_clique_graphs_and_connections} Statement (2) tells us that some path in $[G']$ between $[w_1]$ and $[w_4]$ does not contain any vertex in $[C^{-1}] \setminus \{[w_1],[w_4]\}$.  
    Note that the only vertices in this path that are also in $[H^{-1}]$ are $[w_1]$ and $[w_4]$.   
    Therefore, Lemma \ref{lem:path_in_graph_path_in_atom} in \ref{sec:atoms} shows that some path in $J$ between $[w_1]$ and $[w_4]$ does not contain any vertex in $[H^{-1}] \setminus \{[w_1],[w_4]\}$.  
    The above facts imply that $J$ has an $[f]$-preserving forbidden minor, and so $uw_1$ is reducing, as desired.  

    If $H^{-1}$ is the graph in Figure \ref{fig:5_k5_chord-e}, then let $[G \cup f]$ be the minor obtained by contracting $uw_1$.  
    Observe that $[H^{-1}]$ is an atom with an $f$-preserving forbidden minor.  
    Hence, since some atom of $[G \cup f]$ contains $[H^{-1}]$, $uw_1$ is reducing, as desired.  
    A similar argument shows that $vw_4$ is reducing if $H^{-1}$ is the graph in Figure \ref{fig:5_k5_chord}, and $vw_3$ is reducing if $H^{-1}$ is the graph in Figure \ref{fig:5_k5_chord_arg_2}.  

    Next, assume that $H^{-1}$ is the graph in either Figure \ref{fig:5_cycle_4_u3}, \ref{fig:5_cycle_arg_3_2}, \ref{fig:5_cycle_3_u3}, or \ref{fig:5_cycle_3_arg_3_2}, or $H^{-1}$ is the graph in \ref{fig:5_cycle_4_arg_4} and $uw_2$ is an edge.  
    Since $I^{-1}$ has an $f$-preserving forbidden minor, $G \cup f$ has an $f$-preserving forbidden minor in which some edge $e$ with $v$ as an endpoint is contracted.  
    We will show that $e$ is reducing.  
    The previous observation shows that $e$ is not of Type (1).  
    Let $[G \cup f]$ be the minor obtained by contracting $e$.  
    Note that $f$ is retained in $[G \cup f]$, and so $e$ is not of Type (3).  
    We complete the proof by showing that $e$ is not of Type (2).  
    Since $G'$ is the union of $C^{-1}$ and some non-empty subset of connected components of $(G \cup f) \setminus C^{-1}$ and since $G \cup f$ is an atom, Lemma \ref{lem:I'_sep_components_contain_C'-E} in \ref{sec:atoms} states that, for any CMS $E$ of $G'$, each connected component of $G' \setminus E$ contains some vertex of $C^{-1} \setminus E$.  
    Along with the observation that $C^{-1}$ is an atom, this shows that $G'$ is an atom.  
    Since $G'$ is weakly retained in $[G \cup f]$, we get that $[G']$ is an atom.  
    Additionally, note that $[H^{-1}]$ is an atom.  
    Since $[C^{-1}]$ is not a clique, $[G \cup f] = [G'] \cup [H^{-1}]$, and $[C^{-1}] = [G'] \cap [H^{-1}]$, Lemma \ref{lem:gluing_min_k-clique-sum_graphs} in \ref{sec:atoms} shows that $[G \cup f]$ is an atom, and so $e$ is not of Type (2).  
    Therefore, $e$ is reducing, as desired.  

    Finally, assume that $H^{-1}$ is the graph in Figure \ref{fig:5_cycle_4_u2_edge}.  
    If either $uw_1$ or $vw_3$ is not of Type (1), then an argument similar to those above shows that at least one of these edges is reducing.  
    Hence, we complete the proof by showing that at least one of these edges is not of Type (1).  
    By definition, $I^{-1}$ has an $f$-preserving forbidden minor $[I^{-1}]_1$.  
    If either $w_1w_4$ or $w_3w_5$ is preserved in $[I^{-1}]_1$, then there exists an extension $[G \cup f]$ of $[I^{-1}]_1$ in which either $uw_1$ or $vw_3$ is contracted, and we are done.  
    Otherwise, note that $w_2w_3$ is preserved in $[I^{-1}]_1$ since $f$ is preserved in $[I^{-1}]_1$.  
    Therefore, there exists an extension $[G \cup f]$ of $[I^{-1}]_1$ in which both $uw_1$ and $vw_3$ are contracted, and hence neither of these edges is of Type (1).  
    Similar arguments apply if either $H^{-1}$ is the graph in Figure \ref{fig:5_cycle_3_u2_edge} or $H^{-1}$ is the graph in Figure \ref{fig:5_cycle_4_arg_4} and $uw_2$ is a nonedge.  
    
    \smallskip
    The cases above are exhaustive, and so the proposition is proved.  
\end{proof}

The next subsection is dedicated to proving Lemma \ref{prop:expC_3or4} from Section \ref{sec:prop-8-9}.  

\subsubsection{Proof of Lemma \ref{prop:expC_3or4} [Unpruned possible graphs for top-level $H^{-1}$]}
\label{sec:proof_top-level_H'_cases}

The proof is highly technical and has many cases, which are split among Lemmas \ref{lem:3top_uin}, \ref{lem:3top_uout}, and \ref{lem:top_remaining}, below.  
The proofs of these lemmas require Lemmas \ref{lem:nonedge_C'_no_fbm_H'+f'} and \ref{lem:top-level_contraction_not_type_2}, below.  
Lemma \ref{lem:nonedge_C'_no_fbm_H'+f'} is proved using Lemma \ref{lem:I'_one_vertex}, also below.  

\begin{lemma}[Conditions such that $I^{-1} \setminus C^{-1}$ contains exactly one vertex]
    \label{lem:I'_one_vertex}
    Let $C^{-1}$ be a top-level expanded $f$-separating CMS of a a minimal graph-nonedge pair $(G,f)$, and let its expanded $f$-component and expanded minor component be $H^{-1}$ and $I^{-1}$, respectively.  
    Also, consider the minor $[G \cup f]$ obtained by contracting all edges in $I^{-1} \setminus C^{-1}$.  
    If $F$ is a non-empty set of nonedges of $C^{-1}$ such that (i) some vertex of $C^{-1}$ is an endpoint of each nonedge in $F$, (ii) the vertex $[I^{-1} \setminus C^{-1}]$ neighbors the endpoints of each nonedge in $[F]$, and (iii) $H^{-1} \cup F$ has an $f$-preserving forbidden minor, then $I^{-1} \setminus C^{-1}$ contains exactly one vertex.  
\end{lemma}

\begin{proof}
    Assume to the contrary that $F$ is a non-empty set of nonedges of $C^{-1}$ that satisfy Statements (i)-(iii) but $I^{-1} \setminus C^{-1}$ contains more than one vertex.  
    Then, there exists some edge $e$ in $I^{-1} \setminus C^{-1}$.  
    We will show that $e$ is reducing with respect to $(G,f)$, which contradicts the minimality of this pair.  
    Since neither endpoint of $e$ is contained in $H^{-1}$, $f$ is retained in be the minor $[G \cup f]_1$ obtained by contracting $e$, and so $e$ is not of Type (3).  
    Combining this with Statements (i)-(iii) shows that $[G \cup f]_1$ has an $f$-preserving forbidden minor, and so $e$ is not of Type (1).  
    We complete the proof by showing that $e$ is not of Type (2).  
    
    Since $G \cup f$ is an atom, any CMS of $[G \cup f]_1$ must contain the vertex $[e]_1$.  
    Hence, since $[C^{-1}]_1$ is a separator of $[G \cup f]_1$, every CMS of $[G \cup f]_1$ is contained in $[I^{-1}]_1$.  
    Furthermore, observe that $G' = (G \cup f) \setminus (I^{-1} \setminus C^{-1})$ is weakly retained in $[G \cup f]_1$.  
    Therefore, Lemma \ref{lem:G'_connected} applies and shows that $[G']_1$ is contained in some atom $K$ of $[G \cup f]_1$.  
    Note that $K$ contains $[H^{-1}]_1$, and so $K \cup [F]_1$ has an $[f]_1$-preserving forbidden minor, by Statement (iii).  
    Also, since any CMS of $[G \cup f]_1$ contains $[e]_1$, $K$ contains $[e]_1$.  
    For each endpoint $y$ of each nonedge in $[F]_1$, Statement (ii) can be used to obtain some path in $[G \cup f]_1$ between $[e]_1$ and $y$ that contains exactly one vertex in $[H^{-1}]_1$.  
    Consequently, Lemma \ref{lem:path_in_graph_path_in_atom} in Section \ref{sec:atoms} allows us to choose this path to be contained in $K$.  
    Combining this with Statements (i)-(iii) shows that $K$ has an $[f]_1$-preserving forbidden minor, and thus $e$ is not of Type (2).  
\end{proof}
% \end{comment}

\begin{lemma}[$H^{-1}$ has no $f$-preserving forbidden minor]
    \label{lem:nonedge_C'_no_fbm_H'+f'}
    Consider a minimal graph-nonedge pair $(G,f)$ with a top-level expanded $f$-separating CMS $C^{-1}$ whose expanded $f$-component is $H^{-1}$.  
    Then, the following statements are true.  
    \begin{enumerate}
        \item If $C^{-1}$ is the graph in either Figure \ref{fig:exp_clique_sep_k2}, \ref{fig:exp_clique_sep_k3_chord}, or \ref{fig:exp_clique_sep_k3_degree_1}, then, for any non-empty subset $F$ of nonedges of $C^{-1}$, $H^{-1} \cup F$ has no $f$-preserving forbidden minor.  
        \item If $C^{-1}$ is the graph in Figure \ref{fig:exp_clique_sep_k3_cycle}, then there exists a nonedge $f'$ of $C'$ such that $H^{-1} \cup f'$ has no $f$-preserving forbidden minor.  
    \end{enumerate}
    Consequently, $H^{-1}$ has no $f$-preserving forbidden minor.  
\end{lemma}

\begin{proof}
    Let $I^{-1}$ be the expanded minor component of $C^{-1}$ and consider the minor $[G \cup f]$ obtained by contracting every edge in $I^{-1} \setminus C^{-1}$.  
    Also, let $[G \cup f]_1$ be an $f$-separating minor such that $C$ is one of its $f$-separating CMSs.  
    If Statement (1) is false, then some non-empty set $F$ of nonedges $C^{-1}$ is such that $H^{-1} \cup F$ has an $f$-preserving forbidden minor.  
    Using Lemma \ref{lem:expanded_clique_graphs_and_connections} Statement (2), we see that the vertex $[I^{-1} \setminus C^{-1}]$ neighbors the endpoints of each nonedge in $F$.  
    If Statement (2) is false, then, for any nonedge $f'$ of $C^{-1}$, $H^{-1} \cup f'$ has an $f$-preserving forbidden minor.  
    Using Statement (2) of Lemma \ref{lem:expanded_clique_graphs_and_connections}, we see that the vertex $[I^{-1} \setminus C^{-1}]$ neighbors the endpoints of some nonedge of $C^{-1}$.  
    In both cases, Lemma \ref{lem:I'_one_vertex} applies and shows that $I^{-1} \setminus C^{-1}$ contains exactly one vertex.  
    However, since $|V(C^{-1})| \leq 4$, this implies that $[I^{-1}]_1$ has no forbidden minor, which is a contradiction.  
    This completes the proof.  
\end{proof}

Lemma \ref{lem:top-level_contraction_not_type_2} is proved using Lemma \ref{lem:cliques_in_top}, below.  

\begin{lemma}[Properties of CMS in minor of $H^{-1}$]
\label{lem:cliques_in_top}
    Consider a minimal graph-nonedge pair $(G,f)$ with a top-level expanded $f$-separating CMS $C^{-1}$ whose expanded $f$-component is $H^{-1}$.  
    Also, let $G' = G \setminus (H^{-1} \setminus C^{-1})$ be weakly retained in some $f$-separating minor $[G \cup f]$ that is obtained by contracting some edge in $H^{-1}$.  
    If $E$ is a CMS of $[G \cup f]$ that is contained in $[H^{-1}]$, then any $E$-component of $[G \cup f]$ that does not contain some endpoint of $[f]$ does not have an $[f]$-preserving forbidden minor.  
\end{lemma}

% \noindent\textbf{Proof sketch:} if there exists a CMS $E$ of $[G \cup f]$ that is contained in $[H^{-1}]$ and such that some $E$-component $J$ of $[G \cup f]$ that does not contain some endpoint of $[f]$ has an $[f]$-preserving forbidden minor, then we can show that $C^{-1}$ is not in the top-level of $G \cup f$, which is a contradiction.  
% \begin{comment}
\begin{proof}
    Assume to the contrary that there exists a CMS $E$ of $[G \cup f]$ that is contained in $[H^{-1}]$ and such that some $E$-component $J$ of $[G \cup f]$ that does not contain some endpoint of $[f]$ has an $[f]$-preserving forbidden minor.  
    Wlog, we can assume that $J$ contains $E$.  
    Let $M$ be the atom of $[G \cup f]$ that is contained in $J$ and has an $[f]$-preserving forbidden minor, given by Lemma \ref{cor:minimal_k-clique-sum_component_containing_fbm_and_f}.  
    Then, there exists a subgraph $E'$ of $E$ such that $(E',M)$ is an $f$-separating pair of $[G \cup f]$.  
    We will show that the $f$-component $K$ of $(E',M)$ is a proper subgraph of $[H^{-1}]$.  
    If this is true, then $K^{-1}$ is a proper subgraph of $H^{-1}$, since the edge contracted to obtain $[G \cup f]$ is contained in $H^{-1}$.  
    However, this contradicts the fact that $C^{-1}$ is in the top-level of $G \cup f$, and so the lemma is proved.  

    We complete the proof by showing that the $f$-component $K$ of $(E',M)$ is a proper subgraph of $[H^{-1}]$.  
    Since $J$ contains $E$, $J$ is contained in some $E'$-component $J'$ of $[G \cup f]$.  
    Also, since $G'$ is weakly retained in $[G \cup f]$ and $E'$ is contained in $[H^{-1}]$, Lemma \ref{lem:G'_connected} tells us that $[G']$ is contained in some $E'$-component of $[G \cup f]$.  
    If $J'$ does not contain $[G']$, then $J'$ is contained in $[H^{-1}]$.  
    However, this implies that $H^{-1}$ has an $f$-preserving forbidden minor.  
    Combining this with Lemma \ref{prop:C'_leq_4vert} contradicts Lemma \ref{lem:nonedge_C'_no_fbm_H'+f'}.  
    Therefore, assume that $J'$ contains $[G']$.  
    Since $C^{-1}$ is weakly retained in $[G \cup f]$ and is not a clique, by Lemma \ref{lem:expanded_clique_graphs_and_connections} Statement (1), $[C^{-1}] \setminus E'$ is non-empty.  
    Note that $[C^{-1}] \setminus E'$ is contained in both $J' \setminus E'$ and $[H^{-1}]$.  
    Consequently, we get that $K$ is a proper subgraph of $[H^{-1}]$, as desired.  
\end{proof}
% \end{comment}

\begin{lemma}[Edges that are not Type (2) in $H^{-1}$]
\label{lem:top-level_contraction_not_type_2}
    Consider a minimal graph-nonedge pair $(G,f)$ with a top-level expanded $f$-separating CMS $C^{-1}$ whose expanded $f$-component is $H^{-1}$.  
    Also, let $u$ be an endpoint of $f$ not in $C^{-1}$ and let $G' = G \setminus (H^{-1} \setminus C^{-1})$.  
    If there exists an edge $ab$ of $H^{-1} \setminus u$ such that $G'$ is weakly retained in the minor obtained by contracting $ab$ and either 
    \begin{enumerate}[(i)]
        \item for any vertex $z$ of $C^{-1} \setminus \{a,b\}$, there exists some path in $H^{-1}$ between $z$ and $u$ that does not contain any vertex of $(C^{-1} \setminus z) \cup \{a,b\}$ or

        \item for any subgraph $X$ of $C^{-1}$ that is a clique and contains $V(C^{-1}) \cap \{a,b\}$ and for any connected component $Y$ of $C^{-1} \setminus X$, there exists some path in $H^{-1}$ between $u$ and some vertex in $Y$ that does not contain any vertex of $(C^{-1} \setminus Y) \cup \{a,b\}$,
    \end{enumerate}
    then $ab$ is not of Type (2).  
\end{lemma}

\begin{proof}
    Assume to the contrary that some edge $ab$ as described in the Lemma statement is of Type (2).  
    Then, the minor $[G \cup f]$ obtained by contracting $ab$ is $f$-separating, and hence has an $f$-separating pair $(E,M)$, by Lemma \ref{lem:separating_minor_has_separating_pair}.  
    We first make some observations about $E$ and the $E$-component $I$ of $[G \cup f]$ that contains $M$.  
    By definition, $I$ does not contain some endpoint of $f$ and has an $[f]$-preserving forbidden minor.  
    Hence, Lemma \ref{lem:cliques_in_top} applies and tells us that $E$ is not contained in $[H^{-1}]$.  
    Also, since $u$ is not an endpoint of $ab$, $[u]$ is contained in $[H^{-1}] \setminus [C^{-1}]$.  
    Hence, $[G \cup f] \setminus [C^{-1}]$ is disconnected.  
    Since $E$ is a clique that is not contained in $[H^{-1}]$, this implies that $E$ is contained in $[G']$.  
    Consequently, neither $E$ nor $I$ contains $[u]$.  

    Next, by the definition of $(E,M)$, for some endpoint $[v]$ of $f$ is and each vertex $[x]$ of $M \setminus E$, $E$ is a minimal $[vx]$-separator of $[G \cup f]$.  
    Since $[f]$ is an edge of $[G \cup f]$ and $I$ does not contain $[u]$, this true when we set $v=u$.  
    Furthermore, since $I \setminus E$ is connected, we get that $E$ is a minimal $[uy]$-separator of $[G \cup f]$ for each vertex $[y]$ of $I$.  
    We complete the proof by contradicting this statement.  
    
    First, letting $X = E^{-1} \cap C^{-1}$, we show that $C^{-1} \setminus X$ is non-empty.  
    Since $G \cup f$ is an atom, $E^{-1}$ must contain $ab$.  
    Let $A = E^{-1} \cap G'$.  
    The above facts show that $E^{-1} = A \cup a$.  
    Since $G'$ is weakly retained in $[G \cup f]$, $A$ is a clique.  
    Consequently, $X$ is a clique.  
    Also, since $C^{-1}$ is contained in $G'$ and is not a clique, by Lemma \ref{prop:C'_leq_4vert}, $C^{-1} \setminus X$ is non-empty.  

    Second, we show that $I^{-1}$ is contained in $H^{-1} \cup A$.  
    Consider any vertex $w$ of $G \cup f$ and let $J^{-1}$ be the expanded $f$-component of $E^{-1}$.  
    Note that $X$ contains $V(C^{-1}) \cap \{a,b\}$.  
    If $w$ is contained in $C^{-1} \setminus X$, then Statements (i) and (ii) and the properties of $E^{-1}$ show that $x$ is contained in $J^{-1}$.  
    Else, if $x$ is contained in $G' \setminus (C^{-1} \cup A)$, then, since $G \cup f$ is an atom, $A$ is a clique, and $C^{-1} \setminus X$ is non-empty, there exists some path in $G'$ between $x$ and some vertex $z$ in $C^{-1} \setminus X$ that does not contain any vertex in $A$.  
    Since we have shown that $z$ is contained in $J^{-1}$, we get that $x$ is contained in $J^{-1}$.  
    Therefore, the only vertices potentially not in $J^{-1}$ are those in $H^{-1} \cup A$.  

    Finally, let $[z]$ be any vertex of $I \setminus E$.  
    Since $I^{-1}$ is contained in $H^{-1} \cup A$ and each vertex of $C^{-1} \setminus X$ is contained in $J^{-1}$, $[z]$ is contained in $[H^{-1}] \setminus [C^{-1}]$.  
    Furthermore, since $[C^{-1}]$ is a separator of $[G \cup f]$ and every path between $[z]$ and $[u]$ contains some vertex in $[E]$, by definition, every path between $[z]$ and $[u]$ contains some vertex in $E \cap [C^{-1}]$.  
    Thus, since $E \cap [C^{-1}]$ is both a clique and a proper subgraph of $E$, $E$ is not a minimal clique $[uz]$-separator, which is a contradiction.  
    This completes the proof.  
\end{proof}

% a, b, c, d, e, f, g, h, i j, k, l, m, n, o, p, q, r, s, t, u, v

\begin{lemma}[Cases for $H^{-1}$ when $C^{-1}$ is graph in Figure \ref{fig:exp_clique_sep_k2} and contains endpoint of $f$]
    \label{lem:3top_uin}
    Let $C^{-1}$ be a top-level expanded $f$-separating CMS of a minimal graph-nonedge pair $(G,f)$, and let its expanded $f$-component be $H^{-1}$.
    If $C^{-1}$ is the graph in Figure \ref{fig:exp_clique_sep_k2} and contains exactly one endpoint of $f$, then $H^{-1}$ is one of the graphs in Figures \ref{fig:3top_vin_uonly_2}, \ref{fig:3top_vin_uplus_2}, or \ref{fig:3top_vin_uonly_1} with $f$ as the red edge and $C^{-1}$ as the subgraph induced by $\{w_i\}$.  
    % some isomorphism between $H^{-1}$ and one of the graphs in Figures \ref{fig:3top_vin_uonly_2}, \ref{fig:3top_vin_uplus_2}, or \ref{fig:3top_vin_uonly_1} maps $f$ to the red edge and $C^{-1}$ to the subgraph induced by $\{w_i\}$.  
\end{lemma}

\begin{proof}
    Let $f=uv$, where $u$ is not contained in $C^{-1}$.  
    If $u$ is the only vertex of $H^{-1} \setminus C^{-1}$, then the fact that $G \cup f$ is an atom implies that $H^{-1}$ is one of the graphs in Figures \ref{fig:3top_vin_uonly_2} or \ref{fig:3top_vin_uonly_1} and has the desired properties.  
    % some isomorphism between $H^{-1}$ and one of the graphs in Figures \ref{fig:3top_vin_uonly_2} or \ref{fig:3top_vin_uonly_1} has the desired properties.  
    Otherwise, we will consider the two cases for $v$ in $C^{-1}$.

    \smallskip
    \noindent\textbf{Case 1:} $v$ is either $w_1$ or $w_3$.  
    \smallskip

    We will show that some $H^{-1}$ is the graph in Figure \ref{fig:3top_vin_uplus_2} has the desired properties.  
    % isomorphism between $H^{-1}$ and the graph in Figure \ref{fig:3top_vin_uplus_2} has the desired properties.  
    Wlog, let $v = w_1$.  
    By Lemma \ref{cor:top-level_structure}, each vertex of $H^{-1} \setminus (C^{-1} \setminus u)$ neighbors $u$, $w_1$, and at least one vertex of $C^{-1} \setminus w_1$.  
    First we show that no such vertex neighbors $w_2$.  
    If this is not the case, then let $x$ be any such vertex that neighbors $w_2$.  
    We will show that $xw_2$ is reducing with respect to $(G,f)$, which contradicts the minimality of this pair.  
    Since $f$ is retained in the minor $[G \cup f]$ obtained by contracting $xw_2$, $xw_2$ is not of Type (3).  
    To see that $xw_2$ is not of Type (1) or (2), consider the $f$-separating minor $[G \cup f]_1$ obtained by contracting some edge in $\{w_1w_2,w_2w_3\}$ and that has $C$ as one of its $f$-separating CMSs.  
    Since $G \cup f$ is an atom, there exists some path in $G \cup f$ from $u$ to $w_3$ that does not contain any vertex in $\{x,w_1,w_2\}$.  
    Also, the definition of $C$ shows that some $C$-component of $[G \cup f]_1$ that does not contain $[f]_1$ has an $[f]_1$-preserving forbidden minor.  
    By combining these facts, it is easy to see that $[G \cup f]$ has an $[f]$-preserving forbidden minor, and so $xw_2$ is not of Type (1).  
    Lastly, observe that $G' = (G \cup f) \setminus (H^{-1} \setminus C^{-1})$ is weakly retained in $[G \cup f]$.  
    Furthermore, Statement (i) of Lemma \ref{lem:top-level_contraction_not_type_2} is true if we set $ab = xw_2$.  
    Consequently, this lemma shows that $xw_2$ is not of Type (2), and so the above facts show that $xw_2$ is reducing.  

    Next, we show that $u$ neighbors $w_2$.  
    Assume this is not the case.  
    The above argument shows that the neighborhood of any vertex $x$ in $H^{-1} \setminus (C^{-1} \setminus u)$ is $\{u,w_1,w_3\}$.  
    Since $f$ is retained in the minor $[G \cup f]$ obtained by contracting $xw_3$, $xw_3$ is not of Type (3).  
    Also, since some $C$-component of $[G \cup f]_1$ that does not contain $[f]_1$ has an $[f]_1$-preserving forbidden minor, $[G \cup f]$ clearly has an $[f]$-preserving forbidden minor, and so $xw_3$ is not of Type (1).  
    Hence, $xw_3$ must be of Type (2).  
    Let $E$ be the subgraph of $[G \cup f]$ induced by $\{[w_1],[w_3]\}$ and let $I^{-1}$ be the expanded minor component of $C^{-1}$.  
    Observe that $I^{-1}$ is preserved in $[G \cup f]$ and $E$ is a $[uy]$-separator for any vertex $[y]$ of $[I^{-1}] \setminus E$.  
    Since $I^{-1}$ has an $f$-preserving forbidden minor, Lemma \ref{cor:minimal_k-clique-sum_component_containing_fbm_and_f} can be used to show that there exists an atom $M$ of $[G \cup f]$ that has an $[f]$-preserving forbidden minor and such that $E$ is a $[uy]$-separator for any vertex $[y]$ of $M \setminus E$.  
    Furthermore, since $G \cup f$ is an atom and $E$ is a clique separator of $[G \cup f]$ that contains exactly two vertices, it is easy to see that $E$ is a minimal $[uy]$-separator for any vertex $[y]$ of $M \setminus E$.  
    Therefore, $(E,M)$ is an $f$-separating pair of $[G \cup f]$.  
    However, note that the $[f]$-component of $E$ is $[H^{-1}] \setminus [w_2]$, and so the expanded $f$-component of $E^{-1}$ is $H^{-1} \setminus w_2$.  
    This shows that $C^{-1}$ is not in the top-level of $G \cup f$, which is a contradiction.  

    Finally, we show that either $u$ neighbors $w_3$ or $H^{-1} \setminus (C^{-1} \setminus u)$ contains at least two vertices, which completes the proof.  
    If this is not the case, then $H^{-1} \setminus (C^{-1} \setminus u)$ contains exactly one vertex $x$ and all paths in $H^{-1}$ between $u$ and $w_3$ contain some vertex in $\{x,w_1,w_2\}$.  
    If $w_1w_2$ is contracted in $[G \cup f]_1$, then $C^{-1}$ is clearly not in the top-level of $G \cup f$, which is a contradiction.  
    Otherwise, we arrive at a similar contradiction by considering the minor $[G \cup f]$ obtained by contracting $xw_1$ and the subgraph of $[G \cup f]$ induced by the set $\{[w_1],[w_2]\}$.  

    \smallskip
    \noindent\textbf{Case 2:} $v$ is $w_2$.  
    \smallskip
    
    We will show this case is not possible.  
    Consider the case where some vertex $x$ of $H^{-1} \setminus (C^{-1} \setminus u)$ neighbors each vertex of $C^{-1}$.  
    Since $G \cup f$ is an atom, there exists a path in $H^{-1}$ between $u$ and $w_1$ that does not contain any vertex in $\{x,w_2,w_3\}$.  
    Similarly, there exists a path in $H^{-1}$ between $u$ and $w_3$ that does not contain any vertex in $\{x,w_1,w_2\}$.  
    Hence, $H^{-1} \cup w_1w_3$ has an $f$-preserving $K_5$ minor, which contradicts Lemma \ref{lem:nonedge_C'_no_fbm_H'+f'}.  

    Next, consider the case where no vertex of $H^{-1} \setminus (C^{-1} \setminus u)$ neighbors each vertex of $C^{-1}$.  
    A similar argument as in Case 1 shows that no vertex of $H^{-1} \setminus (C^{-1} \setminus u)$ neighbors either $w_1$ or $w_3$.  
    However, this implies that some vertex in $H^{-1} \setminus (C^{-1} \setminus u)$ neighbors only $u$ and $v$, which contradicts Lemma \ref{cor:top-level_structure}.  

    \smallskip
    The cases above are exhaustive, and so the lemma is proved.  
\end{proof}

\begin{lemma}[$C^{-1}$ contains endpoint of $f$ if it is graph in Figure \ref{fig:exp_clique_sep_k2}]
    \label{lem:3top_uout}
    If any top-level expanded $f$-separating CMS $C^{-1}$ of a minimal graph-nonedge pair $(G,f)$ is the graph in Figure \ref{fig:exp_clique_sep_k2}, then $C^{-1}$ contains some endpoint of $f$.  
\end{lemma}

\begin{proof}
    Assume to the contrary that $C^{-1}$ does not contain any endpoint of $f$.  
    Let $f=uv$ and let $H^{-1}$ be the expanded $f$-component of $C^{-1}$.  
    There are five cases.  

    \smallskip
    \noindent\textbf{Case 1:} $H^{-1} \setminus (C^{-1} \setminus \{u,v\})$ is empty and either $u$ or $v$ neighbors $w_2$.  
    \smallskip

    Wlog, assume that $v$ neighbors $w_2$.  
    Since $G \cup f$ is an atom, $u$ neighbors some vertex in $\{w_1,w_3\}$, wlog say $w_1$, and either $u$ or $v$ neighbors $w_3$.  
    Consider the case where $v$ neighbors $w_3$.  
        If $v$ neighbors $w_1$, then $u$ must neighbor $w_3$.  
            If, additionally, $u$ neighbors $w_2$, then $H^{-1} \cup w_1w_3$ has an $f$-preserving forbidden minor, which contradicts Lemma \ref{lem:nonedge_C'_no_fbm_H'+f'}.  
            Otherwise, an argument as in Case 1 of Lemma \ref{lem:3top_uin} shows that $vw_2$ is reducing with respect to $(G,f)$, which contradicts the minimality of this pair.  
                % Note that $f$ is retained in the minor $[G \cup f]$ obtained by contracting $vw_2$, and so $vw_2$ is not of Type (3).  
                % To see that $vw_2$ is not of Type (1) or (2), consider the $f$-separating minor $[G \cup f]_1$ obtained by contracting some edge in $\{w_1w_2,w_2w_3\}$ and that has $C$ as one of its $f$-separating CMSs.  
                % The definition of $C$ shows that some $C$-component of $[G \cup f]_1$ that does not contain $[f]_1$ has an $[f]_1$-preserving forbidden minor.  
                % Hence, it is easy to see that $[G \cup f]$ has an $[f]$-preserving forbidden minor, and so $vw_2$ is not of Type (1).  
                % Lastly, observe that $G' = (G \cup f) \setminus (H^{-1} \setminus C^{-1})$ is weakly retained in $[G \cup f]$ and Statement (i) of Lemma \ref{lem:top-level_contraction_not_type_2} is true if we set $ab = vw_2$.  
                % Consequently, this lemma shows that $vw_2$ is not of Type (2), and so the above facts show that $vw_2$ is reducing.  
        Therefore, assume $v$ does not neighbor $w_1$.  
            If, additionally, $u$ does not neighbor $w_3$, then a similar argument shows that $uw_1$ is reducing, which is a contradiction.  
            Otherwise, 
                if $u$ neighbors $w_2$, then $\{u,w_2,w_3\}$ is a clique separator of $G \cup f$, which is a contradiction.  
                Else, a similar argument shows that $vw_2$ is reducing, which is a contradiction.  
    
    Next, consider the case where $u$ neighbors $w_3$ and $v$ does not.  
        If $u$ neighbors $w_2$, then the $\{u,w_1,w_2\}$ is a clique separator of $G \cup f$, which is a contradiction.  
        Otherwise, a similar argument shows that $vw_2$ is reducing, which is a contradiction.  

    \smallskip
    \noindent\textbf{Case 2:} $H^{-1} \setminus (C^{-1} \setminus \{u,v\})$ is empty and no vertex in $\{u,v\}$ neighbors $w_2$.  
    \smallskip

    Since $G \cup f$ is an atom, if either $u$ or $v$ neighbors both $w_1$ and $w_3$, then $u$ and $v$ each neighbor both $w_1$ and $w_3$.  
    Consider the $f$-separating minor $[G \cup f]_1$ obtained by contracting some edge in $\{w_1w_2,w_2w_3\}$ and that has $C$ as one of its $f$-separating CMSs.  
    The definition of $C$ shows that some $C$-component of $[G \cup f]_1$ that does not contain $[f]_1$ has an $[f]_1$-preserving forbidden minor.  
    This implies that some edge $e \in \{vw_1,vw_3\}$ is contracted in $[G \cup f]_1$, and so $e$ is not of Type (1).  
    Note that the unique atom of the minor $[G \cup f]$, obtained by contracting $e$, that contains $[u]$ is $K_3$, and hence $e$ is clearly not of Type (3).  
    Consequently, $e$ is of Type (2).  
    Let $E$ be the CMS of $[G \cup f]$ induced by the set $\{[w_1],[w_3]\}$.  
    Using an argument as in Case 1 of Lemma \ref{lem:3top_uin}, we see that $E^{-1}$ is an expanded $f$-separating CMS of $G \cup f$ whose expanded $f$-component is a proper subgraph of $H^{-1}$.  
    However, this contradicts the fact that $C^{-1}$ is in the top-level of $G \cup f$.  

    Next, assume that neither $u$ nor $v$ neighbors both $w_1$ and $w_3$.  
    Since $G \cup f$ is an atom, we can assume wlog that $u$ neighbors $w_1$ and $v$ neighbors $w_3$.  
    An argument as in Case 1 of Lemma \ref{lem:3top_uin}, except using Statement (ii) of Lemma \ref{lem:top-level_contraction_not_type_2} instead of Statement (i), shows that $uw_1$ is reducing.  
 
    \smallskip
    \noindent\textbf{Case 3:} $H^{-1} \setminus (C^{-1} \setminus \{u,v\})$ contains some vertex that neighbors $w_2$.  
    \smallskip

    Let $x$ be any vertex in $H^{-1} \setminus (C^{-1} \setminus \{u,v\})$ that neighbors $w_2$.  
    Also, assume wlog that $w_1w_2$ is contracted in $[G \cup f]_1$.  
    If every path in $G \cup f$ between either $u$ or $v$ and $w_3$ contains some vertex in $\{x,w_1,w_2\}$, then the subgraph induced by these three vertices is an expanded $f$-separating CMS of $G \cup f$ whose expanded $f$-component is a proper subgraph of $H^{-1}$, which is a contradiction.  
    Otherwise, 
        if some path in $G \cup f$ between either $u$ or $v$ and $w_1$ does not contain any vertex in $\{x,w_2,w_3\}$, then an argument as in Case 1 of Lemma \ref{lem:3top_uin} shows that $xw_2$ is reducing, which is a contradiction.  
        Else, note that $f$ is retained in the minor obtained by contracting $xw_2$, and so $xw_2$ is not of Type (3).  
        Also, an argument as in Case 1 of Lemma \ref{lem:3top_uin} shows that $xw_2$ is not of Type (1).  
        Hence, $xw_2$ of of Type (2).  
        Combining our assumption on the paths in $G \cup f$ with an argument as in Case 1 of Lemma \ref{lem:3top_uin}, we see that the subgraph of $G \cup f$ induced by $\{x,w_2,w_3\}$ is an expanded $f$-separating CMS of $G \cup f$ whose expanded $f$-component is a proper subgraph of $H^{-1}$, which is a contradiction.  

    \smallskip
    \noindent\textbf{Case 4:} $H^{-1} \setminus (C^{-1} \setminus \{u,v\})$ does not contain any vertex that neighbors $w_2$ and contains some vertex that neighbors exactly one vertex in $\{w_1,w_3\}$.  
    \smallskip

    By our assumption and Lemma \ref{cor:top-level_structure}, any vertex in $H^{-1} \setminus (C^{-1} \setminus \{u,v\})$ neighbors some vertex in $\{w_1,w_3\}$.  
    Consider the case where some vertex $x$ in this subgraph neighbors exactly one vertex in $\{w_1,w_3\}$, wlog say $w_1$.  
        Also, assume that every path in $G \cup f$ between either $u$ or $v$ and $w_3$ contains some vertex in $\{x,w_1,w_2\}$.  
            If $w_1w_2$ is contracted in $[G \cup f]_1$, then the subgraph induced by these three vertices is an expanded $f$-separating CMS of $G \cup f$ whose expanded $f$-component is a proper subgraph of $H^{-1}$, which is a contradiction.  
            Otherwise, since $f$ is retained in the minor $[G \cup f]$ obtained by contracting $xw_1$, $xw_1$ is not of Type (3).  
            Also, since $G \cup f$ is an atom it has some path between either $u$ or $v$ and $w_2$ that does not contain any vertex in $\{w_1,w_2\}$.  
            Hence, $[G \cup f]_1$ has an $[f]_1$-preserving forbidden minor, it is easy to see that $[G \cup f]$ has an $[f]$-preserving forbidden minor, and so $xw_1$ is not of Type (1).  
            Consequently, $xw_1$ is of Type (2).  
            We arrive at the same contradiction as above by using an argument as in Case 1 of Lemma \ref{lem:3top_uin}.  

        Next, assume some path in $G \cup f$ between either $u$ or $v$ and $w_3$ does not contain any vertex in $\{x,w_1,w_2\}$.  
        We will show that $xw_1$ is reducing with respect to $(G,f)$, which is a contradiction.  
        As noted above, $xw_1$ is not of Type (3).  
        A similar argument as above along with the existence of the above path shows that $xw_1$ is not of Type (1).  
        Lastly, observe that $G' = (G \cup f) \setminus (H^{-1} \setminus C^{-1})$ is weakly retained in $[G \cup f]$ and Statement (ii) of Lemma \ref{lem:top-level_contraction_not_type_2} is true if we set $ab = xw_1$.  
        Therefore, this lemma shows that $xw_1$ is not of Type (2), and so $xw_1$ is reducing.  

    \smallskip
    \noindent\textbf{Case 5:} $H^{-1} \setminus (C^{-1} \setminus \{u,v\})$ is non-empty and each of its vertices neighbors both $w_1$ and $w_3$ but not $w_2$.  
    \smallskip
    
    If either $H^{-1} \setminus (C^{-1} \setminus \{u,v\})$ contains at least three vertices or $u$ and $v$ each neighbor both $w_1$ and $w_3$, then $H^{-1} \cup w_1w_3$ has an $f$-preserving forbidden minor, which contradicts Lemma \ref{lem:nonedge_C'_no_fbm_H'+f'}.  
    Also, if exactly one vertex in $\{u,v\}$, wlog say $v$, neighbors either $w_2$ or exactly one vertex in $\{w_1,w_3\}$, then an argument as in Case 1 shows that some edge incident on $v$ is reducing, which is a contradiction.  
    Hence, assume none of these situations occur.  
        Additionally, assume that $H^{-1} \setminus (C^{-1} \setminus \{u,v\})$ contains exactly two vertices, say $x$ and $y$.  
            If all paths between either $u$ or $v$ and any vertex in $C^{-1}$ contains some vertex in $\{x,y\}$, then consider the minor $[G \cup f]$ obtained by contracting $vx$.  
                The unique atom of $[G \cup f]$ that contains $[f]$ is $K_3$, and so $vx$ is not of Type (3).  
                Also, since $[G \cup f]_1$ has an $[f]_1$-preserving forbidden minor, it is easy to see that $[G \cup f]$ has an $[f]$-preserving forbidden minor, and so $vx$ is not of Type (3).  
                Consequently, $vx$ is of Type (2).  
                Using as argument as in Case 1 of Lemma \ref{lem:3top_uin}, we see that the subgraph of $G \cup f$ induced by $\{v,x,y\}$ is an expanded $f$-separating CMS of $G \cup f$ whose expanded $f$-component is a proper subgraph of $H^{-1}$, which is a contradiction.  
                
            Next, assume that some path between either $u$ or $v$ and some vertex in $C^{-1}$ does not contain any vertex in $\{x,y\}$.  
            Since $x$ and $y$ are the only vertices in $H^{-1} \setminus (C^{-1} \setminus \{u,v\})$, every such path contains exactly one edge.  
            For all cases of the neighborhoods of $u$ and $v$ that satisfy the above assumptions, we see that $H^{-1} \cup w_1w_3$ has an $f$-preserving forbidden minor, which contradicts Lemma \ref{lem:nonedge_C'_no_fbm_H'+f'}.  

        Finally, assume that $H^{-1} \setminus (C^{-1} \setminus \{u,v\})$ contains exactly one vertex, say $x$.  
            An argument as in Case 4 shows that some path in $G \cup f$ between either $u$ or $v$ and $w_3$ does not contain any vertex in $\{x,w_1,w_2\}$.  
            Similarly, some path in $G \cup f$ between either $u$ or $v$ and $w_1$ does not contain any vertex in $\{x,w_2,w_3\}$.  
            These facts along with our assumption show that $w_1$ and $w_3$ each neighbor some vertex in $\{u,v\}$.  
            If either $u$ or $v$ neighbors $w_2$, then the above facts imply the existence of edges such that $H^{-1} \cup w_1w_3$ has an $f$-preserving forbidden minor, which contradicts Lemma \ref{lem:nonedge_C'_no_fbm_H'+f'}.  
            Otherwise, the above facts imply the existence of edges such that $G \cup f$ is not an atom, which is a contradiction.  

    \smallskip
    The above cases are exhaustive, and so the lemma is proved.  
\end{proof}

\begin{lemma}[Remaining cases for $H^{-1}$]
    \label{lem:top_remaining}
    Let $C^{-1}$ be a top-level expanded $f$-separating CMS of a minimal graph-nonedge pair $(G,f)$, and let its expanded $f$-component be $H^{-1}$.  
    If $C^{-1}$ is the graph in Figure \ref{fig:exp_clique_sep_k3_cycle}, \ref{fig:exp_clique_sep_k3_chord}, or \ref{fig:exp_clique_sep_k3_degree_1}, then $H^{-1}$ is one of the graphs in Figure \ref{fig:top-level_H'_cases} not mentioned in Lemma \ref{lem:3top_uin} with $f$ as the red edge and $C^{-1}$ as the subgraph induced by $\{w_i\}$.  
    % some isomorphism between $H^{-1}$ and one of the graphs in Figure \ref{fig:top-level_H'_cases} not mentioned in Lemma \ref{lem:3top_uin} maps $f$ to the red edge and $C^{-1}$ to the subgraph induced by $\{w_i\}$.  
\end{lemma}

The proof of Lemma \ref{lem:top_remaining} is omitted here since it is very similar to the proofs of Lemmas \ref{lem:3top_uin} and \ref{lem:3top_uout}, but it is available upon request.

\begin{proof}[Proof of Lemma \ref{prop:expC_3or4}]
    The lemma follows immediately from Lemmas \ref{prop:C'_leq_4vert} and \ref{lem:3top_uin}-\ref{lem:top_remaining}.  
\end{proof}

Lemma \ref{lem:C'_a-n} in Section \ref{sec:prop-8-9} is proved using Lemma \ref{lem:f'_and_J}, below.  
The proof of Lemma \ref{lem:f'_and_J} requires Lemma \ref{lem:I'_fbm_f_and_C'_nonedges_not_contracted}, also below.  

\begin{figure}[htb]
    \centering
    \begin{subfigure}{0.49\textwidth}
        \centering
        \includegraphics[width = 0.25\textwidth]{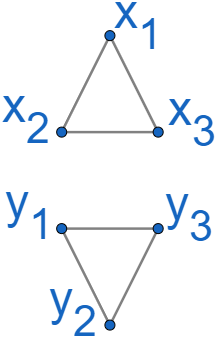}
        \subcaption{}
        \label{fig:k5_f'_from_exp_clique_sep_k3_degree_1}
    \end{subfigure}
    \begin{subfigure}{0.49\textwidth}
        \centering
        \includegraphics[width = 0.25\textwidth]{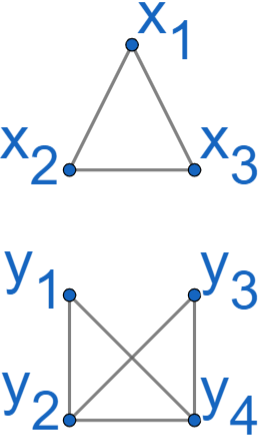}
        \subcaption{}
        \label{fig:k222_f'_from_exp_clique_sep_k3_degree_1}
    \end{subfigure}
    \caption{Graphs in the proof of Lemma \ref{lem:I'_fbm_f_and_C'_nonedges_not_contracted}.  }
    \label{fig:one_contraction_fbm}
\end{figure}

\begin{lemma}[$I^{-1}$ has $f'$-preserving forbidden minor for each nonedge of $C^{-1}$]
    \label{lem:I'_fbm_f_and_C'_nonedges_not_contracted}
    Let $C^{-1}$ be a top-level expanded $f$-separating CMS of a minimal graph-nonedge pair $(G,f)$ whose expanded minor component is $I^{-1}$.  
    $I^{-1}$ has an $f'$-preserving forbidden minor for each nonedge $f'$ of $C^{-1}$.
\end{lemma}

\begin{proof}
    By the definition of $C^{-1}$, $C$ is an $f$-separating CMS of some minor $[G \cup f]$, and the minor component $I$ of $C$ has an $[f]$-preserving forbidden minor $[I]_1$.  
    Note that each nonedge $f'$ of $C^{-1}$ is preserved in $[G \cup f]$.  
    Hence, it suffices to show that $I$ has an $[f']$-preserving forbidden minor, which we demonstrate now.  
    By Lemma \ref{prop:C'_leq_4vert}, $C^{-1}$ is one of the graphs in Figures \ref{fig:exp_clique_sep_k2}-\ref{fig:exp_clique_sep_k3_degree_1}.  
    This implies that $C$ is either $K_2$ or $K_3$.  
    If $C$ is $K_2$, then $[I]_1$ is $[f']$-preserving, since it is $[f]$-preserving.  
    Otherwise, let $[I]_2$ be the minor obtained by contracting all edges that are contracted in $[I]_1$ except for any edge of $C$.  
    If $[I]_2 = [I]_1$, then we are done.  
    Otherwise, since $[I]_1$ is $[f]$-preserving, it is obtained from $[I]_2$ by contracting exactly one edge of $[C]_2$.  
    Since $[I]_1$ is either $K_5$ or $K_{2,2,2}$, this implies that some spanning subgraph of $[I]_2$ is one of the graphs in Figure \ref{fig:one_contraction_fbm} with $[C]_2$ as the subgraph induced by $\{x_i\}$.  
    % some isomorphism between some spanning subgraph of $[I]_2$ and one of the graphs in Figure \ref{fig:one_contraction_fbm} maps $[C]_2$ to the subgraph induced by $\{x_i\}$.  

    Next, assume wlog that contracting $x_1x_2$ yields $[I]_1$.  
    If the above-mentioned some spanning subgraph of $[I]_2$ is the graph in Figure \ref{fig:k5_f'_from_exp_clique_sep_k3_degree_1}, then $[I]_1$ is $K_5$.  
    % the image of the above isomorphism is the graph in Figure \ref{fig:k5_f'_from_exp_clique_sep_k3_degree_1}, then $[I]_1$ is $K_5$.  
    This implies that $x_3$ neighbors each vertex in $\{y_1,y_2,y_3\}$, the union of the neighborhoods of $x_1$ and $x_2$ is $\{y_1,y_2,y_3,x_3\}$, and wlog $x_1$ has at least two neighbors in $\{y_1,y_2,y_3\}$, say $y_1$ and $y_2$.  
    For all cases of the neighborhoods of $x_1$ and $x_2$, we can clearly obtain an $[f']_2$-preserving $K_5$ minor of $[I]_2$, and hence an $[f']$-preserving $K_5$ minor of $I$.  

    Finally, if the above-mentioned some spanning subgraph of $[I]_2$ is the graph in Figure \ref{fig:k222_f'_from_exp_clique_sep_k3_degree_1}, then $[I]_1$ is $K_{2,2,2}$.  
    % the image of the above isomorphism is the graph in Figure \ref{fig:k222_f'_from_exp_clique_sep_k3_degree_1}, then $[I]_1$ is $K_{2,2,2}$.  
    This implies wlog that $x_3$ neighbors each vertex in $\{y_1,y_3,y_4\}$, the union of the neighborhoods of $x_1$ and $x_2$ is $\{y_1,y_2,y_3,x_3\}$, and $x_1$ wlog has at least two neighbors in $\{y_1,y_2,y_3\}$, say $y_1$ and $y_2$.  
    For all cases of the neighborhoods of $x_1$ and $x_2$, we can clearly obtain an $[f']_2$-preserving forbidden minor of $[I]_2$, and hence an $[f']$-preserving forbidden minor of $I$.  
    This completes the proof.  
\end{proof}

The \emph{atom graph} $\mathcal{G}(G)$ of a graph $G$ is the graph defined as follows.  
The vertex set of $\mathcal{G}(G)$ is the disjoint union $X(G) \sqcup Y(G)$ such that there exist bijections between $X(G)$ and the CMSs of $G$ and between $Y(G)$ and the atoms of $G$.  
For any vertex $z \in X(G) \sqcup Y(G)$, $G(z)$ is the subgraph of $G$ that corresponds to $z$.  
$\mathcal{G}(G)$ has an edge $xy$ if and only if $x \in X(G)$, $y \in Y(G)$, and $G(x)$ is a subgraph of $G(y)$.  

\begin{lemma}[Structure necessary for induction]
    \label{lem:f'_and_J}
    Let $C^{-1}$ be any top-level expanded $f$-separating CMS of a minimal graph-nonedge pair $(G,f)$ whose expanded $f$-component is $H^{-1}$, and let $G' = G \setminus (H^{-1} \setminus C^{-1})$.  
    Some nonedge $f'$ of $C^{-1}$ is contained in some atom $J$ of $G' \cup f'$ that has an $f'$-preserving forbidden minor.  
    If $C^{-1}$ is the graph in Figure \ref{fig:exp_clique_sep_k3_cycle}, then this is true for any nonedge of $C^{-1}$.  
    Furthermore, no clique subgraph of $C^{-1}$ separates the endpoints of $f'$ in $J \setminus f'$.  
\end{lemma}

\begin{proof}
    Lemma \ref{lem:I'_fbm_f_and_C'_nonedges_not_contracted} states that the expanded minor component of $C^{-1}$ has an $f'$-preserving forbidden minor for each nonedge $f'$ of $C^{-1}$.  
    Hence, Lemma \ref{cor:minimal_k-clique-sum_component_containing_fbm_and_f} shows that some atom $J$ of $G' \cup f'$ has an $f'$-preserving forbidden minor.  
    First, we demonstrate that we can choose $f'$ to be contained in $J$.  
    Note that $G \cup f$ is an atom, $(G \cup f) \setminus C^{-1}$ is disconnected, and $G'$ is the union of $C^{-1}$ and some non-empty set of connected components of $(G \cup f) \setminus C^{-1}$.
    Furthermore, by inspecting all possibilities for $C^{-1}$ given by Lemma \ref{prop:C'_leq_4vert}, we see that either $C^{-1}$ is an atom or, for any CMS $E$ of $C^{-1}$, there are exactly two $E$-components of $C^{-1}$, both of which are atoms.  
    Therefore, Lemma \ref{lem:I'_path} in \ref{sec:atoms} shows that the atom graph of $G'$ is a path.  

    Next, if $G'$ is an atom, then Lemma \ref{lem:I'_f'_cms} Statement (1) in \ref{sec:atoms} shows that any choice of $f'$ is contained in $J$.  
    Note that if $C^{-1}$ is the graph in Figure \ref{fig:exp_clique_sep_k3_cycle}, then $C^{-1}$ is an atom.  
    Hence, Lemma \ref{lem:I'_sep_components_contain_C'-E} in \ref{sec:atoms} shows that $G'$ is an atom.  
    Combining these facts demonstrates that any choice for $f'$ is contained in $J$.  
    Otherwise, if $G'$ is not an atom, then the atom graph of $G'$ has two distinct leaves $y_1$ and $y_2$ with not necessarily distinct neighbors $x_1$ and $x_2$, respectively.  
    Using an argument similar to the proof of Lemma \ref{lem:I'_path}, we see that $G'(y_1) \setminus G'(x_1)$ contains some vertex $u$ of $C^{-1}$, $G'(y_2) \setminus G'(x_2)$ contains some vertex $v$ of $C^{-1}$, and $G'(y_1) \setminus G'(x_1)$ and $G'(y_2) \setminus G'(x_2)$ do not share any vertices.  
    These facts show that $uv$ is a nonedge of $C^{-1}$.  
    Therefore, Lemma \ref{lem:I'_f'_cms} Statement (2) shows that $J$ contains $uv$.  

    Lastly, we show that no clique subgraph of $C^{-1}$ separates the endpoints of $f'$ in $J \setminus f'$.  
    It suffices to consider any clique subgraph $E$ of $C^{-1}$ that is contained in $J$ and does not contain any endpoint of $f'$.  
    Since $J$ has a forbidden minor and $|V(C^{-1})| \leq 4$, $J$ contains some vertex $x$ of $G' \setminus C^{-1}$.  
    Note that $E$ does not contain $x$.  
    Consider the case where $C^{-1}$ is the graph in Figure \ref{fig:exp_clique_sep_k3_cycle}, and let $f = uv$.  
    If $J$ either contains $C^{-1}$ or does not contain any vertex in $C^{-1} \setminus \{u,v\}$, then the claim is immediate.  
    Otherwise, let $y$ be the unique vertex contained in both $J$ and $C^{-1} \setminus \{u,v\}$.  
    Note that the only case for $E$ is the subgraph induced by $\{y\}$.  
    Since $J$ is an atom $vy$ is an edge, there exists a path in $J$ between $x$ and $u$ that does not contain any vertex in $\{v,y\}$.  
    Similarly, there exists a path in $J$ between $x$ and $v$ that does not contain any vertex in $\{u,y\}$.  
    Since these paths do not contain $f'$ or $y$, they connect $u$ and $v$ in $(J \setminus f') \setminus y$.  

    Finally, if $C^{-1}$ is not the graph in Figure \ref{fig:exp_clique_sep_k3_cycle}, then Lemma \ref{lem:expanded_clique_graphs_and_connections} Statement (2) shows that there exist paths in $G'$ that connect $x$ to the endpoints of each nonedge of $C^{-1}$ and that each contain exactly one vertex in $C^{-1}$.  
    Furthermore, Lemma \ref{lem:path_in_graph_path_in_atom} shows that we can choose these paths to be contained in $J$.  
    Therefore, clearly no clique subgraph of $C^{-1}$ separates the endpoints of $f'$ in $J \setminus f'$.  
\end{proof}

\begin{proof}[Proof of Lemma \ref{lem:C'_a-n}]
    By Lemma \ref{prop:expC_3or4}, $H^{-1}$ is one of the graphs in Figure \ref{fig:top-level_H'_cases}.  
    Let $G' = G \setminus (H^{-1} \setminus C^{-1})$ and consider the nonedge $f'$ of $C^{-1}$ and the atom $J$ of $G' \cup f'$ given by Lemma \ref{lem:f'_and_J}.  
    For each case of $H^{-1}$ in Figures \ref{fig:3top_vin_uonly_1} - \ref{fig:cycle_vin_uplus_4} or Figures \ref{fig:cycle_vin_uonly_2}, \ref{fig:cycle_vin_uplus_3}, \ref{fig:cycle_vin_uplus_2}, \ref{fig:deg1_vin_uplus_3}, or \ref{fig:deg1_vin_uplus_4} such that every dashed line-segment is an edge, observe that there exists a minor $[G \cup f]$ obtained via a single contraction such that $f$ is retained in $[G \cup f]$, $[f] = [f']$, and some induced subgraph of $[G \cup f]$ is $G' \cup f'$ with $[f]$ as $f'$.  
    % xxxsome isomorphism between some induced subgraph of $[G \cup f]$ and $G' \cup f'$ maps $[f]$ to $f'$.  
    These facts along with the existence of $J$ show that the edge contracted in $[G \cup f]$ is reducing with respect to $(G,f)$, which contradicts the minimality of this pair.  
\end{proof}

The next subsection is dedicated to proving Lemma \ref{lem:G'_f'_no_reducible_edge} from Section \ref{sec:prop-8-9}.  

\subsubsection{Proof of Lemma \ref{lem:G'_f'_no_reducible_edge}}
\label{sec:proof_IH_2}

The proof requires Lemmas \ref{lem:J+H'_atom} and \ref{lem:k_f-preserving_fbm}, below.  

\begin{lemma}[Properties of graphs with $H^{-1}$ in Figures \ref{fig:3top_vin_uonly_2} - \ref{fig:deg1_vout_uvonly}]
    \label{lem:J+H'_atom}
    Let $G$ be a graph with subgraphs $G_1$, $G_2$, and $S$ such that $G = G_1 \cup G_2$, $G_1$ is one of the graphs in Figures \ref{fig:3top_vin_uonly_2} - \ref{fig:deg1_vout_uvonly}, $S = G_1 \cap G_2$ is the subgraph of $G_1$ induced by some subset of $\{w_i\}$ that contains some nonedge $e$, and $G \setminus S$ is disconnected.  
    If $G_2 \cup e$ is an atom and no clique subgraph of $S$ separates the endpoints of $e$ in $G_2$, then $G$ is an atom.  
\end{lemma}

\begin{proof}
    The lemma will follow by applying Lemma \ref{lem:gluing_top_part_to_atom} after we show that Statements (1)-(4) of this lemma are satisfied.  
    Statements (1) and (4) are true by assumption.  
    For each case of $G_1$, it is easy to verify that either $G_1$ or $G_1 \cup e$ is an atom, and so Statement (2) is satisfied.  
    Also, for any clique subgraph $A$ of $S$, $G_1 \setminus A$ is connected, and so Statement (3) is satisfied.  
\end{proof}

\begin{lemma}[Additional properties of graphs with $H^{-1}$ in Figures \ref{fig:3top_vin_uonly_2} - \ref{fig:deg1_vout_uvonly}]
    \label{lem:k_f-preserving_fbm}
    Let $G$ be a graph with subgraphs $G_1$, $G_2$, and $S$ such that $G = G_1 \cup G_2$, $G_1$ is one of the graphs in Figures \ref{fig:3top_vin_uonly_2} - \ref{fig:deg1_vout_uvonly}, $S = G_1 \cap G_2$ is the subgraph of $G_1$ induced by some subset of $\{w_i\}$ that contains some nonedge $e$, and $G \setminus S$ is disconnected.  
    Also, let $e = w_1w_3$ if $G_1$ is the graph in Figure \ref{fig:cycle_vin_uonly_2}.  
    If $G_2 \cup e$ has an $e$-preserving forbidden minor, then the red edge of $G_1$ is preserved in some forbidden minor of $G$.  
\end{lemma}

\begin{proof}
    Assume that $G_2 \cup e$ has an $e$-preserving forbidden minor.  
    We will show that $G$ has an $f$-preserving forbidden minor, where $f$ is the red edge of $G_1$.  
    If there exists a minor $[G]$ in which $G_2$ is preserved and $[f] = [e]$, then this follows from the fact that $G_2 \cup e$ has an $e$-preserving forbidden minor.  
    Using this fact, it is easy to see that the lemma is true if either (i) $G_1$ is one of the graphs in Figures \ref{fig:3top_vin_uplus_2}, \ref{fig:deg1_vin_uplus_1}, \ref{fig:deg1_vin_uplus_3}, \ref{fig:deg1_vin_uplus_4}, or \ref{fig:deg1_vout_uvonly}, (ii) $e = w_1w_3$ and $G_1$ is one of the graphs in Figures \ref{fig:cycle_vin_uplus_3} or \ref{fig:cycle_vin_uplus_2}, or (iii) $e = w_1w_4$ and $G_1$ is one of the graphs in Figures \ref{fig:deg1_vin_uonly_2} or \ref{fig:deg1_vin_uplus_5}.  
    
    Next, consider the case where $G_1$ is the graph in Figure \ref{fig:3top_vin_uonly_2}.  
    Note that either $w_1w_2$ or $w_2w_3$ is preserved in any $e$-preserving minor of $G_2$.  
    Also, observe that there exists a minor $[G]$ in which $G_2$ is preserved, $[e]$ is an edge, and $[f]$ is either $[w_1w_2]$ or $[w_2w_3]$.  
    Combining these facts shows that $G$ has an $f$-preserving forbidden minor.  
    A similar argument applies if either (i) $e = w_1w_3$ and $G_1$ is one of the graphs in Figures \ref{fig:chord_vin_uonly_2}, \ref{fig:cycle_vin_uonly_2}, \ref{fig:deg1_vin_uonly_2}, \ref{fig:deg1_vin_uplus_5}, or \ref{fig:cycle_vout_uvonly} or (ii) $e = w_2w_4$ and $G_1$ is one of the graphs in Figures \ref{fig:cycle_vin_uplus_3}, \ref{fig:cycle_vin_uplus_2}, or \ref{fig:cycle_vout_uvonly}.  

    The above cases are exhaustive, and so the proof is complete.  
\end{proof}

\begin{proof}[Proof of Lemma \ref{lem:G'_f'_no_reducible_edge}]
    For any nonedge $f'$ of $C^{-1}$, consider the statements (a) some atom of $G' \cup f'$ contains $f'$ and has an $f'$-preserving forbidden minor and (b) the pair $(G',f')$ has no reducing edge.  
    By Lemma \ref{lem:f'_and_J}, (a) is true for some choice of $f'$, and for any choice of $f'$ if $C^{-1}$ is the graph in Figure \ref{fig:exp_clique_sep_k3_cycle}.  
    Assume that none of these choices satisfy (b).  
    We will show that $(G,f)$ has a reducing edge, which, when combined with the assumption that $(G,f)$ is minimal, proves the lemma.  
    
    Since any choice for $f'$ that satisfies (a) does not satisfy (b), the pair $(G',f')$ has some reducing edge $h$.  
    Consider the minor $[G \cup f]$ obtained by contracting $h$.  
    By definition, $f'$ is retained in $[G']$ and some induced subgraph $K$ of $[G']$ is such that $K \cup [f']$ is an atom that contains $[f']$ and has an $[f']$-preserving forbidden minor.  
    Let $S = [C^{-1}]$, $e = [f']$, $G_1 = [H^{-1}]$, and $G_2 = K$.  
    By Lemma \ref{lem:C'_a-n}, $H^{-1}$ is one of the graphs in Figures \ref{fig:3top_vin_uonly_2} - \ref{fig:deg1_vout_uvonly} and $C^{-1}$ is the subgraph induced by the vertex set $\{w_i\}$.  
    There are two cases.  

    \smallskip
    \noindent\textbf{Case 1:} $G_1$ is preserved in $[G \cup f]$.  
    \smallskip

    We will show that $h$ is reducing with respect to $(G,f)$.  
    Observe that some endpoint of $h$ is contained in $G' \setminus C^{-1}$.  
    Hence, $f$ is retained in $[G \cup f]$, and so $h$ is not of Type (3).  
    To see that $h$ is not of Type (1) or (2), we show $G_1 \cup G_2$ is an atom that contains $[f]$ and has an $[f]$-preserving forbidden minor.  
    By construction, $G_1 \cup G_2$ contains $[f]$.  
    If $C^{-1}$ is the graph in Figure \ref{fig:exp_clique_sep_k3_cycle}, then Lemma \ref{lem:f'_and_J} allows us to choose any of its nonedges to be $f'$.  
    Hence, Lemma \ref{lem:k_f-preserving_fbm} states that $G_1 \cup G_2$ has an $[f]$-preserving forbidden minor.  
    
    We complete the proof by showing that $G_1 \cup G_2$ is an atom.  
    If $G_1$ is not the graph in Figure \ref{fig:3top_vin_uonly_2}, then it is easy to verify that $G_1$ is an atom.  
    Since $G_2 \cup e$ is an atom, Lemma \ref{lem:gluing_min_k-clique-sum_graphs} shows that $G_1 \cup G_2$ is an atom.  
    Hence, assume that $G_1$ is the graph in Figure \ref{fig:3top_vin_uonly_2}.  
    We will demonstrate that no clique subgraph of $S$ separates the endpoints of $e$ in $G_2$ and then apply Lemma \ref{lem:J+H'_atom} to show that $G_1 \cup G_2$ is an atom.  
    The only clique subgraph to check is the one induced by $\{[w_2]\}$.  
    Since $G_2 \cup e$ has an $e$-preserving forbidden minor and $S$ contains exactly three vertices, there exists some vertex $x$ in $G_2 \setminus S$.  
    Let $e = [uv]$.  
    If $\{[w_2]\}$ separates $[u]$ and $[v]$ in $G_2$, then either every path in $G_2$ between $x$ and $[u]$ contains some vertex in $\{[v],[w_2]\}$ or every path in $G_2$ between $x$ and $[v]$ contains some vertex in $\{[u],[w_2]\}$.  
    This implies that either $\{[u],[w_2]\}$ or $\{[v],[w_2]\}$ is a clique separator of $G_2 \cup e$.  
    Therefore, since $G_2 \cup e$ is an atom, this implies that $\{[w_2]\}$ does not separate $[u]$ and $[v]$ in $G_2$.  

    \smallskip
    \noindent\textbf{Case 2:} $G_1$ is not preserved in $[G \cup f]$.  
    \smallskip

    Since $h$ is contained in $G'$, it must be contained in $S$.  
    Combining this with the fact that $e$ is retained in $[G \cup f]$ shows that $H^{-1}$ is one of the graphs in Figures \ref{fig:chord_vin_uonly_2}, \ref{fig:deg1_vin_uplus_1}, \ref{fig:deg1_vin_uonly_2}, \ref{fig:deg1_vin_uplus_5}, \ref{fig:deg1_vin_uplus_3}, \ref{fig:deg1_vin_uplus_4}, or \ref{fig:deg1_vout_uvonly}.  
    There are $9$ subcases.  

    \smallskip
    \noindent\textbf{Subcase 1:} $H^{-1}$ is the graph in Figure \ref{fig:chord_vin_uonly_2}.  
    \smallskip
    
    Note that $h = w_2w_4$ and $G_1$ is the graph in Figure \ref{fig:3top_vin_uonly_2}.  
    We will show that $h$ is reducing with respect to $(G,f)$.  
    Observe that $f$ is retained in $[G \cup f]$, and so $h$ is not of Type (3).  
    To see that $h$ is not of Type (1) or (2), we show $G_1 \cup G_2$ is an atom that contains $[f]$ and has an $[f]$-preserving forbidden minor.  
    By construction, $G_1 \cup G_2$ contains $[f]$.  
    Also, Lemma \ref{lem:k_f-preserving_fbm} shows that $G_1 \cup G_2$ has an $[f]$-preserving forbidden minor.  
    We will demonstrate that no clique subgraph of $S$ separates the endpoints of $e$ in $G_2$ and then apply Lemma \ref{lem:J+H'_atom} to show that $G_1 \cup G_2$ is an atom.  
    The only clique subgraph to check is the one induced by $\{[w_2]\}$.  
    Since $f'$ satisfies (a), there exists an atom $J$ of $G' \cup f'$ that contains $f'$ and has an $f'$-preserving forbidden minor.  
    Lemma \ref{lem:f'_and_J} shows that no clique subgraph of $C^{-1}$ separates the endpoints of $f'$ in $J \setminus f'$.  
    If $\{[w_2]\}$ separates the endpoints of $e$ in $G_2$, then $\{w_2,w_4\}$ separates the endpoints of $f'$ in $J \setminus f'$.  
    Combining these statements implies that $\{[w_2]\}$ does not separate the endpoints of $e$ in $G_2$.  

    \smallskip
    \noindent\textbf{Subcase 2:} $H^{-1}$ is one of the graphs in Figures \ref{fig:deg1_vin_uplus_1}, \ref{fig:deg1_vin_uplus_3}, \ref{fig:deg1_vin_uplus_4}, or \ref{fig:deg1_vout_uvonly}.  
    \smallskip

    If either $H^{-1}$ is not the graph in Figure \ref{fig:deg1_vin_uplus_1} or $h \neq w_2w_4$, then we show that $h$ is reducing with respect to $(G,f)$.  
    % Note that there is exactly one possibility for $f'$ for each case of $h$ since $f'$ is retained in $[G \cup f]$.  
    Observe that $f$ is retained in $[G \cup f]$, and so $h$ is not of Type (3).  
    To see that $h$ is not of Type (1) or (2), we show $G_1 \cup G_2$ is an atom that contains $[f]$ and has an $[f]$-preserving forbidden minor.  
    By construction, $G_1 \cup G_2$ contains $[f]$.  
    Also, note that $G_1$ is an atom.  
    Hence, since $G_2 \cup e$ is an atom, Lemma \ref{lem:gluing_min_k-clique-sum_graphs} shows that $G_1 \cup G_2$ is an atom.  
    Lastly, since $G_2 \cup e$ has an $e$-preserving forbidden minor, $G_1 \cup G_2$ clearly has an $[f]$-preserving forbidden minor.  

    Next, we show that if $H^{-1}$ is the graph in Figure \ref{fig:deg1_vin_uplus_1}, then $h \neq w_2w_4$.  
    If this is not the case, then $f' = w_1w_3$ since it is retained in $[G \cup f]$.  
    Observe that $[G \cup f]$ is $f$-retaining and has an $[f]$-preserving forbidden minor, since it has an $e$-preserving forbidden minor.  
    Let $x$ be the vertex in $H^{-1} \setminus C^{-1}$ that is not an endpoint of $f$.  
    Note that the subgraph $E$ of $[G \cup f]$ induced by $\{[x],[w_1],[w_2]\}$ is a CMS of $[G \cup f]$ such that the $E$-component of $[G \cup f]$ containing $[f]$ is a proper subgraph of $G_1$ that has no $[f]$-preserving forbidden minor.  
    These facts imply that $E^{-1}$ is an expanded $f$-separating CMS whose expanded $f$-component is a proper subgraph of $H^{-1}$.  
    However, this contradicts the fact that $C^{-1}$ is in the top-level of $(G,f)$.  
    Therefore, $h$ is not $w_2w_4$.  

    \smallskip
    \noindent\textbf{Subcase 3:} $H^{-1}$ is one of the graphs in Figures \ref{fig:deg1_vin_uonly_2} or \ref{fig:deg1_vin_uplus_5}, $uw_1$ is an edge, and $h$ is $w_2w_3$ or $w_2w_4$.  
    \smallskip

    Arguments similar to those in Subcases 1 and 2 show that $h$ is reducing with respect to $(G,f)$.  

    \smallskip
    \noindent\textbf{Subcase 4:} $H^{-1}$ is one of the graphs in Figures \ref{fig:deg1_vin_uonly_2} or \ref{fig:deg1_vin_uplus_5}, $uw_1$ is an edge, $h$ is $w_3w_4$, and $f' = w_1w_3$.  
    \smallskip
    
    Let $u$ be the endpoint of $f$ not contained in $C^{-1}$.  
    We will show that $uw_1$ is reducing with respect to $(G,f)$.  
    Let $[G \cup f]_1$ be the minor obtained by contracting $uw_1$.  
    Observe that $f$ is retained in $[G \cup f]_1$, and so $uw_1$ is not of Type (3).  
    We show that $uw_1$ is not of Type (1) or (2) by identifying some atom of $[G \cup f]_1$ that contains $[f]_1$ and has an $[f]_1$-preserving forbidden minor.  
    It suffices to show that some atom of $G' \cup \{w_1w_3,w_1w_4\}$ contains $w_1w_4$ and has a $w_1w_4$-preserving forbidden minor.  
    
    Recall the induced subgraph $K$ of $[G']$ such that $K \cup [f']$ is an atom that contains $[f']$ and has an $[f']$-preserving forbidden minor.  
    Note that the induced subgraph $K^{-1} \cup f'$ of $G' \cup f'$ contains $w_1w_4$ and has a $w_1w_4$-preserving forbidden minor.  
    If $K^{-1} \cup f'$ is an atom, then it is contained in some atom $X$ of $G' \cup f'$ since it is an induced subgraph.  
    Observe that each CMS of $X \cup w_1w_4$ contains $w_1w_4$.  
    Hence, using Lemma \ref{cor:minimal_k-clique-sum_component_containing_fbm_and_f}, we see that some atom of $X \cup w_1w_4$, and therefore some atom of $G' \cup \{w_1w_3,w_1w_4\}$, contains $w_1w_4$ and has a $w_1w_4$-preserving forbidden minor, as desired.  
    
    Otherwise, if $K^{-1} \cup f'$ is not an atom, then it contains exactly one CMS $E$, since $K \cup [f']$ is an atom.  
    Furthermore, $E$ contains exactly one vertex in $\{w_3,w_4\}$ and there are exactly two $E$-components of $K^{-1} \cup f'$, both of which are atoms and one of which is $E \cup \{w_3,w_4\}$.  
    These facts imply that the atom graph of $K^{-1} \cup f'$ is a path of length two.  
    Let $x$ be the vertex in $\{w_3,w_4\}$ that is not contained in $E$.  
    The neighborhood of $x$ in $K^{-1} \cup f'$ is $V(E)$, which is a clique.  
    Since $K^{-1} \cup f'$ contains the vertices $w_1$, $w_3$, and $w_4$, the edges $w_1w_3$ and $w_3w_4$, and the nonedge $w_1w_4$, this implies that $x = w_4$ and $E$ does not contain $w_1$.  
    Therefore, we can use Lemma \ref{lem:I'_f'_cms} to see that $K^{-1} \cup \{f',w_1w_4\}$ is an atom.  
    Consequently, it is contained in some atom of $G' \cup \{w_1w_3,w_1w_4\}$ that contains $w_1w_4$ and has a $w_1w_4$-preserving forbidden minor, as desired.  

    \smallskip
    \noindent\textbf{Subcase 5:} $H^{-1}$ is one of the graphs in Figures \ref{fig:deg1_vin_uonly_2} or \ref{fig:deg1_vin_uplus_5}, $uw_1$ is an edge, $h$ is $w_3w_4$, and $f' = w_1w_4$.  
    \smallskip
    
    This case is similar to Subcase 4.  
    We will show that $uw_1$ is reducing with respect to $(G,f)$.  
    Let $[G \cup f]_1$ be the minor obtained by contracting $uw_1$.  
    Observe that $f$ is retained in $[G \cup f]_1$, and so $uw_1$ is not of Type (3).  
    We show that $uw_1$ is not of Type (1) or (2) by identifying some atom of $[G \cup f]_1$ that contains $[f]_1$ and has an $[f]_1$-preserving forbidden minor.  
    It suffices to show that some atom of $G' \cup \{w_1w_3,w_1w_4\}$ contains $w_1w_4$ and has a $w_1w_4$-preserving forbidden minor.  
    
    If $K^{-1} \cup f'$ is an atom, then it is contained in some atom $X$ of $G' \cup f'$, since it is an induced subgraph of $G' \cup f'$.  
    Note that $X$ has a forbidden minor in which $w_3w_4$ is contracted and both $w_1w_3$ and $w_1w_4$ are preserved.  
    Hence, if $X \cup w_1w_3$ is an atom, then it is contained in some atom of $G' \cup \{w_1w_3,w_1w_4\}$ that has the desired properties.  
    Otherwise, each CMS of $X \cup w_1w_3$ contains $w_1w_3$.  
    Hence, Lemma \ref{cor:minimal_k-clique-sum_component_containing_fbm_and_f} shows that some atom $Y$ of $X \cup w_1w_3$ contains $w_1w_3$ and has a forbidden minor in which $w_3w_4$ is contracted and both $w_1w_3$ and $w_1w_4$ are preserved.  
    If $Y$ contains $w_4$, then the atom of $G' \cup \{w_1w_3,w_1w_4\}$ that contains $Y$ has the desired properties.  
    Otherwise, $Y$ is an atom of $(X \cup w_1w_3) \setminus w_1w_4$.  
    Therefore, $w_1w_3$ satisfies (a), and so we can instead choose $f' = w_1w_3$ and apply Subcase 4.  
    
    Next, if $K^{-1} \cup f'$ is not an atom, then an argument similar to the one in Subcase 4 shows that it contains exactly one CMS $E$ that contains $w_4$ but not $w_3$ and such that there are exactly two $E$-components of $K^{-1} \cup f'$, both of which are atoms and one of which is $E \cup w_3$.  
    These facts imply that the atom graph of $K^{-1} \cup f'$ is a path of length two.  
    Since the neighborhood of $w_3$ in $K^{-1} \cup f'$ is $E$, which is a clique, and $w_1w_3$ is a nonedge, $E$ does not contain $w_1$.  
    Therefore, we can use Lemma \ref{lem:I'_f'_cms} to see that $K^{-1} \cup \{f',w_1w_3\}$ is an atom, and so it is contained in some atom of $G' \cup \{w_1w_3,w_1w_4\}$ that has the desired properties.  

    \smallskip
    \noindent\textbf{Subcase 6:} $H^{-1}$ is the graph in Figure \ref{fig:deg1_vin_uplus_5}, $uw_1$ is a nonedge, and $h = w_2w_3$.  
    \smallskip

    An similar argument to that in Subcase 2 applies.  

    \smallskip
    \noindent\textbf{Subcase 7:} $H^{-1}$ is the graph in Figure \ref{fig:deg1_vin_uplus_5}, $uw_1$ is a nonedge, and $h = w_2w_4$.  
    \smallskip
    
    We will show that $h$ is reducing with respect to $(G,f)$.  
    Observe that $f$ is retained in $[G \cup f]$, and so $h$ is not of Type (3).  
    To see that $h$ is not of Type (1) or (2), we show $G_1 \cup G_2$ is an atom that contains $[f]$ and has an $[f]$-preserving forbidden minor.  
    By construction, $G_1 \cup G_2$ contains $[f]$.  
    Also, since $G_2 \cup e$ has an $e$-preserving forbidden minor, it is easy to see that $G_1 \cup G_2$ has an $[f]$-preserving forbidden minor.  
    We will show that $G_1 \cup G_2$ is an atom by applying Lemma \ref{lem:gluing_top_part_to_atom}, after we show that Statements (1) - (4) in this lemma are true.  
    Statement (1) follows from our assumption that $G_2 \cup e$ is an atom.  
    Statement (2) follows from the observation that $G_1 \cup e$ is an atom.  
    For Statement (3) and (4), consider any clique subgraph $A$ of $S$.  
    It is easy to see that $G_1 \setminus A$ is connected, and so Statement (3) is true.  
    Lastly, to show that Statement (4) is true, we must argue that $\{[w_2]\}$ does not separate the endpoints of $e$ in $G_2$.  
    Since $f'$ satisfies (a), there exists an atom $J$ of $G' \cup f'$ that contains $f'$ and has an $f'$-preserving forbidden minor.  
    Lemma \ref{lem:f'_and_J} shows that no clique subgraph of $C^{-1}$ separates the endpoints of $f'$ in $J \setminus f'$.  
    If $\{[w_2]\}$ separates the endpoints of $e$ in $G_2$, then $\{w_2,w_4\}$ separates the endpoints of $f'$ in $J \setminus f'$.  
    Combining these statements implies that $\{[w_2]\}$ does not separate the endpoints of $e$ in $G_2$.  
    Therefore, applying Lemma \ref{lem:gluing} completes the proof.  

    \smallskip
    \noindent\textbf{Subcase 8:} $H^{-1}$ is the graph in Figure \ref{fig:deg1_vin_uplus_5}, $uw_1$ is a nonedge, $h = w_3w_4$, and $f' = w_1w_4$.  
    \smallskip

    Let $x$ be any vertex of $H^{-1} \setminus C^{-1}$ other that an endpoint of $f$.  
    We will show that $xw_1$ is reducing with respect to $(G,f)$.  
    Observe that the minor $[G \cup f]_1$ obtained by contracting $xw_1$ is $f$-retaining, and so $xw_1$ is not of Type (3).  
    We show that $xw_1$ is not of Type (1) or (2) by identifying some atom of $[G \cup f]_1$ that contains $[f]_1$ and has an $[f]_1$-preserving forbidden minor.  

    Since $f'$ satisfies (a), some atom $J$ of $G' \cup f'$ contains $f'$ and has an $f'$-preserving forbidden minor.  
    We will show that either $J$ contains $w_3$ or $w_2w_3$ is reducing with respect to $(G',f')$.  
    In the latter case, Subcase 6 applies.  
    If $J$ does not contain $w_3$, then note that the minor $[G \cup f']_2$ obtained by contracting $w_2w_3$ is $f'$-retaining, so it $w_2w_3$ is not of Type (3).  
    Also, some CMS $E$ contained in $J$ is such that all paths in $G' \cup f'$ between $w_3$ and any vertex of $J \setminus E$ contain some vertex of $E$.  
    Hence, $J$ is weakly retained in $[G' \cup f']_2$.  
    Therefore, some atom of $[G' \cup f']_2$ contains $[J]_2$, and consequently contains $[f']_2$ and has an $[f']_2$-preserving forbidden minor.  
    This shows that $w_2w_3$ is not of Type (1) or (2), and so it is reducing with respect to $(G',f')$.  

    Next, since $J$ contains $w_3$, $J$ is some induced subgraph $G_2$ of $[G']_1$ with $f'$ as $[f']_1$.  
    % xxxsome isomorphism between $J$ and some induced subgraph $G_2$ of $[G']_1$ maps $f'$ to $[f']_1$.  
    Let $G_1 = [C^{-1}]_1 \cup [u]_1$, $S = G_1 \cap G_2$, and $e = [w_1w_3]_1$.  
    Then, $G_1 \cup G_2$ contains $[f]_1$.  
    Also, since $G_2$ has an $[f']_1$-preserving forbidden minor, it is easy to see that $G_1 \cup G_2$ has an $[f]_1$-preserving forbidden minor.  
    Furthermore, since $G_1$ and $G_2$ are atoms and $e$ is a nonedge, Lemma \ref{lem:gluing_min_k-clique-sum_graphs} shows that $G_1 \cup G_2$ is an atom.  
    Furthermore, $G_1 \cup G_2$ is contained in some atom of $[G \cup f]_1$ since it is an induced subgraph of $[G \cup f]_1$, and this atom has the desired properties.  

    \smallskip
    \noindent\textbf{Subcase 9:} $H^{-1}$ is the graph in Figure \ref{fig:deg1_vin_uplus_5}, $uw_1$ is a nonedge, $h = w_3w_4$, and $f' = w_1w_3$.  
    \smallskip

    First, we show that some atom of $G' \cup \{w_1w_3,w_1w_4\}$ contains $w_1$, $w_3$, and $w_4$ and has a $w_1w_4$-preserving forbidden minor.  
    Recall the induced subgraph $K$ of $[G']$ such that $K \cup [f']$ is an atom that contains $[f']$ and has an $[f']$-preserving forbidden minor.  
    Note that the induced subgraph $K^{-1} \cup f'$ of $G' \cup f'$ contains $w_1$, $w_3$, and $w_4$ and has a $w_1w_4$-preserving forbidden minor.  
    If $K^{-1} \cup f'$ is an atom, then it is contained in some atom $X$ of $G' \cup f'$ since it is an induced subgraph.  
    Observe that each CMS of $X \cup w_1w_4$ contains $w_1w_4$.  
    Hence, Lemma \ref{cor:minimal_k-clique-sum_component_containing_fbm_and_f} shows that some atom $Y$ of $X \cup w_1w_4$ contains $w_1w_4$ and has a $w_1w_4$-preserving forbidden minor.  
    If $Y$ contains $w_3$, then we are done.  
    Otherwise, some atom of $G' \cup w_1w_4$ contains $Y$.  
    Therefore, $w_1w_4$ satisfies (a), and so we can instead choose $f' = w_1w_4$ and apply Subcase 8.  
    
    Next, if $K^{-1} \cup f'$ is not an atom, then an argument similar to that in Subcase 4 shows that $K^{-1} \cup f'$ contains exactly one CMS $E$ that contains $w_3$ but not $w_4$ and such that there are exactly two $E$-components of $K^{-1} \cup f'$, both of which are atoms and one of which is $E \cup w_4$.  
    These facts imply that the atom graph of $K^{-1} \cup f'$ is a path of length two.  
    Since the neighborhood of $w_4$ in $K^{-1} \cup f'$ is $E$, which is a clique, and $w_1w_4$ is a nonedge, $E$ does not contain $w_1$.  
    Therefore, we can use Lemma \ref{lem:I'_f'_cms} to see that $K^{-1} \cup \{f',w_1w_4\}$ is an atom, and so it is contained in some atom of $G' \cup \{w_1w_3,w_1w_4\}$ that has the desired properties.  

    Finally, by the above arguments, some induced subgraph $I$ of $[G']_1$ is such that $I \cup [w_1w_3]_1$ is an atom that contains $[w_1]_1$, $[w_3]_1$, and $[w_4]_1$ and has a $[w_1w_4]_1$-preserving forbidden minor.  
    Let $G_1 = [C^{-1}]_1 \cup [u]_1$, $G_2 = I$, $S = G_1 \cap G_2$, and $e = [w_1w_3]_1$.  
    Similar to the argument in Subcase 8, we can use Lemma \ref{lem:gluing_min_k-clique-sum_graphs} to show that some atom of $[G \cup f]_1$ that contains $[f]_1$ and has an $[f]_1$-preserving forbidden minor.  

    \smallskip
    The above cases are exhaustive, and so the lemma is proved.  
\end{proof}

\subsection{Proof of Lemmas for Proposition \ref{prop:final_top_graphs}}
\label{sec:proof_final_top_graphs}

We present proofs of Lemmas \ref{lem:C'_a-i} and \ref{lem:C'_degree_1_reducible_3} in Section \ref{sec:prop-8-9}.  
Lemma \ref{lem:C'_a-i} is proved using Lemma \ref{lem:J_contains_C'_degree_1}, also below.  

\begin{lemma}[Vertices in atom containing $f'$]
    \label{lem:J_contains_C'_degree_1}
    Consider a minimal graph-nonedge pair $(G,f)$ with a top-level expanded $f$-separating CMS $C^{-1}$ whose expanded $f$-component is $H^{-1}$, and let $G' = G \setminus (H^{-1} \setminus C^{-1})$.  
    For any nonedge $f'$ of $C^{-1}$, if some atom $J$ of $G' \cup f'$ contains $f'$ and has an $f'$-preserving forbidden minor and the pair $(G',f')$ has no reducing edge, then
    \begin{enumerate}[(i)]
        \item $J$ contains $w_1$, $w_3$, and $w_4$ if $C^{-1}$ is the graph in Figure \ref{fig:exp_clique_sep_k3_cycle} and
        \item $J$ contains $C^{-1}$ if $C^{-1}$ is the graph in Figure \ref{fig:exp_clique_sep_k3_degree_1}.  
    \end{enumerate}
\end{lemma}

\begin{proof}
    Assume wlog that $f' = w_1w_3$.  
    Then, $J$ contains $w_1$ and $w_3$.  
    Assume $C^{-1}$ is the graph in Figure \ref{fig:exp_clique_sep_k3_cycle}.  
    If $J$ does not contain $w_4$, then Lemma \ref{lem:G'_is_J+decorations} states that the neighborhood of $w_4$ in $G'$ is $\{w_1,w_3\}$.  
    However, Lemma \ref{lem:expanded_clique_graphs_and_connections} Statement (2) can be used to show that $w_4$ neighbors some vertex of $G'$ not in $\{w_1,w_3\}$.  
    This is a contradiction, and thus $J$ contains $w_4$.  
    A similar argument yields the same contradiction if $C^{-1}$ is the graph in Figure \ref{fig:exp_clique_sep_k3_degree_1}.  
    %
    % Wlog, assume that $f' = w_1w_3$.  
    % By assumption, $J$ contains $w_1$ and $w_3$.  
    % If $J$ does not contain some vertex in $\{w_2,w_4\}$, then Lemma \ref{lem:G'_is_J+decorations} shows that the neighborhood of this vertex in $G'$ is $\{w_1,w_3\}$.  
    % However, observe that $w_2w_4$ is an edge of $G'$.  
    % Thus, $J$ contains $C^{-1}$.  
\end{proof}

\begin{proof}[Proof of Lemma \ref{lem:C'_a-i}]
    By Lemma \ref{lem:C'_a-n}, the expanded $f$-component $H^{-1}$ of $C^{-1}$ is one of the graphs in Figures \ref{fig:3top_vin_uonly_2} - \ref{fig:deg1_vout_uvonly}, where $C^{-1}$ is the subgraph induced by $\{w_i\}$.  
    Furthermore, if $H^{-1}$ is one of the graphs in Figures \ref{fig:cycle_vin_uonly_2}, \ref{fig:cycle_vin_uplus_3}, \ref{fig:3top_vin_uplus_2}, \ref{fig:cycle_vin_uplus_2}, \ref{fig:deg1_vin_uplus_3}, or \ref{fig:deg1_vin_uplus_4}, then at least one dashed line-segment is a nonedge.  
    For each case of $H^{-1}$ in Figures \ref{fig:cycle_vin_uplus_2} - \ref{fig:deg1_vout_uvonly}, we will show that $(G,f)$ has a reducing edge, which contradicts the minimality of this pair.  

    \smallskip
    \noindent\textbf{Case 1:} $H^{-1}$ is one of the graphs Figures \ref{fig:cycle_vin_uplus_2} or \ref{fig:cycle_vout_uvonly}.  
    \smallskip

    We prove the lemma in the case where $H^{-1}$ is the graph Figure \ref{fig:cycle_vin_uplus_2}.  
    The proof in the case where $H^{-1}$ is the graph Figure \ref{fig:cycle_vout_uvonly} is almost identical, so we omit it.  
    Let $G' = G \setminus (H^{-1} \setminus C^{-1})$.  
    By Lemma \ref{lem:G'_f'_no_reducible_edge}, some atom $J$ of $G' \cup w_1w_3$ contains $w_1w_3$ and has a $w_1w_3$-preserving forbidden minor and the pair $(G',w_1w_3)$ has no reducing edge.  
    Also, Lemma \ref{lem:J_contains_C'_degree_1} shows that $J$ contains $w_4$.  
    If $J$ contains $w_2$, then let $x$ be any vertex of $H^{-1} \setminus C^{-1}$ that neighbors $w_3$.  
    We will show that $xw_3$ is reducing with respect to $(G,f)$.  
    Note that $f$ is retained in the minor $[G \cup f]$ obtained by contracting $xw_3$, and so this edge is not of Type (3).  
    To see that $xw_3$ is not of Type (1) or (2), observe that $J$ is some induced subgraph $G_2$ of $[G']$ with $w_1w_3$ as $[w_1x]$.  
    % xxxisomorphism between $J$ and some induced subgraph $G_2$ of $[G']$ maps $w_1w_3$ to $[w_1x]$.  
    Hence, $G_2$ is an atom that contains $S = [C^{-1}]$ and has a $[w_1x]$-preserving forbidden minor.  
    Let $G_1$ be the $\{[u],[w_1],[w_3]\}$-component of $[H^{-1}]$ that contains $S$.  
    Then, $G_1$ is an atom and $S = G_1 \cap G_2$.  
    Since $S$ contains the nonedge $[w_2w_4]$, Lemma \ref{lem:gluing_min_k-clique-sum_graphs} shows that $G_1 \cup G_2$ is an atom.  
    Lastly, $G_1 \cup G_2$ clearly has an $[f]$-preserving forbidden minor.  
    Therefore, the atom of $[G \cup f]$ that contains $G_1 \cup G_2$ contains $[f]$ and has an $[f]$-preserving forbidden minor.  

    Next, if $J$ does not contain $w_2$, then let $x$ be the vertex of $H^{-1} \setminus C^{-1}$ that neighbors $w_2$.  
    We will show that $xw_2$ is reducing with respect to $(G,f)$.  
    Note that $f$ is retained in the minor $[G \cup f]$ obtained by contracting $xw_2$, and so this edge is not of Type (3).  
    To see that $xw_2$ is not of Type (1) or (2), we first show that the graph $K = ((J \cup w_2) \setminus w_1w_3) \cup w_2w_4$ is an atom that has a $w_1w_3$-preserving forbidden minor.  
    The latter claim follows easily by considering the minor of $K$ obtained by contracting $w_1w_2$.  
    The former claim will follow by applying Lemma \ref{lem:gluing_top_part_to_atom}.  
    Let $G_1$ and $G_2$ be the subgraphs of $K$ induced by $\{w_1,w_2,w_3,w_4\}$ and $V(J)$, respectively.  
    Observe that $K = G_1 \cup G_2$ and the $S = G_1 \cap G_2$ is the subgraph of $K$ induced by $\{w_1,w_3,w_4\}$ that  contains the nonedge $w_1w_3$.  
    Furthermore, both $G_1 \cup w_1w_3$ and $G_2 \cup w_1w_3$ are atoms and $G_1 \setminus A$ is connected for any clique subgraph of $S$.  
    Lastly, we must show that $w_4$ does not separate $w_1$ and $w_3$ in $G_2$.  
    Since $G_2 \cup w_1w_3$ has a forbidden minor and $S$ contains exactly three vertices, there exists some vertex in $G_2 \setminus S$.  
    Additionally, since $G_2 \cup w_1w_3$ is an atom, there exist paths that connect this vertex to $w_1$ and $w_3$ that each contain exactly one vertex in $S$.  
    Therefore, $w_4$ does not separate $w_1$ and $w_3$ in $G_2$, and so Lemma \ref{lem:gluing_top_part_to_atom} shows that $K$ is an atom.  

    Finally, observe that $K$ is some induced subgraph $G_1$ of $[G']$ with $w_2w_4$ as $[xw_4]$.  
    % xxxsome isomorphism between $K$ and some induced subgraph $G_1$ of $[G']$ maps $w_2w_4$ to $[xw_4]$.  
    Let $G_2 = [H^{-1}]$ and observe that $S = G_1 \cap G_2 = [C^{-1}]$ contains the nonedge $[w_1w_3]$.  
    Since $G_1$ and $G_2$ are atoms, Lemma \ref{lem:gluing_min_k-clique-sum_graphs} shows that $G_1 \cup G_2$ is an atom.  
    Observe that $G_1 \cup G_2$ has an $[f]$-preserving forbidden minor.  
    These facts show that the atom of $[G \cup f]$ that contains $G_1 \cup G_2$ contains $[f]$ and has an $[f]$-preserving forbidden minor.  
    Therefore, $xw_2$ is not of Type (1) or (2), as desired.  

    \smallskip
    \noindent\textbf{Case 2:} $H^{-1}$ is one of the graphs in Figures \ref{fig:deg1_vin_uplus_3}, \ref{fig:deg1_vin_uplus_4}, or \ref{fig:deg1_vout_uvonly}.  
    \smallskip

    By Lemma \ref{lem:G'_f'_no_reducible_edge}, there exists a nonedge $f'$ of $C^{-1}$ such that some atom $J$ of $G' \cup f'$ contains $f'$ and has an $f'$-preserving forbidden minor and the pair $(G',f')$ has no reducing edge.  
    Furthermore, Lemma \ref{lem:J_contains_C'_degree_1} shows that $J$ contains $C^{-1}$.  
    The remainder of the argument proceeds similarly to Case 1, so we omit it.  
    
    \smallskip
    The above cases are exhaustive, and so the proof is complete.  
\end{proof}

Lemma \ref{lem:C'_degree_1_reducible_3} follows immediately from Lemmas \ref{lem:C'_degree_1_reducible_0}-\ref{lem:C'_degree_1_reducible_2}, below.  

\begin{lemma}[$H^{-1}$ is not graph in Figure \ref{fig:deg1_vin_uplus_1}]
    \label{lem:C'_degree_1_reducible_0}
    The graph in Figure \ref{fig:deg1_vin_uplus_1} is not the expanded $f$-component of any top-level expanded $f$-separating CMS of a minimal graph-nonedge pair.  
\end{lemma}

\begin{proof}
    Assume to the contrary that the graph in Figure \ref{fig:deg1_vin_uplus_1} is the expanded $f$-component $H^{-1}$ of some top-level expanded $f$-separating CMS $C^{-1}$ of some minimal graph-nonedge pair $(G,f)$.  
    We will show that $uw_4$ is reducing with respect to $(G,f)$, which contradicts the fact that $(G,f)$ is minimal.  
    Observe that $f$ is retained in the minor $[G \cup f]$ obtained by contracting $uw_4$, and so $uw_4$ is not of Type (3).  
    We show that $uw_4$ is not of Type (1) or (2) by identifying some atom of $[G \cup f]$ contains $[f]$ and has an $[f]$-preserving forbidden minor.  
    Let $G' = G \setminus (H^{-1} \setminus C^{-1})$.  
    By Lemma \ref{lem:G'_f'_no_reducible_edge}, there exists a nonedge $f'$ of $C^{-1}$ such that some atom $J$ of $G' \cup f'$ contains $f'$ and has an $f'$-preserving forbidden minor and the pair $(G',f')$ has no reducing edge.  
    If $f' = w_1w_4$, then some atom of $[G \cup f]$ clearly has the desired properties.  
    
    In the case where $f' = w_1w_3$, we first argue that some atom of $G' \cup \{w_1w_3,w_1w_4\}$ contains $w_1w_4$ and has a $w_1w_4$-preserving forbidden minor.  
    By Lemma \ref{lem:J_contains_C'_degree_1}, $J$ contains $w_4$.  
    Hence, since $J$ is an atom, each CMS and atom of $J \cup w_1w_4$ contains $w_1w_4$.  
    Combining this with the fact that $J \cup w_1w_4$ has a $w_1w_3$-preserving forbidden minor and Lemma \ref{cor:minimal_k-clique-sum_component_containing_fbm_and_f} shows that some atom of $J \cup w_1w_4$, and therefore some atom $K$ of $G' \cup \{w_1w_3,w_1w_4\}$, contains $w_1w_4$ and has a $w_1w_3$-preserving forbidden minor $[K]_1$.  
    If $w_1w_4$ is preserved in $[K]_1$, then the claim is true.  
    Else, we define a vertex $x$ as follows.  
    If $K$ contains $w_3$, then set $x = w_3$.  
    Otherwise, some CMS $E$ of $G' \cup \{w_1w_3,w_1w_4\}$ contained in $K$ is such that all paths between $w_3$ and any vertex in $K \setminus E$ contains some vertex in $E$.  
    Since $[K]_1$ preserves $w_1w_3$ but not $w_1w_4$, $E \setminus \{w_1,w_4\}$ contains some vertex $x$ such that $[x]_1 = [w_3]_1$.  
    
    Next, for either case of $x$, there clearly exists a $w_1x$-preserving minor $[K]_2$ from which $[K]_1$ can be obtained by contracting $[w_1w_4]_2$.  
    Hence, $[K]_2$ is one of the graphs in Figures \ref{fig:k5_k222_e_contracted}, \ref{fig:k5_k222_e_contracted_degree_1}, or \ref{fig:k5_not_wing} with $[w_1w_4]_2$ as the green edge.  
    % some isomorphism between $[K]_2$ and one of the graphs in Figures \ref{fig:k5_k222_e_contracted}, \ref{fig:k5_k222_e_contracted_degree_1}, or \ref{fig:k5_not_wing} maps $[w_1w_4]_2$ to the green edge.  
    Since $K$ is an atom, it is $2$-connected, and so Lemma \ref{lem:G_not_k5_k222_u4v0} allows us to assume that $[K]_2$ is one of the graphs in Figures \ref{fig:k5_k222_e_contracted} or \ref{fig:k5_not_wing}.  
    % the image of this isomorphism is one of the graphs in Figures \ref{fig:k5_k222_e_contracted} or \ref{fig:k5_not_wing}.  
    Furthermore, since $[K]_1$ preserves $w_1w_3$ but not $w_1w_4$, $[w_3]_2$ is neither $[w_1]_2$ nor $[w_4]_2$.  
    Combining this with the fact that $w_3$ neighbors both $w_1$ and $w_4$ shows that $[K]_2$ cannot be any graph in Figure \ref{fig:k5_k222_e_contracted}.  
    % the image of this isomorphism cannot be any graph in Figure \ref{fig:k5_k222_e_contracted}.  
    Therefore, Lemma \ref{lem:k5_k222_not_wing_fm_e_not_contracted} shows that $[K]_2$ has a $[w_1w_4]_2$-preserving forbidden minor, and hence $K$ has a $w_1w_4$-preserving forbidden minor.  
    
    Finally, we show that there exists some atom of $[G \cup f]$ contains $[f]$ and has an $[f]$-preserving forbidden minor.  
    Note that $K \setminus w_1w_3$ is some induced subgraph $N$ of $[G']$ with $w_1w_4$ as $[w_1w_4]$.  
    % xxxsome isomorphism between $K \setminus w_1w_3$ and some induced subgraph $N$ of $[G']$ maps $w_1w_4$ to $[w_1w_4]$.  
    If $K$ does not contain $w_3$, then $N$ is an atom that contains $[f]$ and has an $[f]$-preserving forbidden minor, since $[f] = [w_1w_4]$.  
    Hence, the atom of $[G \cup f]$ that contains $N$ has the desired properties.  
    Otherwise, if $K$ contains $w_3$, then let $G_1 = [H^{-1}]$, $G_2 = N$, $S = G_1 \cap G_2$, and $e = [w_1w_3]$.  
    Observe that $G_1 \cup G_2$ contains $[f]$ and has an $[f]$-preserving forbidden minor.  
    Lastly, since $G_1$ and $G_2 \cup e$ are atoms, Lemma \ref{lem:gluing_min_k-clique-sum_graphs} shows that $G_1 \cup G_2$ is an atom.  
    Therefore, the atom of $[G \cup f]$ that contains $G_1 \cup G_2$ has the desired properties.  
\end{proof}

\begin{lemma}[$H^{-1}$ is not graph in Figures \ref{fig:deg1_vin_uonly_2} or \ref{fig:deg1_vin_uplus_5}]
    \label{lem:C'_degree_1_reducible}
    Neither of the graphs in Figures \ref{fig:deg1_vin_uonly_2} or \ref{fig:deg1_vin_uplus_5} such that $uw_1$ is an edge is the expanded $f$-component of any top-level expanded $f$-separating CMS of a minimal graph-nonedge pair.  
\end{lemma}

\begin{proof}
    Assume to the contrary that one of the graphs in Figures \ref{fig:deg1_vin_uonly_2} or \ref{fig:deg1_vin_uplus_5} such that $uw_1$ is an edge is the expanded $f$-component $H^{-1}$ of some top-level expanded $f$-separating CMS $C^{-1}$ of a minimal graph-nonedge pair $(G,f)$.  
    We will show that either $uw_1$ or $w_2w_4$ is reducing with respect to $(G,f)$, which contradicts the fact that $(G,f)$ is minimal.  
    Observe that $f$ is retained in the minor $[G \cup f]$ obtained by contracting $uw_1$, and so $uw_1$ is not of Type (3).  
    We attempt to show that $uw_1$ is not of Type (1) or (2) by identifying some atom of $[G \cup f]$ contains $[f]$ and has an $[f]$-preserving forbidden minor.  
    Let $G' = G \setminus (H^{-1} \setminus C^{-1})$.  
    It suffices to show that some atom of $G' \cup \{w_1w_3,w_1w_4\}$ contains $w_1w_4$ and has a $w_1w_4$-preserving forbidden minor.  
    By Lemma \ref{lem:G'_f'_no_reducible_edge}, there exists a nonedge $f'$ of $C^{-1}$ such that some atom $J$ of $G' \cup f'$ contains $f'$ and has an $f'$-preserving forbidden minor and the pair $(G',f')$ has no reducing edge.  
    If $f = w_1w_3$, then the argument is identical to the proof of Lemma \ref{lem:C'_degree_1_reducible_0}, so we omit the details.  
    
    Next, consider the case where $f' = w_1w_4$.  
    By Lemma \ref{lem:J_contains_C'_degree_1}, $J$ contains $w_3$.  
    Hence, since $J$ is an atom, each CMS and atom of $J \cup w_1w_3$ contains $w_1w_3$.  
    Combining this with the fact that $J \cup w_1w_3$ has a $w_1w_4$-preserving forbidden minor and Lemma \ref{cor:minimal_k-clique-sum_component_containing_fbm_and_f} shows that some atom of $J \cup w_1w_3$, and therefore some atom $K$ of $G' \cup \{w_1w_3,w_1w_4\}$, contains $w_1w_3$ and has a $w_1w_4$-preserving forbidden minor $[K]_1$.  
    If $K$ contains $w_4$, then the claim is true.  
    Otherwise, we show that $w_2w_4$ is reducing with respect to $(G,f)$.  
    First, we argue that $K$ has a $w_1w_3$-preserving forbidden minor.  
    If $[K]_1$ preserves $w_1w_3$, then we are done.  
    Else, since $K$ does not contain $w_4$, some CMS $E$ of $G' \cup \{w_1w_3,w_1w_4\}$ contained in $K$ is such that all paths between $w_4$ and any vertex in $K \setminus E$ contains some vertex in $E$.  
    Also, since $[K]_1$ preserves $w_1w_4$ but not $w_1w_3$, $E \setminus \{w_1,w_3\}$ contains some vertex $x$ such that $[x]_1 = [w_4]_1$.  
    Therefore, there clearly exists a $w_1x$-preserving forbidden minor $[K]_2$ from which $[K]_1$ can be obtained by contracting $[w_1w_3]_2$.  
    The remainder of the proof that $K$ has a $w_1w_3$-preserving forbidden minor is identical to part of the proof of Lemma \ref{lem:C'_degree_1_reducible_0}, so we omit the details.  

    Finally, observe that $f$ is retained in the minor $[G \cup f]_4$ obtained by contracting $w_2w_4$, and so $w_2w_4$ is not of Type (3).  
    We complete the proof by showing that $w_2w_4$ is not of Type (1) or (2).  
    Since $K$ does not contain $w_4$, observe that $K$ is an atom of $G \cup w_1w_3$ and is weakly retained in the minor of $G \cup w_1w_3$ obtained by contracting $w_2w_4$.  
    Hence, $K \setminus w_1w_3$ is some induced subgraph $N$ of $[G']_4$ with each vertex $w_i$ other than $w_4$ as $[w_i]_4$.  
    % xxxsome isomorphism between $K \setminus w_1w_3$ and some induced subgraph $N$ of $[G']_4$ maps each vertex $w_i$ other than $w_4$ to $[w_i]_4$.  
    Let $G_1$ be the subgraph of $[G \cup f]_4$ induced by $\{[u]_4,[w_1]_4,[w_2]_4,[w_3]_4\}$, $G_2 = N$, $S = G_1 \cap G_2$, and $e = [w_1w_3]_4$.  
    By construction, $G_1 \cup G_2$ contains $[f]_4$.  
    Also, it is easy to see that $G_1 \cup G_2$ has an $[f]_4$-preserving forbidden minor.  
    We will apply Lemma \ref{lem:gluing_top_part_to_atom} to show that $G_1 \cup G_2$ is an atom.  
    Note that $e$ is a nonedge of $G_1 \cup G_2$ and both $G_1 \cup e$ and $G_2 \cup e$ are atoms.  
    Furthermore, $G_1 \setminus A$ is connected for any clique subgraph $A$ of $S$.  
    Lastly, the only clique subgraph of $S$ that could separate $[w_1]_4$ and $[w_3]_4$ in $G_2$ is $[w_2]_4$.  
    Since $G_2 \cup e$ has a forbidden minor and $S$ contains at most three vertices, there exists some vertex in $G_2 \setminus S$.  
    Additionally, since $G_2 \cup e$ is an atom, there exist paths that connect this vertex to $[w_1]_4$ and $[w_3]_4$ that each contain exactly one vertex in $S$.  
    Therefore, $[w_2]_4$ does not separate $[w_1]_4$ and $[w_3]_4$ in $G_2$, and so Lemma \ref{lem:gluing_top_part_to_atom} shows that $K$ is an atom.  
    Thus, $G_1 \cup G_2$ is contained in some atom of $[G \cup f]_4$ that contains $[f]_4$ and has an $[f]_4$-preserving forbidden minor, and the proof is complete.  
\end{proof}

\begin{lemma}[$H^{-1}$ is not graph in Figure \ref{fig:deg1_vin_uplus_5}]
    \label{lem:C'_degree_1_reducible_2}
    The graph in Figure \ref{fig:deg1_vin_uplus_5} such that $uw_1$ is a nonedge is not the expanded $f$-component of any top-level expanded $f$-separating CMS of a minimal graph-nonedge pair.  
\end{lemma}

\begin{proof}
    Assume to the contrary that the graph in Figure \ref{fig:deg1_vin_uplus_5} such that $uw_1$ is a nonedge is the expanded $f$-component $H^{-1}$ of some top-level expanded $f$-separating CMS $C^{-1}$ of some minimal graph-nonedge pair $(G,f)$.  
    Let $x$ be any vertex of $H^{-1} \setminus (C^{-1} \cup u)$.  
    We show that either $xw_1$ or $w_2w_3$ is reducing with respect to $(G,f)$, which contradicts the fact that $(G,f)$ is minimal.  
    Observe that $f$ is retained in the minor $[G \cup f]$ obtained by contracting $xw_1$, and so $xw_1$ is not of Type (3).  
    We attempt to show that $xw_1$ is not of Type (1) or (2) by identifying some atom of $[G \cup f]$ contains $[f]$ and has an $[f]$-preserving forbidden minor.  
    Let $G' = G \setminus (H^{-1} \setminus C^{-1})$.  
    By Lemma \ref{lem:G'_f'_no_reducible_edge}, there exists a nonedge $f'$ of $C^{-1}$ such that some atom $J$ of $G' \cup f'$ contains $f'$ and has an $f'$-preserving forbidden minor and the pair $(G',f')$ has no reducing edge.  
    There are two cases.  

    \smallskip
    \noindent\textbf{Case 1:} $f' = w_1w_3$.  
    \smallskip

    First, we show that some atom of $G' \cup \{w_1w_3,w_1w_4\}$ contains $w_1w_3$ and has a $w_1w_3$-preserving forbidden minor.  
    By Lemma \ref{lem:J_contains_C'_degree_1}, $J$ contains $w_4$.  
    Hence, since $J$ is an atom, each CMS and atom of $J \cup w_1w_4$ contains $w_1w_4$.  
    Combining this with the fact that $J \cup w_1w_4$ has a $w_1w_3$-preserving forbidden minor and Lemma \ref{cor:minimal_k-clique-sum_component_containing_fbm_and_f} shows that some atom of $J \cup w_1w_4$, and therefore some atom $K$ of $G' \cup \{w_1w_3,w_1w_4\}$, contains $w_1w_4$ and has a $w_1w_3$-preserving forbidden minor $[K]_1$.  
    If $K$ contains $w_3$, then we are done.  
    Otherwise, we show that $w_2w_3$ is reducing with respect to $(G,f)$.  
    
    Observe that $f$ is retained in the minor $[G \cup f]_1$ obtained by contracting $w_2w_3$, and so $w_2w_3$ is not of Type (3).  
    To see that $w_2w_3$ is not of Type (1) or (2), note that an argument similar to to part of the proof of Lemma \ref{lem:C'_degree_1_reducible} shows that $K$ has a $w_1w_4$-preserving forbidden minor.  
    Also, since $K$ does not contain $w_3$, $K$ is an atom of $G \cup w_1w_4$ and is weakly retained in the minor of $G \cup w_1w_4$ obtained by contracting $w_2w_3$.  
    Hence, $K \setminus w_1w_4$ is some induced subgraph $N$ of $[G']_1$ with each vertex $w_i$ other than $w_3$ as $[w_i]_1$.  
    % xxxsome isomorphism between $K \setminus w_1w_4$ and some induced subgraph $N$ of $[G']_1$ maps each vertex $w_i$ other than $w_3$ to $[w_i]_1$.  
    Let $G_1 = [H^{-1}_A$, $G_2 = N$, $S = G_1 \cap G_2$, and $e = [w_1w_4]_1$.  
    By construction, $G_1 \cup G_2$ contains $[f]_1$.  
    Furthermore, it is easy to see that $G_1 \cup G_2$ has an $[f]_1$-preserving forbidden minor.  
    Lastly, since $e$ is a nonedge of $G_1 \cup G_2$ and both $G_1$ and $G_2 \cup e$ are atoms, Lemma \ref{lem:gluing_min_k-clique-sum_graphs} shows that $G_1 \cup G_2$ is an atom.  
    Therefore, $G_1 \cup G_2$ is contained in some atom of $[G \cup f]_1$ that contains $[f]_1$ and has an $[f]_1$-preserving forbidden minor, and the proof is complete.  
    
    Finally, since $K$ contains $w_3$, $K \setminus w_1w_3$ is some induced subgraph $N$ of $[G \cup f]$ with each vertex $w_i$ as $[w_i]$.  
    % xxxsome isomorphism between $K \setminus w_1w_3$ some induced subgraph $N$ of $[G \cup f]$ maps each vertex $w_i$ to $[w_i]$.  
    Let $G_1 = [C^{-1}] \cup [u]$, $G_2 = N$, $S = G_1 \cap G_2$, and $e = [w_1w_3]$.  
    Observe that $G_1 \cup G_2$ contains $[f]$ and has an $[f]$-preserving forbidden minor.  
    Lastly, since $G_1$ and $G_2 \cup e$ are atoms, Lemma \ref{lem:gluing_min_k-clique-sum_graphs} shows that $G_1 \cup G_2$ is an atom.  
    Thus, $G_1 \cup G_2$ is contained in some atom of $[G \cup f]$ that contains $[f]$ and has an $[f]$-preserving forbidden minor, as desired.  

    \smallskip
    \noindent\textbf{Case 2:} $f' = w_1w_4$.  
    \smallskip

    By Lemma \ref{lem:J_contains_C'_degree_1}, $J$ contains $w_2$ and $w_3$.  
    Hence, $J$ is some induced subgraph $G_2$ of $[G' \cup w_1w_4]$ with each vertex $w_i$ as $[w_i]$.  
    % xxxsome isomorphism between $J$ and some induced subgraph $G_2$ of $[G' \cup w_1w_4]$ maps each vertex $w_i$ to $[w_i]$.  
    Let $G_1 = [C^{-1}] \cup [u]$, $S = G_1 \cap G_2$, and $e = [w_1w_3]$.  
    Observe that $G_1 \cup G_2$ contains $[f]$ and has an $[f]$-preserving forbidden minor.  
    Also, since $e$ is a nonedge and both $G_1$ and $G_2$ are atoms, Lemma \ref{lem:gluing_min_k-clique-sum_graphs} shows that $G_1 \cup G_2$ is an atom.  
    Therefore, $G_1 \cup G_2$ is contained in an atom of $[G \cup f]$ that contains $[f]$ and has an $[f]$-preserving forbidden minor.  
    Thus, $xw_1$ is reducing, as desired.  

    \smallskip
    The above cases are exhaustive, and so the proof is complete.  
\end{proof}

\begin{proof}[Proof of Lemma \ref{lem:C'_degree_1_reducible_3}]
    The lemma follows immediately from Lemmas \ref{lem:C'_degree_1_reducible_0}-\ref{lem:C'_degree_1_reducible_2}.  
\end{proof}

\subsection{Proof of Lemma \ref{lem:degree_3_decorations} for Proposition \ref{prop:IS}}
\label{sec:prop-10}

\begin{proof}[Proof of Lemma \ref{lem:degree_3_decorations}]
    We can clearly choose $\ell'_i$ to satisfy (i) and (ii), and we will show that some such choice satisfies (iii).  
    Let $\{v_i\}$ be the vertex set of $G$ with $v_44$ not in $H$ and let $\ell_1(v_1v_2)$ and $\ell_2(v_1v_2)$ be distinct positive real numbers.  
    Consider the $2$-realizations $p$ and $q$ of $(H,\ell_1)$ and $(H,\ell_2)$, respectively, with $p_i = p(v_i)$, $q_i = q(v_i)$, $(0,0) = p_1 = q_1$, $p_2 = (\ell_1(v_1v_2),0)$, $q_2 = (\ell_2(v_1v_2),0)$, $p_3 = (x,y)$ with $y \geq 0$, and $q_3 = (w,z)$ with $z \geq 0$.  
    % Note that there exists a $2$-realization $p$ of $(H,\ell_1)$ that maps $v_1$ to $a_1 = (0,0)$, $v_2$ to $b_1 = (\ell_1(v_1v_2),0)$, and $v_3$ to some point $c_1 = (x,y)$ with $y \geq 0$.  
    % Similarly, there exists a $2$-realization $q$ of $(H,\ell_2)$ that maps $v_1$ to $a_1 = (0,0)$, $v_2$ to $b_2 = (\ell_2(v_1v_2),0)$, and $v_3$ to some point $c_2 = (w,z)$ with $z \geq 0$.  
    Since $\ell_1(v_1v_2)$ and $\ell_2(v_1v_2)$ are distinct, the line-segment between $p_2$ and $q_2$ has a perpendicular bisector $L_1$.  
    If $p_3=q_3$, then taking any point $x$ on $L_1$ and setting $\ell'_i(v_4v_j) = \|x-p_j\|^2$ completes the proof.  
    Otherwise, let $L_2$ be the perpendicular bisector of the line-segment between $p_3$ and $q_3$.  
    If some point $x$ lies on both $L_1$ and $L_2$, then setting $\ell'_i(v_4v_j) = \|x-p_j\|^2$ completes the proof; see Figure \ref{fig:triangle_1}.  
    Else, $L_1$ and $L_2$ are distinct and parallel; see Figure \ref{fig:triangle_2}.  
    In this case, let $p'$ be a realization obtained by rotating $p$ about the origin by an angle of $\theta$, and let $p'_i = p'(v_i)$.  
    It suffices to show that we can choose $p'$ such that the perpendicular bisector of the line-segment between $p'_2$ and $q_2$ intersects the perpendicular bisector of the line-segment between $p'_3$ and $q_3$.  

    We complete the proof by showing how to choose $\theta$ so that $p'$ has the desired properties.  
    Let $p'_i = p(v_i)$.  
    It suffices to choose $\theta$ so that the line containing the points $p'_2$ and $q_2$ and the line containing the points $p'_3$ and $q_3$ have distinct slope magnitudes.  
    Note that $p'_2 = (a\cos \theta, -a\sin \theta)$, where $a = \ell_1(v_1v_2)$, and $p'_3 = (x\cos \theta + y\sin \theta,-x\sin \theta + y\cos \theta)$.  
    Let $a' = \ell_2(v_1v_2)$.  
    Then, $\theta$ must not be a solution to
    \begin{align*}
        \left|\frac{-a\sin \theta}{a\cos \theta - a'}\right| =
        \left|\frac{-x\sin \theta + y\cos \theta - z}{x\cos \theta + y\sin \theta - w}\right|.
    \end{align*}
    This is equivalent to $\theta$ not being in the union of the solution sets to any of the following:
    \begin{align*}
        &\left|a\cos \theta - a'\right| = 0,\\
        &\left|x\cos \theta + y\sin \theta - w\right| = 0,\\
        &\left|(az + a'y)\cos \theta + (aw - a'x)\sin \theta - (ay + a'z)\right| = 0,\\
        &\left|(az + a'y)\cos \theta - (aw + a'x)\sin \theta + 2ax\cos \theta \sin \theta - ay\cos 2\theta - a'z\right| = 0.
    \end{align*}
    Each of these equations is periodic and has finitely many solutions on the interval $[0,2\pi]$, and thus we can choose $\theta$ as needed.  
\end{proof}

\section{Results about atoms}
\label{sec:atoms}

Here we collect various results about atoms.  

\begin{lemma}[Induced minors of $G$ that are atoms extend induced minors of atoms of $G$]
\label{lem:3-flat_fm_containing_e_pass_through_minimal_3-clique-sum_components}
    If an induced minor $[G]$ of a graph $G$ is an atom, then $[G]$ is an extension of some induced minor of some atom of $G$.  
\end{lemma}

\begin{proof}
    Assume that $[G]$ is an atom.  
    We proceed by induction on the number $n$ of CMSs of $G$.  
    When $n=0$, the lemma is immediate.  
    Assume the lemma is true when $n \leq k$, for any integer $k \geq 0$, and we will prove it when $n=k+1$.  
    Let $C$ be any CMS of $G$.  
    Since $[G]$ is an induced minor and an atom, at most one $C$-component of $G$ contains some vertex that is not contained in $[C]^{-1}$.  
    Let $H$ be any $C$-component of $G$ that contains $C$, which exists since $C$ is a minimal separator.  
    If no $C$-component of $G$ contains some vertex that is not contained in $[C]^{-1}$, then each vertex of $G \setminus C$ is contained in $[C]^{-1}$, and hence $[H]^{-1}$.  
    Otherwise, we can wlog choose $C$ and $H$ such that $H$ additionally contains some vertex that is not contained in $[C]^{-1}$.  
    In this case, each vertex of $G \setminus H$ is contained in $[C]^{-1}$, and hence $[H]^{-1}$.  
    Therefore, in either case, $[G]$ is an extension of some induced minor $[H]_1$ that is an atom.  
    
    Next, observe that $H$ has strictly fewer CMSs than $G$.  
    Therefore, by the inductive hypothesis, $[H]_1$ is an extension of some induced minor $M$ of some atom of $H$.  
    Finally, the lemma follows from the observation that $[G]$ is an extension of $M$. 
\end{proof}

\begin{lemma}[Conditions under which nonedge $f$ is contained in each atom of $G \cup f$]
    \label{lem:I'_f'_cms}
    For any nonedge $f$ of a graph $G$ whose atom graph is a path, each CMS, and consequently each atom, of $G \cup f$ contains $f$ if either of the following statements are true:
    \begin{enumerate}
        \item $G$ is an atom.
        \item The atom graph of $G$ contains two distinct leaves $y_1$ and $y_2$, whose not necessarily distinct neighbors are $x_1$ and $x_2$, respectively, and the endpoints of $f$ are contained in the graphs $G(y_1) \setminus G(x_1)$ and $G(y_2) \setminus G(x_2)$.  
    \end{enumerate}
\end{lemma}

\begin{proof}
    If Statement (1) is true, then the lemma is easy to verify.  
    Consider the case where Statement (2) is true.  
    If $G \cup f$ is an atom, then the lemma is proved.  
    Otherwise, assume to the contrary that some CMS $E$ of $G \cup f$ does not contain $f$, then $E$ is clearly a CMS of $G$.  
    Furthermore, by assumption, the endpoints of $f$ are contained in distinct $E$-components of $G$.  
    Since the atom of graph of $G$ is a path, there are exactly two $E$-components of $G$.  
    Hence, these facts show that there exists some path in $G \cup f$ between any two vertices in $(G \cup f) \setminus E$.  
    This contradicts the fact that $E$ is a CMS of $G \cup f$, and so $E$ must contain $f$.  
    Thus, each CMS of $G$ contains $f$, and so the proof is complete.  
\end{proof}

\begin{lemma}[How a CMS partitions atoms and CMSs]
    \label{lem:clq_sep_props}
    For any CMS $C$ of a graph $G$, the set of all atoms of $G$ is equal to the disjoint union of the sets of all atoms of all $C$-components of $G$.  
    Also, if no proper subgraph of any CMS of $G$ is a CMS of $G$, then set of all CMSs of $G$ is equal to the disjoint union of the set $\{C\}$ and the sets of all CMSs of all $C$-components of $G$.  
\end{lemma}

\begin{proof}
    The first claim follows from the observations that any atom of $G$ is contained in a unique $C$-component of $G$ and the set of all atoms of $G$ is equal to the union of the sets of all atoms of all $C$-components of $G$.  
    Next, assume that no proper subgraph of any CMS of $G$ is a CMS of $G$.  
    First, note that $C$ is not a CMS of any $C$-component of $G$.  
    Second, any CMS $E$ of a $C$-component of $G$ is clearly a CMS of $G$.  
    Also, since $E$ is not a subgraph of $C$, it is contained in a unique $C$-component of $G$.  
    Third, by a similar argument, any CMS $E$ of $G$ other than $C$ is contained in a unique $C$-component $H$ of $G$.  
    Hence, it is easy to see that $E$ is a CMS of $H$.  
    Combining these facts completes the proof.  
\end{proof}

\begin{lemma}[How a CMS partitions the vertices of a separator]
    \label{lem:I'_sep_components_contain_C'-E}
    Let $S$ be a separator of a graph $G$ and $H$ be the union of $S$ and any subset connected components of $G \setminus S$.  
    If $G$ is an atom, then, for any CMS $E$ of $H$, each connected component of $H \setminus E$ contains at least one vertex of $S \setminus E$.  
\end{lemma}

\begin{proof}
    If some CMS $E$ of $H$ is such that some connected component of $H \setminus E$ does not contain any vertex of $S \setminus E$, then clearly $E$ is a clique separator of $G$, and thus $G$ is not as atom.  
\end{proof}

\begin{lemma}[Condition under which an atom graph is a tree]
    \label{lem:dual_min_clique-sum_graph_tree}
    If $G$ is a graph such that no proper subgraph of any of its CMSs is a CMS of $G$, then the atom graph $\mathcal{G}(G)$ is a tree whose leaf nodes are contained in $Y(G)$ and, for any vertex $x \in X(G)$ and any connected component $H$ of $\mathcal{G}(G) \setminus x$, $\bigcup_{v \in V(H)} G(v)$ is a $G(x)$-component of $G$.
\end{lemma}

\begin{proof}
    We proceed by induction on the number $n$ of CMSs of $G$.  
    When $n=0$, the lemma is immediate, and so the base case is true.  
    For any integers $k \geq 0$ and $0 \leq j \leq k$, assume the lemma is true when $n=j$, and we will prove it when $n=k+1$.  
    Since $k + 1 > 0$, $G$ has a CMS $C$.  
    Consider the set $\{H_i\}$ of $C$-components of $G$.  
    Lemma \ref{lem:clq_sep_props} shows that each graph $H_i$ has strictly fewer CMSs than $G$, and any CMS $E$ of $H_i$ is a CMS of $G$.  
    Hence, no proper subgraph of $E$ is a CMS of $H_i$.  
    Therefore, the inductive hypothesis applies to each atom graph $\mathcal{G}(H_i)$.  

    Next, we will construct $\mathcal{G}(G)$ from the trees $\mathcal{G}(H_i)$ and show that it has the desired properties.  
    First, we show that there exists a unique atom $J_i$ of each $C$-component $H_i$ that contains $C$.  
    Then, we show that $\mathcal{G}(G)$ can be obtained from the graph $\bigcup_i \mathcal{G}(H_i)$ by adding a vertex $x$, corresponding to $C$, and, for each $J_i$, an edge between $x$ and the vertex of $\mathcal{G}(H_i)$ that corresponds to $J_i$.  
    This clearly proves the lemma.  
    
    For the first claim, since $C$ is a CMS of $G$ but no proper subgraph of $C$ is a CMS of $G$, $C$ is contained in some atom of each $C$-component $H_i$.  
    Note that the intersection of any two atoms of $G$ is either empty or a CMS of $G$.  
    Hence, if there exist distinct atoms of $H_i$ that both contain $C$, then $C$ must be a proper subgraph of some of $G$, which is a contradiction.  
    Therefore, $C$ is contained in a unique atom $J_i$ of $H_i$.  

    For the second claim, let $G'$ be the graph obtained as described above.  
    By Lemma \ref{lem:clq_sep_props}, the vertex sets of $G'$ and $\mathcal{G}(G)$ are equal.  
    Furthermore, the edge set of $G'$ is contained in the edge set of $\mathcal{G}(G)$, by definition.  
    Next, by assumption, any CMS $E$ of a $C$-component $H_i$ contains some vertex in $H_i \setminus C$.  
    Therefore, only an atom of $G$ that is contained in $H_i$ can contain $C_i$.  
    The facts above show that the edge sets of $G'$ and $\mathcal{G}(G)$ are the same, and so we have $G' = \mathcal{G}(G)$, which completes the proof.  
\end{proof}

\begin{lemma}[Condition under which an atom graph is a path]
    \label{lem:I'_path}
    Let $G$ be a graph that is an atom and let $S$ be one of its subgraphs and separators such that, for any CMS $C$ of $S$, there are exactly two $C$-components of $S$, both of which are atoms.  
    If $H$ is the union of $S$ and any subset of connected components of $G \setminus S$, then the atom graph of $H$ is a path.
\end{lemma}

\begin{proof}
    If $H$ is an atom, then the lemma is proved.  
    Otherwise, for any CMS $E$ of $H$, we can use Lemma \ref{lem:I'_sep_components_contain_C'-E} and our assumption on $S$ to see that there are exactly two $E$-components of $H$.  
    This implies that no proper subgraph of $E$ is a CMS of $H$.  
    Since $E$ was arbitrary, this shows that Lemma \ref{lem:dual_min_clique-sum_graph_tree} applies.  
    Hence, each vertex in $X(H)$ has degree $2$, and it suffices to show that each vertex in $Y(H)$ has degree $2$.  
    If this is not the case, then there exist three leaf nodes $y_1$, $y_2$, and $y_3$ in $\mathcal{G}(H)$ whose neighbors are $x_1$, $x_2$, and $x_3$, respectively.  
    We will show that this implies the existence of a CMS $C$ of $S$ such that some $C$-component of $S$ is not as atom, which contradicts our assumption.  
    
    By the discussion above, $H(y_1)$ is a $H(x_1)$-component of $H$.  
    Hence, by Lemma \ref{lem:I'_sep_components_contain_C'-E}, the connected component $H(y_1) \setminus H(x_1)$ contains a vertex $z_1$ of $S$.  
    Similarly, $H(y_2) \setminus H(x_2)$ and $H(y_3) \setminus H(x_3)$ contain the vertices $z_2$ and $z_3$ of $S$, respectively.  
    Consider any CMS $E$ of $H$ other than $H(x_1)$.  
    If $H(y_1) \setminus H(x_1)$ contains some vertex in $E$, then $H(y_1)$ contains $E$.  
    However, this implies that $y_1$ has degree $2$, which contradicts the fact that $y_1$ is a leaf node.  
    Hence, $H(y_1) \setminus H(x_1)$ does not contain any vertex in $E$.  
    Similarly, $H(y_2) \setminus H(x_2)$ and $H(y_3) \setminus H(x_3)$ do not contain any vertex in $E$.  
    Since the intersection of any two atoms of $H$ is either empty or a CMS of $H$, this implies that $z_1$, $z_2$, $z_3$ are distinct vertices.  
    However, these facts imply that $H(x_1) \cap S$ is a CMS of $S$ such that some $H(x_1) \cap S$-component of $S$ is not as atom, which is a contradiction.  
    Thus, setting $C = H(x_1)$ completes the proof.  
\end{proof}

\begin{lemma}[Gluing atoms to obtain an atom 1]
\label{lem:gluing_min_k-clique-sum_graphs}
    Let $G$ be a graph with with subgraphs $G_1$, $G_2$, and $S$ such that $G = G_1 \cup G_2$, $S = G_1 \cap G_2$ contains some nonedge $e$, and $G \setminus S$ is disconnected.  
    If $G_1$ is an atom and either $G_2$ or $G_2 \cup e$ is an atom, then $G$ is an atom.  
\end{lemma}

\begin{proof}
    Assume that $G_1$ and either $G_2$ or $G_2 \cup e$ is an atom, and consider any clique subgraph $E$ of $G$.  
    Since $S$ contains a nonedge, the graph $S \setminus E$ is non-empty.  
    Also, since $G_1$ is an atom, $G_1 \setminus E$ is connected.  
    Lastly, since either $G_2$ or $G_2 \cup e$ is an atom, it is easy to see that, for any vertex $x$ of $G_2 \setminus E$, there is some path in $G_2 \setminus E$ between $x$ and some vertex in $S \setminus E$.  
    Combining the above facts clearly shows that $G \setminus E$ is connected.  
    Thus, since $E$ was an arbitrary clique subgraph of $G$, the lemma is proved.  
\end{proof}

\begin{lemma}[Gluing atoms to obtain an atom 2]
    \label{lem:gluing_top_part_to_atom}
    Let $G$ be a graph with with subgraphs $G_1$, $G_2$, and $S$ such that $G = G_1 \cup G_2$, $S = G_1 \cap G_2$ contains some nonedge $e$, and $G \setminus S$ is disconnected.  
    Then, $G$ is an atom if the following statements are true:
    \begin{enumerate}
        \item $G_2 \cup e$ is an atom, 
        \item either $G_1$ or $G_1 \cup e$ is an atom, 
        \item $G_1 \setminus A$ is connected for any clique subgraph $A$ of $S$, and 
        \item $A$ does not separate the endpoints of $e$ in $G_2$.  
    \end{enumerate}
\end{lemma}

\begin{proof}
    Assume that Statements (1)-(4) are true.  
    If $G_1$ is an atom, then the lemma follows from Lemma \ref{lem:gluing_min_k-clique-sum_graphs}.  
    Otherwise, we must show that $G \setminus E$ is connected for any clique subgraph $E$ of $G$.  
    First, we show that $G_1 \setminus E$ is connected if $E$ is contained in $G_2$, and $G_2 \setminus E$ is connected otherwise.  
    Note that $A = E \cap S$ is either empty or a clique.  
    If $E$ is contained in $G_2$, then $G_1 \setminus E = G_1 \setminus A$, which is connected by Statement (3).  
    Otherwise, $E$ must be contained in $G_1$ since $S$ is a separator of $G$.  
    Hence, we have $G_2 \setminus E = G_2 \setminus A$.  
    Statement (1) shows that $(G_2 \cup e) \setminus A$ is connected.  
    Therefore, if $G_2 \setminus A$ is disconnected, then $A$ must separate the endpoints of $e$ in $G_2$.  
    Combining this with Statement (4) shows that $G_2 \setminus E$ is connected.  

    Next, assume that $G_1 \setminus E$ is connected.  
    Since $S$ contains the nonedge $e$, $S \setminus E$ is non-empty.  
    Using Statement (1), it is easy to see that, for any vertex $x \in G_2 \setminus E$, there exists a path in $G_2 \setminus E$ between $x$ and some vertex in $S \setminus E$.  
    Hence, since $G_1 \setminus E$ contains $S \setminus E$, this implies that $G \setminus E$ is connected.  
    A similar argument using Statement (2) instead of (1) shows that $G \setminus E$ is connected if $G_2 \setminus E$ is connected.  
    This completes the proof.  
\end{proof}

\begin{lemma}[Paths between vertices of atoms]
    \label{lem:path_in_graph_path_in_atom}
    Given an atom $J$ of a graph $G$, if there exists a path in $G$ between two vertices in $J$, then some subsequence of the vertices in this path is a path in $J$ between the same two vertices.  
\end{lemma}

\begin{proof}
    Let $x$ and $y$ be any two vertices of $J$ such that there exists a path $P$ in $G$ between them.  
    If $P$ is contained in $J$, then we are done.  
    Otherwise, moving along $P$ from $x$ to $y$, let $b$ be the first vertex not contained in $J$ and let $a$ be the vertex that immediately precedes $b$.  
    Then, there must exist a CMS of $G$ containing $a$ and some vertex $c$ after $b$ in $P$.  
    Hence, $ac$ is an edge of $G$.  
    Note that replacing the path in $P$ between $a$ and $c$ with the edge $ac$ yields a shorter path $P'$ between $x$ and $y$, which is a subsequence of $P$.  
    Repeating this argument for $P'$, and so on, eventually yields a path with the desired properties.  
\end{proof}

\end{document}